\documentclass[a4paper,11pt,twoside]{book}

\usepackage[applemac]{inputenc}

\usepackage[T1]{fontenc}
\usepackage[english]{babel} 
\usepackage{verbatim}
\usepackage[pdftex]{graphicx}
\usepackage[pdftex]{hyperref}
\usepackage[dvipsnames]{xcolor}
\usepackage{soul}

\hypersetup{%
   colorlinks,%
   linkcolor=Brown,%
   citecolor=violet,%
   urlcolor=blue,%
}

\usepackage{enumerate}

\usepackage{multirow} 
\usepackage{amsfonts,amsmath,amssymb,amsthm}
\usepackage[english]{minitoc} 
\usepackage{fancyhdr}
\usepackage{lscape} 
\usepackage{array}
\usepackage{natbib}

\usepackage{enumitem}
\usepackage{pgf,tikz}
\usetikzlibrary{decorations.markings}
\usetikzlibrary{arrows}
\usetikzlibrary{positioning}
\usetikzlibrary{snakes}
\tikzset{->-/.style={decoration={
  markings,
  mark=at position .97 with {\arrow{>}}},postaction={decorate}}}

\usepackage[lined, vlined, ruled, linesnumbered, algochapter]{algorithm2e}
\usepackage{color}

\usepackage{xspace}

\usepackage{subfig} 

\newdimen\decalage
\oddsidemargin=\paperwidth
\advance\oddsidemargin by -\textwidth
\divide\oddsidemargin by 2
\advance\oddsidemargin by -1 in
\evensidemargin=\oddsidemargin
\advance\oddsidemargin by \decalage 
\advance\evensidemargin by -\decalage

\newtheorem{lemma}{Lemma}[chapter]
\newtheorem{proposition}[lemma]{Proposition}
\newtheorem{theorem}[lemma]{Theorem}

\newtheorem{definition}[lemma]{Definition}
\newtheorem{problem}[lemma]{Problem}

\newcommand{\cyclebubble}{\ensuremath{\mathtt{cycle}}}
\newcommand{\unblock}{\ensuremath{\mathtt{unblock}}}

\newcommand{\listpaths}{\ensuremath{\mathtt{list\_paths}}}
\newcommand{\listpathsiter}{\ensuremath{\mathtt{list\_paths\_iterative}}}
\newcommand{\lcp}{\ensuremath{\mathtt{lcp}}}

\newcommand{\listbubbles}{\ensuremath{\mathtt{enumerate\_bubbles}}}

\newcommand{\ks}{{\sf \sc KisSplice}\xspace}
\newcommand{\tri}{{\sf \sc Trinity}\xspace}

\newcommand{\setofcycles}{\ensuremath{\mathcal{C}}}
\newcommand{\setofpaths}{\ensuremath{\mathcal{P}}}
\newcommand{\beadstring}{\ensuremath{B_{u,t}}}
\newcommand{\sbeadstring}{\ensuremath{B_{s,t}}}
\newcommand{\vbeadstring}{\ensuremath{B_{v,t}}}
\newcommand{\head}{\ensuremath{H_u}}

\newcommand{\Chead}{\ensuremath{H_X}}

\newcommand{\chooseedge}{\ensuremath{\mathtt{choose}}}

\newcommand{\oracleleft}{\ensuremath{\mathtt{left\_update}}}
\newcommand{\oracleright}{\ensuremath{\mathtt{right\_update}}}
\newcommand{\undooracle}{\ensuremath{\mathtt{restore}}}

\newcommand{\bcc}{\textsc{bcc}}
\newcommand{\liststpaths}{\ensuremath{\mathtt{list\_paths}_{s,t}}}

\newcommand{\return}{\ensuremath{\mathtt{return}}}
\newcommand{\routput}{\ensuremath{\mathtt{output}}}

\newif\ifcomment\commentfalse
\def\commentON{\commenttrue}

\long\outer\def\bc#1\ec{{\ifcomment \sloppy  $[${\bf Begin new}]
{{#1}} \textbf{[end new]} \fi }}

\long\outer\def\br#1\er{{\ifcomment \sloppy  $[${\bf Vicente suggest remove}]
{{#1}} \textbf{[end]} \fi }}

\long\outer\def\bo#1\eo{{\ifcomment \sloppy  $[${\bf instead of}]
{\textit{#1}} \textbf{[end]}  \fi }}

\long\outer\def\BC#1\EC{{\ifcomment \sloppy \par \#  \dotfill
{\textsc{#1}} \dotfill \# \par \fi }}

\commentON   

\begin{document}

\renewcommand{\bibname}{Bibliography}


\pagenumbering{roman}


\thispagestyle{empty}

\begingroup
\makeatletter

\begin{titlepage}
\begin{large}
\hskip -20ex
\begin{center}
  \begin{tabular}{lp{7.2cm}l}
    N$^o$ d'ordre: 43-2014 && Année 2014\\
    
  \end{tabular}
  \begin{tabular}{llll}
    &&&\\
    &\multicolumn{2}{c}{\textsc{Th\`ese}}&\\ [1.5ex]
    &\multicolumn{2}{c}{Présentée}&\\ [1.5ex]
    &\multicolumn{2}{c}{devant \textsc{l'Université Claude Bernard - Lyon 1}}&\\ [1.5ex]
    &\multicolumn{2}{c}{pour l'obtention}&\\[1.5ex]
    &\multicolumn{2}{c}{du \textsc{Diplôme de Doctorat}}&\\ [1.5ex]
    &\multicolumn{2}{c}{\textit{(arrêté du 7 août 2006)}}&\\ [1.5ex]
    &\multicolumn{2}{c}{et soutenue publiquement le}&\\
    &\multicolumn{2}{c}{6 Mars 2014}&\\ [1.5ex]
    &\multicolumn{2}{c}{par}&\\
    &&&\\
    &\multicolumn{2}{c}{\large{Gustavo Akio \textsc{Tominaga Sacomoto}}}&\\
    &&&\\
    \cline{2-3}\\
    &\multicolumn{2}{c}{\bf \Large{Efficient Algorithms for de novo Assembly of Alternative}}&\\
    &\multicolumn{2}{c}{\bf \Large{\strut  Splicing Events from RNA-seq Data}}&\\
    &&&\\
    \cline{2-3}\\
    &&&\\
    &\multicolumn{2}{c}{Directeur de thèse: Marie-France \textsc{Sagot}} &\\
      & \multicolumn{2}{c}{Co-Directeur: Pierluigi \textsc{Crescenzi}} &\\
      & \multicolumn{2}{c}{Co-Encadrant: Vincent \textsc{Lacroix}} \\
    &&&\\
  \end{tabular}
  \begin{tabular}{lll}
    \textsc{Jury}: 
    & Céline \textsc{Brochier-Armanet},& Examinateur\\
    & Michael \textsc{Brudno},& Rapporteur\\
    & Rodéric \textsc{Guigo},& Rapporteur\\
    & Thierry \textsc{Lecroq},& Examinateur\\
    & Peter \textsc{Widmayer}, & Rapporteur
  \end{tabular}
\end{center}
\end{large}
\end{titlepage}
\endgroup
 
\thispagestyle{empty}
\cleardoublepage
\thispagestyle{empty}

\begingroup
\makeatletter
\pagestyle{empty}
\begin{center}
\Large{UNIVERSITÉ CLAUDE BERNARD-LYON 1}
\end{center} 

\begin{center}
\begin{tabular}{lcl}
\bf Président de l'Université&&\bf M. le Professeur M. A. BONMARTIN\\
Vice-Président du Conseil d'Administration&&M. le Professeur G. ANNAT\\
Vice-Président du Conseil des Etudes et&&M. le Professeur D. SIMON\\
 de la Vie Universitaire&& \\
Vice-Président du Conseil Scientifique&&M. le Professeur J-F. MORNEX\\
Secrétaire Général&&  M. G. GAY\\
\end{tabular} 

\begin{center}
\large SECTEUR SANTÉ
\end{center}

\begin{tabular}{p{7.2cm}ll}
\it Composantes && \\
Faculté de Médecin Lyon-Est - Claude Bernard && Directeur: M. le Professeur J. ETIENNE\\
Faculté de Médecine et de Maïeutique Lyon && Directeur: M. le Professeur F-N. GILLY\\
Lyon Sud – Charles Mérieux&&\\
UFR d'Ontologie && Directeur: D. BOURGEOIS\\
Institut des Sciences Pharmaceutiques && Directeur: M. le Professeur F. LOCHER \\
et Biologiques&&\\
Institut Techniques de Réadaptation && Directeur: M. le Professeur MATILLON\\
Département de Formation et Centre de && Directeur: M. le Professeur P. FARGE\\
Recherche en Biologie Humaine&&\\

\end{tabular}

\begin{center}
\large SECTEUR SCIENCES
\end{center}

\begin{tabular}{p{7.3cm}ll}
\it Composantes && \\
Faculté des Sciences et Technologies && Directeur: M. le Professeur S. De MARCHI\\
Département Biologie && Directeur: M. le Professeur F. FLEURY\\
Département Chimie Biochimie && Directeur: Mme. le Professeur H. PARROT\\
Département Génie Electrique et && Directeur: M. N. SIAUVE\\
des Procédés&&\\
Département Informatique && Directeur: M. le Professeur S. AKKOUCHE\\
Département Mathématiques&& Directeur: M. le Professeur A. GOLDMAN\\
Département  Mécanique&& Directeur: M. le Professeur H. BEN HADID\\
Département Physique&& Directeur: M. le Professeur S. FLECK\\
Département Sciences de la Terre && Directeur: M. le Professeur I. DANIEL\\
des Activités Physiques et Sportives&&\\
UFR Sciences et Techniques && Directeur: M. C. COLLIGNON\\
Observatoire de Lyon && Directeur: M. B. GUIDERDONI\\
Ecole Polytechnique Universitaire de Lyon 1&& Directeur: M. P. FOURNIER\\
Ecole Supérieure de Chimie && Directeur: M. G. PIGNAULT\\
Physique Electronique&&\\
Institut Universitaire de Technologie de Lyon 1&& Directeur: M. le Professeur C. COULET\\
Institut de Science Financière && Directeur: M. le Professeur  J. C. AUGROS\\
et d'Assurances&&\\
\end{tabular}
\end{center} 
\hfill

\clearpage
\pagestyle{empty}
\pagebreak[4]
$^{}$
\clearpage
\pagestyle{empty}
\pagebreak[4]
\endgroup

\cleardoublepage
\thispagestyle{empty}

\chapter*{Acknowledgments}
First and foremost, I would like to thank my advisors, Marie, Pilu and
Vincent. \emph{This thesis would not been possible without your
  guidance and support}. A special thanks to Vincent, who was a
constant presence during the last three years, particularly in the
beginning when I was trying to figure out my thesis topic.

My gratitude to my co-authors, Pavlos Antoniou, \'Etienne Birmel\'e,
Rayan Chikhi, Pierluigi Crescenzi, Rui Ferreira, Roberto Grossi,
Vincent Lacroix, Alice Julien-Laferri\`ere, Andrea Marino, Vincent
Miele, Nadia Pisanti, Janice Kielbassa, Gregory Kucherov, Pierre
Peterlongo, Marie-France Sagot, Kamil Salikhov, Romeo Rizzi and Raluca
Uricaru. It has been a great pleasure to work with all of you. I
learned a lot from you, not only new techniques or algorithms, but
different ways to do research.

I would like to thank all KisSplice team, Alice Julien-Laferri\`ere,
Camille Marchet, Vincent Miele and Vincent Lacroix. Thanks for all the
bug reports, patches and discussions. I would also like to thank all
Bamboo team, who created a great work environment. You made me feel
home in a distant country.

Last but not least, I would like to write a few words in Portuguese
dedicaded to my family: \emph{``Essa tese com certeza n\~ao seria
  poss\'ivel sem o suporte incondicional da minha fam\'ilia, minha
  m\~ae, Miltes, meu pai, Jo\~ao, e minha irm\~a, Nat\'alia. Voc\^es
  que desde cedo me ensinaram que toda conquista \'e uma mistura de
  talento e muito suor. Voc\^es que sempre me encorajaram a seguir os
  meu sonhos. Obrigado por isso e muito mais!''}

\pagestyle{empty} 
\cleardoublepage

\chapter*{Abstract}

In this thesis, we address the problem of identifying and quantifying
variants (alternative splicing and genomic polymorphism) in RNA-seq
data when no reference genome is available, without assembling the
full transcripts. Based on the fundamental idea that each variant
corresponds to a recognizable pattern, a bubble, in a de Bruijn graph
constructed from the RNA-seq reads, we propose a general model for all
variants in such graphs. We then introduce an exact method, called
KisSplice, to extract alternative splicing events. Finally, we show
that it enables to identify more correct events than general purpose
transcriptome assemblers (Grabherr et al. (2011)).

In order to deal with ever-increasing volumes of NGS data, an extra
effort was put to make our method as scalable as possible. The main
time bottleneck in the KisSplice is the bubble enumeration step. Thus,
in order to improve the running time of KisSplice, we propose a new
algorithm to enumerate bubbles. We show both theoretically and
experimentally that our algorithm is several orders of magnitude
faster than the heuristics based on cycle enumeration. The main memory
bottleneck in KisSplice is the construction and representation of the
de Bruijn graph. Thus, in order to reduce the memory consumption of
KisSplice, we propose a new compact way to build and represent a de
Bruijn graph improving over the state of the art (Chikhi and Rizk
(2012)). We show both theoretically and experimentally that our
approach uses 30\% to 40\% less memory than such state of the art,
with an insignificant impact on the construction time.

Additionally, we show that the same techniques used to list bubbles
can be applied in two classical enumeration problems: cycle listing
and the K-shortest paths problem. In the first case, we give the first
optimal algorithm to list cycles in undirected graphs, improving over
Johnson's algorithm, the long-standing state of the art. This is the
first improvement to this problem in almost 40 years. We also give the
first optimal algorithm to list $st$-paths in undirected graphs. In
the second case, we consider a different parameterization of the
classical K-shortest simple (loopless) paths problem: instead of
bounding the number of st-paths, we bound the weight of the
st-paths. We present new algorithms with the same time complexities
but using exponentially less memory than previous approaches.

\dominitoc 
\setcounter{tocdepth}{2} 
\setcounter{minitocdepth}{2} 


\pagestyle{fancy} 
\renewcommand{\chaptermark}[1]{%
\markboth{\chaptername~\thechapter.\ #1}{}}
\renewcommand{\sectionmark}[1]{\markright{\thesection\ #1}} 
\fancyhf{} 
\fancyhead[LE,RO]{\bfseries\thepage} 
\fancyhead[LO]{\bfseries\rightmark} 
\fancyhead[RE]{\bfseries\leftmark} 
\fancypagestyle{plain}{%
\fancyhead{} 
\renewcommand{\headrulewidth}{0pt} 
}

\tableofcontents
\clearpage
\thispagestyle{empty} 

%

\chapter*{Introduction}
\markboth{\bf Introduction}{\bf Introduction}
\addstarredchapter{Introduction}
The general question addressed in this thesis is how to extract
biologically meaningful information from next generation sequencing
(NGS) data without using a reference genome. The NGS technology allows
to read, although in a fragmented way, the full content of the genetic
material (DNA) in a given organism. The main difficulty lies in the
fragmented nature of the NGS information. The NGS data forms a huge
``jigsaw puzzle'' that needs to be, at least partially solved (or
assembled) in order to retrieve some biologically meaningful
information. The challenge is made harder by the no reference genome
assumption, which means that the puzzle needs to be solved relying
only on the intrinsic information that two ``pieces'' (reads) are
compatible, and thus should be together; there is no prior information
about the ``full picture''.

The usual route to solve this kind of problem is to first assemble the
NGS reads and then analyze the result to draw conclusions about a
given biological question. The difficulty with this approach is that
even the simplest formalization of genome assembly as an optimization
problem is NP-hard (shortest superstring problem). In addition, there
is no guarantee that a solution to the optimization problem is unique
or really corresponds to the original sequence.  Actually, since the
genome may contain repeats much larger than the read size, we may not
have enough information to completely solve this problem regardless of
the formulation. In practice, several heuristics are applied to this
problem. The main goal of these heuristics is to produce long contigs,
i.e. contiguous \emph{consensus} sequences of overlapping reads.  Of
course, being heuristics, there are no strong guarantees about the
results. Additionally, since they try to maximize the length of the
contigs, genomic polymorphism (SNPs, indels and CNVs), corresponding
to local variations in a (diploid) genome, are not modeled explicitly
and thus systematically overlooked, for each variation only a
consensus sequence is produced.  In this thesis, we propose an
alternative strategy avoiding the use of heuristics. We argue that for
certain biological questions, it is not necessary to first solve the
hard, often ill-posed, problem of completely assembling the NGS reads;
it is instead sufficient, and sometimes even better (as in the case of
searching for genomic polymorphism), to only \emph{locally} assemble
the data.

The classical view of the information flow inside the cell, or central
dogma of molecular biology, can be summarized as: genes (DNA sequence)
are transcribed into messenger molecules (mRNA) which are then
translated into proteins. The NGS technology is not restricted to the
sequence in the first step of this flow (DNA), it can as well be
applied to the entire set of mRNAs of a cell (that is, to the
\textit{transcriptome}) through what is called an RNA-seq experiment.
In this case, the assembly problem is not anymore to solve a single
``jigsaw puzzle'', but several puzzles where the pieces are mixed
together. Intuitively, this is a generalization of the genome assembly
problem, and thus certainly no easier than it.

The specific problem we address here is the identification of
variations (including alternative splicing) in RNA-seq data. The way
the central dogma was stated may induce us to think that there is a
one-to-one correspondence between genes and proteins. In general that
is not true, a single gene can produce several distinct
proteins. Alternative splicing is one of the main factors responsible
for this variability. It is a mechanism where several distinct mRNAs
are produced from a same gene through an RNA sequence editing
process. The local assembly strategy is specially suited to identify
alternative splicing events because these are intrinsically local: an
alternative splicing event usually generates two similar mRNAs
molecules (isoforms) sharing the majority of their sequence
(constitutive exons). The heuristics used in transcriptome assemblers
tend to overlook such similar, but not identical, sequences.

Since the first automatic sequencing instruments (Sanger) were
introduced in 1998, the cost per base sequenced has decreased
dramatically: from \$2,400 to currently \$0.07 per million of bases
sequenced (Illumina). This is due to an exponential increase in
throughput; while the early Sanger machines produced 10 Kb ($10^4$
base pairs) per run, the current Illumina HiSeq 2000 produces 600
Gb/run. This represents an impressive $10^{7}$-fold increase. During
the same time period, the processing capacity of an off-the-shelf
computer had only a $10$-fold increase, and the memory a $10^2$-fold
increase. For that reason, any algorithm dealing with NGS data has to
be highly efficient both in terms of memory usage and time
consumption, and not rely on hardware improvements to compensate for
ever-increasing volumes of data. In this thesis, we focus on time and
memory efficient algorithms, from both the theoretical and the
practical point of view.

The first step towards a solution to our variation identification
problem is to have a suitable representation for the set of RNA-seq
reads. A natural way to represent NGS data is to consider a directed
graph (overlap graph), where each vertex corresponds to a read and the
``compatibility'' information, i.e.  suffix-prefix overlaps, is stored
in the arcs. This representation, however, does not scale well to
large volumes of NGS data, since to compute the arcs, in principle, a
quadratic number of read comparisons is necessary. A more suitable
one, proposed by \cite{Pevzner01}, is to use de Bruijn graphs. This is
the representation used here, and also in the majority of the recent
NGS assemblers. A more in-depth comparison between several possible
representations of NGS data, along with the necessary biological and
mathematical background to follow this thesis, is given in
Chapter~\ref{chap:back}.

We then show that variations in RNA-seq correspond to certain
subgraphs in the de Bruijn graph built from the set of RNA-seq
reads. More specifically, a variation creates a \emph{bubble}, that is
a pair of vertex-disjoint paths, in a de Bruijn graph. Hence, the
problem of finding variations can be reduced to the problem of listing
bubbles satisfying certain properties in the de Bruijn graph built
from the set of reads. In Chapter~\ref{chap:kissplice}, based on our
paper \cite{Sacomoto12}, we describe a method, called \ks,
implementing this strategy, along with a complete description of the
relationship between variations and bubbles. We then show that, for
the specific case of alternative splicing identification, our method
is more sensitive than general purpose transcriptome assemblers and,
although using relatively simple algorithms to build the graph and
list the bubbles, uses roughly the same amount of memory and time.

The main time bottleneck in the \ks algorithm is the bubble
enumeration step. Thus, in an effort to make our method as scalable as
possible, in the first part of Chapter~\ref{chap:unweighted}, which is
based on our paper \cite{Birmele12}, we modified Johnson's cycle
listing algorithm (\cite{Johnson75}) to enumerate bubbles in general
directed graphs, while maintaining the same time complexity. For a
directed graph with $n$ vertices and $m$ arcs containing $\eta$
bubbles, the method we propose lists all bubbles with a given source
in $O((n + m)(\eta + 1))$ total time and $O(m+n)$ delay (time elapsed
between the output of two consecutive solutions).  For the general
problem of listing bubbles, this algorithm is exponentially faster
than the listing algorithm of \ks. However, in the particular case of
listing bubbles corresponding to alternative splicing events, this
algorithm is outperformed by \ks. This issue is addressed in the first
part of Chapter~\ref{chap:weighted}, which is based on our paper
\cite{Sacomoto13}. Using a different enumeration technique, we
propose an algorithm to list bubbles with path length constraints in
weighted directed graphs. The method we propose lists all bubbles with
a given source in $O(n(m + n \log n))$ delay. Moreover, we
experimentally show that this algorithm is several orders of magnitude
faster than the listing algorithm of \ks to identify bubbles
corresponding to alternative splicing events.

The main memory bottleneck in \ks is the construction and
representation of the de Bruijn graph. Thus, again with the goal to
make our method as scalable as possible, in Chapter~\ref{chap:dbg},
which is based on our paper \cite{Salikhov13}, we propose a new
compact way to build and represent a de Bruijn graph improving over
the state of the art \cite{Chikhi12}. We show both theoretically and
experimentally that our approach uses 30\% to 40\% less memory than
such state of the art, with an insignificant impact on the
construction time. Our de Bruijn graph representation is general, in
other words it is not restricted to the variation finding or RNA-seq
context, and can be used as part of any algorithm that represents NGS
data with de Bruijn graphs.

The central result of this thesis can be summarized as \ks version
2.0, the current version of \ks that includes both improvements, time
and memory, discussed above. This version is able to treat
medium-sized datasets (up to 100M Ilumina reads) in a desktop computer
(8GB of RAM) and large datasets (already tested with 1G Ilumina reads)
in a standard high-memory server (100GB of RAM). This is however not
the only result of this thesis. As another example of the role of
serendipity in scientific research, the techniques we developed while
designing the two bubble listing algorithms turned out to be useful in
the context of two classical enumeration problems: listing simple
cycles in undirected graphs and listing the $K$-shortest paths.

The problem of efficiently listing all the simple cycles in a graph
has been studied since the early 70s. For a graph with $n$ vertices
and $m$ edges, containing $\eta$ cycles, the most efficient solution
was presented by \cite{Johnson75} and takes \mbox{$O((\eta+1)(m+n))$}
time.  This solution is not optimal for undirected
graphs. Nevertheless, no theoretical improvements have been proposed
in the past decades. In the second part of
Chapter~\ref{chap:unweighted}, which is based on our paper
\cite{Birmele13}, we present the first optimal solution to list all
the simple cycles in an undirected graph $G$. Specifically, let
$\setofcycles(G)$ denote the set of all these cycles
($|\setofcycles(G)| = \eta$).  For a cycle $c \in \setofcycles(G)$,
let $|c|$ denote the number of edges in~$c$. Our algorithm requires
$O(m + \sum_{c \in \setofcycles(G)}{|c|})$ time and is asymptotically
optimal: $\Omega(m)$ time is necessarily required to read $G$ as
input, and $\Omega(\sum_{c \in \setofcycles(G)}{|c|})$ time is
required to list the output. We also present the first optimal
solution to list all the simple paths from $s$ to $t$ (shortly,
$st$-paths) in an undirected graph $G$. Let $\setofpaths_{st}(G)$
denote the set of $st$-paths in $G$ and, for an $st$-path $\pi \in
\setofpaths_{st}(G)$, let $|\pi|$ be the number of edges in $\pi$.
Our algorithm lists all the $st$-paths in~$G$ optimally in $O(m +
\sum_{\pi \in \setofpaths_{st}(G)}{|\pi|})$ time.

The $K$-shortest paths problem, that is returning the first $K$
distinct shortest simple $st$-paths, has also been studied for more
than 30 years, since the early 60s. For a weighted graph with $n$
vertices and $m$ edges, the most efficient solution is an $O(K(mn +
n^2 \log n))$ time algorithm for directed graphs
(\cite{Yen71,Lawler72}), and an $O(K(m+n \log n))$ time algorithm for
undirected graphs (\cite{Katoh82}), both algorithms using $O(Kn + m)$
memory. In the second part of Chapter~\ref{chap:weighted}, which is
based on our paper
\cite{Grossi14} (in preparation), we consider a different
parameterization for this problem: instead of bounding the number of
$st$-paths, we bound the weight of the $st$-paths. In other words, we
consider the problem of listing $st$-paths with a weight bounded by
$\alpha$ in a weighted graph.  We present a general scheme to list
bounded length $st$-paths in weighted graphs that takes $O(nt(n,m)
\eta)$ time, where $t(n,m)$ is the time for a single source shortest
path computation and $\eta$ is the number of paths. This algorithm
uses memory linear in the size of the graphs, independent of the
number of paths output. For undirected non-negatively weighted
graphs, we also show an improved algorithm that lists all $st$-paths
with length bounded by $\alpha$ in $O( (m + t(n,m)) \eta)$ total
time. In particular, this is $O(m \eta)$ for unit weights and $O((m +
n \log n) \eta)$ for general non-negative weights. The time spent per
path by our algorithms in both directed and undirected graphs matches
the complexity of the best algorithms for the $K$-shortest path
problem, while only using memory linear in the size of the
graph. Moreover, we also show how to modify the general scheme to
output the paths in increasing order of their lengths, providing an
alternative solution to the $K$-shortest paths problem.

A summary of the thesis organization, and the relation between the
chapters, is given in the figure below. As suggested by the figure,
Chapters 3, 4 and 5 can be read more or less independently of each
other.

\bigskip
\bigskip

\begin{figure}[htbp]
  \includegraphics[width=\linewidth]{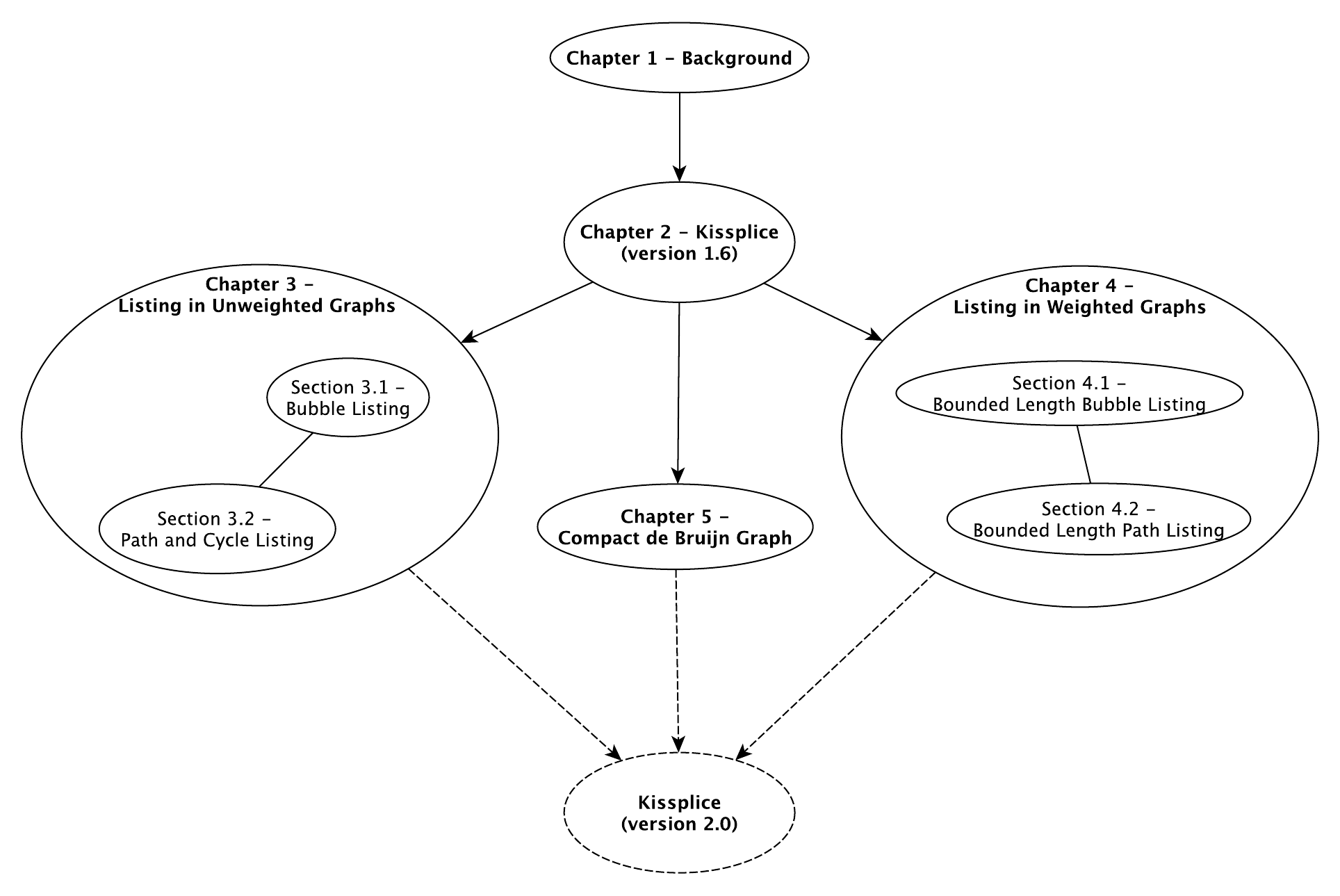}
  \caption*{Thesis organization and the corresponding versions of \ks.}
\end{figure}


\mainmatter
\chapter{Background}
\label{chap:back}
\minitoc
In this chapter, we cover the background and introduce the main
notations necessary to follow the rest of this thesis. Certainly it
would be infeasible to cover all the material with sufficient detail
to make a thesis self-contained; we therefore do not attempt to be
exhaustive or even complete. For the majority of the topics presented
in this chapter, we give only a brief introduction, whereas we spend
more time on a few others, which we consider to be important to the
thesis and less familiar to the reader.  Whenever possible, we give
the main intuition behind the concepts presented and their
inter-relationship. The chapter is divided in two sections:
\emph{biological concepts} (Section~\ref{sec:back:bio}) and
\emph{mathematical concepts} (Section~\ref{sec:back:math}).

The purpose of Section~\ref{sec:back:bio} is to present three main
topics. The first is the \emph{central dogma of molecular biology},
which intuitively gives a roadmap for the information flow inside the
cell, from the ``blueprints'' (DNA) to the ``workers'' (proteins). The
next topic is \emph{alternative splicing}, one of the steps of the
information flow in ``complex organisms'' (eukaryotes), where the
``message'' can be \emph{modulated}, that is from the same input
message, several distinct output messages (mRNAs) can be
produced. Finally, the third topic is \emph{next-generation sequencing
(NGS)} with emphasis on \emph{RNA-seq}, a technology that allows to
read, although in a fragmented way, the mRNAs inside the cell. A
central question in this thesis is to find, from RNA-seq data, all the
alternative ways in which the mRNAs are \emph{spliced} (i.e. to find
alternative splicing events).

The main goal of Section~\ref{sec:back:math} is to introduce, along
with the main definitions, notations and some properties of the
mathematical structures used in this thesis, two seemingly unrelated
topics: modeling and assembling NGS data, and enumeration
algorithms. The relationship, although not immediately apparent,
stands at the very core of this thesis. By modeling NGS data, more
specifically RNA-seq, using a special kind of directed graphs, namely
de Bruijn graphs, the problem of identifying biologically interesting
structures (e.g. alternative splicing events) can be seen as an
enumeration problem of special structures in those graphs. This
relationship is further detailed and explored in
Chapter~\ref{chap:kissplice}.

It should be noted that several standard computer science topics used
throughout this thesis are not covered in
Section~\ref{sec:back:math}. For the analysis of algorithms, the
asymptotic big $O$ notation and basic data structures (e.g. stacks,
queues and heaps), we refer to \cite{Cormen01}. For the computational
complexity theory and a compendium of NP-hard problems, we refer to
\cite{Ausiello99} and \cite{Garey79}. For basic graph algorithms,
e.g. depth-first search (DFS), breadth-first search (BFS) and
Dijkstra's algorithm, we refer to \cite{Cormen01} and
\cite{Sedgewick01}. Finally, for further information about graphs and
digraphs, we refer to \cite{Diestel05} and \cite{Bang-Jensen08},
respectively.


\bigskip
\bigskip


\section{Biological Concepts} \label{sec:back:bio}

\subsection{DNA, RNA and Protein} \label{sec:back:central_dogma}

\emph{Deoxyribonucleic acid} (\emph{DNA}) is a long biopolymer that
carries the hereditary information in almost all known organisms.
They contain the ``blueprint'' for a complete organism. In other
words, it is a long molecule composed by a huge number of repeated
subunits, called \emph{nucleotides}, chemically bonded together. There
are four different nucleotides, namely: \emph{adenine} (A),
\emph{cytosine} (C), \emph{thymine} (T) and \emph{guanine} (G).  A DNA
molecule has a \emph{double-stranded} structure, where each molecule
is composed by two strands, i.e. two chains of nucleotides, of the
same length, running in opposite directions and respecting a fixed
pairing rule between the corresponding nucleotides in each strand. The
rules for nucleotide pairing, or \emph{hybridization}, are: A with T
(and vice-versa), and C with G (and vice-versa). We say that A
(resp. C) is \emph{complementary} to T (resp. G). That way, each
strand is the \emph{reverse complement} of the other, i.e. the
sequence of nucleotides in one strand is equal to the reverse sequence
of the other strand after substituting each nucleotide by its
complementary. A \emph{genome} is the set of all the genetic material,
in the form of DNA (except for some viruses), of a given organism.

\emph{Ribonucleic acid} (\emph{RNA}) is, as DNA, a long
biopolymer composed by adenine, cytosine, guanine and \emph{uracil}
(U) instead of thymine. However, unlike DNA, an RNA is not
double-stranded: it contains only a single chain of nucleotides, and
the sequence is usually much shorter. Despite being single-stranded,
the nucleotide hybridization rules still hold for RNA molecules, with
U substituting T. The hibridization can occur inside the same molecule
of RNA, with complementary stretches of RNA folding and binding
together; between two molecules of RNA; or, under certain
circumstances, between an RNA and partially single-stranded DNA
molecules. In terms of function, except for some viruses, the main
role of RNA is not to carry hereditary information, but to transport
genetic information from the DNA to other parts of the cell. An RNA is
the main ``messenger'' within the cell. The \emph{transcriptome} is
the set of all RNAs present in the organism. Unlike the genome, it is
not the same in all cells at all times.

\emph{Proteins} are another kind of biopolymers where, unlike
nucleic acids (DNA and RNA), the subunits are called \emph{amino
  acids} and, instead of 4, there are 20 different types. A protein
  may contain several linear chains of amino acids, called
\emph{polypeptides}, but there is no strict pairing rules for amino
acids like for nucleic acids. Proteins perform a large number of
functions within an organism, including: catalyzing certain reactions,
replicating DNA, and transporting molecules from one location to
another. The proteins are the main ``workers'' of the cell.

The \emph{central dogma of molecular biology} states that the ``coded
genetic information hard-wired into DNA is transcribed into individual
transportable cassettes, composed of messenger RNA (mRNA); each
mRNA\footnote{This is the classical view of molecular biology, which
  we present for simplicity. However, it is known that not every
  transcribed RNA is later translated into protein
  (\cite{Encode07,Bakel10}). These molecules are known as non-coding
  RNA (ncRNA). } cassette contains the program for synthesis of a
particular protein (or small number of proteins)''
(\cite{Lodish00}). In other words, the ``blueprint'' for each protein
encoded in the DNA sequence is transcribed to RNA which is then
translated to proteins. Note that, for the passage from DNA to RNA, we
use the term transcription whereas from RNA to protein we use
translation. That is because RNA and DNA use the same ``language'',
they are encoded using the same set of letters (nucleotides), while
proteins use a different set, the amino acids, so when passing from an
RNA to a protein, there is a translation from one language
(nucleotides) to another (amino acids). A diagram of the central dogma
is shown in Fig.~\ref{fig:back:central_dogma}(a).

\begin{figure}[htb]
  \center
  \subfloat[]{\includegraphics[width=5cm]{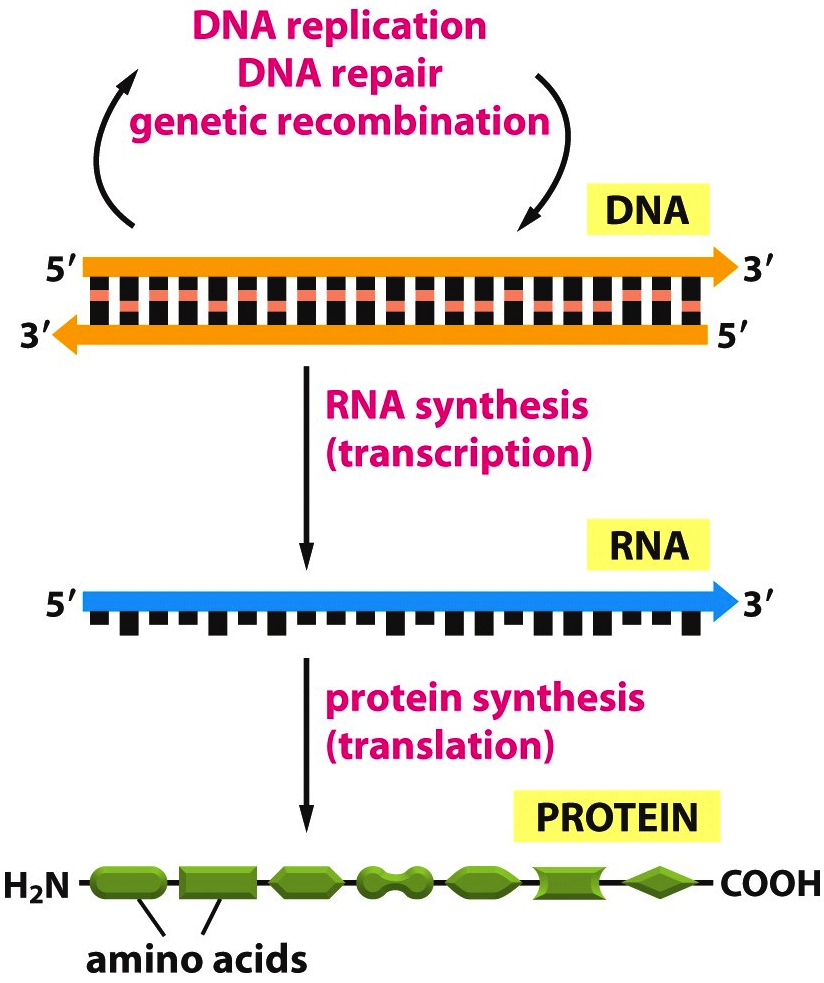}} \hspace{1cm} 
  \subfloat[]{\includegraphics[width=7.5cm]{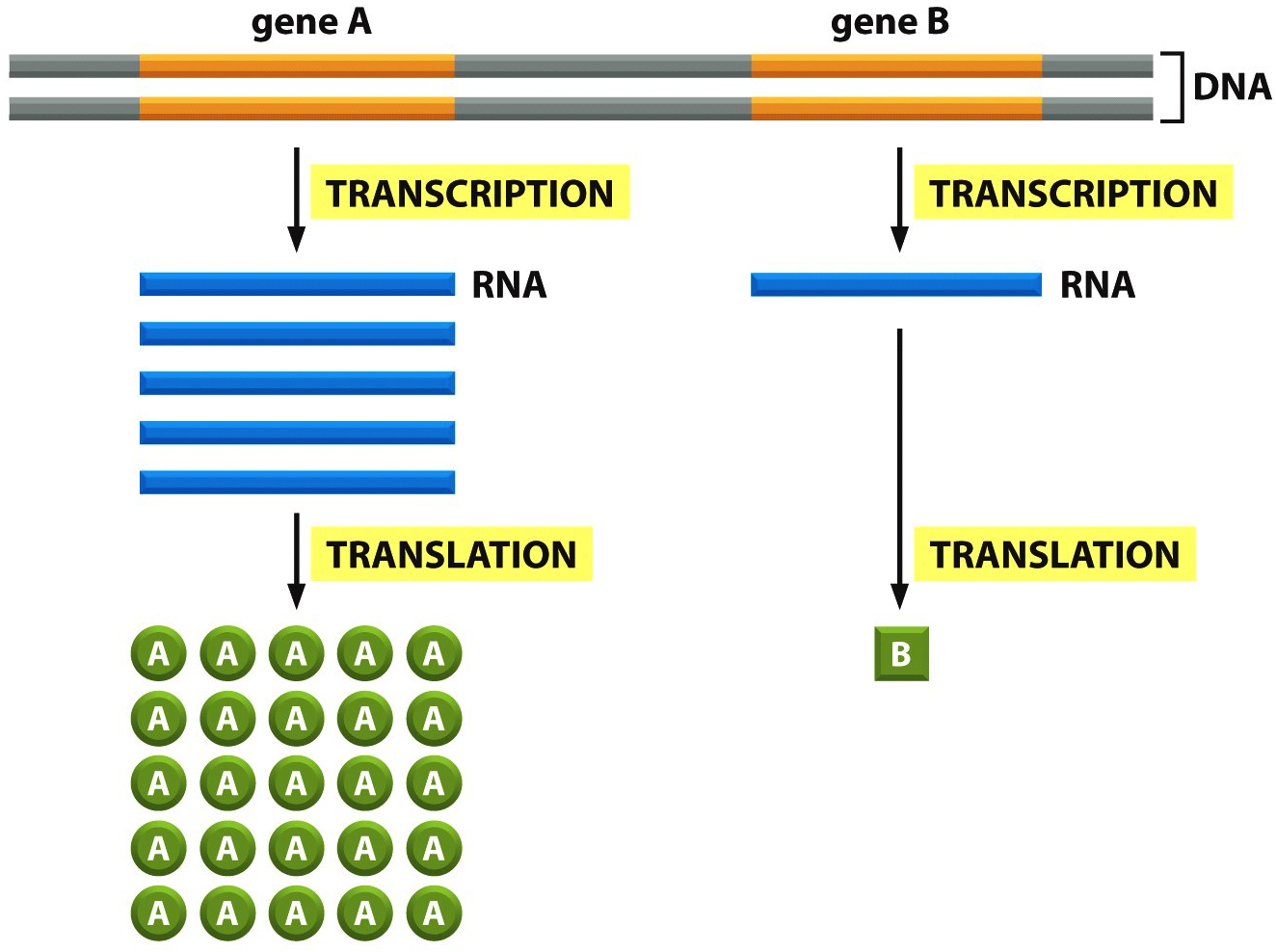}}
  \caption{The central dogma of molecular biology and gene
    expression. (a) The central dogma states that genetic
    information flows from DNA to RNA to proteins. In the diagram,
    DNA is transcribed to RNA, which is then translated to
    proteins. (b) The genes are transcribed and translated at different
    rates to respond to different demands of the proteins corresponding
    to each gene. Two genes are showed in the left, a highly
    expressed gene and, in the right, a lowly expressed
    gene. Reproduced from \cite{Alberts03}.}
  \label{fig:back:central_dogma}
\end{figure}

Classically, a \emph{gene} is a region of a DNA molecule that encodes
for a protein (actually, a polypeptide chain). As stated before,
translation and transcription are two main ways in which the cell
reads out, or \emph{expresses}, their genetic information, or
genes. In a simplified view of the genome, we can assume that each
gene is present in only one copy, and in this case, to respond to
different demands of each protein, the cell has to translate and
transcribe each gene with different efficiencies. A gene that is
transcribed with higher rates is called \emph{highly expressed}; on
the other hand, a gene transcribed at lower rates is
called \emph{lowly expressed}. A diagram of the variable ranges of
gene expression is shown in Fig.~\ref{fig:back:central_dogma}(b).

\subsection{Alternative Splicing} \label{sec:back:AS}

In eukaryotic organisms (that comprises all living organisms except
for bacteria and archaea); the cells have separate compartments for
the nucleus and other structures (organelles). The genetic material
(DNA) is stored in the nucleus which is separated from the cytoplasm
by a membrane. In these organisms, transcription is done inside the
nucleus, the RNA is then processed into an mRNA and exported to the
cytoplasm to be translated into proteins. Another relevant
particularity of eukaryotes is that their genes contain two different
types of regions: \emph{exons} and \emph{introns}. The gene is then
composed by alternating sequences of exons and introns. In the RNA
processing step to produce an mRNA, the majority of the introns are
removed, \emph{spliced}, and a long chain of A's, the \emph{poly-A
  tail}, is added to one of the ends. See
Fig.~\ref{fig:back:alternative_splicing}(a) for a diagram of gene
expression in eukaryotes.

\begin{figure}[htb]
  \vspace{0.5cm}
  \centering
  \subfloat[]{\includegraphics[width=7cm]{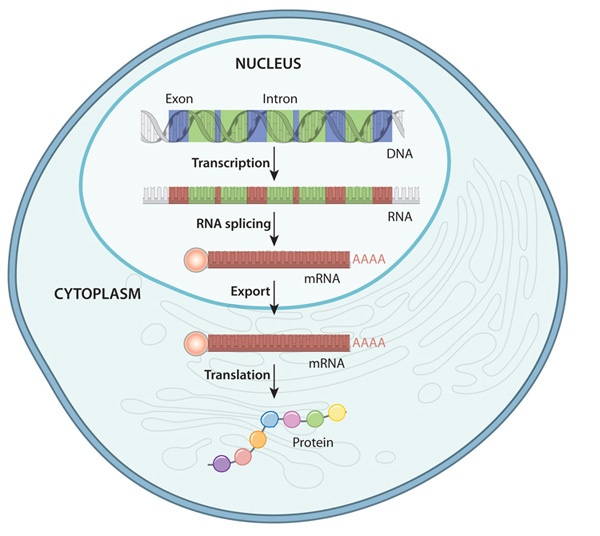}}  
  \subfloat[]{\includegraphics[width=9cm]{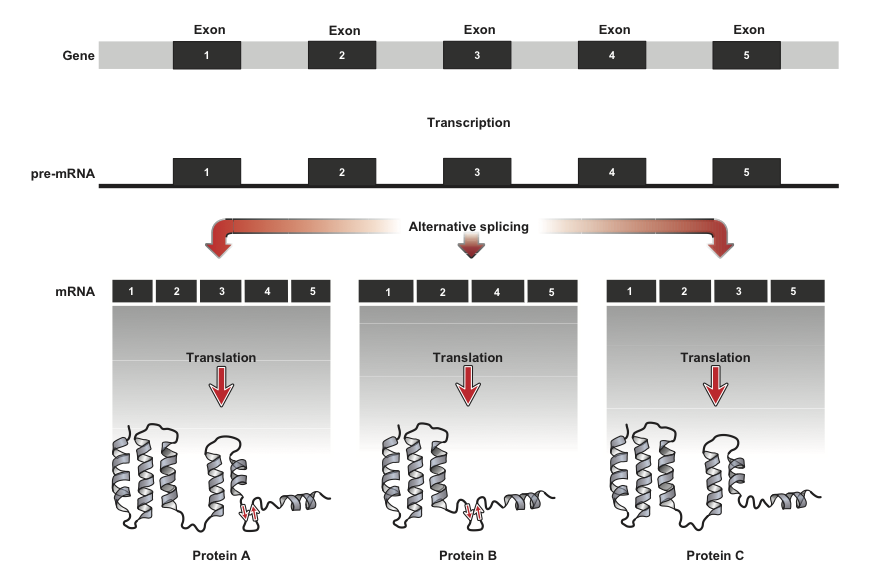}}
  \caption{From DNA to RNA to protein in eukaryotes. (a) An
    eukaryotic gene transcription and translation is shown. The DNA
    inside the nucleus is transcribed into RNA, which is then processed
    into mRNA still inside the nucleus. Finally, the mature mRNA is
    transported to the cytoplasm to be translated into a polypeptide
    chain. (b) An eukaryotic gene with the introns and exons
    highlighted is shown. The same gene can produce different proteins
    through alternative splicing. In the example, three possible mRNAs
    are shown: one containing all the exons and two others with one of
    the exons skipped. Reproduced from \cite{Schulz10}.}
  \label{fig:back:alternative_splicing}
\end{figure}

The discussion about the central dogma and gene expression may lead us
to think that there is a one-to-one correspondence between genes and
mRNAs (and proteins). In general, that is not the case, one gene can
give rise to several distinct mRNAs (and proteins). Actually, it is
estimated that 95\% of all human genes give rise to more than one mRNA
(\cite{Pan08}). There are three main mechanism responsible for this
variability: \emph{alternative splicing}, \emph{alternative promoters}
and \emph{alternative polyadenilation}. Alternative promoter and
polyadenilation sites change the start and the end of the RNA
transcription, respectively. Alternative splicing, our main interest
here, is a post-transcription modification of the transcribed RNA
(pre-mRNA).  A diagram is shown in
Fig.~\ref{fig:back:alternative_splicing}(b). Different mRNAs
originating from the same gene are called \emph{alternative isoforms}
or simply \emph{isoforms}. An exon is \emph{constitutive} if it is
present in all isoforms, and \emph{alternative} otherwise.

Alternative splicing takes place when the transcribed RNA (pre-mRNA)
is spliced to produce the mature mRNA, instead of only remove all
introns, some of exons may be skipped, included, shortened, or
extended, and some introns may be retained; in each case a single
pre-mRNA produces different mRNA variants (isoforms). An overview of
the five types of alternative splicing events are shown in
Fig.~\ref{fig:back:as_types}. A \emph{splice site} is a sequence
marking the border of a spliced region, usually the beginning or end of
an intron, but it can also occur inside an exon in the case of
alternative 3' or 5' splice sites; in this case, when performing
alternative splicing, the exon is shortened or extended,
respectively. 

\begin{figure}[htbp]
  \centering
  \includegraphics[]{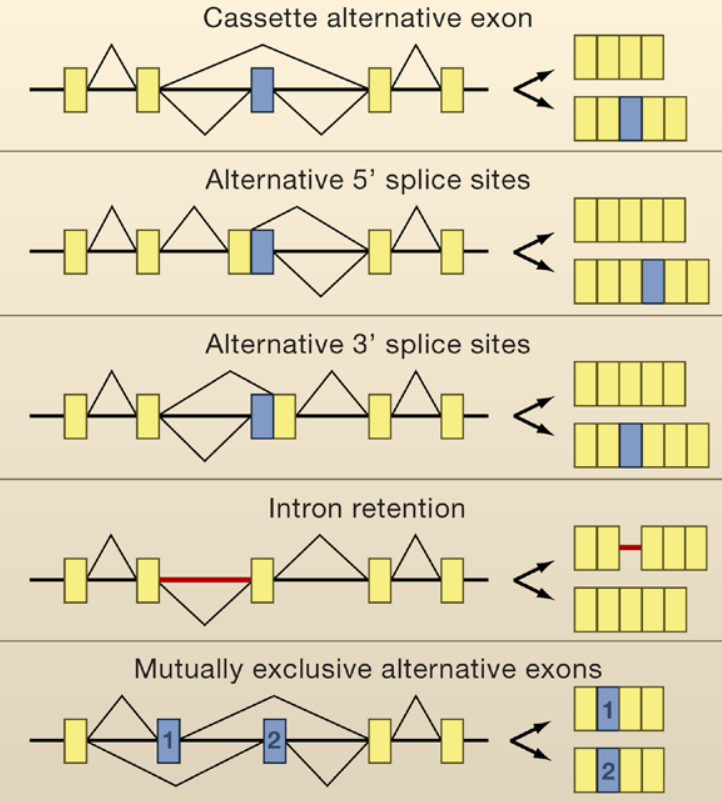}
  \caption{The five types of alternative splicing events are shown. In
    the left hand side, alternating sequences of exons (rectangles)
    and introns (straight lines) are represented; the zig-zag lines
    represent the removal (splicing) of a region. In the right hand
    side, in each case we have the two alternative isoforms after the
    splicing indicated by the zig-zag lines. Constitutive exons are
    shown in yellow, alternative exons in blue and retained introns in
    red. Reproduced and modified from \cite{Blencowe06}.}
  \label{fig:back:as_types}
\end{figure}

\subsection{Next Generation Sequencing (NGS)} \label{subsec:ngs}
The \emph{shotgun sequencing} process was proposed by \cite{Staden79}
to overcome the limitations of the semi-automated chain termination
sequencing method introduced by \cite{Sanger77}. The main difficult of
Sanger's method when applied to whole genome sequencing was the length
of the DNA strands that could be sequenced. Typically, only DNA
strands with at most 1000 base pairs could be \emph{read}, while even
the smallest eukaryotic genome is several orders of magnitude
larger. In order to overcome this, Staden's key idea was to randomly
shear the whole genome of an organism into small fragments, and to
independently sequence each fragment using Sanger's method. The
resulting DNA \emph{reads} would then be combined together, \emph{in
silico}, to reconstruct the original genome.

In the past few years, several sequencing technologies have been
developed to replace Sanger sequencing in the shotgun sequencing
context. These new methods, although producing shorter reads, are
fully-automated and massively parallel and therefore can sequence a
huge number of fragments in a same \emph{run}, resulting in a huge
number of reads in comparable time and for a fraction of the cost of a
Sanger sequencing. These high-throughput approaches are collectively
known as \emph{next generation sequencing (NGS)} methods.  These
technologies have been released as commercial products by several
companies, e.g., the Solexa Genome Analyzer (Illumina, San Diego), the
SOLiD platform (Applied Biosystems; USA) and 454 Genome Sequencers
(Roche Applied Science; Basel). Although the specific details vary
from one method to another, they can be seen as implementations of the
cyclic-array sequencing (\cite{Shendure08}), which can be summarized
as ``the sequencing of a dense array of DNA features by iterative
cycles of enzymatic manipulation and imaging-based data collection''
(\cite{Mitra99}). See Fig.~\ref{fig:back:ngs} for a comparison between
Sanger and cyclic-array methods for the shotgun sequencing process.

\begin{figure}[htbp]
  \centering
  \includegraphics{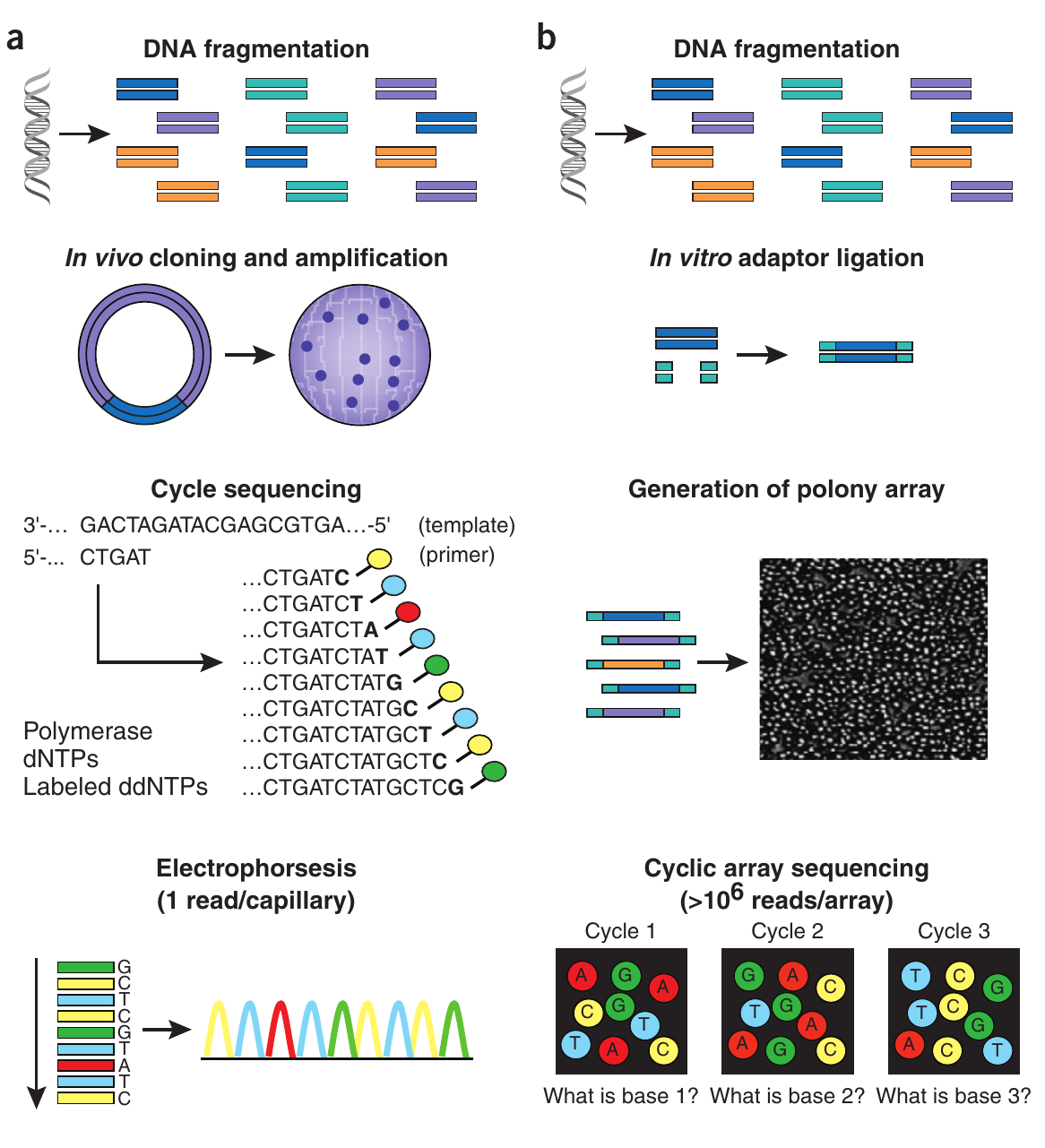}
  \caption{Workflow of shotgun Sanger sequencing versus
    next-generation sequencing. (a) In the high-throughput shotgun
    Sanger sequencing, DNA is first fragmented and subsequently
    integrated into a plasmid vector (a circular bacterial DNA that
    can replicate) that is later inserted into \emph{Escherichia coli}
    to be amplified (copied several times). A single bacterial colony
    is selected for each sequencing reaction and the DNA is
    isolated. Each cycle sequencing reaction creates a ladder of
    dye-labeled products, which are subjected to electrophoretic
    separation in one run of a sequencing instrument. A detector for
    fluorescently labeled fragments of discrete sizes in the
    four-channel emission spectrum facilitates the sequencing
    trace. (b) In the next-generation shotgun sequencing, common
    adaptors are ligated to fragmented genomic DNA. The DNA is treated
    to create millions of immobilized PCR (polymerase chain reaction)
    colonies, called polonies, each containing copies of a single
    shotgun library fragment. In cyclic reactions, sequencing and
    detection of fluorescence labels determine a contiguous
    sequencing read for each polony.  Reproduced
    from \cite{Shendure08}.}  \label{fig:back:ngs}
\end{figure}

The DNA fragmentation process in the shotgun sequencing approach can
be seen as a random sampling process, where each fragment comes from a
random position in the genome. The reads are then obtained from both
ends of each fragment, and usually do not cover the full DNA
fragment. Assuming an uniform sampling, it can be theoretically shown
(\cite{Lander88}) that in order for all bases of the genome to be
included in at least one read with high probability, a given amount of
over-sampling is needed. In other words, for the whole genome to be
sampled at least once, the average number of reads including a given
genomic position, that is the \emph {coverage}, should be higher than
one.

\subsubsection{RNA-seq}

\emph{RNA-seq} or \emph{transcriptome sequencing} is the process of
sequencing transcribed RNA using NGS technologies. Although the term
RNA-seq is usually applied to both mRNA and non-coding RNA (e.g. micro
RNAs) sequencing (\cite{Wang09}), in this work we use it exclusively
for mRNA sequencing. The basic RNA-seq protocol is very similar to
genomic NGS, differing only in the first two phases. The first step is
the extraction of mRNA from the cell, which is possible due to a
distinctive structural property of mRNAs, namely its poly-A
tail\footnote{Actually, there are mRNAs \emph{without} the poly-A tail
  (\cite{Yang11,Encode12}). They are not selected, and thus not
  sequenced, by the standard RNA-seq protocol.}. The second step is
the conversion to complementary DNA (cDNA), or reverse
transcription. Then, the cDNA library is sequenced using the same NGS
methods as for whole genome sequencing. The second step implies that
the result of an RNA-seq is usually not strand specific, i.e. the
reads obtained are a mixture of both the original strand of the mRNA
and its reverse complement. There are alternative RNA-seq protocols
(\cite{Levin10}) where the strand information is not lost, however in
this work we only consider the non strand-specific RNA-seq.

It is important to understand what new information can RNA-seq provide
with regard to genomic sequencing. In addition, of course, to the
information of which region is transcribed in the genome. As stated in
Section~\ref{sec:back:AS}, there is no one-to-one correspondence
between genes and proteins, and in particular mRNAs, which implies
that using only the genomic sequence, it is not possible to infer
which set of mRNAs is expressed. In addition, as stated in
Section~\ref{sec:back:AS}, the expression level may vary from gene to
gene, and this is reflected in the number of mRNA molecules
transcribed for each gene. In that way, the sampling process in the
fragmentation step of the NGS protocol is now a biased sampling
towards the more transcribed mRNAs, i.e. the coverage of the mRNA is a
proxy for its expression level. Therefore, the RNA-seq data brings at
least two new dimensions in comparison to NGS genomic data: the
variability of the mRNAs of a same gene and the expression level of
each mRNA. Since 2008, several works
(\cite{Cloonan08,Cloonan08b,Lister08,Nagalakshmi08,Mortazavi08,Sultan08,Trapnell10,Pickrell10,Peng12})
have used RNA-seq data to shed a new light into the dynamics of gene
expression in eukaryotic cells; see \cite{Cloonan08} and \cite{Wang09}
for comprehensive reviews.


\section{Mathematical Concepts} \label{sec:back:math}

\subsection{Sets, Sequences and Strings}
Given a set $X = \{x_1, \ldots, x_n\}$, the cardinality of $X$ is
denoted by $|X|$. The power set $2^X$ is the set of all subsets of
$X$, including the empty set. A \emph{sequence} $S$ is an ordered
multi-set and is denoted by $(s_1, \ldots, s_n)$. A \emph{subsequence}
of $S$ is a sequence obtained from $S$ by removing some elements,
without changing the order of the others. A \emph{prefix}
(\emph{suffix}) is the subsequence of $S$ obtained after removing
$(s_i, \ldots, s_n)$ ($(s_1, \ldots, s_i)$), where $1 \leq i \leq
n$. The \emph{length} of the sequence, denoted by $|S|$, is the
cardinality of the multi-set. The \emph{concatenation} of $S$ with an
element $x$ is the sequence $(s_1, \ldots , s_n, s_{n+1})$ with
$s_{n+1} = x$ and is denoted by $Sx$. Analogously, the concatenation
of two sequences $S_1,S_2$ is denoted by $S_1S_2$.
 
A \emph{string}\footnote{Although not strictly correct, we may use
  sequence when referring to strings.} is a sequence where each
  element belongs to a set $\Sigma$, the \emph{alphabet}. The prefix,
  suffix, concatenation and notation for the length are defined the
  same way as for sequences. The set of all strings over $\Sigma$ is
  denoted by $\Sigma^*$. An element or \emph{letter} of a string
  $w \in
\Sigma^*$ at the position $i$, where $1 \leq i \leq |w|$, is denoted
by $w[i]$. A \emph{substring} of $w$ is a contiguous subsequence of $w$
denoted by $w[i,j]$, where $1 \leq i \leq j \leq |w|$. A
\emph{$k$-mer} is a substring of length $k$. Given two strings $x,y
\in \Sigma^*$, the \emph{suffix-prefix overlap} or simply
\emph{overlap}, is the longest suffix of $x$ that is also a prefix of
$y$. The \emph{edit distance} for $x,y \in \Sigma^*$ is the minimum of
number edit operations -- substitutions, deletions or insertions --
necessary to transform $x$ in $y$ (or vice-versa, it is
symmetrical). The edit distance for $x,y \in \Sigma^*$ can be computed
in $O(|x||y|)$ time using dynamic programming (\cite{Cormen01}).

\subsection{Graphs} \label{sec:back:graphs}
A \emph{directed graph} $G$ is a pair of sets $(V,E)$ such that $E
\subseteq V^2$ is a set of ordered pairs. An element $v \in V$ is
called a \emph{vertex} of $G$, while an ordered pair $(u,v) \in E$ is
called an \emph{arc} of $G$. Given an arc $e = (u,v) \in E$, the
\emph{head} of $e$ is vertex $u$ and the \emph{tail} is $v$.  An
\emph{undirected graph} $G$ is a pair of sets $(V,E)$ such that $E
\subseteq V^2$ is a set of unordered pairs, i.e. $(u,v) = (v,u)$. An
unordered pair $(u,v) \in E$ is called an \emph{edge}. Given a
directed graph $G = (V,E)$, the \emph{underlying undirected graph} is
the undirected graph $G'$ obtained from $G$ disregarding the arc
directions. Whenever it is clear from the context, or we are referring
to both, we omit terms directed or undirected, saying simply
``graph''. Finally, the graphs considered here, unless otherwise
stated, are simple, that is, do not contain self-loops,
i.e. $(v,v) \notin E$, nor multiple edges, i.e. $E$ is a set not a
multi-set.

Given a directed graph $G = (V,E)$ and a vertex $v \in V$, the in and
out-neighborhoods of $v$ are denoted by $N^+(v)$ and $N^-(v)$,
respectively. For an undirected graph, $N^+(v) = N^-(v)$ and is
denoted by $N(v)$. The \emph{in} and
\emph{out-degree} of $v$ are $d^-(v) = |N^-(v)|$ and $d^+(v) = |N^-(v)|$,
respectively; for an undirected graph, the \emph{degree} of $v$ is
$d(v) = |N(v)|$.  A vertex $v \in V$ is a \emph{source} of $G$ if
$N^-(v) = \emptyset$; symmetrically, it is a \emph{sink} if $N^+(v)
= \emptyset$. The reverse graph of $G = (V,E)$, denoted by $G^R = (V,
E')$, is the directed graph obtained by reversing all arcs of $G$,
i.e. $E' = \{(u,v) | (v,u) \in E \}$. The line graph of a directed
graph $G$ is the directed graph $L(G)$ whose vertex set corresponds to
the arc set of $G$ and there is an arc directed from an arc $e_1$ to
an arc $e_2$ if in $G$, the head of $e_1$ meets the tail of $e_2$. See
Fig.~\ref{fig:line_graph} for an example.

\begin{figure}[htb]
  \center
  \subfloat[]{\includegraphics[width=6cm]{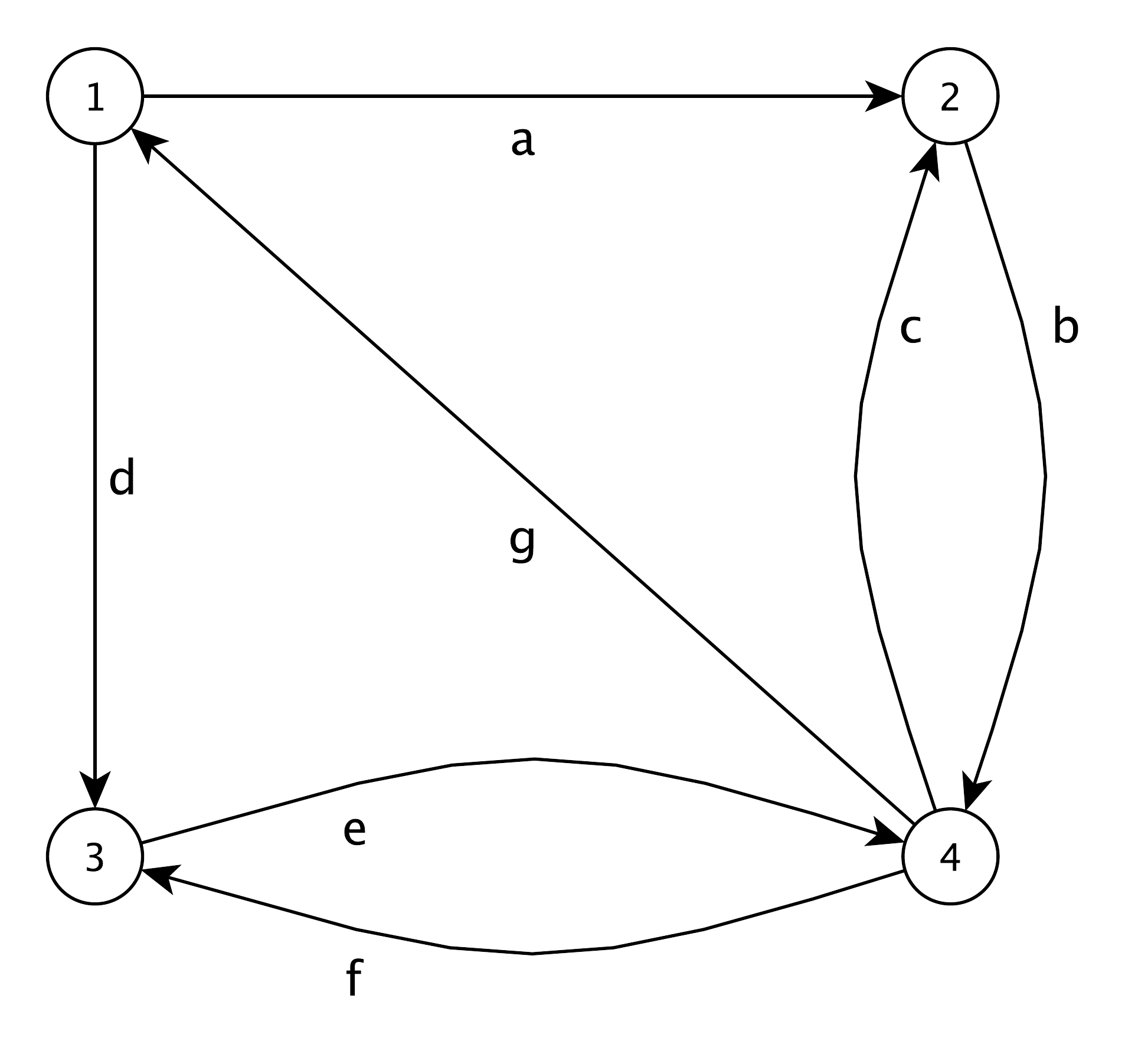}} \hspace{1cm} 
  \subfloat[]{\includegraphics[width=6cm]{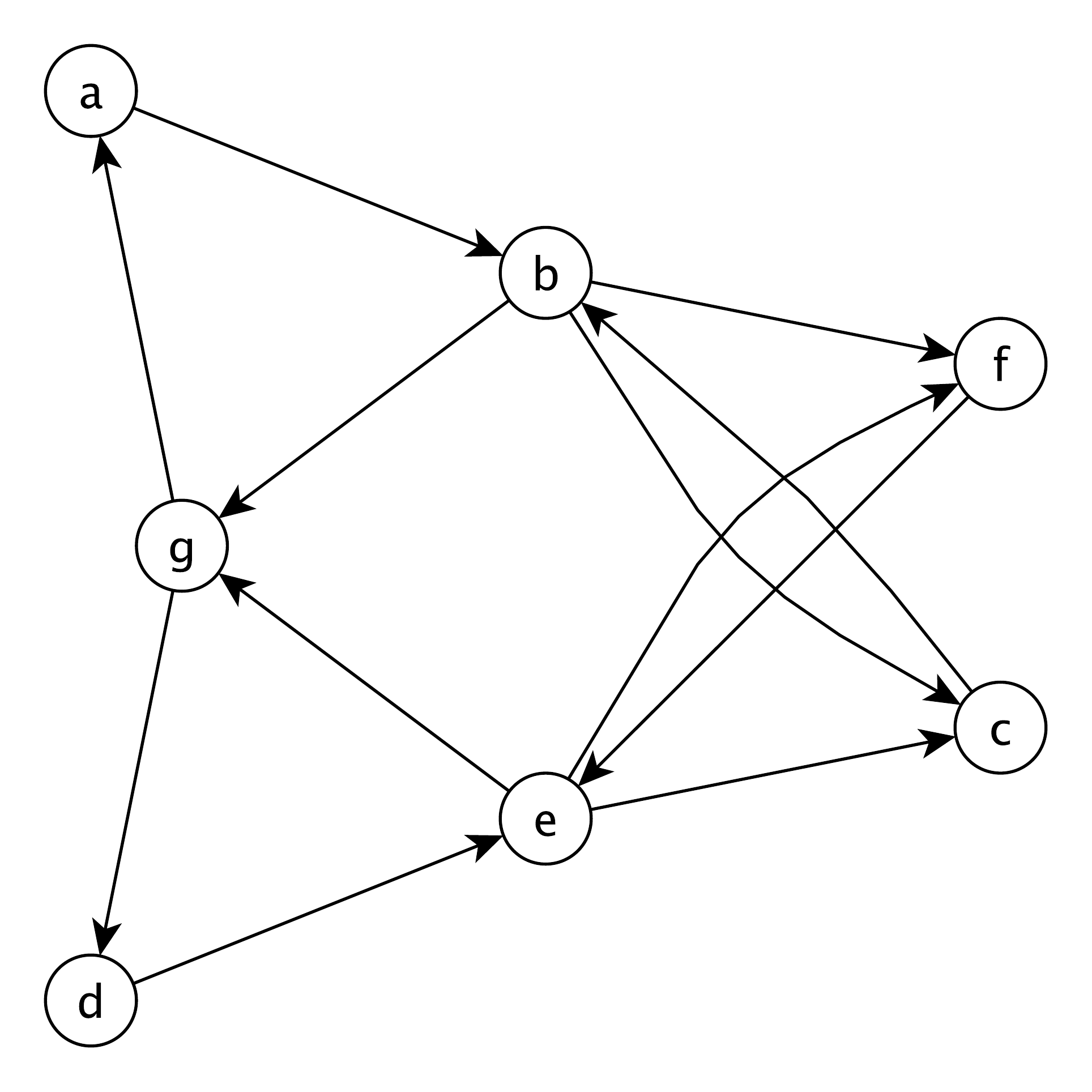}}
  \caption{A line graph example. (a) A directed graph $G$ with the
    vertices and arcs labeled. (b) The line graph $L(G)$, where the
    label of each vertex is the label of the corresponding arc in
    $G$. For example, arc $a$ in $G$ has arc $g$ entering $1$ and arc $b$
    leaving $2$, in $L(G)$ the vertex $a$ has $g$ as an in-neighbor
    and $b$ as an out-neighbor.}
  \label{fig:line_graph}
\end{figure}

A \emph{walk} in $G$ is a sequence of arcs or vertices $p = (v_1,
v_2) \ldots (v_{n-1}, v_n) = (v_1, v_2, \ldots, v_n)$, such that
$(v_{i-1}, v_{i}) \in E$ for $1 \leq i \leq n$. The first (last)
vertex of $p$ is called \emph{source} (\emph{target}).  A
(simple)\footnote{A path is, by definition, simple although we may say
``simple path'' to emphasize the fact that there are no duplicated
vertex.}  \emph{path} is a walk in which all vertices are distinct:
$v_i \neq v_j$ for all distinct $1 \leq i,j \leq n$. The path $p$ from
$s$ to $t$ is called \emph{$st$-path} and is denoted by $s \leadsto
t$, $p_{st}$ or $\pi_{st}$. A \emph{trail} is a walk in which all arcs
are distinct, i.e. $(v_{i-1},v_i) \neq (v_{j-1},v_j)$ for all distinct
$1 \leq i,j \leq n$. A subwalk of $p$ is a subsequence
$(v_{i-1},v_{i}) \ldots (v_{j-1}, v_j)$ of $p = (v_1,v_2) \ldots
(v_{n-1},v_n)$. It is not hard to prove that every walk $p$, such that
$v_1 \neq v_n$, contains a subwalk $p'$ with the same source and
target such that $p'$ is a path. A (simple)\footnote{A cycle is, by
definition, simple although we may say ``simple cycle'' to emphasize
the fact that there are no duplicated vertex.}  cycle $c = (v_1,v_2)
\ldots (v_{n-1},v_n)$ is a closed path, i.e. $v_1 = v_n$. A path or
cycle $p$ is \emph{Hamiltonian} if it includes all vertices of $G$. A
trail $p$ is \emph{Eulerian} if it includes all arcs of $G$.

A graph $H = (V_H, E_H)$ is a \emph{subgraph} of $G = (V,E)$ if $V_H
\subseteq V$ and $E_H \subseteq E$. The subgraph \emph{induced} by a
set of vertices $V' \subseteq V$ is the subgraph $G'=(V',E')$, where
$E'=\{ (u,v):\, (u,v)\in E,\, u,v\in V'\}$ which is denoted by
$G[V']$. The induced subgraph $G[ V\setminus \{u\}]$ for $u \in V$ is
denoted by $G - u$. Similarly for an edge $e \in E$, we adopt the
notation $G-e = (V,E \setminus \{e\})$, and for any $F\subseteq E$,
$G-F=(V,E \setminus F)$. Let $p$ be a walk of $G$; the induced
subgraph $G[V \setminus p]$ is denoted by $G-p$.

A \emph{weighted directed graph} $G = (V,E)$ is a directed graph with
weights $w : E \rightarrow Q$ associated to the arcs. The weight of a
walk $p = (v_1, v_2) \ldots (v_{n-1}, v_n)$ is the sum of the weights
of the arcs and is denoted by $w(p)$. The \emph{distance} from $s$ to
$t$, denoted by $d_G(s,t)$ (we drop the subscript when the graph is
clear from the context), is the weight of the shortest path from $s$
to $t$. In an unweighted graph $G$, the distance between two vertices
is equal to the distance in the corresponding weighted graph with all
the weights equal to one.

An undirected graph $G = (V,E)$ is \emph{connected} if, for any two
vertices $x,y \in V$, there exits a path from $x$ to $y$. A
\emph{connected component} is a maximal connected subgraph. Any
undirected graph can be uniquely partition into connected
components. The connected components of $G = (V,E)$ can be found using
any graph traversal algorithm; in particular, they can be computed in
$O(|V| + |E|)$ time with a depth-first search (DFS) or breadth-first
search (BFS) (\cite{Cormen01}). A directed graph $G = (V,E)$
\emph{strongly connected} if for any $x,y \in V$ there exists the
paths $x \leadsto y$ and $y \leadsto x$. It is \emph{weakly connected}
if its underlying undirected graph is connected. Of course, any
strongly connected graph is also weakly connected

\subsubsection{Trees}
A connected acyclic undirected graph $G$ is called a \emph{(unrooted)
  tree}. A \emph{rooted tree} $T$ is a tree with a special vertex $r$
called \emph{root}. The \emph{parent} of a vertex $v$ in $T$ is the
neighbor of $v$ closer to the root. Every vertex, except the root, has
a unique parent. The root has no parent. A \emph{child} of $v$ is a
vertex of which $v$ is the parent. Intuitively, a rooted tree is a
tree where the edges are directed away from the root. The set of all
children of $v$ is denoted by $N^+(v)$. A vertex $w$ is an
\emph{ancestor} of $v$ if it belongs to the path $v \leadsto
r$. Conversely, $w$ is a \emph{descendent} of $v$ if $v$ belongs to $w
\leadsto r$. The descendent or ancestor is \emph{proper} if it is
different from $v$. A \emph{subtree} of $T$ rooted at $v$, denoted by
$T_v$, is the subgraph of $T$ induced by all descendents of $v$, which
is also a tree, with root at $v$. A \emph{leaf} is a vertex without
any children. The \emph{depth} of a vertex is the length of its unique
path to the root. The \emph{height} of a vertex is the length of the
longest downward path to a leaf from that vertex.

\subsubsection{Biconnected Graphs}
An undirected graph $G = (V,E)$ is \emph{biconnected} if it is
connected and for any $x \in V$ the graph $G - x$ is still
connected. Generalizing the definition of a connected graph, an
undirected graph $G = (V,E)$ is \emph{$2$-connected} (or
\emph{$2$-vertex-connected}) if for any $x,y \in V$ there exist two
internally vertex-disjoint paths from $x$ to $y$. By Menger's theorem
(\cite{Diestel05}), the two definitions are equivalent, except when
$G$ is a single vertex\footnote{We are using the convention that the
  \emph{null} graph, i.e. the graph containing no vertices, is
  connected.} or a single edge; in those cases the graphs are
biconnected but not $2$-connected. Before giving a simple
characterization of the structure of $2$-connected graphs
(Lemma~\ref{lem:back:ear}), we need another definition. Given an
undirected graph $H$, a path $p$ is an \emph{$H$-path} or \emph{ear}
of $H$ if $p$ meets $H$ exactly at its endpoints, i.e. the only
vertices of $p$ in common with $H$ are its endpoints.

\begin{lemma}[\cite{Diestel05}] \label{lem:back:ear}
  A graph is $2$-connected if and only if it can be constructed from a
  cycle by successively adding $H$-paths to the graphs $H$ already
  constructed.
\end{lemma}

This process of adding $H$-paths to construct a $2$-connected graph is
also known as an \emph{ear decomposition} (\cite{Bang-Jensen08}). A
similar characterization, using directed ears, can be stated for
strongly connected graphs. Biconnected graphs have other interesting
properties. For instance, given a biconnected graph $G = (V,E)$ and
three distinct vertices $x,y,z \in V$, there is a $xy$-path passing
through $z$. Indeed, let us construct $G'$ by adding a new vertex $w$
and the edges $(x,w)$, $(y,w)$ to $G$; the graph $G'$ is also
biconnected, so it contains two vertex-disjoint paths $p_1, p_2$ from
$w$ to $z$, one passing through $x$ and the other through $y$. Thus,
since $p_1, p_2$ are vertex-disjoint, the concatenation contains a
path in $G$ from $x$ to $y$ passing through $z$. Using a similar
argument, we can also prove the following. For any distinct $x,y \in
V$ and an edge $e \in E$, there is a $xy$-path passing through $e$ in
$G$.

Similarly to connected components, for any undirected graph, a
\emph{biconnected component} (BCC) is a maximal biconnected
subgraph. An \emph{articulation point} or \emph{cut vertex} is a
vertex such that its removal increases the number of connected
components. A biconnected component decomposition uniquely defines a
partition on the edges, but not on the vertices. In other words, two
distinct BCCs may share vertices but not edges.  Actually, as stated
in Lemma~\ref{lem:back:bcc}, a decomposition into BCCs can be defined
as an equivalence relation for the edges, where each equivalence class
is a BCC, and the common vertices are exactly the articulation points.

\begin{lemma}[\cite{Tarjan72}] \label{lem:back:bcc}
  Let $G = (V,E)$ be an undirected graph. We define an equivalence
  relation on the set of edges as follows: two edges are equivalent if
  and only if they belong to a common cycle. Let the distinct
  equivalence classes under this relation be $E_i$, $1 \leq i \leq l$,
  and let $B_i = (V_i, E_i)$, where $V_i$ is the set of edges incident
  to $E_i$ in $G$. Then:
  \begin{enumerate}
    \item $\{B_1, B_2, \ldots, B_l \}$ is the set of biconnected
      components of $G$;
    \item Each articulation point of $G$ occurs more than once among
      the $V_i$, $1 \leq i \leq l$. 
    \item Each non-articulation point of $G$ occurs exactly once among
      the $V_i$, $1 \leq i \leq l$;
    \item The set $V_i \cap V_j$ contains at most one vertex, for any
      $1 \leq i,j \leq l$. Such vertex is an articulation point of
      $G$.
  \end{enumerate}
\end{lemma}

As a corollary, we have that a BCC decomposition is also a partition
on the cycles of $G$, i.e. every cycle is contained in exactly one
biconnected component. In addition, for any connected graph $G$, we
have that the BCCs form a tree-like structure, the \emph{block tree},
where two BCCs are adjacent if they share an articulation point. See
Fig.~\ref{fig:back:block_tree} for an example. More precisely, let $A$
be the set of articulation points of $G$ and $\mathcal{B}$ its set of
biconnected components. Then, consider the graph $\mathcal{T}$ whose
vertices are $A \cup \mathcal{B}$ and there is an edge from $a \in A$
to $B \in \mathcal{B}$ if $a \in B$ (there are no edges between two
vertices of $A$ or $\mathcal{B}$, i.e. it is a bipartite graph). The
graph $\mathcal{T}$ is a tree.

\begin{figure}[htb]
  \vspace{0.5cm}
  \centering
    \begin{tikzpicture}
[nodeDecorate/.style={shape=circle,inner sep=1pt,draw,thick,fill=black},%
  lineDecorate/.style={-,dashed},%
  elipseDecorate/.style={color=gray!30},
  scale=0.7]

\draw (5,10) circle (2);
\draw (9,10) circle (2);
\draw (2,10) circle (1);
\draw (-1,10) circle (2);
\draw (5,7) circle (1);
\draw (9,7) circle (1);
\draw [rotate around={55:(-1,8)}] (-1,7) circle (1);
\draw [rotate around={-55:(-1,8)}] (-1,7) circle (1);
\draw (13,10) circle (2);
\draw (13,7) circle (1);
\draw (-4,10) circle (1);
\begin{footnotesize}
\node (7) at (-2,7) [nodeDecorate] {};
\node (14) at (9,11) [nodeDecorate] {};
\end{footnotesize}
\foreach \nodename/\x/\y in {
  0/7/10, 1/5/11,
  2/3/10, 3/1/10, 4/-3/10, 5/-1/10, 6/-1/8, 7/-2/7, 8/-1/11,
  9/0/7,
  11/5/10, 12/4/9, 13/5/8, 14/9/11, 15/11/10, 16/9/9,
  17/9/7 , 18/9/8, 50/5/7, 51/13/11, 52/13/8, 53/13/7, 54/-4/10}
{
  \node (\nodename) at (\x,\y) [nodeDecorate] {};
}

\path
\foreach \startnode/\endnode in {6/7, 6/9, 5/6, 5/3, 5/8, 4/5, 3/2,
2/11, 1/11, 11/12, 11/0, 11/13, 13/50, 0/14, 14/15, 15/16, 16/18,
18/17, 15/51, 15/52, 51/52, 52/53, 54/4}
{
  (\startnode) edge[lineDecorate] node {} (\endnode)
};

\path
\foreach \startnode/\endnode/\bend in { 8/3/20, 6/3/20, 4/8/10,
12/13/20, 1/0/20, 2/1/20, 13/0/20, 0/18/10}
{
  (\startnode) edge[lineDecorate, bend left=\bend] node {} (\endnode)
};

\end{tikzpicture}
  \caption{An example of connected graph $G$ with its biconnected
    components highlighted. The articulation points are precisely the
    nodes in the intersection of the circle (BCCs). The circles plus
    the intersections form the block tree of $G$.}
  \label{fig:back:block_tree}
\end{figure}
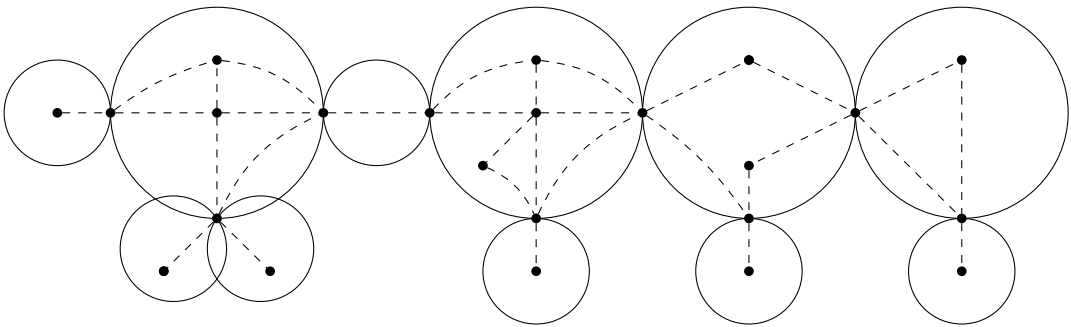

The biconnected components of $G = (V,E)$ can be computed in $O(|V| +
|E|)$ using a modified DFS (\cite{Tarjan72,Cormen01}). The same
algorithm can also be used to find all the articulation points of $G$
in linear time.

\subsection{Modeling and assembling NGS data} \label{sec:back:modeling}
As stated in Section~\ref{subsec:ngs}, the NGS technology when applied
to a genome (or a transcriptome, in the case of RNA-seq) produces a
large number of fragments, or reads, from the original sequences. In
this context, the most natural question is, given the set of reads,
how to reconstruct the original sequence by combining the reads. This
is called genome (transcriptome) assembly problem. At this point, it
is important to make a clear distinction between \emph{de novo} genome
(transcriptome) assembly, which aims to reconstruct the genome
(transcriptome) without using any other information but the reads, and
comparative (re-sequencing) approaches that use knowledge on the
genome of a closely related organism to guide the reconstruction. The
\emph{de novo} assembly problem is, as we show in this section,
NP-hard under three common formalizations. On the other hand,
comparative assembly is a considerably easier task admitting a
polynomial algorithm, basically, it is sufficient to map the reads
back to the reference genome (transcriptome)
(\cite{Flicek09,Pop09}). Provided there exists a close enough
reference, otherwise a mixture of both problems could be
considered. Our main interest is in \emph{de novo}
assembly. Hereafter, we omit the term \emph{de novo} when referring to
it.

One basic assumption commonly made when modeling the assembly problem
is that every read in the input must be present in the original genome
(transcriptome). This neglects the fact that the reads may contain
errors.  Under this hypothesis, the genome (transcriptome) assembly
problem can be formally stated as, given a set of strings $\mathcal{R}
\subset \Sigma^* = \{A,C,T,G\}^*$, such that $r \in \mathcal{R}$ is a
substring of an unknown string $S \in \Sigma^*$ (set of strings $S
\subset \Sigma^*$), reconstruct the original string $S$ (set of
strings $S$). From now on, for the sake of a clear exposition, we
consider only the genome assembly problem, and in the end of the
section we highlight the differences with transcriptome assembly.

A simple way to formulate the assembly problem as an optimization
problem, i.e. a problem of maximizing or minimizing a given objective
function, is to require the reconstructed string to be of minimal
length. Formally, given a set of strings $r \in \mathcal{R}$, find the
minimum length string $S^*$ such that every $r \in \mathcal{R}$ is a
substring of $S^*$. This is precisely the \emph{shortest superstring
  problem}, which is known to be NP-hard for $|\Sigma| \geq 2$
(\cite{Garey79}). Despite that fact, some assemblers
(\cite{Ssake,Sharcgs}) employed this formulation. Of course, they do
not solve the shortest superstring problem exactly; an exponential
algorithm would indeed be impractical for all, but very small,
instances. Instead, they employ variations of the following iterative
greedy heuristic: at a given step the algorithm maintains a
superstring $S'$ for a subset $\mathcal{R'} \subseteq R$, then extends
$S'$ with the read $r \in \mathcal{R} \setminus \mathcal{R'}$ such
that the suffix-prefix overlap with $S'$ is maximum and then adds $r$
to $R'$.

\begin{figure}[htbp]
  \centering 
  \includegraphics[width=0.9\linewidth]{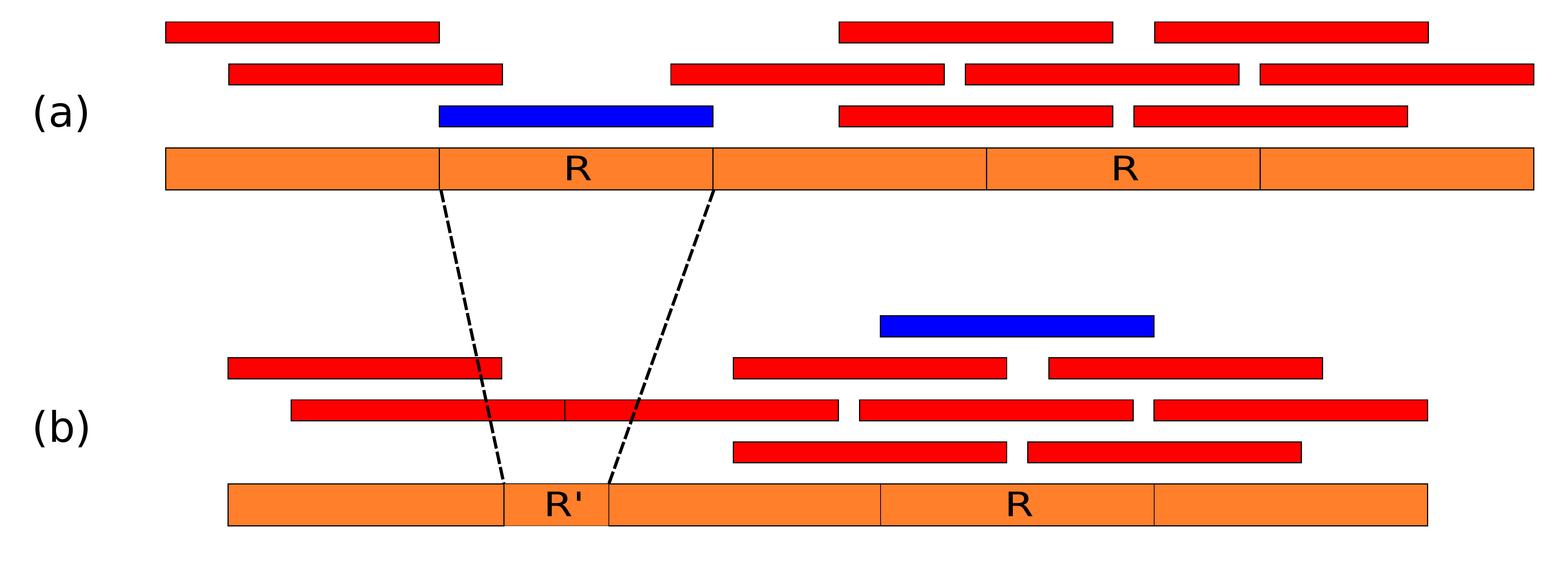} 
  \caption{An example of a shortest superstring with an over-colapsed
    repeat. (a) The original string contains two exact copies of $R$;
    the set of reads is shown abave it, and the read in blue is
    entirely contained in one of the copies of $R$. (b) The shortest
    superstring for the same set of reads; the first copy $R'$ is
    \emph{over-collapsed}, and the blue read is now assigned to the
    second copy of $R$.}
 \label{fig:greedy_counter_ex}
\end{figure}

There are two main problems with the greedy strategy. The first issue
is mainly due to the problem formulation: requiring the superstring to
be of minimal length, although motivated by parsimony, is in the best
case questionable. The reason is that the majority of the genomes have
repeats, i.e. multiple identical, or nearly identical, substrings,
while requiring a superstring of minimum length tends to over-collapse
these substrings in the obtained solution. Consider for instance the
example shown in Fig.~\ref{fig:greedy_counter_ex}. The second problem
is due to the inherent local nature of the greedy heuristic: the
choices are iteratively made without taking into account the global
relationships between the reads. It is likely that in the true
solution, the genome from which the reads were generated, several
suffix-prefix read overlaps are not locally optimal. In order to
address these issues, actually more the second one than the first, two
high-level strategies were proposed (\cite{Pop09}): the
overlap-layout-consensus (OLC) and the Eulerian path. In the core of
each strategy is a representation of the reads set $\mathcal{R}$ in
terms of a (weighted) directed graph, the overlap graph for the OLC
and the de Bruijn graph for the Eulerian path. The graph
representation of $\mathcal{R}$ allows for non-local analysis of the
reads which is not possible with the greedy strategy.

\subsubsection{Overlap-layout-consensus (OLC) Strategy}
The overlap-layout-consensus strategy divides the assembly problem
into three major stages. In the overlap stage, similarly to the
approach used in the greedy strategy, for each pair of reads in
$\mathcal{R}^2$ the maximal suffix-prefix overlaps are computed. In
the layout stage, the overlap graph is constructed. That is, a
complete weighted directed graph where each read of $\mathcal{R}$ is a
vertex, and there is a directed edge between any pair of reads $(r_1,
r_2)$ with weight equal to the length of the suffix-prefix overlap
between $r_1$ and $r_2$. The formal definition is given below
(Definition~\ref{def:overlap_graph}). Next, still in the layout stage,
the overlap graph is simplified. Finally, in the consensus stage, the
genome sequence is obtained as consensus sequence, corresponding to a
path or walk in the simplified overlap graph.

\begin{definition}[Overlap Graph] \label{def:overlap_graph}
  Given a set of reads $\mathcal{R} \subseteq \Sigma^*$, the
  \emph{overlap graph} $G(\mathcal{R})= (V,E)$, $w : E \rightarrow
  \mathbf{N}$ is a complete weighted directed graph such that:
  
  \begin{enumerate}
    \item $V = \mathcal{R}$ and $E = \mathcal{R}^2$;
    \item $w(u,v) = $ length of the maximal suffix of $u$ that is
      equal\footnote{For simplicity, we do not consider the more
        general definition where a small number of mismatches is
        allowed for the suffix-prefix overlap.} to a prefix of $v$.
  \end{enumerate}
\end{definition}

The main goal of the graph simplification in the layout stage is to
reduce the complexity of the overlap graph. The first step is usually
to remove all arcs that have weights below a given threshold. In
practice, those edges are not even added to the original graph. A
possible way to further reduce the complexity, proposed by
\cite{Myers05}, is to perform a transitive reduction in the graph,
that is, to remove from the graph all the edges that are transitive
inferable, i.e. consider the edges $(x,y)$, $(y,z)$ and $(x,z)$, the
last edge $(x,z)$ is transitive inferable since there is still a path
from $x$ to $z$ after removing $(x,z)$. The subgraph of the overlap
graph obtained by this process is called \emph{string graph}
(\cite{Myers05,Medvedev07,Pop09}).

In the consensus stage, the problem of finding a walk in the graph
corresponding to a consensus sequence can be formulated as an
optimization problem by considering a constrained walk in the string
graph and requiring it to be of minimum length
(\cite{Medvedev07}). The walk is constrained in the sense that some
arcs of the string graph should be present at least once and others
exactly once. Formally, we have a selection function $s$ that
classifies the arcs of string graph in the three categories:
\emph{optional} (no constraint), \emph{required} (present at least
once) and \emph{exact} (present exactly once). The rationale for this
classification is that some portions of the graph correspond to
repeats in the genome, implying that they should be present more than
once in the consensus sequence (walk), whereas other correspond to
unique sequences, that should be present exactly once. This arc
classification can be computed using the A-statistics (\cite{Celera}),
as shown in \cite{Myers05}.  For a given selection function $s$ and a
string graph $G$, a walk of $G$ respecting $s$ is called an
\emph{$s$-walk}. \cite{Medvedev07} showed, using a reduction from
Hamiltonian path, that the problem of finding a minimum length
$s$-walk is NP-hard.

Similarly to the shortest superstring problem, despite the fact that
the minimum $s$-walk problem is NP-hard, several assemblers employed
this formulation, using diverse heuristics to simplify the graph,
i.e. make it as linear as possible, and find the consensus
sequence. The OLC strategy was used by assemblers for various whole
genome shotgun sequencing technologies, not only NGS technologies, for
instance, the Celera assembler (\cite{Celera}), Arachne2
(\cite{Arachne2}) and Cap3 (\cite{Cap3}) for Sanger reads; Newbler
(\cite{Newbler}) and Cabog (\cite{Cabog}) for 454 reads; Edena
(\cite{Edena}) and Shorty (\cite{Shorty}) for Illumina reads. For
longer reads, i.e. Sanger and 454, the OLC seemed to be the more
suitable approach (\cite{Pop09}). However, for shorter reads and much
deeper coverages, the overlap computation step becomes a computational
bottleneck. For that reason, most of the more recent assemblers use
the Eulerian path strategy. With the notable exception of SGA
(\cite{SGA}) where they manage to overcome the overlap computation
bottleneck using a FM-index (\cite{Ferragina05}), which is a full-text
compressed index based on the Burrows-Wheelers transformation allowing
for fast substring queries.

\subsubsection{Eulerian Path Strategy}

The de Bruijn graph of order $k \in \mathbf{N}$ of a set of reads
$\mathcal{R}$ is a directed graph where each $k$-mer present in
$\mathcal{R}$ corresponds\footnote{From now on, when considering de
  Bruijn graphs, we make no distinction between the $k$-mer
  corresponding to a vertex and the vertex itself.} to a vertex and
there is an arc $(u,v)$ if the $k$-mers corresponding to $u$ and $v$
share a suffix-prefix overlap of size $k-1$ and the corresponding
$(k+1)$-mer, the $k$-mer $u$ concatenated with the last character of
$v$, is present in $\mathcal{R}$. The formal definition is given below
(Definition~\ref{def:dbg}). Actually, this is a \emph{subgraph} of the
de Bruijn graph under its classical combinatorial definition
(\cite{Bang-Jensen08}).  However, following the terminology common to
the bioinformatics literature, we still call it a de Bruijn graph. One
of the most important aspects of de Bruijn graphs is that, unlike
overlap graphs, they are not subjected to the overlap computation
bottleneck.  De Bruijn graphs can be efficiently computed using
hashing or sorting. Indeed, given a read set $\mathcal{R}$, we can
build a de Bruijn graph $G_k(\mathcal{R})$ using a hash table
(\cite{Cormen01}) to store all $(k+1)$-mers present in
$\mathcal{R}$. As each insertion and membership query in the hash
table takes $O(1)$ (expected) time, the de Bruijn graph can be built
in time linear in the size of $\mathcal{R}$, i.e. $O(\sum_{r \in
  \mathcal{R}} |r|)$.

\begin{definition}[De Bruijn Graph] \label{def:dbg}
  Given a set of reads $\mathcal{R} \subseteq \Sigma^*$ and a
  parameter $k \in \mathbf{N}$, the \emph{de Bruijn graph}
  $G_k(\mathcal{R})= (V,E)$ is a directed graph such that:
  
  \begin{enumerate}
  \item $V = $ the set of $k$-mers of $\mathcal{R}$;
  \item $E = $ the set of $(k+1)$-mers of $\mathcal{R}$, in the sense
    that, given a $(k+1)$-mer $e$ of $\mathcal{R}$, we have that $e =
    (u,v)$, where $u = e[1,k]$ and $v = e[2,k+1]$.
  \end{enumerate}
\end{definition}

Although not apparent from their definitions, intuitively, a de Bruijn
can be seen as a special case of the overlap graph, where all the
reads were further divided in $k$-mers and all the suffix-prefix
overlaps have length exactly $k-1$.  Indeed, the arcs in the overlap
and de Bruijn graph represent the same structure, a suffix-prefix
overlap between the strings corresponding to the vertices. In fact, in
the particular case where all the reads of $\mathcal{R}$ have length
exactly $k+1$, the line graph of $G_k(\mathcal{R})$ is exactly the
overlap graph of $\mathcal{R}$ with the arcs of weight zero
removed. Moreover, in a de Bruijn graph there is a loss of information
with regard to the overlap graph: in de Bruijn graphs we do not have
the information that two $k$-mers came from the same read. As a
consequence there are walks in the de Bruijn graph that are not
\emph{read coherent}, i.e. are not entirely covered by an ordered set
of reads where two adjancent reads have a non-empty suffix-prefix
overlap (a tilling of the reads). An example of a de Bruijn and an
overlap graph built from the same set of reads is shown in
Fig.~\ref{fig:dbg_overlap}.

\begin{figure}[htb] 
  \centering 
  \includegraphics[width=15cm]{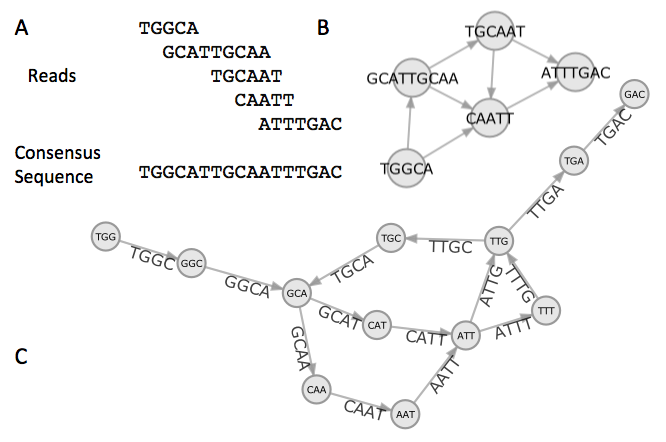} 
  
  \caption{(a) An example of a read set $\mathcal{R} = \{$TGGCA,
    GCATTGCAA, TGCAAT, CAATT, ATTTGAC$\}$ from the genome
    TGGCATTGCAATTGAC. (b) The overlap graph of $\mathcal{R}$ with the
    zero weight arcs not represented is shown: each vertex is labeled
    with the sequence of the corresponding read. (c) The de Bruijn
    graph of $\mathcal{R}$ with $k = 3$ is shown, each vertex is
    labeled with the sequence of the corresponding $k$-mer and each
    arc with the corresponding $(k+1)$-mer. Reproduced from
    \cite{Phast13}.}  \label{fig:dbg_overlap}
\end{figure}

Interestingly, de Bruijn graphs were first used in computational
biology in the context of sequencing by hybridization (SBH)
(\cite{Pevzner89}). The outcome of a SBH experiment is the set of all
distinct substrings of size $k$ in the original sequence. Years later,
it re-appeared in a pre-NGS context as an alternative to the OLC that
could potentially lead to a polynomial algorithm for the genome
assembly problem, although no such algorithm was provided
(\cite{Pevzner01}). The intuition was that differently from the OLC
strategy that models the genome assembly problem as special case of
the Hamiltonian path problem where the goal is to visit all vertices
in the graph, genome assembly in a de Bruijn graph could be modelled
as an Eulerian path (trail) problem, where the goal is to visit all
\emph{arcs} of the graph, for which there are polynomial algorithms
(\cite{Cormen01}). Unfortunately, there can be an exponential number
of Eulerian trails in a graph (\cite{Diestel05}) and in order to
select the one corresponding to the original sequence it is necessary
to impose some constraints to the Eulerian trail, resulting in an
NP-hard problem (\cite{Medvedev07}).

As with the shortest superstring formulation for genome assembly, it
is natural to require that all reads should be substrings of the
solution of the genome assembly problem. In order to transpose this to
the de Bruijn graph context, we observe that every read $r \in
\mathcal{R}$ corresponds to a walk in the de Bruijn graph
$G_k(\mathcal{R})$, possibly containing repeated vertices and
arcs. This means that the solution should be a walk $S$ in the de
Bruijn graph $G_k(\mathcal{R})$, such that the each walk $r_w$
corresponding to a read $r \in \mathcal{R}$ is a subwalk of $S$,
i.e. $S$ is a \emph{superwalk} of $G_k(\mathcal{R})$. Now, motived by
parsimony, the optimization problem can be formulate as the problem of
finding a minimum length superwalk in $G_k(\mathcal{R})$. Using a
reduction from the shortest superstring problem, \cite{Medvedev07}
showed this problem is NP-hard.

The efficiency of a hash-based approach to construct a de Bruijn graph
made it the ideal structure to represent NGS reads as the throughput
of new technologies continued to increase. This is clear as the
majority of the recent genome assemblers, although not trying to solve
the minimum superwalk problem exactly (as with the OLC approaches,
several heuristics to linearize the graph are employed), use de Bruijn
graphs, namely, in chronological order: Euler-SR (\cite{Eulersr}),
Velvet (\cite{Velvet}), ABySS (\cite{Abyss}), Allpaths
(\cite{Allpaths}), SOAPdenovo (\cite{Soapdenovo}), IDBA (\cite{Idba})
and SPAdes (\cite{Spades}). The main steps of a de Bruijn graph based
assembler are shown in Fig.~\ref{fig:assemblers_heuristic}.
 
\begin{figure}[htp]
  \centering
  \includegraphics[width=0.8\linewidth]{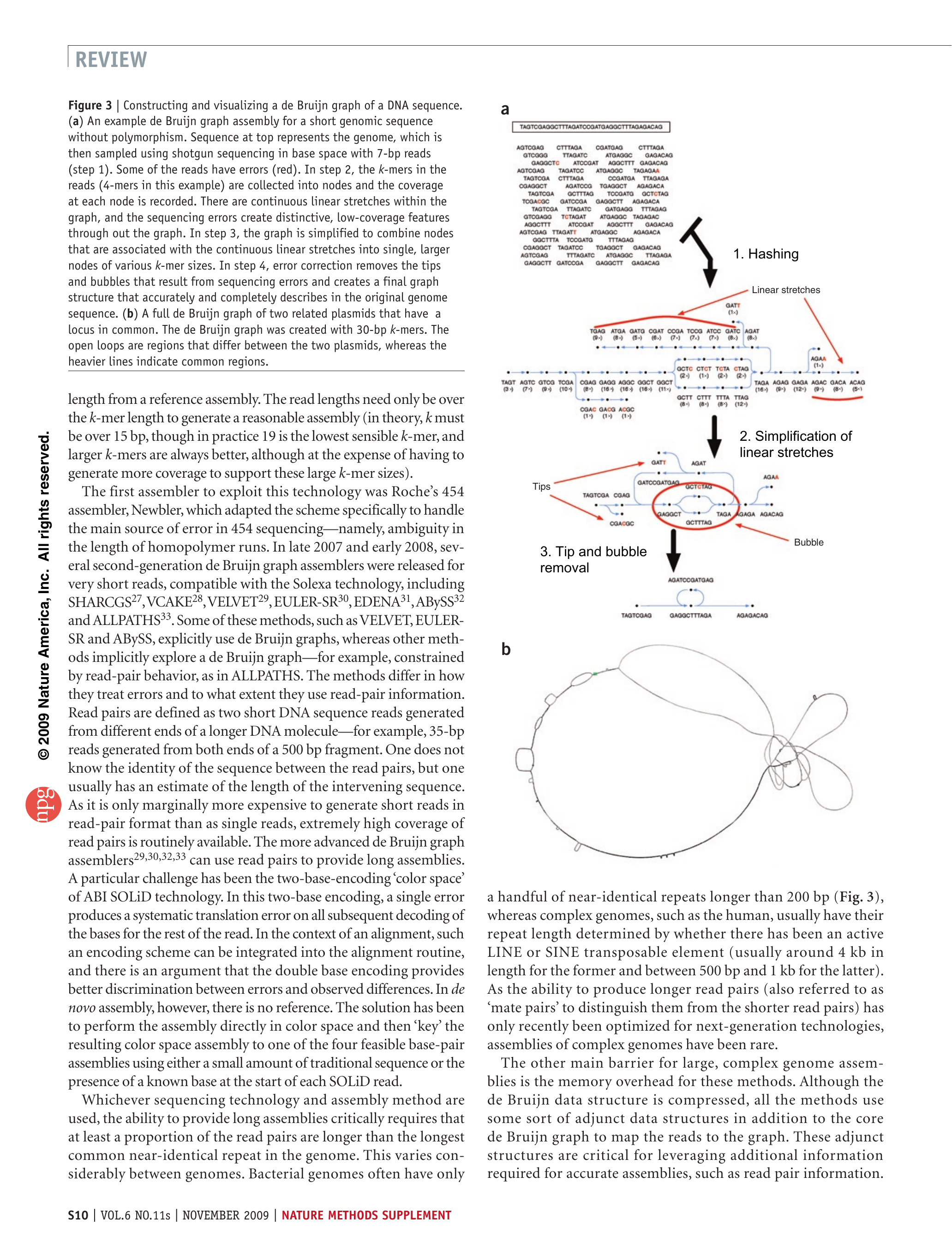} 
  \caption{The three main steps of an assembler based on a de Bruijn
    are shown. (1) The de Bruijn graph is built from the set of reads
    using a hashing-based approach. (2) Lossless graph simplification,
    the linear stretches (non-branching paths) of the graph are
    compressed. (3) Lossy graph simplification, the graph is further
    linearized by removing tips (dead-ends) and bubbles (alternative
    paths). Reproduced and modified from \cite{Flicek09}.}
  \label{fig:assemblers_heuristic}
\end{figure}

\subsubsection{Transcriptome Assembly}
In terms of mathematical formulation, the main difference between
genome and transcriptome assembly is the number of sequences
reconstructed. Indeed, the goal of a genome assembly problem is to
reconstruct one\footnote{This is a simplification used in the
  theoretical models for genome assembly (\cite{Medvedev07}). A genome
  is usually composed of several chromossomes, that is in genome
  assembly, similarly to transcriptome assembly, several strings
  should be reconstructed. There is, however, an important difference
  in scale when compared to transcriptomes: the majority of the known
  species have less than 100 chromossomes, whereas the number of
  transcripts is several orders of magnitude larger.} string (the
genome), whereas the goal of a transcriptome assembly is to
reconstruct a set of strings (the set of transcripts). The three
formulations for the genome assembly problem as optimization problems,
shortest superstring, minimum $s$-walk / superwalk, can easily be
generalized in such a way that the solution is a set of strings in the
first case, or a set of walks in the corresponding graph in the last
two cases. For instance, the minimum superwalk in $G_k(\mathcal{R})$
can be generalized to the problem of finding the set of walks
$\mathcal{S}$, such that $\alpha \leq |\mathcal{S}| \leq \beta$,
i.e. $\mathcal{S}$ contains at least $\alpha$ \emph{non-empty} walks
and at most $\beta$; each read $r \in \mathcal{R}$ is a subwalk of
some $s \in \mathcal{S}$; and the total sum of the lengths of the
walks in $\mathcal{S}$ is minimum. The bounds $\alpha, \beta$ are part
of the input of the problem and reflect the expected number of
transcripts. Now, by choosing $\alpha = \beta = 1$ we obtain exactly
the minimum superwalk problem, thus the generalization is also
NP-hard. The same holds for the other two generalizations.

The majority of the transcriptome assemblers use de Bruijn graphs to
represent the set of reads, for instance, in chronological order:
Trans-ABySS (\cite{Transabyss}), Trinity (\cite{Trinity}), Oases
(\cite{Oases}) and IDBA-tran (\cite{Idbatran}). As in the genome
assembly case, in practice the transcriptome assemblers do not attempt
to solve the generalized minimum superwalk problem exactly, employing
instead various heuristics. There are, however, important differences
in the heuristics used in both cases. Unlike the heuristics for genome
assemblers where the main goal of the heuristics is to simplify the
graph by linearizing it, the heuristics for transcriptome assembly
have three main steps (\cite{Trinity,Oases}): 

\begin{enumerate}
  \item \textbf{Graph simplification.} This step is very similar to
    the de Bruijn graph simplification in genome assembly (see
    Fig.~\ref{fig:assemblers_heuristic}); the goal is the same, remove
    branching structures, that ideally correspond to sequencing
    errors, to transform the graph in a path. However, since
    alternative isoforms also produce branching structures, compared
    to genome assembly this simplification is done in a less
    aggressive way.
 \item \textbf{Graph partition.} In the ideal case where two distinct
   genes do not share any $k$-mer (vertex), each connected component
   of the graph corresponds to the set of alternative isoforms of each
   gene. Unfortunately, genomes contain repeats, so two unrelated
   genes may share $k$-mers. The goal of this step is then to deal
   with these $k$-mers (vertices) linking two genes (connected
   components), in such a way that allows for the graph to be
   partition in subgraphs corresponding to genes.  In the case of
   \cite{Oases} this is, approximately, done by identifying the
   vertices corresponding to repeats and, for each vertex, duplicating
   it and dividing the arcs among the two copies.
 \item \textbf{Path decomposition.} The goal of this step is, for each
   subgraph (gene) obtained in the last step: find the set of paths
   corresponding to the set of alternative isoforms of the gene. In
   the case of \cite{Oases} this done by first applying an heuristic
   to remove cycles, and then iteratively applying a dynamic
   programming algorithm to find the path with largest read support
   (coverage).
\end{enumerate}

We should stress that this is an \emph{heuristic}, there are no
guarantees that all transcripts are going to be (correctly)
assembled. For instance, in the graph simplification step, even using
a less aggressive approach, there are no guarantees that only
sequencing errors are removed, actually it is quite likely that some
alternative isoforms are also removed. Moreover, in the graph
partition step, a gene can be split in two or more subgraphs, that
way, in the path decomposition step, the resulting isoforms are
necessarily fragmented. On the other hand, if two unrelated genes are
in the same subgraph, in the path decomposition step, the resulting
paths can be \emph{chimeras}, containing parts of two unrelated
transcripts. In general, transcripts from highly expressed genes are
better assembled than lowly expressed ones; within the same gene,
dominant isoforms are better assembled than minor ones.

\subsection{Enumeration Algorithms}

In this section, which is based on our paper \cite{Marino14}, we give
a brief introduction to the area of an enumeration algorithms
area. Naturally, the goal of enumeration is to list all feasible
solutions of a given problem. For instance, given a graph $G = (V,E)$,
listing the paths or shortest paths from a vertex $s \in V$ to a
vertex $t \in V$, enumerating cycles, or enumerating all feasible
solutions of a knapsack problem, are classical examples of
\emph{enumeration problems} or \emph{listing problems}. An algorithm
to solve an enumeration problem is called \emph{enumeration algorithm}
or \emph{listing algorithm}.

While an optimization problem aims to find just \emph{the best}
solution according to an objective function, an enumeration problem
aims to find \emph{all} solutions satisfying a given set of
constraints. This is particularly useful when the data is incomplete
or the objective function is not clear: in these cases the best
solution should be chosen among the results obtained by
enumeration. Moreover, enumeration algorithms can be also applied to
solve exactly NP-hard problems, by listing all feasible solutions and
choosing the best one, as well as counting the number of feasible
solutions in \#P-hard problems (\cite{Valiant79}).

\subsubsection{Complexity Classes}
The classical complexity classes: P, NP, co-NP, NP-complete, \#P,
among others; are extremely useful but can only deal with problems
with small (polynomial) outputs with regard to input size, e.g. the
decision problems return 1 (true) or 0 (false). Quite often the number
of solutions in a enumeration problem is exponential in the size of
input, e.g. listing $st$-paths in a graphs, there is $G = (V,E)$ such
that the number of $st$-paths is $\Omega(2^{|V|/2})$. To overcome this
problem, the enumeration complexity classes are defined in an
output-sensitive way; in other words, the time complexity takes into
account the size of input and the \emph{output}. In this way, if the
number of solutions is small, an efficient algorithm has to terminate
after a short (polynomial) time, otherwise it is allowed to spend more
time. According to this idea, the following complexity classes were
defined in \cite{Johnson88}.
\begin{definition}[Polynomial Total Time]
\label{def:poly_total_time}
An enumeration algorithm is \emph{polynomial total time} if the time
required to output all the configurations is bounded by a polynomial
input size and the number of configurations.
\end{definition}

\begin{definition}[Incremental Polynomial Time]
\label{def:inc_poly_time}
An enumeration algorithm is \emph{incremental polynomial time} if it
generates the configurations, one after the other in some order, in
such a way that the time elapsed (delay) until the first is output,
and thereafter the delay between any two consecutive solutions, is
bounded by a polynomial in the input size and the number of
configurations output so far.
\end{definition}

\begin{definition}[Polynomial Delay] 
\label{def:poly_delay}
An enumeration algorithm is \emph{polynomial delay} if it generates
the configurations, one after the other in some order, in such a way
that the delay until the first is output, and thereafter the time
elapsed (delay) between any two consecutive solutions, is bounded by a
polynomial in the input size.
\end{definition}

Intuitively, the polynomial total time definition means that the delay
between any two consecutive solutions is polynomial on \emph{average},
while the polynomial delay definition implies that the maximum delay
is polynomial. Hence, Definition~\ref{def:poly_delay} contains
Definition~\ref{def:poly_total_time} and
Definition~\ref{def:inc_poly_time} is in between. It is important to
stress that these complexity classes impose no restriction on the
space complexity of the algorithms, e.g. a polynomial total time
algorithm can use memory exponential in the input size. However, in
this thesis, we are mainly concerned about polynomial delay algorithm
with space complexity polynomial on the \emph{input size}. In some
sense, these are efficient listing algorithm, or as defined in
\cite{Fukuda97} \emph{strongly P-enumeration} algorithms.

The basic framework for efficient listing algorithms are:
\emph{backtracking} (unconstrained depth-first search with
lexicographic ordering), \emph{binary partition} (branch and
bound-like recursive partition algorithm) and \emph{reverse search}
(traversal on the tree defined by the parent-child relation). In the
remaining of this section, we give a brief introduction for the first
two strategies; the backtracking method is used in Chapters
\ref{chap:kissplice} and \ref{chap:unweighted}, while the binary
partition method is used in Chapter~\ref{chap:unweighted}, and a
variation of if it in Chapter~\ref{chap:weighted}. An introduction of
the reverse search method can be found in \cite{Avis93} and
\cite{Marino14}.

\subsubsection{The Backtracking Method}
The backtracking method is a recursive\footnote{Of course, it can also
  be implemented in an iterative way, but it is not as natural as the
  recursive implementation.} listing technique based on the following
simple idea: given a partial solution (i.e. a set that can be extended
to a solution) recursively try all possible extensions leading to a
solution. This technique has been successfully applied by several
algorithms to list cycles in directed graphs
(\cite{Tiernan70,Tarjan73,Johnson75,Szwarcfiter76}); list bubbles in
directed graphs (\cite{Birmele12}); list maximal cliques in undirected
graphs (\cite{Bron73,Koch01,Eppstein10,Eppstein11}); and list maximal
independent set in undirected graphs (\cite{Johnson88}). The last two
problems are particular cases of the more general problem of listing
(maximal) sets in an \emph{independence system}. This is not a
coincidence, the backtracking method is particularly useful for
listing problems that can be described as the enumeration of (maximal)
sets in an independence system.

A collection of sets $I$ is an independence system if for any set $X
\in I$ all its subsets, $X' \subseteq X$, are also in $I$. More
formally, a family of sets $I$ over an universe $U$, i.e. $I \subseteq
2^U$, is an \emph{independence system} if it satisfies the following
properties: (i) the empty set belongs to $I$; and (ii) every subset of
some set in $I$ also belongs to $I$, i.e. $X' \subseteq X$ and $X \in
I$ implies that $X' \in I$. The sets in an independence system can be
listed using the backtracking method: given a set $X$ (initially the
empty set) we recursively try to extended it, adding one new element
$e \in U \setminus X$, and obtaining $X \cup \{e\}$. However, this is
not enough to guarantee an efficient algorithm; for that, we also need
to efficiently decide if the extension $X \cup \{ e \}$ belongs to
$I$, i.e. we need a polynomial membership oracle for $I$. An example
of an efficient listing algorithm for the sets in an independence
system is shown next. See \cite{Marino14} for an example of an
algorithm to list only the \emph{maximal} sets in an independence
system.

\paragraph{Enumerating all the subsets of a collection $U=\{a_1,\ldots,
a_n\} \subset \mathbb{Z}_{\geq 0}$ whose sum is less than $b$.} The
family of sets $I$ over $U = \{a_1,\ldots, a_n\}$ whose sum is less
than $b$, form an independence system. Indeed, the empty set has sum
0, so it belongs to $I$; since $a_i \geq 0$, the sum of $X' \subseteq
X$ is not greater than $X$, so if $X \in I$ then $X' \in I$. Moreover,
for a given subset $X$ of $U$ we can decide if it belongs to $I$ in
linear time, we just have to compute the sum of elements of $X$. We
already have all requirements to design a backtracking-based algorithm
for this problem: starting from the empty set $S = \emptyset$, we
recursively try to add a new element $a_i \in U \setminus S$ to $S$,
provided the resulting set $S \cup \{a_i\}$ belongs to $I$ (the sum is
smaller than $b$). The pseudocode is shown in
Algorithm~\ref{alg:four}. Now, each iteration outputs a solution, and
takes $O(n)$ time, where $n = |U|$, thus the algorithm spend $O(n)$
time per solution. It is worth observing that by sorting the elements
of $U$, each recursive call can generate a solution in $O(1)$ time,
resulting in optimal $O(1)$ time per solution.

\begin{algorithm}[htbp]
\KwIn{$S$ a set (initially empty) of integers belonging to the collection $U=\{a_1,\ldots, a_n\} \subset \mathbb{Z}_{\geq 0}$}
\KwOut{The subsets of $U$ whose sum is less than $b$.}
  output $S$ \\
  \ForEach{$a_i \in U \setminus S$}{
    \If{$a_i + \sum_{x\in S}x \leq b$}{
      \textsc{SubsetSum}$(S\cup \{a_i\})$}
  }
\caption{\textsc{SubsetSum}$(S)$}
\label{alg:four}
\end{algorithm}

\subsubsection{The Binary Partition Method}
The binary partition method, similarly to the backtracking method, is
a recursive technique; based on the following simple idea: recursively
divide the solution space into two \emph{disjoint} parts until it
becomes trivial, i.e. each part contains exactly one solution. More
formally, let $X$ be a subset of $F$, the set solutions, such that all
elements of $X$ satisfy a property $\mathcal{P}$; recursively
\emph{partition} $X$ into two subsets $X_1$ and $X_2$ (i.e. $X = X_1
\cup X_2$ and $X_1 \cap X_2 = \emptyset$), characterized by disjoint
properties $\mathcal{P}_1$ and $\mathcal{P}_2$, respectively. This
procedure is repeated until the current set of solutions is a
singleton. This technique has been successfully applied to many
listing problems in graphs, including: $st$-paths in undirected graphs
(\cite{Birmele13}), cycles in undirected graphs (\cite{Birmele13}),
perfect matchings in bipartite graphs (\cite{Uno01}), and $k$-trees in
undirected graphs (\cite{Ferreira11}).

The recursion tree of any algorithm implementing the binary partition
method is binary, since there are at most two recursive calls in the
algorithm, one for each set $X_1$ and $X_2$ partitioning
$X$. Moreover, unlike the backtracking method, the solutions are
output only in the leaves of the tree, when the partition is a
singleton. In order to design an efficient algorithm based on this
technique, in any given call, when $X$ is partition into $X_1,X_2$,
before proceeding with the recursion, we have to decide if $X_1$ and
$X_2$ are non-empty, otherwise we would have many calls leading to no
solution. Assuming that we have a polynomial (in the input size)
oracle to decide if $X_1$ and $X_2$ are non-empty, and the height of
the tree is bounded by a polynomial in the input size; the resulting
algorithm has polynomial delay. Indeed, considering the recursion
tree, the time elapsed between two solutions being output is
bounded by the time spent in the nodes in any leaf-to-leaf path in the
tree (recall that the solutions are only output in the leaves). As the
height of the tree is polynomial in the input size, the number of
nodes in any leaf-to-leaf path is also polynomial in the input size;
and since the emptiness oracle is polynomial in the input size, the
time spent in each node is also polynomial in the input size. An
example of application of binary partition method is presented bellow.

\paragraph{Enumerating all the $st$-paths in an undirected graph $G = (V,E)$.} 
The first requirement to apply the binary partition method is a
property that allows to recursively partition the set of solutions,
$st$-paths in this case.  Let $v$ be \emph{any neighbor} of $s$, the
set $X$ of all $st$-paths in $G$ can be partition in two sets: $X_1$,
the set of $st$-paths that do not include the edge $(s,v)$; and $X_2$,
the set of $st$-paths that include it. Actually, these sets can be
described in a more ``recursive'' way: $X_1$, the set of $st$-paths in
$G - (s,v)$; and $X_2$, the set of $vt$-paths in $G - s$ concatenated
with the edge $(s,v)$; in that way, both $X_1$ and $X_2$ are described
in terms of the sets of $xt$-paths in a graph $G'$.  Thus, a procedure
$st$\textsc{Paths}$(s,t,G)$, to list $st$-paths in $G$, can be
implemented with recursive calls to $st$\textsc{Paths}$(s,t,G-(s,v))$,
corresponding to the $st$-paths in the partition $X_1$; and
$st$\textsc{Paths}$(v,t,G-s)$ with paths prepended with $(s,v)$,
corresponding to the $st$-paths in the partition $X_2$. In the base
case, where $s = t$ and the current partition has only one solution,
the corresponding $st$-path is output. The pseudocode in shown in
Algorithm~\ref{alg:six}.

Recall that in order to have an efficient algorithm, before performing
a recursive call we need to efficiently decide if the corresponding
partition, $X_1$ or $X_2$, is not empty, and only perform the call in
that case. Clearly, $X_1$ is not empty if and only if there is at
least one $st$-path in $G - (s,v)$, and $X_2$ is not empty if and only
if there is at least one $vt$-path in $G - s$. In both cases, the test
can be done in $O(|V| + |E|)$ time using a DFS traversal. Let us now
analyzed the delay of the algorithm. The height of the recursion is
bounded by $|V| + |E|$, since at every call one vertex or one edge is
removed from $G$, after $|V| + |E|$ calls the graph is empty. Hence,
there are at most $2 (|V| + |E|)$ nodes in any leaf-to-leaf path in
the recursion tree. As the time spend in each node is $O(|V| + |E|)$,
the delay is thus $O((|V| + |E|)^2)$.

\begin{algorithm}
  \KwIn{An undirected graph $G$, vertices $s$ and $t$, and a path $\pi$ 
    (initially empty).}
  \KwOut{The paths from $s$ to $t$ in $G$.}
  \If {$s=t$}{
    output S\\ 
    \Return
  }
  choose an edge $e=(s,v)$ \\
  \If {there is a $vt$-path in $G-s$}{
    $st$\textsc{Paths}$(G-s, v, t, \pi(s,v))$\\ 
  }
  \If {there is a $st$-path in $G-e$}{
    $st$\textsc{Paths}$(G-e, s, t,\pi)$ \\
  }
  \caption{$st$\textsc{Paths}$(G, s, t, \pi)$}
  \label{alg:six}
\end{algorithm}

\chapter[Kissplice: calling alternative splicing from RNA-seq]{Kissplice: de novo calling alternative splicing events from RNA-seq data}
\label{chap:kissplice}
\minitoc This chapter is strongly based on our paper \cite{Sacomoto12}. Here,
we address the problem of identifying and quantifying variations
(alternative splicing and genomic polymorphism) in RNA-seq data when
no reference genome is available, without assembling the full
transcripts. Based on the fundamental idea that each variation
corresponds to a recognizable pattern in a de Bruijn graph constructed
from the RNA-seq reads, we propose a general model for all variations
in such graphs. We then introduce an exact algorithm, called \ks, to
extract alternative splicing events. Finally, we show that it enables
to identify more correct events than general purpose transcriptome
assemblers. The algorithm presented in this chapter corresponds to \ks
version 1.6. Further improvements in time and memory efficiency are
shown in Chapters
\ref{chap:weighted} and \ref{chap:dbg}, respectively. The current
implementation of \ks (version 2.0) already includes those
improvements.


\bigskip
\bigskip


\section{Introduction}
Thanks to recent technological advances, sequencing is no longer
restricted to genomes and can now be applied to many new areas,
including the study of gene expression and splicing. As stated in
Section~\ref{subsec:ngs}, the so-called RNA-seq protocol consists in
applying fragmentation and reverse transcription to an RNA sample
followed by sequencing the ends of the resulting cDNA fragments. The
short sequencing reads then need to be reassembled to get back to the
initial RNA molecules. As stated in Section~\ref{sec:back:modeling}, a
lot of effort has been put on this assembly task, whether in the
presence or in the absence of a reference genome but the general goal
of identifying and quantifying all RNA molecules initially present in
the sample remains hard to reach.

The main challenge is certainly that reads are short, and can
therefore be ambiguously assigned to multiple transcripts.  In
particular, in the case of alternative splicing (AS for short), reads
stemming from constitutive exons can be assigned to any alternative
transcript containing this exon. Finding the correct transcript is
often not possible given the data we have, and any choice will be
arguable.  As pointed out in \cite{Martin11}, reference-based and de
novo assemblers each have their own limitations. Reference-based
assemblers (\cite{Scripture,Cufflinks,Fluxcapacitor,Ireckon,Express})
depend on the quality of the reference while only a small number of
species currently have a high-quality reference genome available. De
novo assemblers (\cite{Transabyss,Trinity,Oases,Idbatran}), as stated
in Section~\ref{sec:back:modeling}, implement reconstruction
heuristics which may lead them to miss infrequent alternative
transcripts while genes sharing repeats are likely to be assembled
together and create chimeras.

We argue here that it is not always necessary to aim at the difficult
goal of assembling full-length molecules. Instead, identifying the
variable parts between molecules is already very valuable and does not
require to solve the problem of assigning a read from a constitutive
exon to the correct transcript.  We therefore focus in this work on
the simpler task of identifying variations in RNA-seq data. Three kinds
of variations have to be considered: (i) alternative splicing (AS) that
produces several alternative transcripts for a same gene, (ii) single
nucleotide polymorphism (SNPs) that may also produce several
transcripts for a same gene whenever they affect transcribed regions,
and (iii) genomic insertions or deletions (indels).  Our contribution
in this chapter is double: we first give a general model which
captures these three types of variations by linking them to
characteristic structural patterns called ``bubbles'' in the de Bruijn
graph (DBG for short) built from a set of RNA-seq reads, and second,
we propose a method dedicated to the problem of identifying AS events
in a DBG, including read-coverage quantification.  We notice here that
only splicing events but not transcriptional events, such as
alternative start and polyadenylation sites, are covered by our
method.

The identification of bubbles or bulges in DBG has been studied before
in the context of genome assembly (\cite{Pevzner04,Velvet,Abyss}), but
the goal was not list them as variation-related structures, instead to
simplify the de Bruijn graph. On the other hand, methods to identify
variations as a restricted type of bubbles were proposed
(\cite{Peterlongo10,Cortex,Bubbleparse}), these works deal only with
genomic NGS data and the variations considered are genomic
polymorphisms, mainly SNPs and small indels. More
recently, \cite{Marygold} presented a method to list bubble-like
structures in the metagenomic context using the same graph
decomposition previously proposed in \cite{Sacomoto12}.

When no reference genome is available, efforts have focused on
assembling the full-length RNA molecules, not the variable parts which
are our interest here. As stated in Section~\ref{sec:back:modeling},
most RNA-seq assemblers (\cite{Transabyss,Trinity,Oases,Idbatran}) do
rely on the use of a DBG, but, since the primary goal of an assembler
is to produce the longest contigs, heuristics are applied, such as tip
or bubble removal, in order to linearize the graph.  The application
of such heuristics results in a loss of information which may in fact
be crucial if the goal is to study expressed variations (alternative
splicing and genomic polymorphism).

To our knowledge, this work is the first attempt to characterize
variations in RNA-seq data without assembling full-length transcripts.
We stress that it is not a general purpose transcriptome assembler and
when we benchmark it against such methods, we only focus on the
specific task of AS event calling.  Finally, our method can be used in
a comparative framework with two or more conditions and our
quantification module outputs a coverage (number of reads mapped) for
both the shorter and the longer isoform(s) of each AS event, in each
experiment.

The chapter is organized as follows. We first present the model
(Section~\ref{sec:kissplice:dbg_models}) linking structures of the DBG
for a set of RNA-seq reads to variations (AS, SNPs and indels), and
then introduce a method, that we call \ks, for identifying DBG
structures associated with AS events
(Section~\ref{sec:kissplice:algorithm}).  We show in
Section~\ref{sec:results} the results of using \ks compared with other
methods on simulated and real data.

\section{Methods}
\subsection{De Bruijn graph models} \label{sec:kissplice:dbg_models}
\subsubsection{De Bruijn graph}

In Section~\ref{sec:back:modeling}, we defined a de Bruijn graph for a
read set $\mathcal{R} \subset \{A,C,T,G\}^*$ as a directed graph where
each vertex corresponds to a $k$-mer and the arcs represent
suffix-prefix overlaps of size $k-1$ and correspond to a
$(k+1)$-mer. See Fig.~\ref{fig:kissplice:bidir_vs_dir}(a) for an
example of \emph{directed} de Bruijn graph. One problem with this
definition is that it does not capture very well the double stranded
nature of the DNA molecule, that is each $k$-mer present in the reads
has a reverse complementary $k$-mer essentially representing the same
information. Recall that, even though we are dealing with mRNAs
sequencing data, and RNA is single stranded, one of the early steps of
the RNA-seq protocol is reverse transcription, where the more stable
double stranded cDNA is obtained from the mRNA extracted from the
cell. Thus, RNA-seq data is also double\footnote{As stated in
Section~\ref{subsec:ngs}, there are strand specific RNA-seq
protocols. } stranded.

\begin{figure}[htbp]
  \centering
  \includegraphics[width=\linewidth]{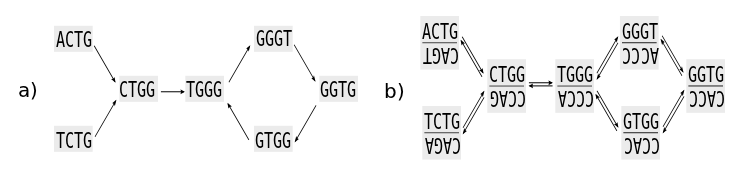}
  \caption{(a) The \emph{directed} de Bruijn graph, with $k=4$, for
    the set of reads $\mathcal{R} = \{ACTGG, TCTGGG, CTGGGTGGG\}$ is
    shown. (b) The \emph{bidirected} de Bruijn graph, with $k=4$, for
    the same set of reads is shown.}
  \label{fig:kissplice:bidir_vs_dir}
\end{figure}

In order to better model the DNA double stranded nature,
\cite{Medvedev07}, based on \cite{Kececioglu92}, modified the de
Bruijn graph definition to associate to each vertex not only a $k$-mer
$w \in \{A,C,T,G\}^k$ but its reverse complement $\overline{w} \in
\{A,C,T,G\}^k$. In such a context, a de Bruijn graph is a
directed\footnote{The original definition of \cite{Medvedev07} for the
  bidirected de Bruijn graph is based on \emph{bidirected} graphs
  (\cite{Edmonds70}). However, for the sake of a clearer exposition,
  Definition~\ref{def:biDBG} is based on directed multigraph with arc
  labels instead. It can be shown that they are equivalent.}
multigraph $G = (V,E)$, where each vertex $v \in V$ associated to a
$k$-mer $w$ and its reverse complement $\overline{w}$. The sequence
$w$, denoted by $F(v)$, is the \emph{forward sequence} of $v$, while
$\overline{w}$, denoted by $R(v)$, is the reverse complement sequence
of $v$.  An arc exists from vertex $v_1$ to vertex $v_2$ if the suffix
of length $k-1$ of $F(v_1)$ or $R(v_1)$ overlaps perfectly with the
prefix of $F(v_2)$ or $R(v_2)$. See
Fig.~\ref{fig:kissplice:bidir_vs_dir}(b) for an example of
\emph{bidirected} de Bruijn graph. This is formally stated in the
following definition.

\begin{definition}[Bidirected de Bruijn Graph] \label{def:biDBG}
  Given a set of reads $\mathcal{R} \subseteq \Sigma^*$ and a
  parameter $k \in \mathbf{N}$, the \emph{bidirected de Bruijn graph}
  $B_k(\mathcal{R})= (V,E)$ is a directed multigraph such that:
  
  \begin{enumerate}
  \item $V = \{ \{ w,\overline{w} \} | w \text{ is a $k$-mer of }
    \mathcal{R} \}$, 
  \item $E = \{ (x,y) \in V^2 | F(x) \text{ or } R(x) \text{ has a
    $k-1$ suffix-prefix overlap with } F(y) \mbox{ or } R(y)\}$,
  \end{enumerate}
  where $\overline{w}$ is the reverse complement of $w$, and $F,R : V
  \rightarrow \Sigma^k$ are functions such that, for\footnote{Given a
    vertex $v \in V$ and its two corresponding $k$-mers, it is
    arbitrary, but fixed, which $k$-mer is the forward $w = F(v)$ and
    reverse $\overline{w} = R(v)$.} $v = \{w, \overline{w}\} \in V$,
  $F(v) = w$ and $R(v) = \overline{w}$.
\end{definition}

It is convenient to augment this definition with arc labels in the set
$\{F,R\}^2$.  The first letter of the arc label indicates which of
$F(v_1)$ or $R(v_1)$ has a suffix-prefix overlap with $F(v_2)$ or
$R(v_2)$, this latter choice being indicated by the second letter.
Moreover, because of the reverse complements, a suffix-prefix overlap
from $v_1$ to $v_2$ induces a symmetrical suffix-prefix overlap from
the corresponding complementary $k$-mers of $v_2$ to $v_1$. As a
result, there is an even number of arcs in the bidirected de Bruijn
graph: if there is an arc from $v_1$ to $v_2$ then, necessarily, there
is a \emph{twin arc} from $v_2$ to $v_1$ with the corresponding label
(i.e. if the first arc has label $FF,RF,FR,RR$ then the second has
label $RR,RF,FR,FF$, respectively). An example of a bidirected de
Bruijn graph with the corresponding arc labels is shown in
Fig.~\ref{fig:kissplice:bidir_arc_labels}. The de Bruijn graphs
considered in this chapter are all bidirected, so we omit this term,
referring to them simply as de Bruijn graphs (or DBG for short).

\begin{figure}[htbp]
  \centering
  \includegraphics[width=0.8\linewidth]{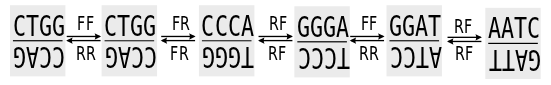}
  \caption{An example of a bidirected de Bruijn graph with arc labels.}
  \label{fig:kissplice:bidir_arc_labels}
\end{figure}

\begin{definition}[Valid path] \label{def:kissplice:validpath}
  Given a bidirected de Bruijn graph $B_k(\mathcal{R}) = (V,E)$, a
  simple path $p = (v_1,v_2) \ldots (v_{n-1},v_n)$ is \emph{valid} if
  for any two adjacent arcs $(v_{i-1},v_i)$ and $(v_i,v_{i+1})$ the
  labels are of the form $L_1L_2$ and $L_2L_3$, respectively, where
  $L_1,L_2,L_3 \in \{R,F\}$.
\end{definition}

Consider the arc $e = (x,y)$ with label $L_1L_2 \in \{R,F\}^2$, we say
that $e$ enters $y$ in the forward or reverse direction if $L_2 = F$
or $L_2 = R$, respectively, analogously for $e$ leaving
$x$. Basically, Definition~\ref{def:kissplice:validpath} says that for
a path to be valid all pairs of adjacent arcs should enter and leave a
vertex in the same direction. For instance, for the graph shown in
Fig.~\ref{fig:kissplice:bidir_arc_labels}, the path from the leftmost
vertex ($CTGG / CCAG$) going to the vertex $GGAT / ATCC$ is valid,
with $(FF, FR, RF, FF)$ being the corresponding sequence of arc-labels. On
the other hand, the path from the leftmost vertex ($CTGG / CCAG$) to the
rightmost vertex ($AATC / GATT$) is not valid, since there is no arc
leaving the \emph{forward} part of $GGAT / ATCC$ and entering $AATC /
GATT$. Finally, due the reverse complement relationship between the
pair of labels, every valid path $p = s \leadsto t$ induces a
complementary valid path $\overline{p} = t \leadsto s$, where each arc
is substituted by its twin.

\begin{figure}[htbp]
  \centering
  \includegraphics[width=0.8\linewidth]{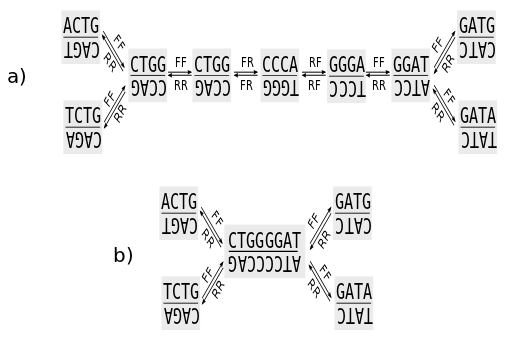}
  \caption{(a) An example of a de Bruijn graph, with $k=4$. (b) The
    corresponding compressed de Bruijn graph. The compressed path in
    the DBG has 5 vertices, and the corresponding vertex in the cDBG
    has a pair of sequences each of length $k + (i-1) = 4 + 4 = 8$.}
  \label{fig:kissplice:cdbg}
\end{figure}

A DBG can be compressed without loss of information by merging the
vertices of non-branching valid paths. A non-branching valid path $p =
s \leadsto t$ is a valid path such that for any internal vertex there
is only one valid extension. See Fig.~\ref{fig:kissplice:cdbg}(a) for
an example of a non-branching valid path from $CTGG / CCAG$ to $GGAT /
ATCC$. Let $d_F^+(v)$ be the number of arcs leaving $v$ in the forward
direction, analogously for incoming arcs and the reverse
direction. More explicitly, a path is non-branching if $d_F^+(s) = 1$
and $d_F^-(t) = 1$, and for every internal vertex $v$ we have that
$d_F^+(v) = d_F^-(v) = 1$ (the conditions can be stated considering
the reverse direction instead, as one implies the other).  Two
adjacent vertices in a non-branching valid path are merged into one by
removing the redundant information, that is keeping only one copy of
the $k-1$ suffix-prefix overlap. A valid path composed by $i>1$
vertices is merged into one vertex containing as labels, not a pair of
$k$-mers, but a pair of sequences of length $k+(i-1)$ as each vertex
in the path adds one new character to the first vertex.  See
Fig.~\ref{fig:kissplice:cdbg} for an example of DBG and the
corresponding compressed DBG. In the remaining of the chapter, we
denote by cDBG a compressed DBG. Moreover,
Definition~\ref{def:kissplice:validpath} also applies to cDBG.

\subsubsection{Bubble patterns in the cDBG}\label{ssec:bubbleevents}

Variations (alternative splicing events and genomic polymorphisms) in
a transcriptome, correspond to recognizable patterns in the cDBG,
which we call a
\emph{bubble}. Intuitively, the variable parts will correspond to
alternative paths and the common parts will correspond to the
beginning and end points of these paths.  See
Fig.~\ref{fig:kissplice:bubble} for an example of a bubble in a
cDBG. We now formally define the notion of bubble, taking carefully
into account the bidirected and arc labeled nature of the cDBG.

\begin{definition}[Switching Vertex] \label{def:kissplice:sv}
  Given a path $p = (v_1,v_2) \ldots (v_{n-1},v_n)$ in a bidirected de
  Bruijn graph (or a cDBG), a vertex $v_i \in p$ is a \emph{switching
    vertex} of $p$ if the arc $(v_i,v_{i+1})$ leaves $v_i$ in the
  complementary direction the arc $(v_i, v_{i-1})$ enters $v_i$.
\end{definition}

\begin{definition}[Bubble] \label{def:kissplice:bubble}
  Given a bidirected de Bruijn graph (or a cDBG) $B_k(\mathcal{R}) =
  (V,E)$, a \emph{bubble} is a cycle with at least four distinct
  vertices such that there are exactly two switching vertices, denoted
  $S_{left}$ and $S_{right}$.
\end{definition}

It follows directly from this definition, that for any bubble there
are two valid paths, not sharing any internal vertex, from $S_{left}$
to $S_{right}$.  In the remaining of the chapter, we refer to these
two paths as the paths of the bubble. If they differ in length, we
refer to, respectively, the longer and the shorter path of the
bubble. Where the length of a valid path in a cDBG is the length of
the corresponding sequence of that path, not the number of
vertices. In the example of Fig.~\ref{fig:kissplice:bubble} the
switching vertices are encircled in blue and the longer path is shown
above the shorter path.

\begin{figure}[htbp]
  \centering
  \includegraphics[width=0.4\linewidth]{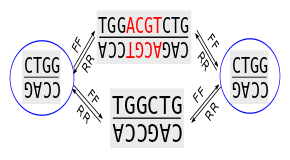}
  \caption{An example of a bubble in a bidirected de Bruijn graph. The
    bubble was generated by the sequences: CTGG{\color{red}ACGT}CTGG
    ($asb$) and CTGGCTGG ($ab$). The switching vertices are encircled
    in blue.}
  \label{fig:kissplice:bubble}
\end{figure}

In general, any process generating patterns $asb$ and $as'b$ in the
sequences, with $a,b,s,s' \in \Sigma^*$, $|a| \geq k, |b|\geq k$ and
$s$ and $s'$ not sharing any $k$-mer, creates a bubble in the cDBG.
Indeed, all $k$-mers entirely contained in $a$ (resp. $b$) compose the
vertex $S_{left}$ (resp. $S_{right}$). Since $|a| \geq k$ and $s \neq
s'$, there is at least one pair of $k$-mers, one in $as$ and the other
in $as'$, sharing the $k-1$ prefix and differing by the last letter,
thus creating a branch in $S_{left}$ from which the two paths in the
bubble diverge. The same applies for $sb$, $s'b$ and $S_{right}$,
where the paths merge again. All $k$-mers contained in $s$
(resp. $s'$) and in the junctions $as$ and $sb$ (resp. $as'$ and
$s'b$) compose the paths of the bubble. In the case where $s$ is
empty, the shorter path is composed of $k$-mers covering the junction
$ab$. As we show later most AS events fall into this case.

We show next that this model is general as it captures SNPs, indels
and AS events. However, the main focus of the algorithm we present in
this work is the detection of bubbles generated by AS events.

\subsubsection{Bubbles generated by AS events} 

As stated in Section~\ref{sec:back:AS}, a single gene may give rise to
multiple alternative spliceforms through the process of AS.
Alternative spliceforms differ locally from each other by the
inclusion or exclusion of subsequences. These subsequences may
correspond to exons (exon skipping), exon fragments (alternative donor
or acceptor sites) or introns (intron retention) as shown in
Fig.~\ref{fig:splicing_events}(a).  We should stress that we do not
model mutually exclusive exons (another less frequent type of AS),
since, as we show next, it does not correspond to the same pattern in
terms of path lengths, and is therefore harder to treat.
Additionally, alternative start and polyadenylation sites
(\cite{Alberts03}), which are not considered as AS events but as
transcriptional events, are also not taken into account.

An alternative splicing event corresponds to a local variation between
two alternative transcripts. It is characterized by two common
sequences ($a$ and $b$ in the AS events given in
Fig.~\ref{fig:splicing_events}(a)) and a single variable part ($s$ in
Fig.~\ref{fig:splicing_events}(a)). As stated in the last section, if
$|a| \geq k$, $|b| \geq k$, then the patterns $asb$ and $ab$ generate
a bubble in the cDBG. See Fig.~\ref{fig:splicing_events}(b) for an
example of bubble generated by an AS event. In this example, the
flanking sequences $a$ and $b$ correspond to the switching vertices,
the variable part $s$ to the longer path, and the shorter path
corresponds to the $k$-mers covering the junction $ab$. Moreover, as
there are $k-1$ $k$-mers covering the junction between the two common
sequences $a$ and $b$, the shorter path is composed of \emph{exactly}
$k-1$ $k$-mers. This however is not true in general. These two
properties -- correspondence between flanking sequences and switching
vertices and exactly $k-1$ $k$-mers in the shorter path -- do not hold
in general. They are not true when $a$ and $s$ share a suffix or $b$
and $s$ share a prefix. This case (which actually happens in more than
50\% of the AS events, since it suffices that 1 out of 4 possible
nucleotides are shared) is illustrated in
Fig~\ref{fig:splicing_events}(c). In this example, the sequence of the
switching vertex opening the bubble is $a$ concatenated with the
longest common prefix between $b$ and $s$, and the shorter path
contains $k-3$ $k$-mers. In general, the length of the shorter path
for a bubble generated by the pattern $asb$ and $ab$ is
$k-1-lcp(s,b)-lcs(s,a)$, where $lcp(s,b)$ (resp. $lcs(s,a)$) is the
length of the longest common prefix (resp. suffix) between $s$ and $b$
(resp. $a$).  Overall, a bubble generated by an AS event always
corresponds to a local variation between two RNA sequences. The
shorter variant always has a length bounded by $2k-2$. In human, 99\%
of the annotated exon skipping events yield a bubble with a shorter
path length between $2k-8$ and $2k-2$ (\cite{Kuhn09}).

\begin{figure}[htbp]
\centering 
\includegraphics[width=0.8\linewidth]{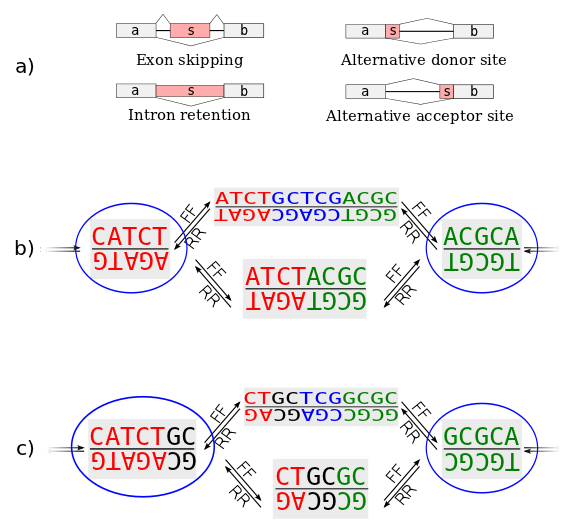}
\caption{(a) AS events generating a bubble in the DBG. These events
  create a bubble in the DBG or cDBG, in which the shorter path is
  composed by $k$-mers covering the $ab$ junction. This path, composed
  by $k-1$ vertices in the DBG, is compressed into a sequence of
  length $2k-2$ in the cDBG. (b) A bubble in a cDBG, with $k=5$, due
  to the variable part $\textcolor{blue}{GCTCG}$ ($s$). This bubble is
  generated by the sequences
  $\textcolor{red}{CATCT}\textcolor{green}{ACGCA}$ ($ab$) and
  $\textcolor{red}{CATCT}\textcolor{blue}{GCTCG}\textcolor{green}{ACGCA}$
  ($asb$). The shorter path has length $2k-2 = 8$. (c) A bubble in a
  cDBG, with $k=5$, due to the skipped exon $GC\textcolor{blue}{TCG}$
  ($s$) with the flanking sequences $\textcolor{red}{CATCT}$ ($a$) and
  $GC\textcolor{green}{GCA}$ ($b$). This bubble is generated by the
  sequences $\textcolor{red}{CATCT}GC\textcolor{green}{GCA}$ ($ab$)
  and
  $\textcolor{red}{CATCT}GC\textcolor{blue}{TCG}\textcolor{green}{GCGCA}$
  ($asb$). Observe that $s$ and $b$ share the prefix $GC$. As a
  result, the $k$-mers $ATCTG$ and $TCTGC$ are common to both paths
  and and represented only once; and the length of the shorter path is
  $2k-2-2= 6$.}
\label{fig:splicing_events}
\end{figure}

\subsubsection{Bubbles generated by SNPs, indels and repeats} \label{sec:kissplice:tandem_repeats}
Variations at the genomic level will necessarily also be present at the
transcriptomic level whenever they affect transcribed regions.  Two
major types of variations can be observed at the genomic level: SNPs and
indels.  As shown in Fig.~\ref{fig:kissplice:other_bubbles}(a) and
Fig.~\ref{fig:kissplice:other_bubbles}(b), they
also generate bubbles in the cDBG.

However, these bubbles have characteristics which enable to
differentiate them from bubbles generated by AS events. Indeed,
bubbles generated by SNPs exhibit two paths of length exactly $2k-1$,
which is larger than $2k-2$, the maximum size of the shorter path in a
bubble generated by an AS event.

Genomic insertions or deletions (indels for short) may also generate
bubbles with similar path lengths as bubbles generated by splicing
events. In this case, the difference of length between the two paths
is usually smaller, less than 3 nt for 85\% of indels in human
transcribed regions (\cite{Sherry01}) whereas it is more than 3 nt for
99\% of AS events. This suggests an initial criterion to separate
between AS and indels: when the difference of path lengths is strictly
below 3 we classify them as an indel; and AS event, otherwise. In
Section~\ref{subsec:kissplice:novel}, we refine the classification by
considering that in an AS event in a coding region the difference of
length is more likely to be a multiple of 3; since each codon is
composed of 3 bases, an AS event with the length of the variable part
not a multiple of 3 would cause a frame shift, potentially change
completely the amino acid sequence.

Finally, inexact repeats may generate bubbles with a similar path
length as bubbles generated by splicing events, but the sequences of
the paths exhibit a clear pattern which can be easily identified: the
longer path contains an inexact repeat. More precisely, as outlined in
Fig~\ref{fig:kissplice:other_bubbles}(c), it is sufficient to compare
the shorter path with one of the ends of the longer path. We treat
this kind of event as false positive, bubbles that do not correspond
to a true variation in the dataset. However, there is a type of true
genomic polymorphism that may be include in this group: copy number
variations (CNVs).

\begin{figure}[htbp]
  \centering 
  \includegraphics[width=\linewidth]{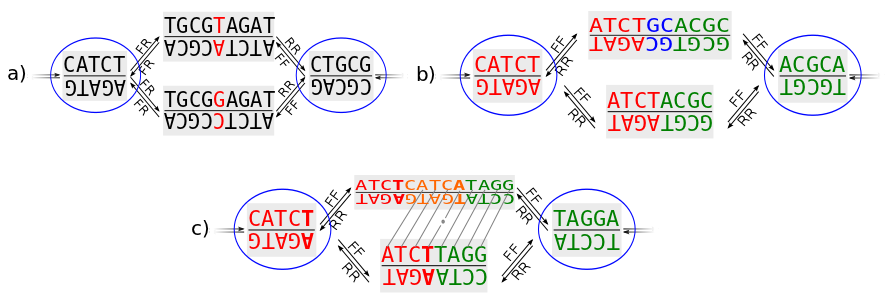} 
  \caption{(a) Bubble due to a SNP (substitution is the red
  letter). Starting from the forward strand in the leftmost
  (switching) vertex would generate the sequences
  $CATCT\textcolor{red}{A}CGCAG$ (upper path) and
  $CATCT\textcolor{red}{C}CGCAG$ (lower path). (b) Bubble due to the
  deletion $\textcolor{blue}{GC}$. This bubble is generated by the
  sequences $\textcolor{red}{CATCT}\textcolor{green}{ACGCA}$ and
  $\textcolor{red}{CATCT}\textcolor{blue}{GC}\textcolor{green}{ACGCA}$.
  (c) Bubble due to an inexact repeat. This bubble is generated by the
  sequences $\textcolor{red}{CATCT}\textcolor{green}{TAGGA}$ and
  $\textcolor{red}{CATCT}\textcolor{orange}{CATCA}\textcolor{green}{TAGGA}$,
  where $\textcolor{red}{CATC{\bf T}}\textcolor{orange}{CATC{\bf A}}$
  is an inexact repeat.}  \label{fig:kissplice:other_bubbles}
\end{figure}

In the following, we focus on bubbles generated by AS events. In the
output of the method we present in the next section, we do provide as
a collateral result three additional collections of bubbles: one
corresponding to putative SNPs, one to short indels, and one to
putative repeats associated bubbles. The post-treatment of these
collections to discard false positives caused by sequencing errors, or
recover the ones corresponding to CNVs, is beyond the scope of this
work.

\subsection{The \ks algorithm} \label{sec:kissplice:algorithm}

The \ks algorithm detects in the cDBG all the bubble patterns
generated by AS events, i.e. the bubbles having a shorter path of
length at most $2k-2$.  Essentially, the algorithm lists all the
\emph{cycles} verifying the following criteria:
\begin{itemize}
\item[{\bf i}] the cycle contains exactly two switching vertices,
  i.e. it corresponds to a bubble;
\item[{\bf ii}] the length of the shorter path linking the two
  switching vertices is smaller than $2k-2$;
\item[{\bf iii}] both paths have length greater than $2k-8$;
\item[{\bf iv}] the length of the longer path is smaller than $\alpha$
  (a parameter, set to 1000 by default).
\end{itemize}
The last condition imposes an upper bound on the length of the exon or
intron skipped in a AS event, and is necessary due to performance
issues, a larger value considerably increases the running
times. Further criteria are applied to make the algorithm more
efficient without loss of information, and to eliminate bubbles that
do not correspond to AS.

Since the number of cycles in a graph may be exponential in the size
of the graph, the naive approach of listing all cycles of the cDBG and
verifying which of them satisfy our conditions is only viable for very
small cases. Nonetheless, \ks is able to enumerate a potentially
exponential number of bubbles for real-sized dataset in very
reasonable time and memory consumption.  This is in part due to the
fact that, previous to cycle enumeration, the graph is pre-processed
in a way that, along with the pruning criteria of Step 4 (see below),
is responsible for a good performance in practice.

\ks is indeed composed of six main steps which are described next.
The pre-processing just mentioned corresponds to Step 2, and the
enumeration algorithm is described in Step 4. This description
corresponds to \ks version 1.6. A memory efficient replacement for
Step 1 is presented in Chapter~\ref{chap:dbg}. And a time efficient
replacement for Step 4 is presented in
Chapter~\ref{chap:weighted}. The current implementation of \ks
(version 2.0) includes both improvements.

\begin{enumerate}

\item \textit{cDBG construction.} Construction of the cDBG of the
  reads of one or several RNA-seq experiments. The first step is to
  obtain the list of unique $k$-mers, with the corresponding
  multiplicities (coverage), from the reads. This is done, using
  constant memory, by applying an algorithm similar to the
  \emph{external merge-sort} (\cite{Knuth98}) to the multiset of
  $k$-mers. Basically, the method works by partitioning the multiset
  of $k$-mers, and performing several iterations where only a fixed
  amount of $k$-mers is loaded in memory, sorted, and re-written on
  the disk. As a result, we obtain a list of $k$-mers and its
  coverage.  In order to get rid of most of the sequencing errors,
  $k$-mers with a minimal $k$-mer coverage of $mkC$ (a parameter) are
  removed. The second step is to actually build the DBG, this is done
  in the naive way by reading the list of $k$-mers and adding the
  corresponding arcs. In the next step, using a greedy non-branching
  path extension algorithm all maximal non-branching valid paths are
  found. Then, we obtain the cDBG by merging each path into a single
  vertex.

\item \textit{Biconnected component (BCC) decomposition.} As stated in
  Section~\ref{sec:back:graphs}, a connected undirected graph is
  biconnected if it remains connected after the removal of any vertex,
  and a BCC of an undirected graph is a maximal biconnected
  subgraph. Moreover, as stated in Lemma~\ref{lem:back:bcc} the BCCs
  of an undirected graph form a partition of the edges with two
  important properties: every cycle is contained in exactly one BCC,
  and every edge not contained in a cycle forms a singleton BCC.

  From Definition~\ref{def:kissplice:bubble}, it is clear that every
  bubble in a cDBG corresponds to a cycle in the underlying undirected
  graph. Thus, applying on the underlying undirected graph of the cDBG
  Tarjan's lowpoint method (\cite{Tarjan72}) which performs a modified
  depth-first search traversal of the graph, Step 2 detects all BCCs,
  and discards the ones with less than 4 vertices, they cannot contain
  any bubble.  Without modifying the results, this considerably
  reduces the memory footprint and the computation time of the whole
  process. To give an idea of the effectiveness of this step, the cDBG
  of a 5M reads dataset had 1.7M vertices, but the largest BCC only
  2961 vertices.

\item \textit{Simple bubbles compression.} Single substitution events
  (SNPs, sequencing errors) generate a large number of cycles
  themselves included into bigger ones, creating a combinatorial
  explosion of the number of possible bubbles. This step of \ks
  detects and compresses all bubbles composed of just four vertices:
  two switching vertices and two \emph{non-branching internal
    vertices} each corresponding to sequences differing by just one
  position. Fig.~\ref{fig:kissplice:other_bubbles}(a) shows an example
  of a simple bubble. Simple bubbles are output as potential SNPs and
  then replaced by a single vertex in the graph.  The two
  non-branching internal vertices are merged into one, storing a
  consensus sequence where the unique substitution is replaced by N.

\item \textit{Bubble enumeration.}  The cycles are detected in the
  cDBG using a simple backtracking procedure proposed
  by \cite{Tiernan70}, which is an unconstrained DFS augmented with
  four pruning criteria. Indeed, from a path prefix $\pi = s \leadsto
  u$ the algorithm recursively explores the vertices of $N^+(u)$ minus
  the internal vertices of $\pi$. Every time a new vertex $v \in
  N^+(u)$ is added to $\pi$ the algorithm checks whether: $\pi \cdot
  (u,v)$ contains more than two switching vertices, the length of the
  shorter path is greater than $2k-2$, the length of the longer path
  is greater than $\alpha$, or the length of one of the paths is
  smaller than $2k-8$; if any of the conditions is satisfied the
  algorithm stops the recursion on that branch. On the other hand, if
  $\pi \cdot (u,v)$ is a cycle, i.e. $v = s$, and it satisfies the
  conditions (i) to (iv) the algorithm outputs a bubble.

  This approach has the same theoretical time complexity as Tiernan's
  algorithm for cycle listing, i.e. in the worst case the complexity
  is proportional to the number of paths in the graph, which might be
  exponential in the size of the graph and the number of
  bubbles. Tiernan's algorithm is worse than Tarjan's
  (\cite{Tarjan73}) or Johnson's (\cite{Johnson75}) polynomial delay
  algorithms, but it appears to be not immediate how to use the
  pruning criteria with them while preserving their theoretical
  complexity.  Moreover, the pruning criteria are very effective for
  the type of instances we are dealing with. In practice, Tiernan's
  algorithm with prunings is faster than a complete cycle listing
  using Tarjan's or Johnson's with a post-processing step to check the
  four conditions.

\item \textit{Results filtration and classification.} The two paths of
  each bubble are aligned. If the whole of the shorter path aligns
  with high similarity to the longer path, we decide that the bubble
  is due to inexact repeats (see
  Section~\ref{sec:kissplice:tandem_repeats}). After this alignment, a
  bubble is classified either as an SNP, AS event, repeat associated
  bubble, or a small indel.

\item \textit{Read coherence and coverage computation.} Reads from
  each input dataset are mapped to each path of the bubble.  If at
  least one nucleotide of a path is covered by no read, the bubble is
  said to be not \emph{read-coherent} and is discarded. The coverage
  of each position of the bubble corresponds to the number of reads
  overlapping this position. 

\end{enumerate}
 
\section{Results} \label{sec:results}
\subsection{Simulated data}
  
In order to assess the sensitivity and specificity of our approach, we
simulated the sequencing of genes for which we are able to control the
number of alternative transcripts.  We show that the method is indeed
able to recover AS events whenever the alternative transcripts are
sufficiently expressed. For our sensitivity tests, we used simulated
RNA-seq single end reads (75 bp) with sequencing errors. We first
tested a pair of transcripts with a 200 nt skipped exon.  Simulated
reads were obtained with MetaSim (\cite{Richter08}) which is a
reference software for simulating sequencing experiments. As in real
experiments, it produces heterogeneous coverage and authorizes to use
realistic error models.

In order to find the minimum coverage for which we are able to work,
we created datasets for several coverages (from 4X to 20X, which
corresponds to 60 to 300 Reads Per Kilobase or RPK for short), with 3
repetitions for each coverage, and tested them with different values
of $k$ ($k=13, \ldots 41$). The purpose of using 3 repetitions for
each coverage was to obtain results which did not depend on
irreproducible coverage biases. For coverages below 8X (120 RPK), \ks
found the correct event in some but not all of the 3 tested
samples. The failure to detect the event was due to the heterogeneous
and thus locally very low coverage around the skipped exon, e.g.  some
nucleotides were not covered by any read or the overlap between the
reads was smaller than $k$-1.  Above 8X (120 RPK), \ks detected the
correct exon skipping event in all samples.

For each successful test, there was a maximal value $k_{max}$ for $k$
above which the event was not found, and a minimal value $k_{min}$
below which \ks also reported false positive events. Indeed, if $k$ is
too small, then the pattern $ab$, $as'b$, with $|a| \geq k, |b|\geq k$
is more likely to occur by chance in the transcripts, therefore
generating a bubble in the DBG. Between these two thresholds, \ks
found only one event: the correct one.  The values of $k_{min}$ and
$k_{max}$ are clearly dependent on the coverage of the gene. At 8X
(120 RPK), the 200 nucleotides exon was found between $k_{min}=17$ and
$k_{max}=29$. At 20X (300 RPK), it was found for $k_{min}=17$ and
$k_{max}=39$.  We performed similar tests on other datasets, varying
the length of the skipped exon. As expected, if the skipped exon is
shorter (longer), \ks needed a lower (higher) coverage to recover it.

Since \ks is, to our knowledge, the first method able to call AS
events without a reference genome, it cannot be easily benchmarked
against other programs.  Here, we compare it to a general purpose
transcriptome assembler, \tri (\cite{Trinity}). Both methods are
compared only on the specific task of AS event calling. The current
version of \tri being restricted to a fixed value of $k=25$, we
systematically verified that this value was included in
$[k_{min},k_{max}]$.

We found out that \tri was able to recover the AS event in all 3
samples only when the coverage was above 18X (270 RPK), which clearly
shows that \ks is more sensitive for this task.  This can be explained
by the fact that \tri uses heuristics which tend to over-simplify the
cDBG.

All these results were obtained using a minimal $k$-mer coverage
($mkC$ for short) of 1. We also tested with $mkC=2$ (i.e. $k$-mers
present only once in the dataset are discarded), leading to the same
main behavior. We noticed however a loss in sensitivity for both
methods, but a significant gain in the running time. \ks found the
event in all 3 samples for a coverage of 12X (180 RPK) which remains
better than the sensitivity of \tri for $mkC=1$.

\subsection{Real data} \label{sec:kissplice:real_data}

We further tested our method on RNA-seq data from human.  Even though
we do not use any reference genome in our method, we applied it to
cases where an annotated reference genome is indeed available in order
to be able to assess if our predictions are correct.

We ran \ks with $k=25$ and $mkC=2$ on a dataset which consists of 32M
reads from human brain and 39M reads from liver from the Illumina Body
Map 2.0 Project (downloaded from the Sequence Read Archive, accession
number ERP000546).  As in all DBG based assemblers, the most memory
consuming step was the DBG construction which we performed on a
cluster.  The memory requirement is directly dependent on the number
of unique $k$-mers in the dataset.

Despite the fact that we do not use any heuristic to discard $k$-mers
(except for the minimum coverage threshold) from our index, our memory
performances are very similar to the ones of Inchworm, the first step
of \tri, as indicated in Fig.~\ref{fig:inchwormDBG}. In addition, for
the specific task of calling AS events, \ks is faster than \tri as
shown in Fig.~\ref{fig:ksVstritime}.

\ks identified 5923 biconnected components which contained at least
one bubble, 664 of which consisted of bubbles generated by repeats
associated events and 1160 which consisted of bubbles generated by
short indels (less than 3 nt).  Noticeably, the BCCs which generated
most cycles and were most time consuming were associated to
repeats. As these bubbles are not of interest for \ks, this
observation prompted us to introduce an additional parameter in \ks to
stop the computation in a BCC if the number of cycles being enumerated
reaches a threshold. This enabled us to have a significant gain of
time.

Out of the 4099 remaining BCCs, we found that 3657 were read-coherent
(i.e. each nucleotide is covered by at least one read) and we next
focused on this set.  For each of the 3657 cases, we tried to align
the two paths of each bubble to the reference genome using Blat
(\cite{Kent02}). If the two paths align with the same initial and
final coordinates, then we consider that the bubble is a real AS
event. If they align with different initial and final coordinates,
then we consider that it is a false positive.  Out of the 3657 BCCs,
3497 (95\%) corresponded to real AS events, while the remaining
corresponded to false positives. A first inspection of these false
positives led to the conclusion that the majority of them correspond
to chimeric transcripts. Indeed, the shorter path and the longer path
both map in two blocks within the same gene, but the second block is
either upstream of the first block, or on the reverse strand, in both
cases contradicting the annotations and therefore suggesting that the
transcripts are chimeric and could have been generated by a genomic
rearrangement or a trans-splicing mechanism.

For each of the 3497 real cases, we further tried to establish if they
corresponded to annotated splicing events. We therefore first computed
all annotated AS events using AStalavista (\cite{Sammeth08}) and the
UCSC Known Genes annotation (\cite{Kuhn09}).  Then, for each aligned
bubble, we checked if the coordinates of the aligned blocks matched
the splice sites of the annotated AS events. If the answer was
positive, then we considered that the AS event we found was known,
otherwise we considered it was novel.  Out of a total of 3497 cases,
we find that only 1538 are known while 1959 are novel. This clearly
shows that current annotations largely underestimate the number of
alternative transcripts per multi-exon genes as was also reported
recently (\cite{Wang08}).

Additionally, we noticed that 719 BCCs contained more than one AS
event, which all mapped to the same gene. This corresponds to complex
splicing events which involve more than 2 transcripts. Such events
have been described in \cite{Sammeth09}. Their existence suggests that
more complex models could be established to characterize them as one
single event, and not as a collection of simple pairwise events.  An
example of novel complex AS event is given in
Fig.~\ref{fig:complexASevent}.

\begin{figure}[htbp]
\centering
\includegraphics[width=\linewidth]{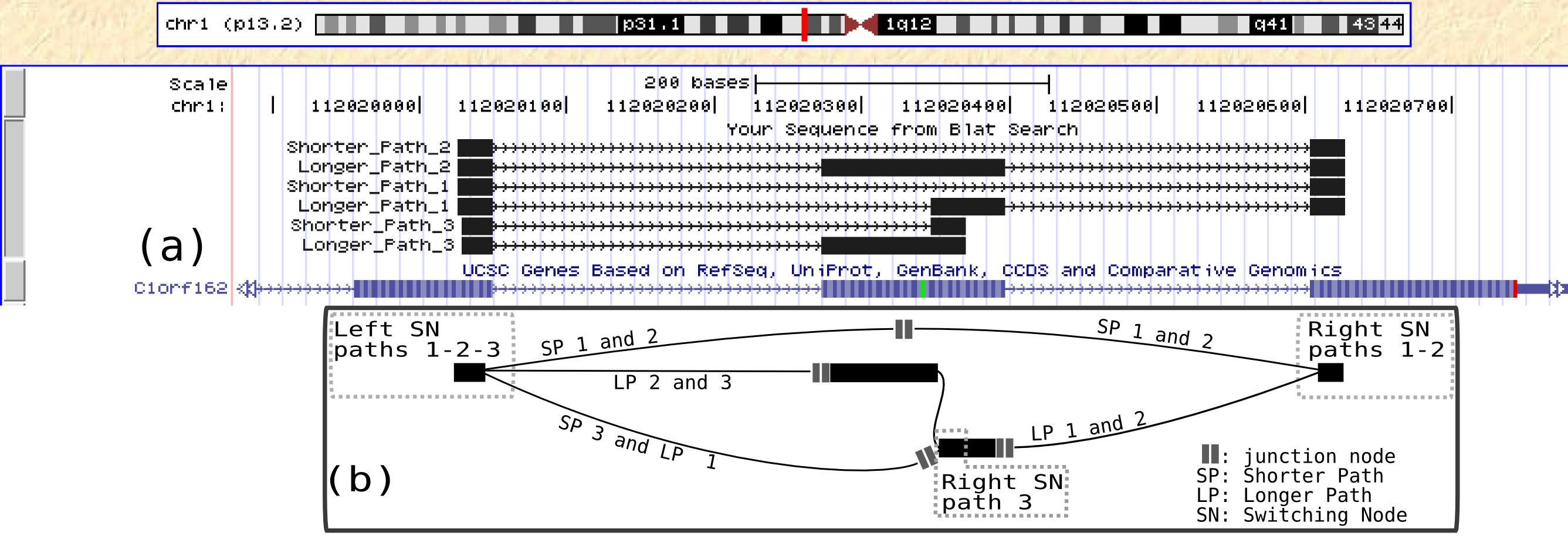}
\caption{BCC corresponding to a novel complex AS event. The
  intermediate annotated exon is either present, partially present, or
  skipped. (a) The annotations (blue track) report only the version
  where it is present while black tracks report all events found by
  \ks. (b) The cDBG associated to this complex event where the
  junction vertices are composed by $2k-2$ nucleotides.}
\label{fig:complexASevent}
\end{figure}

We also found the case where the same AS event maps to multiple
locations on the reference genome (423 cases). We think these
correspond to families of paralogous genes, which are ``collectively''
alternatively spliced. We were able to verify this hypothesis on all
tested instances. In this case, we are unable to decide which of the
genes of the family are producing the alternative transcripts, but we
do detect an AS event.

\subsection{Characterization of novel AS events} \label{subsec:kissplice:novel}

In order to further characterize the 1959 novel AS events we found, we
compared them with annotated events considering their abundance,
length of the variable region and use of splice sites.  For each AS
event, we have 4 abundances, one for each spliceform (i.e. path of the
bubble), and one for each condition. We computed the abundance of an
event as the abundance of the minor spliceform.  As outlined in
Fig.~\ref{fig:abundances}, we show that novel events are less abundant
than annotated events.  This in itself could be one of the reasons why
they had not been annotated so far.  Interestingly, we also found that
while annotated events are clearly more expressed in brain than liver
(median coverage, in reads per nucleotide, of 3.4 Vs 1.2), this trend
was weaker for novel events (2.4 Vs 1.2). This may reflect the fact
that, since tissue-specific splicing in brain has been intensely
studied, annotations may be biased in their favor.

\begin{figure}[htbp]
\centering
\includegraphics[width=0.6\linewidth]{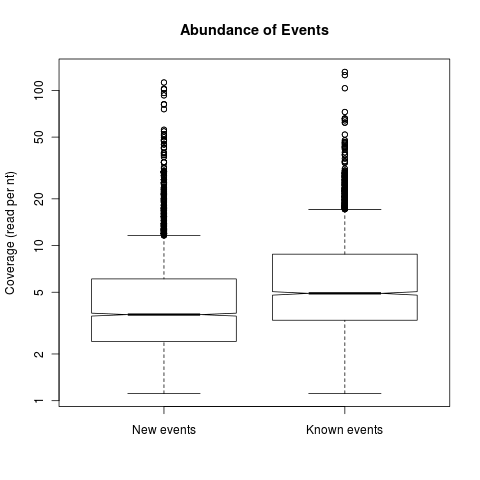}
\caption{Abundance of known and novel events.}
\label{fig:abundances}
\end{figure}

We then computed the length of each event as the difference of the
length between the two paths of the bubble. We found that for
annotated events, there is a clear preference (59\%) for lengths that
are a multiple of 3, which is expected if the event affects a coding
region. However, although still very different from random, this
preference is less strong for novel events (45\%), which, in addition,
are particularly enriched in short lengths as shown in
Fig.~\ref{fig:lengths}.

\begin{figure}[htbp]
\centering
\includegraphics[width=0.6\linewidth]{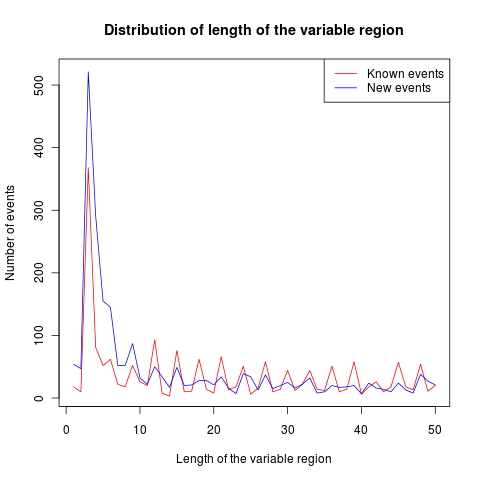}
\caption{Distribution of lengths of the variable regions for known and
  novel events. Only the initial part of the distribution is given.}
\label{fig:lengths}
\end{figure}

Finally, we computed the splice sites of annotated and novel events,
and we found that a vast majority (99.5\%) of known events exhibit
canonical splice sites, while this is again less strong for novel
events (75.3\%). Out of the non canonical cases, 13 correspond to U12
introns, but most correspond to short events.

Altogether, while we cannot discard that short non canonical events do
occur and have been under-annotated so far, we think that the
observations we make on the length and splice site features can be
explained by the presence of genomic indels in our results. We had
indeed already stated in Section~\ref{sec:kissplice:dbg_models} that
while most annotated genomic indels are below 3nt, some may still be
above. In order to assess the proportion of bubbles, with length below
10nt, corresponding to indels and AS events, we mapped them to the
reference genome. The results are shown in
Fig.~\ref{fig:kissplice:indel_vs_as}. It is clear that bubbles with
length smaller than 6nt and not a multiple of 3 are more likely to
correspond to genomic indels than AS events. In \ks (version 2.0) we
changed our criterion to classify events with lengths 1, 2, 4, and 5
nt as indels. Moreover, events larger than 10nt have canonical splice
sites 92.5\% of the cases. More generally, we wish to stress that this
confusion between genomic indels and AS events is currently being made
by all transcriptome assemblers.

\begin{figure}[htbp]
\centering
\includegraphics[width=0.9\linewidth]{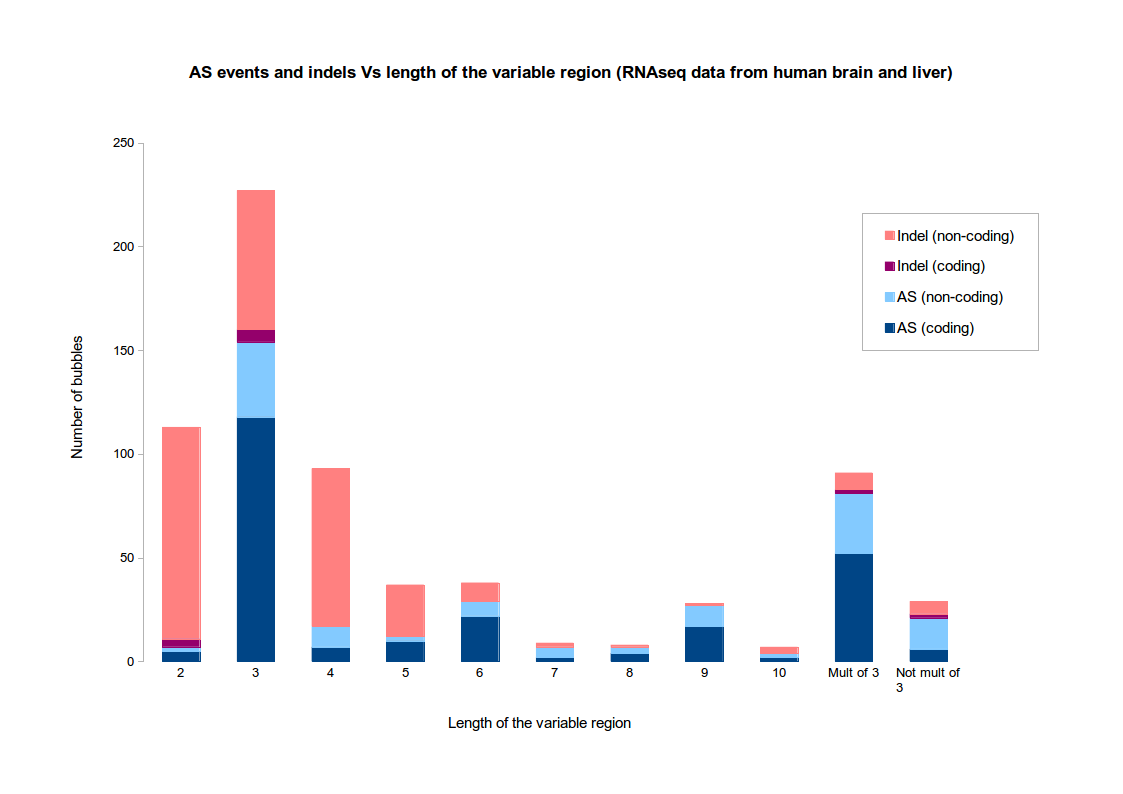}
\caption{Distribution of bubbles corresponding to alternative
  splicing events and indels, according to the length of the variable
  region.}
\label{fig:kissplice:indel_vs_as}
\end{figure}

\subsection{Comparison with \tri}
Finally, in order to further discuss the sensitivity of our method on
real data, we compared our results with \tri. Although \tri is not
tailored to find AS events, we managed to retrieve this information
from the output. Whenever \tri found several alternative transcripts
for one gene, we selected this gene. We further focused on cases which
contained a cycle in the splicing graph reconstructed from this gene
and we compared them with the events found by \ks.  Whenever we found
that both the longer and the shorter path of a bubble were mapping to
the transcripts of a \tri gene, we decided that both methods had found
the same event. In total, \ks found 4099 cases, \tri found 1123 out of
which 553 were common. While the sensitivity is overall larger for
\ks, we see that 570 cases are found by Trinity and not by \ks.  We
then mapped these transcripts to the human genome using Blat. In many
instances (348 cases), the transcripts did not align on their entire
length, or to different chromosomes, indicating that they corresponded
to chimeras.  A first inspection of the remaining 222 cases revealed
that they correspond to the complex BCCs we chose to neglect at an
early stage of the computation, because they contain a very large
number of repeat-associated bubbles. A first simple way to deal with
this issue is to increase the value of $k$. The effect of this is to
break the large BCCs into computable cases, enabling to recover a good
proportion of the missed events. For instance, for $k=35$, we found
back 84 cases.  More generally, this shows that more work on the model
and on the algorithms is still required to characterize better AS
events which are intricate with inexact repeats. We think that \tri
manages to identify some of them because it uses heuristics, which
enables it to simplify these complex graph structures.

\begin{figure}[htbp]
\centering
\subfloat[]{\label{fig:inchwormDBG} \includegraphics[width=0.5\linewidth]{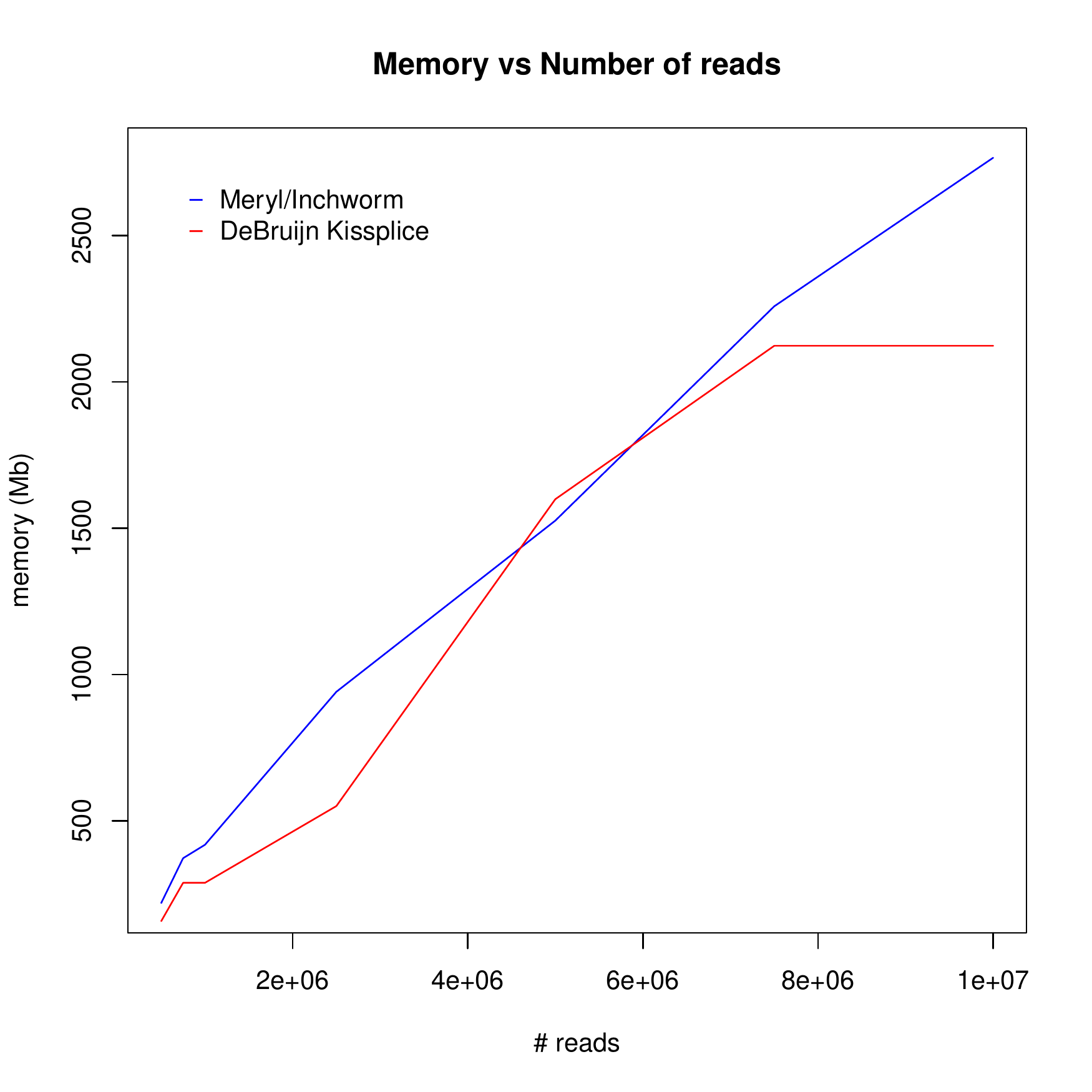}}
\subfloat[]{\label{fig:ksVstritime} \includegraphics[width=0.5\linewidth]{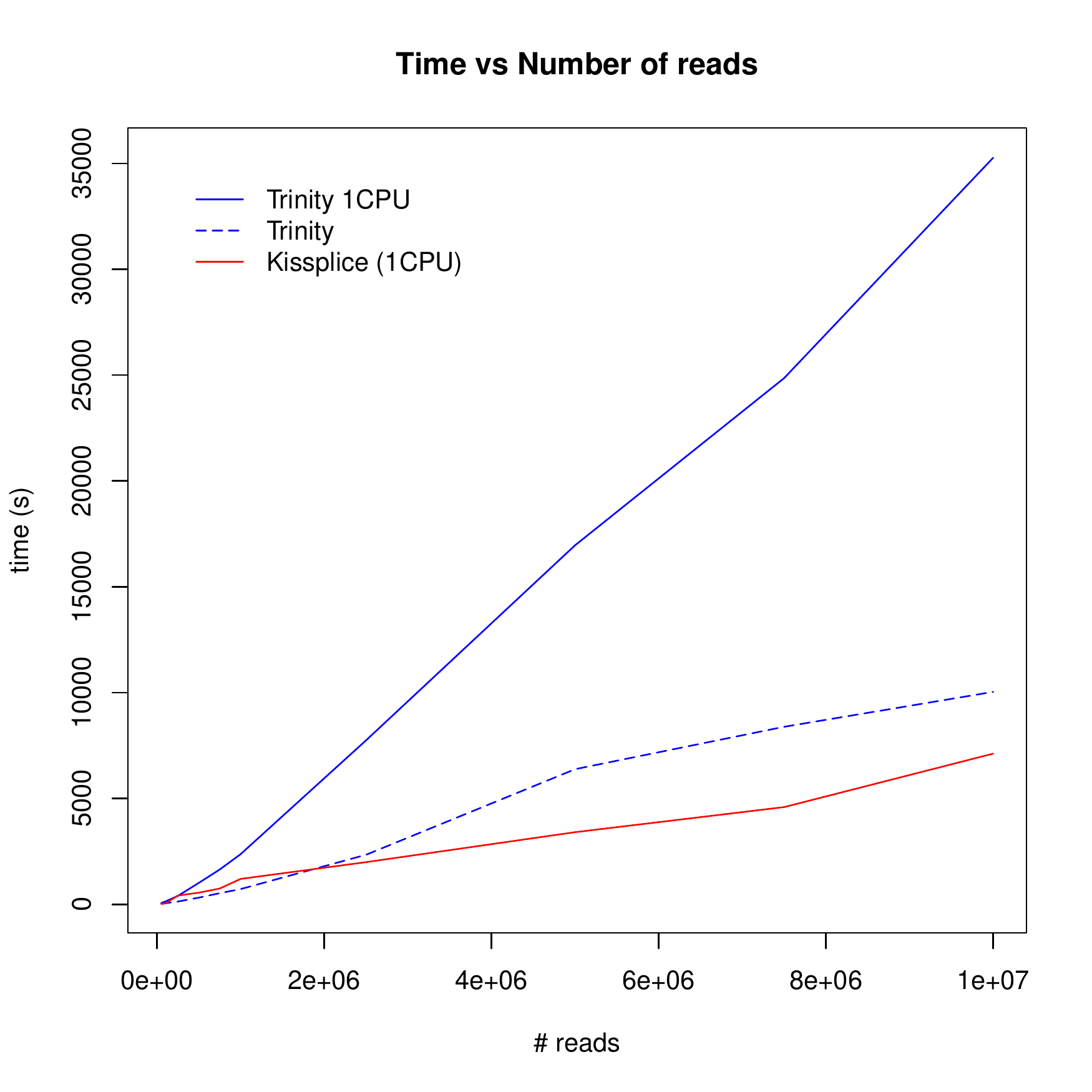}}
\caption{(a) Memory usage of \ks and Inchworm as a function of input
  size. (b) Time performances of \ks and \tri as a function of input
  size.}
\end{figure}

\section{Discussion and conclusions}

This chapter presents two main contributions. First, we introduced a
general model for detecting variations in de Bruijn graphs, and second,
we developed an algorithm, \ks, to detect AS events in such
graphs. This approach enables to tackle the problem of finding AS
events without assembling the full-length transcripts, which may be
time consuming and uses heuristics that may lead to a loss of
information. To our knowledge, this approach is new and should
constitute a useful complement to general purpose transcriptome
assemblers.

Results on human data show that this approach enables de novo calling
of AS events with a higher sensitivity than obtained by the approaches
based on a full assembly of the reads, while using similar memory
requirements and less time. 5\% of the extracted events correspond to
false positives, while the 95\% remaining can be separated into known
(44\%) and novel events (56\%). Novel events exhibit similar sequence
features as known events as long as we focus on events longer than 10
bp. Below this, novel events seem to be enriched in genomic indels.

\ks is an user-friendly tool under active development available for
download at \url{http://kissplice.prabi.fr/}, which is mature enough
to be used in real life projects to establish a more complete catalog
of AS events in any species, whether it has a reference genome or
not.  Despite the fact that more and more genomes are now being
sequenced, the new genome assemblies obtained usually do not reach the
level of quality of the ones we have for model organisms. Hence, we
think that methods which do not rely on a reference genome are not
going to be easily replaced in the near future.

There is of course room for further improvements. For instance, the
current bubble listing algorithm, the core of the \ks pipeline, is not
entirely satisfactory. In Chapter~\ref{chap:unweighted}, we present a
linear delay algorithm to list all cycles satisfying condition (i) of
Section~\ref{sec:kissplice:algorithm}, that is to directly list all
bubbles in a de Bruijn graph. In Chapter~\ref{chap:weighted}, we
propose an improved, completely unrelated, polynomial delay algorithm
to list all cycles satisfying conditions (i), (ii) and (iv), and
experimentally show that this method outperforms the algorithm of
Section~\ref{sec:kissplice:algorithm}.

Another point not satisfactory in the initial \ks (version 1.6)
pipeline is the memory consumption. As stated in
Section~\ref{sec:kissplice:real_data}, the memory bottleneck is the de
Bruijn graph construction. We address this issue in
Chapter~\ref{chap:dbg} where we propose a practical algorithm to build
the de Bruijn improving over the state of the art.

In addition, the coverage could be used for distinguishing SNPs from
sequencing errors, and the splicing site signature, i.e. canonical
splicing sites (\cite{Burset00}) GT-AG, could be used to distinguish
between intron retention and the others AS events.  Moreover, the
sequences surrounding the bubbles could be locally assembled using a
third party tool (\cite{Mapsembler}). This would allow to output their
context or the full contig they belong to.

Last, the complex structure of BCCs associated to repeats seems to
indicate that more work on the model and on the algorithms is required
to efficiently deal with the identification of repeat associated
bubbles, which may be highly intertwined with other events.

\chapter{Listing in unweighted graphs}
\label{chap:unweighted}
\minitoc
In this chapter, we are mainly concerned with listing problems
in \emph{unweighted} graphs. In directed graphs, we consider the
problem of listing bubbles, defined as a pair of internally
vertex-disjoint paths (Chapter~\ref{chap:kissplice}). In undirected
graphs, we consider the classical problem of listing $st$-path and
cycles. The chapter is divided in two main parts.

The first part (Section~\ref{sec:unweighted:bubble}) is strongly based
on our paper \cite{Birmele12}. The goal is to show a non-trivial
adaptation of Johnson's cycle\footnote{Johnson uses the
term \emph{elementary circuits}.}  listing algorithm
(\cite{Johnson75}) to identify all bubbles in a directed graph
maintaining the same complexity. For a directed graph with $n$
vertices and $m$ arcs, containing $\eta$ bubbles, the method we
propose lists all bubbles with a given source in $O((n + m)(\eta +
1))$ total time and $O(m+n)$ delay. For the general problem of listing
bubbles, this algorithm is exponentially faster than the algorithm
based on Tiernan's algorithm (\cite{Tiernan70}) presented in
Chapter~\ref{chap:kissplice}. However, it should be noted that,
contrary to Chapter~\ref{chap:kissplice}, the graph here is not a
bidirected de Bruijn graph.

The second part (Section~\ref{sec:unweighted:cycle}) is strongly based
on our paper \cite{Birmele13}. The goal is to show an algorithm to
list cycles in undirected graphs improving over the state of the art
(Johnson's algorithm). Indeed, we present the first optimal solution
to list all the simple cycles in an undirected graph
$G$. Specifically, let $\setofcycles(G)$ denote the set of all these
cycles.  For a cycle $c \in \setofcycles(G)$, let $|c|$ denote the
number of edges in~$c$. Our algorithm requires $O(m
+ \sum_{c \in \setofcycles(G)}{|c|})$ time and is asymptotically
optimal: $\Omega(m)$ time is necessarily required to read $G$ as
input, and $\Omega(\sum_{c \in \setofcycles(G)}{|c|})$ time is
required to list the output.  We also present the first optimal
solution to list all the simple paths from $s$ to $t$ in an undirected
graph $G$.


\bigskip
\bigskip


\section{Efficient bubble enumeration in directed graphs} \label{sec:unweighted:bubble}
\subsection{Introduction}

In the previous chapter, a method (\ks) to identify variants
(alternative splicing events, SNPs, indels and inexact tandem repeats)
in RNA-seq data without a reference genome was introduced. Each
variant corresponds to a recognizable pattern in a (bidirected) de
Bruijn graph built from the reads of the RNA-seq experiment.  In each
case, the pattern corresponds to a bubble defined as two
vertex-disjoint paths between a pair of source and target vertices $s$
and $t$. Properties on the lengths or sequence similarity of the paths
then enable to differentiate between the different types of variants.

Bubbles have been studied before in the context of genome assembly
(\cite{Idba,Soapdenovo,Abyss,Velvet}) where they also have been called
bulges (\cite{Pevzner04}).  However, the purpose in these works was
not to list all bubbles, but ``only'' to remove them from the graph in
order to provide longer contigs for a genome assembly.  More recently,
ad-hoc listing methods have been proposed but are restricted to
(almost) non-branching bubbles
(\cite{Peterlongo10,Cortex,Bubbleparse}), i.e.  each vertex from the
bubble has in-degree and out-degree 1, except for $s$ and
$t$. Furthermore, in all these applications
(\cite{Pevzner04,Velvet,Abyss,Idba,Soapdenovo,Cortex,Bubbleparse}),
since the patterns correspond to SNPs or sequencing errors, the
authors only considered paths of length smaller than a constant.

On the other hand, bubbles of arbitrary length have been considered in
the context of splicing graphs (\cite{Sammeth09}). However, in this
context, a notable difference is that the graph is a
DAG. Additionally, in the case of
\cite{Cortex} the vertices are colored and only unicolor paths are
then considered for forming bubbles. Finally, the concept of bubble
also applies to the area of phylogenetic networks (\cite{Gusfield04}),
where it corresponds to the notion of a recombination cycle.  Again
for this application, the graph is a DAG. To our knowledge, no
enumeration algorithm for recombination cycles has been proposed.

In this chapter, we consider the more general problem of listing all
bubbles in an arbitrary directed graph. That is, our solution is not
restricted to acyclic or de Bruijn graphs, neither imposes
restrictions on the path length or the degrees of the internal
nodes. This problem is quite general but it remained an open question
whether a polynomial delay algorithm could be proposed for solving it.
The algorithm briefly presented in Chapter~\ref{chap:kissplice} (also
in \cite{Sacomoto12}) was an adaptation of Tiernan's algorithm for
cycle listing (\cite{Tiernan70}) which is not polynomial
delay. Actually, since in the worst case Tiernan's algorithm can
explore all the $st$-paths while the graph only contains a constant
number of cycles, the algorithm of Chapter~\ref{chap:kissplice} is not
even polynomial total time. The time spent by the algorithm is, in the
worst case, exponential in the size of the input graph and the number
of bubbles output.

The first part of this chapter is organized as follows. We start by
discussing in Section~\ref{sec:DBG-bubbles} the correspondence between
bubbles in bidirected de Bruijn graphs (Chapter~\ref{chap:kissplice})
and directed de Bruijn graphs (Chapter~\ref{chap:back}). We then
explain in Section~\ref{sec:bubble-cycle} how to transform the
directed graph where we want to list the bubbles into a new directed
graph such that the bubbles correspond to cycles satisfying some extra
properties. We present in Section~\ref{sec:algo} the algorithm to list
all cycles corresponding to bubbles in the initial graph and prove in
Section~\ref{sec:complexity} that this algorithm has linear delay.
Finally, we briefly describe, in Section~\ref{sec:oneone}, a slightly
more complex version of the algorithm that could lead to a more space
and time efficient implementation, but with the same overall
complexity.

\subsection{De Bruijn graphs and bubbles}
\label{sec:DBG-bubbles}

In the previous chapter, we defined bubbles
(Definition~\ref{def:kissplice:bubble}) as a pair of vertex-disjoint
valid paths in a bidirected de Bruijn graph. Recall that, a bidirected
de Bruijn graph is directed multigraph where each vertex is labeled by
a $k$-mer and its reverse complement and the arcs represent a $k-1$
suffix-prefix overlap and are labeled depending on which $k$-mer,
forward or reverse, the overlap refers to, whereas a (directed) de
Bruijn graph (Definition~\ref{def:dbg}) is a directed graph where each
vertex is labeled by a $k$-mer and the arcs correspond to $k-1$
suffix-prefix overlaps. Here, we consider bubbles in a directed de
Bruijn graph.

\begin{definition}[$(s,t)$-bubble] \label{def:unweighted:bubble}
  Given a directed graph $G = (V,E)$, an $(s,t)$-bubble is a pair of
  internally vertex-disjoint $st$-paths.
\end{definition}

Both de Bruijn graph definitions are roughly equivalent. Indeed, given
a bidirected DBG we transform it into a regular DBG by splitting every
vertex in two vertices, one corresponding to the forward $k$-mer and
the other to the reverse $k$-mer, and maintaining the arcs
accordingly.  This transformation, however, does not induce a
one-to-one correspondence between bubbles in the bidirected DBG
(Definition~\ref{def:kissplice:bubble}) and $(s,t)$-bubbles in the
corresponding DBG.  Indeed, every valid path in the bidirected DBG
corresponds to a simple path in the directed DBG, but the converse is
not true, a simple path in the directed DBG containing a $k$-mer and
its reverse complement is not a valid path in the bidirected DBG,
implying that, every bubble in a bidirected DBG corresponds to a
$(s,t)$-bubble in the directed DBG, but the converse is not true. We,
however, disregard this nonequivalence, since no true bubble is lost
by considering the directed DBG.

From now on, we consider the more general problem of listing bubbles
in an arbitrary directed graph, not necessary a DBG.

\begin{problem}[Listing bubbles] \label{prob:unweighted:bubbles}
  Given a directed graph $G = (V,E)$, output all $(s,t)$-bubbles in
  $G$, for all pairs $s,t \in V$.
\end{problem}

In order to solve Problem~\ref{prob:unweighted:bubbles}, we consider
the problem of listing all bubbles with a given source
(Problem~\ref{prob:unweighted:stbubbles}). Indeed, by trying all
possible sources $s$ we can list all $(s,t)$-bubbles.

\begin{problem}[Listing $(s,*)$-bubbles] \label{prob:unweighted:stbubbles}
  Given a directed graph $G = (V,E)$ and vertex $s$, output all
  $(s,t)$-bubbles in $G$, for all $t \in V$.
\end{problem}

The number of vertices and arcs of $G$ is denoted by $n$ and $m$,
respectively.

\subsection{Turning bubbles into cycles}
\label{sec:bubble-cycle}

Let $G=(V,E)$ be a directed graph, and let $s \in V$. We want to find
all $(s,t)$-bubbles for all possible target vertices $t$. We transform
$G$ into a new graph $G'_s=(V'_s,E'_s)$ where $|V'_s| = 2 |V|$ and
$|E'_s| = O(|V|+|E|)$. Namely,

$$V'_s=\{ v, \overline{v} \ | \ v \in V\}$$
$$E'_s= \{(u,v),(\overline{v},\overline{u}) \ | \ (u,v) \in E \textrm{ and }v\neq s\} \cup \{(v,\overline{v} )\ | \ v \in V\textrm{ and }v\neq s \} \cup \{(\overline{s},s)\}$$

Let us denote by $\overline{V}$ the set of vertices of $G'_s$ that
were not already in $G$, that is $\overline{V}=V' _s\setminus V$.  The
two vertices $x\in V$ and $\overline{x} \in \overline{V}$ are said to
be \emph{twin vertices}.  Observe that the graph $G'_s$ is thus built
by adding to $G$ a reversed copy of itself, where the copy of each
vertex is referred to as its \emph{twin}. The arcs incoming to $s$
(and outgoing from $\overline{s}$) are not included so that the only
cycles in $G'_s$ that contain $s$ also contain $\overline{s}$. New
arcs are also created between each pair of twins: the new arcs are the
ones leading from a vertex $u$ to its twin $\bar{u}$ for all $u$
except for $s$ where the arc goes from $\overline{s}$ to $s$. An
example of a transformation is given in Figure~\ref{fig:graphtrans}.

\begin{figure}[htb]
\centering
\subfloat[\label{fig:graph} Graph $G$]{

\begin{tikzpicture}
[nodeDecorate/.style={shape=circle,inner sep=1pt,draw,thick,fill=black},%
  lineDecorate/.style={thick},%
  scale=0.7]

\node (s) at (0,-1)
[nodeDecorate,label=below left:$s$,font=\small ] {};

\foreach \nodename/\x/\y in {
a/2/1, b/4/1, c/6/1, d/8/1, e/5/3}
{
  \node (\nodename) at (\x,-\y)
  [nodeDecorate,label=below right:$\nodename$,font=\small ] {};
}

\path [->]
\foreach \startnode/\endnode in {
s/a, a/b, b/c, c/d}
{
  (\startnode) edge[lineDecorate] node {} (\endnode)
};

\path [->]
\foreach \startnode/\endnode in {
s/e, b/e, e/d, e/c}
{
  (\startnode) edge[lineDecorate, bend left=-30] node {} (\endnode)
};

\path [->]
\foreach \startnode/\endnode in {
b/s}
{
  (\startnode) edge[lineDecorate, bend left=40] node {} (\endnode)
};
\end{tikzpicture}
} \quad\quad
\subfloat[\label{fig:trans} Graph $G'_s$]{

\begin{tikzpicture}
[nodeDecorate/.style={shape=circle,inner sep=1pt,draw,thick,fill=black},%
  lineDecorate/.style={thick},%
  scale=0.7]
  
  \node (s) at (0,-1)
  [nodeDecorate,label=below left:$s$,font=\small ] {};
  \node (s2) at (0,1)
  [nodeDecorate,label=above left:$\overline{s}$,font=\small ] {};

\foreach \nodename/\x/\y in {
a/2/1, b/4/1, c/6/1, d/8/1, e/5/3}
{
  \node (\nodename) at (\x,-\y)
  [nodeDecorate,label=below right:$\nodename$,font=\small ] {};
  \node (\nodename2) at (\x,\y)
  [nodeDecorate,label=above right:$\overline{\nodename}$,font=\small ] {};
}

\path [->]
\foreach \startnode/\endnode in {
s/a, a/b, b/c, c/d}
{
  (\startnode) edge[lineDecorate] node {} (\endnode)
  (\endnode2) edge[lineDecorate] node {} (\startnode2)
};

\path [->]
\foreach \startnode/\endnode in {
s/e, b/e, e/d, e/c}
{
  (\startnode) edge[lineDecorate, bend left=-30] node {} (\endnode)
  (\endnode2) edge[lineDecorate, bend right=30] node {} (\startnode2)
};

\path [->]
\foreach \startnode in {
a, b, c, d, e}
{
  (\startnode) edge[lineDecorate,dashed] node {} (\startnode2)
};

\path [->]
\foreach \startnode in {
s}
{
  (\startnode2) edge[lineDecorate,dashed] node {} (\startnode)
};
\end{tikzpicture}
}
\caption{\label{fig:graphtrans}Graph $G$ and its transformation
  ${G'}_s$.  We have that $\langle
  s,{e},\overline{e},\overline{b},\overline{a},\overline{s},s \rangle$
  is a bubble-cycle with swap arc $({e},\overline{e})$ that has a
  correspondence to the $(s,e)$-bubble composed by the two
  vertex-disjoint paths $\langle s,e \rangle$ and $\langle s,a,b,e
  \rangle$.  }
\end{figure}
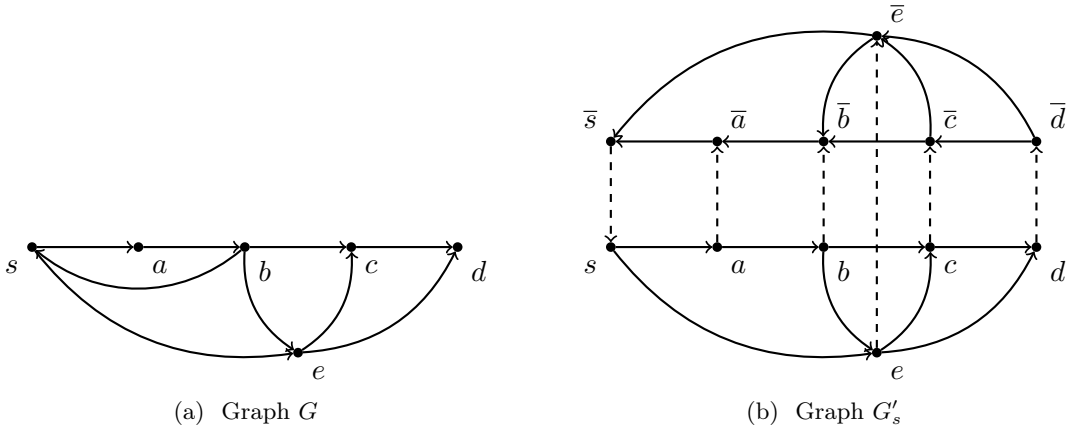

We define a cycle of $G'_s$ as being {\em bipolar} if it contains
vertices of both $V$ and $\overline{V}$. As the only arc from
$\overline{V}$ to $V$ is $(\overline{s},s)$, then every bipolar cycle
$C$ contains also only one arc from $V$ to $\overline{V}$. This arc,
which is the arc $(t,\overline{t})$ for some $t\in V$, is called the
{\em swap arc} of $C$. Moreover, since $(\bar{s},s)$ is the only
incoming arc of $s$, all the cycles containing $s$ are bipolar. We say
that $C$ is {\em twin-free} if it contains no pair of twins except for
$(s,\overline{s})$ and $(t,\overline{t})$.

\begin{definition}[Bubble-cycle] A \emph{bubble-cycle} in $G'_s$ is a twin-free 
cycle of size greater than four\footnote{The only twin-free cycles in
  of size four in $G'_s$ are generated by the outgoing edges of
  $s$. There are $O(|V|)$ of such cycles.}.
\end{definition}

\begin{proposition}\label{onetotwoproposition}
Given a vertex $s$ in $G$, there is a one-to-two correspondence
between the set of $(s,t)$-bubbles in $G$ for all $t \in V$, and the
set of bubble-cycles of $G'_s$.
\end{proposition}

\begin{proof}
Let us consider an $(s,t)$-bubble in $G$ formed by two vertex-disjoint
$st$-paths $P$ and $Q$. Consider the cycle of $G'_s$ obtained by
concatenating $P$ (resp. $Q$), the arc $(t,\overline{t})$, the
inverted copy of $Q$ (resp. $P$), and the arc $(\overline{s},s)$. Both
cycles are bipolar, twin-free, and have $(t,\overline{t})$ as swap
arc.  Therefore both are bubble-cycles.

Conversely, consider any bubble-cycle $C$ and let $(t,\overline{t})$
be its swap arc. $C$ is composed by a first subpath $P$ from $s$ to
$t$ that traverses vertices of $V$ and a second subpath $\overline{Q}$
from $\overline{t}$ to $\overline{s}$ composed of vertices of
$\overline{V}$ only. By definition of $G'_s$, the arcs of the subpath
$P$ form a path from $s$ to $t$ in the original graph $G$; given that
the vertices in the subpath $\overline{Q}$ from $\overline{t}$ to
$\overline{s}$ are in $\overline{V}$ and use arcs that are those of
$E$ inverted, then $Q$ corresponds to another path from $s$ to $t$ of
the original graph $G$. As no internal vertex of $\overline{Q}$ is a
twin of a vertex in $P$, these two paths from $s$ to $t$ are
vertex-disjoint, and hence they form an $(s,t)$-bubble.

Notice that there is a cycle $s,v,\overline{v},\overline{s}$ for each
$v$ in the out-neighborhood of $s$. Such cycles do not correspond to
any bubble in $G$, and the condition on the size of $C$ allows us to
rule them out.
\end{proof}


\subsection{The algorithm}
\label{sec:algo}

\cite{Johnson75} introduced a polynomial delay algorithm for the cycle
enumeration problem in directed graphs. We propose to adapt the
principle of this algorithm, the pruned backtracking, to enumerate
bubble-cycles in $G'_s$. Indeed, we use a similar pruning strategy,
modified to take into account the twin nodes.
Proposition~\ref{onetotwoproposition} then ensures that running our
algorithm on $G'_s$ for every $s\in V$ is equivalent to the
enumeration of (twice) all the bubbles of $G$.  To do so, we explore
$G'_s$ by recursively traversing it while maintaining the following
three variables. We denote by $N^+(v)$ the set of out-neighbors and
$N^-(v)$ as the set of in-neighbors of $v$.

\begin{enumerate}
\item A variable {\em stack} which contains the vertices of a path
  (with no repeated vertices) from $s$ to the current vertex. Each
  time it is possible to reach $\overline{s}$ from the current vertex
  by satisfying all the conditions to have a bubble-cycle, this stack
  is completed into a bubble-cycle and its content output.
\item A variable {\em status$(v)$} for each vertex $v$ which can take
  three possible values:
\begin{description}
\item[$free$:] $v$ should be explored during the traversal of $G'_s$; 
\item[$blocked$:] $v$ should not be explored because it is already in
  the stack or because it is not possible to complete the current
  stack into a cycle by going through $v$ -- notice that the key idea
  of the algorithm is that a vertex may be blocked without being on
  the stack, avoiding thus useless explorations;
\item[$twinned$:] $v\in \overline{V}$ and its twin is already in the
  stack, so that $v$ should not be explored.
\end{description}
\item A set $B(v)$ of in-neighbors of $v$ where vertex $v$ is blocked
  and for each vertex $w \in B(v)$ there exists an arc $(w,v)$ in
  $G'_s$ (that is, $w \in N^-(v)$). If a modification in the stack
  causes that $v$ is unblocked and it is possible to go from $v$ to
  $\bar{s}$ using free vertices, then $w$ should be unblocked if it is
  currently blocked.
\end{enumerate}

Algorithm~\ref{alg:main} enumerates all the bubble-cycles in $G$
(Problem~\ref{prob:unweighted:bubbles}) by fixing the source $s$ of
the $(s,t)$-bubble, computing the transformed graph $G'_s$ and then
listing all bubble-cycles with source $s$ in $G'_s$
(Problem~\ref{prob:unweighted:stbubbles}).  This procedure is repeated
for each vertex $s \in V$. To list the bubble-cycles with source $s$,
procedure $\cyclebubble(s)$ is called.  As a general approach,
Algorithm~\ref{cyclealgo} uses classical backtracking with a pruned
search tree. The root of the recursion corresponds to the enumeration
of all bubble-cycles in $G'_s$ with starting point $s$. The algorithm
then proceeds recursively: for each free out-neighbor $w$ of $v$ the
algorithm enumerates all bubble-cycles that have the vertices in the
current stack plus $w$ as a prefix. If $v\in V$ and $\overline{v}$ is
twinned, the recursion is also applied to the current stack plus
$\overline{v}$, $(v,\overline{v})$ becoming the current swap arc. A
base case of the recursion happens when $\overline{s}$ is reached and
the call to $\cyclebubble(\overline{s})$ completed. In this case, the
path in
\emph{stack} is a twin-free cycle and, if this cycle has more than 4
vertices, it is a bubble-cycle to output.

The key idea that enables to make this pruned backtracking efficient
is the block-unblock strategy. Observe that when $\cyclebubble(v)$ is called,
$v$ is pushed in the stack and to ensure twin-free extensions, $v$ is
blocked and $\bar{v}$ is twinned if $v\in V$. Later, when
backtracking, $v$ is popped from the stack but it is \emph{not
  necessarily} marked as free. If there were no twin-free cycles with
the vertices in the current stack as a prefix, the vertex $v$ would
remain blocked and its status would be set to free only at a later
stage.  The intuition is that either $v$ is a dead-end or there remain
vertices in the stack that block all twin-free paths from $v$ to
$\overline{s}$. In order to manage the status of the vertices, the
sets $B(w)$ are used.  When a vertex $v$ remains blocked while
backtracking, it implies that every out-neighbor $w$ of $v$ has been
previously blocked or twinned. To indicate that each out-neighbor $w
\in N^+(v)$ (also, $v \in N^-(w)$ is an \emph{in-neighbor} of $w$)
blocks vertex $v$, we add $v$ to each $B(w)$. When, at a later point
in the recursion, a vertex $w \in N^+(v)$ becomes unblocked, $v$ must
also be unblocked as possibly there are now bubble-cycles that include
$v$.  Algorithm~\ref{alg:unblock} implements this recursive unblocking
strategy.

\begin{algorithm}
\caption{Main algorithm} \label{alg:main}
\DontPrintSemicolon
\For{$s\in V$}{
  stack = $\emptyset$\;
  \For{$v\in G'_s$}{
    $status(v) = free$\;
    $B(v) = \emptyset$\; 
  }
  $\cyclebubble(s)$\;
}
\end{algorithm}

\begin{algorithm}
\caption{Procedure $\unblock(v)$} \label{alg:unblock}
\DontPrintSemicolon
  \tcc{recursive unblocking of vertices for which popping $v$ creates a path to $\overline{s}$}
  $status(v) = free$\;

  \For{$w\in B(v)$}{
    delete $w$ from $B(v)$\;
    \If {$status(w) = blocked$}{
      $\unblock(w)$\;
    }
  }
\end{algorithm}

\begin{algorithm}
\caption{Procedure $\cyclebubble(v)$} \label{cyclealgo}
\DontPrintSemicolon
$f$ = false\;
push $v$\; \nllabel{stacking}
$status(v) = blocked$\; \nllabel{blockstatus}

\tcc{Exploring forward the edges going out from $v\in V$}
\If {$v\in V$}{
  \If{$status(\overline{v}) = free$}{
    $status(\overline{v}) = twinned$\;
  }
  \For {$w\in N^+(v)\cap V$}{
    \If{$status(w) = free$}{
      \If{$\cyclebubble(w)$}{ $f$ = true\;}
    }
  } 
  \If{$status(\overline{v}) = twinned$}{ \nllabel{condtwintoblock}
    \If{$\cyclebubble(\overline{v})$ }{ \nllabel{twintoblock}
      $f$ = true\;  
    }
  }
}

\tcc{Exploring forward the edges going out from $v \in \overline{V}$}
\Else { \nllabel{swapvertex}
  \For {$w\in N^+(v)$ }{
       \If{$w = \overline{s}$}{
          output the cycle composed by the stack followed by $\overline{s}$ and $s$\;
          $f$ = true\;
       }       
       \ElseIf{status$(w)$ = free}{  \nllabel{swapeffect}
         \If{$\cyclebubble(w)$}{$f$ = true\;}
       }
  } 
}

\If{$f$}{ \nllabel{backtracking}
  $\unblock(v)$\; \nllabel{unblocking}  
}
\Else{
  \For{$w\in N^+(v)$}{
     \If{$v \notin B(w)$}{ 
       $B(w) = B(w) \cup \{ v \}$\; 
     }
  }
}

pop $v$\;
return $f$\;
\end{algorithm}

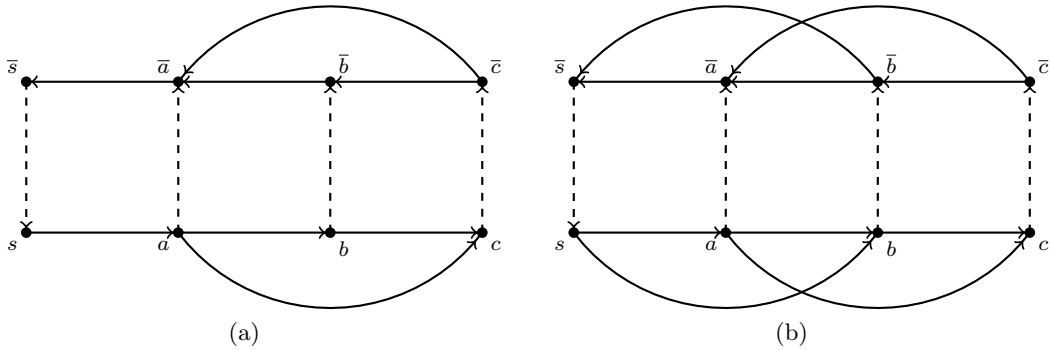
\begin{figure}[htb]
\centering
\subfloat[\label{fig:counter_ex_twined} ]{
\definecolor{qqqqff}{rgb}{0,0,0}
\definecolor{cqcqcq}{rgb}{0,0,0}
\begin{tikzpicture}
[nodeDecorate/.style={shape=circle,inner sep=1pt,draw,thick,fill=black},%
  lineDecorate/.style={thick},%
  scale=1,line width=0.8pt]
\draw [->-,dashed] (-3,4) -- (-3,2);
\draw [->-,dashed] (-1,2) -- (-1,4);
\draw [->-,dashed] (1,2) -- (1,4);
\draw [->-,dashed] (3,2) -- (3,4);
\draw [->-] (-3,2) -- (-1,2);
\draw [->-] (-1,2) -- (1,2);
\draw [->-] (1,2) -- (3,2);
\draw [->-] (3,4) -- (1,4);
\draw [->-] (1,4) -- (-1,4);
\draw [->-] (-1,4) -- (-3,4);
\draw [->-,shift={(1,2.5)}] plot[domain=0.64:2.5,variable=\t]({1*2.5*cos(\t r)+0*2.5*sin(\t r)},{0*2.5*cos(\t r)+1*2.5*sin(\t r)});
\draw [->-,shift={(1,3.5)}] plot[domain=3.79:5.64,variable=\t]({1*2.5*cos(\t r)+0*2.5*sin(\t r)},{0*2.5*cos(\t r)+1*2.5*sin(\t r)});
\begin{scriptsize}
\fill [color=qqqqff] (-3,2) circle (2pt);
\draw[color=qqqqff] (-3,2) node[below left] {$s$};
\fill [color=qqqqff] (-1,2) circle (2pt);
\draw[color=qqqqff] (-1,2) node[below left] {$a$};
\fill [color=qqqqff] (1,2) circle (2pt);
\draw[color=qqqqff] (1,2) node[below right] {$b$};
\fill [color=qqqqff] (3,2) circle (2pt);
\draw[color=qqqqff] (3,2) node[below right] {$c$};
\fill [color=qqqqff] (-3,4) circle (2pt);
\draw[color=qqqqff] (-3,4) node[above left] {$\overline{s}$};
\fill [color=qqqqff] (-1,4) circle (2pt);
\draw[color=qqqqff] (-1,4) node[above left] {$\overline{a}$};
\fill [color=qqqqff] (1,4) circle (2pt);
\draw[color=qqqqff] (1,4) node[above right] {$\overline{b}$};
\fill [color=qqqqff] (3,4) circle (2pt);
\draw[color=qqqqff] (3,4) node[above right] {$\overline{c}$};
\end{scriptsize}
\end{tikzpicture}
}
\subfloat[\label{fig:counter_ex_visits} ]{
\definecolor{qqqqff}{rgb}{0,0,0}
\definecolor{cqcqcq}{rgb}{0,0,0}
\begin{tikzpicture}
[nodeDecorate/.style={shape=circle,inner sep=1pt,draw,thick,fill=black},%
  lineDecorate/.style={thick},%
  scale=1,line width=0.8pt]
\draw [->-,dashed] (-3,4) -- (-3,2);
\draw [->-,dashed] (-1,2) -- (-1,4);
\draw [->-,dashed] (1,2) -- (1,4);
\draw [->-,dashed] (3,2) -- (3,4);
\draw [->-] (-3,2) -- (-1,2);
\draw [->-] (-1,2) -- (1,2);
\draw [->-] (1,2) -- (3,2);
\draw [->-] (3,4) -- (1,4);
\draw [->-] (1,4) -- (-1,4);
\draw [->-] (-1,4) -- (-3,4);
\draw [->-,shift={(-1,2.5)}] plot[domain=0.64:2.5,variable=\t]({1*2.5*cos(\t r)+0*2.5*sin(\t r)},{0*2.5*cos(\t r)+1*2.5*sin(\t r)});
\draw [->-,shift={(1,2.5)}] plot[domain=0.64:2.5,variable=\t]({1*2.5*cos(\t r)+0*2.5*sin(\t r)},{0*2.5*cos(\t r)+1*2.5*sin(\t r)});
\draw [->-,shift={(-1,3.5)}] plot[domain=3.79:5.64,variable=\t]({1*2.5*cos(\t r)+0*2.5*sin(\t r)},{0*2.5*cos(\t r)+1*2.5*sin(\t r)});
\draw [->-,shift={(1,3.5)}] plot[domain=3.79:5.64,variable=\t]({1*2.5*cos(\t r)+0*2.5*sin(\t r)},{0*2.5*cos(\t r)+1*2.5*sin(\t r)});
\begin{scriptsize}
\fill [color=qqqqff] (-3,2) circle (2pt);
\draw[color=qqqqff] (-3,2) node[below left] {$s$};
\fill [color=qqqqff] (-1,2) circle (2pt);
\draw[color=qqqqff] (-1,2) node[below left] {$a$};
\fill [color=qqqqff] (1,2) circle (2pt);
\draw[color=qqqqff] (1,2) node[below right] {$b$};
\fill [color=qqqqff] (3,2) circle (2pt);
\draw[color=qqqqff] (3,2) node[below right] {$c$};
\fill [color=qqqqff] (-3,4) circle (2pt);
\draw[color=qqqqff] (-3,4) node[above left] {$\overline{s}$};
\fill [color=qqqqff] (-1,4) circle (2pt);
\draw[color=qqqqff] (-1,4) node[above left] {$\overline{a}$};
\fill [color=qqqqff] (1,4) circle (2pt);
\draw[color=qqqqff] (1,4) node[above right] {$\overline{b}$};
\fill [color=qqqqff] (3,4) circle (2pt);
\draw[color=qqqqff] (3,4) node[above right] {$\overline{c}$};
\end{scriptsize}
\end{tikzpicture}
}
   \caption{ (a) Example where the twin $\overline{v}$ is already
     blocked when the algorithm starts exploring $v$. By starting in
     $s$ and visiting first $(s,a)$ and $(a,b)$, the vertex
     $\overline{c}$ is already blocked when the algorithm starts
     exploring $c$.  (b) Counterexample for the variant of the
     algorithm visiting first the twin and then the regular
     neighbors. By starting in $s$ and visiting first $(s,a)$ and
     $(a,b)$, the algorithm misses the bubble-cycle $ \langle
     s,a,c,\overline{c},\overline{b},\overline{s} \rangle$.  }
     \label{fig:counterex}
\end{figure}

An important difference between the algorithm introduced here and
Johnson's is that we now have three possible states for any vertex,
\emph{i.e.} free, blocked and twinned, instead of only the first
two. The twinned state is necessary to ensure that the two paths of
the bubble share no internal vertex. Whenever $\overline{v}$ is
twinned, it can only be explored from $v$. On the other hand, a
blocked vertex should never be explored. A twin vertex $\overline{v}$
can be already blocked when the algorithm is exploring $v$, since it
could have been unsuccessfully explored by some other call. In this
case, it is necessary to verify the status of $\overline{v}$, as it is
shown in the graph of Figure~\ref{fig:counterex}a. Indeed, consider
the algorithm starting from $s$ with $(s,a)$ and $(a,b)$ being the
first two arcs visited in the lower part. Later, when the calls
$\cyclebubble(\bar{c})$ and $\cyclebubble(\bar{b})$ are made, since $\bar{a}$ is
twinned, both $\bar{b}$ and $\bar{c}$ remain blocked. When the
algorithm backtracks to $a$ and explores $(a,c)$, the call $\cyclebubble(c)$
is made and $\bar{c}$ is already blocked.

Another important difference with respect to Johnson's algorithm is
that there is a specific order in which the out-neighborhood of a
vertex should be explored. In particular, notice that the order in
which Algorithm~\ref{cyclealgo} explores the neighbors of a vertex $v$
is: first the vertices in $N^+(v) \setminus \{ \bar{v} \}$ and then
$\bar{v}$. A variant of the algorithm where this order would be
reversed, visiting first $\bar{v}$ and then the vertices in $N^+(v)
\setminus \{ \bar{v} \}$, would fail to enumerate all the
bubbles. Indeed, intuitively a vertex can be blocked because the only
way to reach $\bar{s}$ is through a twinned vertex and when that
vertex is untwinned the first one is not unblocked.  Indeed, consider
the graph in Figure~\ref{fig:counterex}b and the twin-first variant
starting in $s$ with $(s,a)$ and $(a,b)$ being the first two arcs
explored in the lower part of the graph. When the algorithm starts
exploring $b$ the stack contains $\langle s,a,b \rangle$. After, the
call $\cyclebubble(\bar{b})$ returns \emph{true} and $\cyclebubble(c)$ returns
\emph{false} because $\bar{a}$ and $\bar{b}$ are twinned. After
finishing exploring $b$, the blocked list $B(b)$ is empty. Thus, the
only vertex unblocked is $b$, $c$ (and $\bar{c}$) remaining
blocked. Finally, the algorithm backtracks to $a$ and explores the
edge $(a,c)$, but $c$ is blocked, and it fails to enumerate $\langle
s,a,c,\overline{c},\overline{b},\overline{s} \rangle$.

One way to address the problem above would be to modify the algorithm
so that every time a vertex $\bar{v}$ is untwinned, a call to
$\unblock(\bar{v})$ is made.  All the bubble-cycles would be correctly
enumerated. However, in this case, it is not hard to find an example
where the delay would then no longer be linear. Intuitively, visiting
first $N^{+}(v) \setminus \{\bar{v} \}$ and, then $\bar{v}$, works
because every vertex $u$ that was blocked (during the exploration of
$N^{+}(v) \setminus \{\bar{v} \}$) should remain blocked when the
algorithm explores $\bar{v}$. Indeed, a bubble would be missed only if
there existed a path starting from $\overline{v}$, going to
$\overline{s}$ through $u$ and avoiding the twinned vertices. This is
not possible if no path from $N^{+}(v) \setminus \{\bar{v} \}$ to $u$
could be completed into a bubble-cycle by avoiding the twinned
vertices, as we will show later on.


\subsection{Proof of correctness and complexity analysis}
\label{sec:complexity}

\subsubsection{Proof of correctness: Algorithm~\ref{cyclealgo} enumerates all bubbles with source $s$}

\begin{lemma}\label{blockedverticeslemma}
Let $v$ be a vertex of $G'_s$ such that $status(v)=blocked$, $S$ the
set of vertices currently in the stack, and $T$ the set of vertices
whose status is equal to twinned. Then $S\cup T$ is a
$(v,\overline{s})$ separator, that is, each path, if any exists, from
$v$ to $\overline{s}$ contains at least one vertex in $S\cup T$.
\end{lemma}

\begin{proof}
The result is obvious for the vertices in $S\cup T$.  Let $v$ be a
vertex of $G'_s$ such that $status(v)=blocked$ and $v \notin S\cup T$.
This means that when $v$ was popped for the last time, $\cyclebubble(v)$ was
equal to {\em false} since $v$ remained blocked.

Let us prove by induction on $k$ that each path to $\overline{s}$ of
length $k$ from a blocked vertex not in $S\cup T$ contains at least
one vertex in $S\cup T$.

We first consider the base case $k=1$.  Suppose that $v$ is a
counter-example for $k=1$.  This means that there is an arc from $v$
to $\overline{s}$ ($\overline{s}$ is an out-neighbor of $v$). However,
in that case the output of $\cyclebubble(v)$ is {\em true}, a contradiction
because $v$ would then be unblocked.

Suppose that the result is true for $k-1$ and, by contradiction, that
there exists a blocked vertex $v \notin S\cup T$ and a path
$(v,w,\ldots,\overline{s})$ of length $k$ avoiding $S\cup T$. Since
$(w,\ldots,\overline{s})$ is a path of length $k-1$, we can then
assume that $w$ is free. Otherwise, if $w$ were blocked, by induction,
the path $(w,\ldots,\overline{s})$ would contain at least one vertex
in $S\cup T$, and so would the path $(v,w,\ldots,\overline{s})$.

Since the call to $\cyclebubble(v)$ returned {\em false} ($v$ remained
blocked), either $w$ was already blocked or twinned, or the call to
$\cyclebubble(w)$ made inside $\cyclebubble(v)$ gave an output equal to {\em
  false}. In any case, after the call to $\cyclebubble(v)$, $w$ was
blocked or twinned and $v$ put in $B(w)$.

The conditional at line~\ref{condtwintoblock} of the $\cyclebubble$
procedure ensures that when untwinned, a vertex immediately becomes
blocked.  Thus, since $w$ is now free, a call to $\unblock(w)$ was made
in any case, yielding a call to $\unblock(v)$.  This contradicts the
fact that $v$ is blocked.
\end{proof}

\begin{theorem}
The algorithm returns only bubble-cycles. Moreover, each of those
cycles is returned exactly once.
\end{theorem}

\begin{proof} 
Let us first prove that only bubble-cycles are output.  As any call to
$\unblock$ (either inside the procedure $\cyclebubble$ or inside the
procedure $\unblock$ itself) is immediately followed by the popping of
the considered vertex, no vertex can appear twice in the stack.  Thus,
the algorithm returns only cycles. They are trivially bipolar as they
have to contain $s$ and $\overline{s}$ to be output.

Consider now a cycle $C$ output by the algorithm with swap arc
$(t,\overline{t})$. Let $(v,w)$ in $C$ with $v \neq s$ and $v \neq
t$. If $\overline{v}$ is free when $v$ is put on the stack, then
$\overline{v}$ is twinned before $w$ is put on the stack and cannot be
explored until $w$ is popped. If $\overline{v}$ is blocked when $v$
is put on the stack, then by Lemma~\ref{blockedverticeslemma} it
remains blocked at least until $v$ is popped. Thus, $\overline{v}$
cannot be in $C$, and consequently the output cycles are twin-free.

So far we have proven that the output produces bubble-cycles. Let us
now show that all cycles
$C=\{v_0=s,v_1,\ldots,v_{l-1},v_{l}=\overline{s},v_0\}$ satisfying
those conditions are output by the algorithm, and each is output
exactly once.

The fact that $C$ is not returned twice is a direct consequence of the
fact that the stack is different in all the leaves of a backtracking
procedure.  To show that $C$ is output, let us prove by induction that
the stack is equal to $\{v_0,\ldots,v_i\}$ at some point of the
algorithm, for every $0\leq i\leq l-1$. Indeed, it is true for
$i=0$. Moreover, suppose that at some point, the stack is
$\{v_0,\ldots,v_{i-1}\}$. 

Suppose that $v_{i-1}$ is different from $t$.  As the cycle contains
no pair of twins except for those composing the arcs
$(s,\overline{s})$ and $(t,\overline{t})$, the path
$\{v_i,v_{i+1},\ldots,v_{l}\}$ contains no twin of
$\{v_0,\ldots,v_{i-1}\}$ and therefore no twinned vertex. Thus, it is
a path from $v_i$ to $\overline{s}$ avoiding $S\cup
T$. Lemma~\ref{blockedverticeslemma} then ensures that at this point
$v_i$ is not blocked. As it is also not twinned, its status is free.
Therefore, it will be explored by the backtracking procedure and the
stack at some point will be $\{v_0,\ldots,v_{i}\}$. If $v_{i-1}=t$,
$v_i=\overline{t}$ is not blocked using the same arguments. Thus it
was twinned by the call to $\cyclebubble(t)$ and is therefore explored
at Line~\ref{twintoblock} of this procedure.  Again, the stack at some
point will be $\{v_0,\ldots,v_{i}\}$.
\end{proof}

\subsubsection{Analysis of complexity: Algorithm~\ref{cyclealgo} has linear delay}

As in \cite{Johnson75}, we show that Algorithm~\ref{cyclealgo} has
delay $O(|V| + |E|)$ by proving that a cycle has to be output between
two successive unblockings of the same vertex and that with linear
delay some vertex has to be unblocked again. To do so, let us first
prove the following lemmas.

\begin{lemma} \label{lemma:returns_cycle}
  Let $v$ be a vertex such that $\cyclebubble(v)$ returns true. Then a
  cycle is output after that call and before any call to $\unblock$.
\end{lemma}
\begin{proof}
 Let $y$ be the first vertex such that $\unblock(y)$ is called inside
 $\cyclebubble(v)$. Since $\cyclebubble(v)$ returns true, there is a
 call to $\unblock(v)$ before it returns, so that $y$
 exists. Certainly, $\unblock(y)$ was called before $\unblock(v)$ if
 $y \neq v$. Moreover, the call $\unblock(y)$ was done inside
 $\cyclebubble(y)$, from line~\ref{unblocking}, otherwise it would
 contradict the choice of $y$. So, the call to $\cyclebubble(y)$ was
 done within the recursive calls inside the call to
 $\cyclebubble(v)$. $\cyclebubble(y)$ must then return true as $y$ was
 unblocked from it.

 All the recursive calls $\cyclebubble(z)$ made inside
 $\cyclebubble(y)$ must return false, otherwise there would be a call
 to $\unblock(z)$ before $\unblock(y)$, contradicting the choice of
 $y$. Since $\cyclebubble(y)$ must return true and the calls to all
 the neighbors returned false, the only possibility is that
 $\overline{s} \in N^+(y)$. Therefore, a cycle is output before
 $\unblock(y)$.
\end{proof}

\begin{lemma} \label{lemma:returns_true}
  Let $v$ be a vertex such that there is a $v\overline{s}$-path
  $P$ avoiding $S \cup T$ at the moment a call to $\cyclebubble(v)$ is
  made. Then the return value of $\cyclebubble(v)$ is true.
\end{lemma}
\begin{proof}
 First notice that if there is such a path $P$, then $v$ belongs to a
 cycle in $G'_s$. This cycle may however not be a bubble-cycle in the
 sense that it may not be twin-free, that is, it may contain more than
 two pairs of twin vertices.  Indeed, since the only constraint that
 we have on $P$ is that it avoids all vertices that are in $S$ and $T$
 when $v$ is reached, then if $v \in V$, it could be that the path $P$
 from $v$ to $\overline{s}$ contains, besides $s$ and ${\overline s}$,
 at least two more pairs of twin vertices.  An example is given in
 Figure~\ref{fig:trans}. It is however always possible, by
 construction of $G'_s$ from $G$, to find a vertex $y \in V$ such that
 $y$ is the first vertex in $P$ with ${\overline y}$ also in $P$. Let
 $P'$ be the path that is a concatenation of the subpath $s \leadsto
 y$ of $P$, the arc $(y,{\overline y})$, and the subpath ${\overline
   y} \leadsto {\overline s}$ in $P$. This path is twin-free, and a
 call to $\cyclebubble(v)$ will, by correctness of the algorithm, return
 true. 
\end{proof}

\begin{theorem}
Algorithm~\ref{cyclealgo} has linear delay.
\end{theorem}

\begin{proof}
Let us first prove that between two successive unblockings of any
vertex $v$, a cycle is output. Let $w$ be the vertex such that a call
to $\unblock(w)$ at line~\ref{unblocking} of Algorithm~\ref{cyclealgo}
unblocks $v$ for the first time. Let $S$ and $T$ be, respectively, the
current sets of stack and twinned vertices after popping $w$.  The
recursive structure of the unblocking procedure then ensures that
there exists a $vw$-path avoiding $S\cup T$. Moreover, as the call to
$\unblock(w)$ was made at line~\ref{unblocking}, the answer to
$\cyclebubble(w)$ is {\em true} so there exists also a
$w\overline{s}$-path avoiding $S\cup T$. The concatenation of both
paths is a again a $v\bar{s}$-path avoiding $S \cup T$. Let $x$ be the
first vertex of this path to be visited again. Note that, if no vertex
in this path is visited again there is nothing to prove, since $v$ is
free, $\cyclebubble(v)$ needs to be called before any $\unblock(v)$
call. When $\cyclebubble(x)$ is called, there is a $x
\overline{s}$-path avoiding the current $S \cup T$. 
of stack and twinned vertices.  Thus, applying
Lemma~\ref{lemma:returns_true} and then
Lemma~\ref{lemma:returns_cycle}, we know that a cycle is output before
any call to $\unblock$. As no call to $\unblock(v)$ can be made before
the call to $\cyclebubble(x)$, a cycle is output before the second
call to $\unblock(v)$.

Let us now consider the delay of the algorithm. In both its
exploration and unblocking phases, the algorithm follows the arcs of
the graph and transforms the status or the $B$ lists of their
endpoints, which overall require constant time. Thus, the delay only
depends on the number of arcs which are considered during two
successive outputs.  An arc $(u,v)$ is considered once by the
algorithm in the three following situations: the exploration part of a
call to $\cyclebubble(u)$; an insertion of $u$ in $B(v)$; a call to
$\unblock(v)$. As shown before, $\unblock(v)$ is called only once
between two successive outputs.  $\cyclebubble(u)$ cannot be called
more than twice. Thus the arc $(u,v)$ is considered at most $5$ times
between two outputs.  This ensures that the delay of the algorithm is
$O(m+n)$.
\end{proof}

\subsection{Practical speedup}
\label{sec:oneone}

\paragraph{Speeding up preprocessing.}
In Section~\ref{sec:bubble-cycle}, the bubble enumeration problem was
reduced to the enumeration of some particular cycles in the
transformed graph $G'_s$ for each $s$. It is worth observing that this
does not imply building from scratch $G'_s$ for each $s$. Indeed,
notice that for any two vertices $s_1$ and $s_2$, we can transform
$G'_{s_1}$ into $G'_{s_2}$ by: (a) removing from $G'_{s_1}$ the arcs
$(\overline{s}_1, {s_1})$, $({s_2},\overline{s}_2)$, $(v, {s_2})$, and
$(\overline{s}_2,\overline{v})$ for each $v\in N^-({s_2})$ in $G$; (b)
adding to $G'_{s_1}$ the arcs $({s_1}, \overline{s}_1)$,
$(\overline{s}_2,{s_2})$, $(v, {s_1})$, and
$(\overline{s}_1,\overline{v})$ for each $v\in N^-({s_1})$ in $G$.

\paragraph{Avoiding duplicate bubbles.}
The one-to-two correspondence between cycles in $G'_s$ and bubbles
starting from $s$ in $G$, claimed by Proposition
\ref{onetotwoproposition}, can be reduced to a one-to-one
correspondence in the following way. Consider an arbitrary order on
the vertices of $V$, and assign to each vertex of $\overline{V}$ the
order of its twin. Let $C$ be a cycle of $G'_s$ that passes through
$s$ and contains exactly two pairs of twin vertices. Denote again by
$t$ the vertex such that $(t,\overline{t})$ is the arc through which
$C$ swaps from $V$ to $\overline{V}$. Denote by {\em swap predecessor}
the vertex before $t$ in $C$ and by {\em swap successor} the vertex
after $\overline{t}$ in $C$.

\begin{proposition}\label{onetooneproposition}
There is a one-to-one correspondence between the set of
$(s,t)$-bubbles in $G$ for all $t \in V$, and the set of cycles of
$G'_s$ that pass through $s$, contain exactly two pairs of twin
vertices and such that the swap predecessor is greater than the swap
successor.
\end{proposition}

\begin{proof}
The proof follows the one of
Proposition~\ref{onetotwoproposition}. The only difference is that, if
we consider a bubble composed of the paths $P_1$ and $P_2$, one of
these two paths, say $P_1$, has a next to last vertex greater than the
next to last vertex of $P_2$. Then the cycle of $G'_s$ made of $P_1$
and $\overline{P_2}$ is still considered by the algorithm whereas the
cycle made of $P_2$ and $\overline{P_1}$ is not.  Moreover, the cycles
of length four which are of the type
$\{s,t,\overline{t},\overline{s}\}$ are ruled out as $\overline{s}$ is
of the same order as $s$.  
\end{proof}


\bigskip
\bigskip


\section{Optimal listing of cycles and $st$-paths in undirected graphs} \label{sec:unweighted:cycle}

\subsection{Introduction}
\label{sec:introduction}
Listing all the simple cycles (hereafter just called cycles) in a
graph is a classical problem whose efficient solutions date back to
the early 70s. For a graph with $n$ vertices and $m$ edges containing
$\eta$ cycles, the best known solution in the literature is given by
Johnson's algorithm (\cite{Johnson75}) and takes $O((\eta+1)(m+n))$
time.

\subsubsection{Previous work}
The classical problem of listing all the cycles of a graph has been
extensively studied for its many applications in several fields,
ranging from the mechanical analysis of chemical
structures~\cite{Sussenguth65} to the design and analysis of reliable
communication networks, and the graph isomorphism problem
(\cite{Welch66}).  In particular, at the turn of the seventies several
algorithms for enumerating all cycles of an undirected graph have been
proposed.  There is a vast body of work, and the majority of the
algorithms listing all the cycles can be divided into the following
three classes (see \cite{Bezem87} and \cite{Mateti76} for excellent
surveys).

\begin{enumerate}
\item \textit{Search space algorithms.}  According to this approach,
  cycles are looked for in an appropriate search space.  In the case
  of undirected graphs, the \emph{cycle vector space}
  (\cite{Diestel05}) turned out to be the most promising choice: from
  a basis for this space, all vectors are computed and it is tested
  whether they are a cycle. Since the algorithm introduced in
  \cite{Welch66}, many algorithms have been proposed: however, the
  complexity of these algorithms turns out to be exponential in the
  dimension of the vector space, and thus in $n$. For planar graphs,
  an algorithm listing cycles in $O((\eta + 1)n)$ time was presented
  in \cite{Syslo81}.

\item \textit{Backtrack algorithms.} By this approach, all paths are
  generated by backtrack and, for each path, it is tested whether it
  is a cycle. One of the first algorithms is the one proposed
  in~\cite{Tiernan70}, which is however exponential in $\eta$. By
  adding a simple pruning strategy, this algorithm has been
  successively modified in~\cite{Tarjan73}: it lists all the cycles in
  $O(nm(\eta+1))$ time. Further improvements were proposed
  in (\cite{Johnson75,Szwarcfiter76,Read75}), leading to
  $O((\eta+1)(m+n))$-time algorithms that work for both directed and
  undirected graphs.  Apart from the algorithm in~\cite{Tiernan70},
  all the algorithms based on this approach are
  \textit{polynomial-time delay}, that is, the time elapsed between
  the outputting of two cycles is polynomial in the size of the graph
  (more precisely, $O(nm)$ in the case of the algorithm
  of~\cite{Tarjan73} and $O(m)$ in the case of the other three
  algorithms).

\item \textit{Using the powers of the adjacency matrix.}  This
  approach uses the so-called \emph{variable adjacency matrix}, that
  is, the formal sum of edges joining two vertices. A non-zero element
  of the $p$-th power of this matrix is the sum of all walks of length
  $p$: hence, to compute all cycles, we compute the $n$th power of the
  variable adjacency matrix. This approach is not very efficient
  because of the non-simple walks. Algorithms based on this approach
  (e.g.\mbox{} \cite{Ponstein66} and \cite{Yau67}) basically differ
  only on the way they avoid to consider walks that are neither paths
  nor cycles.
\end{enumerate}

Almost 40 years after Johnson's algorithm~\cite{Johnson75}, the
problem of efficiently listing all cycles of a graph is still an
active area of research
(e.g.~\cite{Halford04,Horvath04,Liu06,Sankar07,Wild08,Schott11}).  New
application areas have emerged in the last decade, such as
bioinformatics: for example, two algorithms for this problem have been
proposed in \cite{Klamt06} and \cite{Klamt09} while studying
biological interaction graphs. Nevertheless, no significant
improvement has been obtained from the theory standpoint: in
particular, Johnson's algorithm is still the theoretically most
efficient. His $O((\eta+1)(m+n))$-time solution is surprisingly not
optimal for undirected graphs as we show in this chapter.

\subsubsection{Results}  

We present the first optimal solution to list all the cycles in an
undirected graph~$G$.  Specifically, let $\setofcycles(G)$ denote the
set of all these cycles ($|\setofcycles(G)| = \eta$).  For a cycle $c
\in \setofcycles(G)$, let $|c|$ denote the number of edges in~$c$. Our
algorithm requires $O(m + \sum_{c \in \setofcycles(G)}{|c|})$ time and
is asymptotically optimal: indeed, $\Omega(m)$ time is necessarily
required to read $G$ as input, and $\Omega(\sum_{c \in
  \setofcycles(G)}{|c|})$ time is necessarily required to list the
output. Since $|c| \leq n$, the cost of our algorithm never exceeds
$O(m + (\eta+1) n)$ time.

Along the same lines, we also present the first optimal solution to
list all the simple paths from $s$ to $t$ (shortly, $st$-paths) in an
undirected graph $G$. Let $\setofpaths_{st}(G)$ denote the set of
$st$-paths in $G$ and, for an $st$-path $\pi \in \setofpaths_{st}(G)$,
let $|\pi|$ be the number of edges in $\pi$.  Our algorithm lists all
the $st$-paths in~$G$ optimally in $O(m + \sum_{\pi \in
  \setofpaths_{st}(G)}{|\pi|})$ time, observing that $\Omega(\sum_{\pi
  \in \setofpaths_{st}(G)}{|\pi|})$ time is necessarily required to
list the output.

We prove the following reduction to relate $\setofcycles(G)$ and
$\setofpaths_{st}(G)$ for some suitable choices of vertices $s,t$: if
there exists an optimal algorithm to list the $st$-paths in $G$, then
there exists an optimal algorithm to list the cycles in $G$.  Hence,
we can focus on listing $st$-paths.

\subsubsection{Difficult graphs for Johnson's algorithm}

It is worth observing that the analysis of the time complexity of
Johnson's algorithm is not pessimistic and cannot match the one of our
algorithm for listing cycles.  For example, consider the sparse
``diamond'' graph $D_n = (V, E)$ in Fig.~\ref{fig:johnsoncounter} with
$n=2k+3$ vertices in $V = \{a,b,c, v_1, \ldots, v_k, u_1, \ldots,
u_k\}$. There are $m = \Theta(n)$ edges in $E = \{ (a,c)$, $(a,v_i)$,
$(v_i,b)$, $(b,u_i)$, $(u_i,c)$, for $1 \leq i \leq k\}$, and three
kinds of (simple) cycles: (1)~$(a, v_i), (v_i, b), (b, u_j), (u_j, c),
(c, a)$ for $1 \leq i, j \leq k$; (2)~$(a, v_i), (v_i, b), (b, v_j),
(v_j, a)$ for $1 \leq i < j \leq k$; (3)~$(b, u_i), (u_i, c), (c,
u_j), (u_j, b)$ for $1 \leq i < j \leq k$, totalizing $\eta =
\Theta(n^2)$ cycles.  Our algorithm takes $\Theta(n + k^2) =
\Theta(\eta) = \Theta(n^2)$ time to list these cycles.  On the other
hand, Johnson's algorithm takes $\Theta(n^3)$ time, and the discovery
of the $\Theta(n^2)$ cycles in~(1) costs $\Theta(k) = \Theta(n)$ time
each: the backtracking procedure in Johnson's algorithm starting at
$a$, and passing through $v_i$, $b$ and $u_j$ for some $i,j$, arrives
at $c$: at that point, it explores all the vertices $u_l$ $(l \neq i)$
even if they do not lead to cycles when coupled with $a$, $v_i$, $b$,
$u_j$, and $c$.

\begin{figure}[htbp]
\centering
\definecolor{cqcqcq}{rgb}{0,0,0}
\begin{tikzpicture}[scale=1.5,line cap=round,line join=round,>=triangle 45,x=1.0cm,y=1.0cm]
\draw (-2,2)-- (-1,3);
\draw [dash pattern=on 5pt off 5pt] (-2,2)-- (-1,2.52);
\draw [dash pattern=on 5pt off 5pt] (-2,2)-- (-1,2);
\draw (-1,3)-- (0,2);
\draw [dash pattern=on 5pt off 5pt] (-1,2.52)-- (0,2);
\draw [dash pattern=on 5pt off 5pt] (-1,2)-- (0,2);
\draw (-2,2)-- (-1,1);
\draw (-1,1)-- (0,2);
\draw (0,2)-- (1,3);
\draw [dash pattern=on 5pt off 5pt] (0,2)-- (1.04,2.48);
\draw [dash pattern=on 5pt off 5pt] (0,2)-- (1,2);
\draw (0,2)-- (1,1);
\draw (1,3)-- (2,2);
\draw [dash pattern=on 5pt off 5pt] (1.04,2.48)-- (2,2);
\draw [dash pattern=on 5pt off 5pt] (1,2)-- (2,2);
\draw (1,1)-- (2,2);
\draw [dash pattern=on 5pt off 5pt] (-2,2)-- (-1,1.48);
\draw [dash pattern=on 5pt off 5pt] (-1,1.48)-- (0,2);
\draw [dash pattern=on 5pt off 5pt] (0,2)-- (1,1.44);
\draw [dash pattern=on 5pt off 5pt] (1,1.44)-- (2,2);
\draw [shift={(0,2)}] plot[domain=0:3.14,variable=\t]({1*2*cos(\t r)+0*2*sin(\t r)},{0*2*cos(\t r)+1*2*sin(\t r)});
\begin{footnotesize}
\fill [color=black] (-2,2) circle (1.5pt);
\draw[color=black] (-2,2) node[left] {$a$};
\fill [color=black] (-1,3) circle (1.5pt);
\draw[color=black] (-1,3) node[above] {$v_1$};
\fill [color=black] (-1,2) circle (1.5pt);
\fill [color=black] (-1,1) circle (1.5pt);
\draw[color=black] (-1,1) node[below] {$v_k$};
\fill [color=black] (-1,1.48) circle (1.5pt);
\fill [color=black] (-1,2.52) circle (1.5pt);
\fill [color=black] (0,2) circle (1.5pt);
\draw[color=black] (0,2) node[above] {$b$};
\fill [color=black] (1,3) circle (1.5pt);
\draw[color=black] (1,3) node[above] {$u_1$};
\fill [color=black] (1,2) circle (1.5pt);
\fill [color=black] (1.04,2.48) circle (1.5pt);
\fill [color=black] (1,1) circle (1.5pt);
\draw[color=black] (1,1) node[below] {$u_k$};
\fill [color=black] (2,2) circle (1.5pt);
\draw[color=black] (2,2) node[right] {$c$};
\fill [color=black] (1,1.44) circle (1.5pt);
\end{footnotesize}
\end{tikzpicture}
\caption{Diamond graph.}
\label{fig:johnsoncounter}
\end{figure}
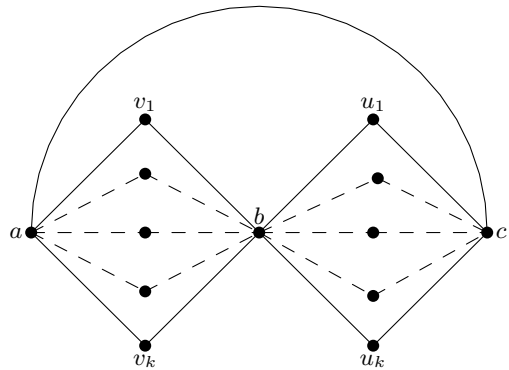

\subsection{Preliminaries}
Let $G=(V,E)$ be an undirected connected graph with $n=|V|$ vertices
and $m=|E|$ edges, without self-loops or parallel edges. For a vertex
$u \in V$, we denote by $N(u)$ the neighborhood of $u$ and by
$d(u)=|N(u)|$ its degree.  $G[V']$ denotes the subgraph \emph{induced}
by $V' \subseteq V$, and $G - u$ is the induced subgraph $G[ V
  \setminus \{u\}]$ for $u \in V$. Likewise for edge $e \in E$, we
adopt the notation $G-e = (V,E \setminus \{e\})$. For a vertex $v \in
V$, the \emph{postorder} DFS number of $v$ is the relative time in
which $v$ was \emph{last} visited in a DFS traversal, i.e. the
position of $v$ in the vertex list ordered by the last visiting time
of each vertex in the DFS.

Paths are simple in $G$ by definition: we refer to a path $\pi$ by its
natural sequence of vertices or edges.  A path $\pi$ from $s$ to $t$,
or $st$-\emph{path}, is denoted by $\pi = s \leadsto t$. Additionally,
$\setofpaths(G)$ is the set of all paths in $G$ and
$\setofpaths_{s,t}(G)$ is the set of all $st$-paths in $G$.  When
$s=t$ we have cycles, and $\setofcycles(G)$ denotes the set of all
cycles in $G$. We denote the number of edges in a path $\pi$ by
$|\pi|$ and in a cycle~$c$ by $|c|$. In this section, we consider the
following problems.

\begin{problem}[Listing $st$-Paths] \label{prob:liststpaths}
  Given an undirected graph $G=(V,E)$ and two distinct vertices $s,t
  \in V$, output all the paths $\pi \in \setofpaths_{s,t}(G)$.
\end{problem} 

\begin{problem}[Listing Cycles] \label{prob:listcycles}
  Given an undirected graph $G=(V,E)$, output all the cycles $c \in
  \setofcycles(G)$.
\end{problem} 

Our algorithms assume without loss of generality that the input graph
$G$ is connected, hence $m \ge n-1$, and use the decomposition of $G$
into biconnected components. Recall that an \emph{articulation point}
(or cut-vertex) is a vertex $u \in V$ such that the number of
connected components in $G$ increases when $u$ is removed. $G$ is
\emph{biconnected} if it has no articulation points. Otherwise, $G$
can always be decomposed into a tree of biconnected components, called
the \emph{block tree}, where each biconnected component is a maximal
biconnected subgraph of $G$ (see Fig.~\ref{fig:beadstring}), and
two biconnected components are adjacent if and only if they share an
articulation point.

\begin{figure}[htbp]
\centering
\begin{tikzpicture}
[nodeDecorate/.style={shape=circle,inner sep=1pt,draw,thick,fill=black},%
  lineDecorate/.style={-,dashed},%
  elipseDecorate/.style={color=gray!30},
  scale=0.7]
\fill [elipseDecorate] (5,10) circle (2);
\fill [elipseDecorate] (9,10) circle (2);
\fill [elipseDecorate] (2,10) circle (1);
\fill [elipseDecorate] (-1,10) circle (2);
\fill [elipseDecorate,rotate around={-55:(-1,8)}] (-1,7) circle (1);

\draw (5,10) circle (2);
\draw (9,10) circle (2);
\draw (2,10) circle (1);
\draw (-1,10) circle (2);
\draw (5,7) circle (1);
\draw (9,7) circle (1);
\draw [rotate around={55:(-1,8)}] (-1,7) circle (1);
\draw [rotate around={-55:(-1,8)}] (-1,7) circle (1);
\draw (13,10) circle (2);
\draw (13,7) circle (1);
\draw (-4,10) circle (1);
\begin{footnotesize}
\node (7) at (-2,7) [nodeDecorate,label=above:$s$] {};
\node (14) at (9,11) [nodeDecorate,label=above:$t$] {};
\end{footnotesize}
\foreach \nodename/\x/\y in {
  0/7/10, 1/5/11,
  2/3/10, 3/1/10, 4/-3/10, 5/-1/10, 6/-1/8, 7/-2/7, 8/-1/11,
  9/0/7,
  11/5/10, 12/4/9, 13/5/8, 14/9/11, 15/11/10, 16/9/9,
  17/9/7 , 18/9/8, 50/5/7, 51/13/11, 52/13/8, 53/13/7, 54/-4/10}
{
  \node (\nodename) at (\x,\y) [nodeDecorate] {};
}

\path
\foreach \startnode/\endnode in {6/7, 6/9, 5/6, 5/3, 5/8, 4/5, 3/2,
2/11, 1/11, 11/12, 11/0, 11/13, 13/50, 0/14, 14/15, 15/16, 16/18,
18/17, 15/51, 15/52, 51/52, 52/53, 54/4}
{
  (\startnode) edge[lineDecorate] node {} (\endnode)
};

\path
\foreach \startnode/\endnode/\bend in { 8/3/20, 6/3/20, 4/8/10,
12/13/20, 1/0/20, 2/1/20, 13/0/20, 0/18/10}
{
  (\startnode) edge[lineDecorate, bend left=\bend] node {} (\endnode)
};

\end{tikzpicture}
\caption{Block tree of $G$ with bead string $\sbeadstring$ in gray.}
\label{fig:beadstring}
\end{figure}

\subsection{Overview and main ideas}
\label{sec:overview}

While the basic approach is simple (see the binary partition in
point~\ref{item:abstract:3}), we use a number of non-trivial ideas to
obtain our optimal algorithm for an undirected (connected) graph $G$
as summarized in the steps below.
\begin{enumerate}
  \item Prove the following reduction. If there exists an optimal
    algorithm to list the $st$-paths in $G$, there exists an optimal
    algorithm to list the cycles in $G$. This relates
    $\setofcycles(G)$ and $\setofpaths_{st}(G)$ for some choices
    $s,t$.


  \item Focus on listing the $st$-paths. Consider the decomposition of
    the graph into biconnected components ({\bcc}s), thus forming a
    tree $T$ where two {\bcc}s are adjacent in $T$ iff they share an
    articulation point. Exploit (and prove) the property that if $s$
    and $t$ belong to distinct {\bcc}s, then $(i)$ there is a unique
    \emph{sequence} $\sbeadstring$ of adjacent {\bcc}s in $T$ through
    which each $st$-path must necessarily pass, and $(ii)$ each
    $st$-path is the concatenation of paths connecting the
    articulation points of these {\bcc}s in $\sbeadstring$.

  \item \label{item:abstract:3} Recursively list the $st$-paths in
    $\sbeadstring$ using the classical binary partition (i.e.\mbox{}
    given an edge $e$ in $G$, list all the cycles containing $e$, and
    then all the cycles not containing~$e$): now it suffices to work
    on the \emph{first} \bcc\ in $\sbeadstring$, and efficiently
    maintain it when deleting an edge $e$, as required by the binary
    partition.

  \item Use a notion of \emph{certificate} to avoid recursive calls
    (in the binary partition) that do not list new $st$-paths.  This
    certificate is maintained dynamically as a data structure
    representing the first \bcc\ in $\sbeadstring$, which guarantees
    that there exists at least one \emph{new} solution in the current
    $\sbeadstring$.

  \item Consider the binary recursion tree corresponding to the binary
    partition.  Divide this tree into \emph{spines}: a spine
    corresponds to the recursive calls generated by the edges $e$
    belonging to the same adjacency list in $\sbeadstring$.  The
    amortized cost for each listed $st$-path $\pi$ is $O(|\pi|)$ when
    there is a guarantee that the amortized cost in each spine $S$ is
    $O(\mu)$, where $\mu$ is a lower bound on the number of $st$-paths
    that will be listed from the recursive calls belonging to~$S$. The
    (unknown) parameter~$\mu$, which is different for each spine~$S$,
    and the corresponding cost $O(\mu)$, will drive the design of the
    proposed algorithms.
\end{enumerate}

\subsubsection{Reduction to \boldmath{$st$}-paths}
\label{sub:reduction-paths}

We now show that listing cycles reduces to listing $st$-paths while
preserving the optimal complexity.

\begin{lemma} \label{lemma:reduction}
  Given an algorithm that solves Problem~\ref{prob:liststpaths} in
  optimal \mbox{$O(m + \sum_{\pi \in \setofpaths_{s,t}(G)}{|\pi|})$}
  time, there exists an algorithm that solves
  Problem~\ref{prob:listcycles} in optimal \mbox{$O(m + \sum_{c \in
      \setofcycles(G)}{|c|})$} time.
\end{lemma}
\begin{proof}
  Compute the biconnected components of $G$ and keep them in a list
  $L$. Each (simple) cycle is contained in one of the biconnected
  components and therefore we can treat each biconnected component
  individually as follows. While $L$ is not empty, extract a
  biconnected component $B=(V_{B},E_{B})$ from $L$ and repeat the
  following three steps: $(i)$ compute a DFS traversal of $B$ and take
  any back edge $b=(s,t)$ in $B$; $(ii)$ list all $st$-paths in $B-b$,
  i.e.~the cycles in $B$ that include edge~$b$; $(iii)$ remove edge
  $b$ from $B$, compute the new biconnected components thus created by
  removing edge~$b$, and append them to $L$. When $L$ becomes empty,
  all the cycles in $G$ have been listed.

  Creating $L$ takes $O(m)$ time. For every $B \in L$, steps $(i)$ and
  $(iii)$ take $O(|E_B|)$ time.  Note that step $(ii)$ always outputs
  distinct cycles in $B$ (i.e.~$st$-paths in $B-b$) in
  $O(|E_{B}|+\sum_{\pi \in \setofpaths_{s,t}(B-b)}{|\pi|})$ time.
  However, $B-b$ is then decomposed into biconnected components whose
  edges are traversed again. We can pay for the latter cost: for any
  edge $e \neq b$ in a biconnected component $B$, there is always a
  cycle in $B$ that contains both $b$ and $e$ (i.e.\mbox{} it is an
  $st$-path in $B-b$), hence $\sum_{\pi \in
    \setofpaths_{s,t}(B-b)}{|\pi|}$ dominates the term $|E_{B}|$,
  i.e.~$\sum_{\pi \in \setofpaths_{s,t}(B-b)}{|\pi|}=
  \Omega(|E_{B}|)$.  Therefore steps $(i)$--$(iii)$ take $O(\sum_{\pi
    \in \setofpaths_{s,t}(B-b)}{|\pi|})$ time. When $L$ becomes empty,
  the whole task has taken $O(m + \sum_{c \in \setofcycles(G)}{|c|})$
  time.
\end{proof}

\subsubsection{Decomposition in biconnected components}
\label{sec:decomposition}

We now focus on listing $st$-paths
(Problem~\ref{prob:liststpaths}). We use the decomposition of $G$ into
a block tree of biconnected components.  Given vertices $s,t$, define
its \emph{bead string}, denoted by $\sbeadstring$, as the unique
sequence of one or more adjacent biconnected components (the
\emph{beads}) in the block tree, such that the first one contains $s$
and the last one contains $t$ (see Fig.~\ref{fig:beadstring}): these
biconnected components are connected through articulation points,
which must belong to all the paths to be listed.

\begin{lemma} \label{lemma:beadstring}
  All the $st$-paths in $\setofpaths_{s,t}(G)$ are contained in the
  induced subgraph $G[\sbeadstring]$ for the bead string
  $\sbeadstring$. Moreover, all the articulation points in
  $G[\sbeadstring]$ are traversed by each of these paths.
\end{lemma}
\begin{proof}
  Consider an edge $e = (u,v)$ in $G$ such that $u \in \sbeadstring$
  and $v \notin \sbeadstring$. Since the biconnected components of a
  graph form a tree and the bead string $\sbeadstring$ is a path in
  this tree, there are no paths $v \leadsto w$ in $G-e$ for any $w \in
  \sbeadstring$ because the biconnected components in $G$ are maximal
  and there would be a larger one (a contradiction).
  Moreover, let $B_1, B_2, \ldots, B_r$ be the biconnected components
  composing $\sbeadstring$, where $s \in B_1$ and $t \in B_r$. If
  there is only one biconnected component in the path (i.e.~$r=1$),
  there are no articulation points in $\sbeadstring$.  Otherwise, all
  of the $r-1$ articulation points in $\sbeadstring$ are traversed by
  each path $\pi \in \setofpaths_{s,t}(G)$: indeed, the articulation
  point between adjacent biconnected components $B_i$ and $B_{i+1}$ is
  their only vertex in common and there are no edges linking $B_i$ and
  $B_{i+1}$.
\end{proof}

We thus restrict the problem of listing the paths in
$\setofpaths_{s,t}(G)$ to the induced subgraph $G[\sbeadstring]$,
conceptually isolating it from the rest of $G$. For the sake of
description, we will use interchangeably $\sbeadstring$ and
$G[\sbeadstring]$ in the rest of the chapter.

\subsubsection{Binary partition scheme}
\label{sec:basic-scheme}

We list the set of $st$-paths in $\sbeadstring$, denoted by
$\setofpaths_{s,t}(\sbeadstring)$, by applying the binary partition
method (where $\setofpaths_{s,t}(G) = \setofpaths_{s,t}(\sbeadstring)$
by Lemma~\ref{lemma:beadstring}): we choose an edge $e = (s,v)$
incident to~$s$ and then list all the $st$-paths that include $e$ and
then all the $st$-paths that do not include $e$. Since we delete some
vertices and some edges during the recursive calls, we proceed as
follows.

\smallskip

\noindent{\it Invariant:} At a generic recursive step on vertex $u$
(initially, $u:=s$), let $\pi_s = s \leadsto u$ be the path discovered
so far (initially, $\pi_s$ is empty $\{\}$). Let $\beadstring$ be the
current bead string (initially, $\beadstring := \sbeadstring$). More
precisely, $\beadstring$ is defined as follows: $(i)$~remove from
$\sbeadstring$ all the vertices in $\pi_s$ but $u$, and the edges
incident to $u$ and discarded so far; $(ii)$~recompute the block tree
on the resulting graph; $(iii)$~$\beadstring$ is the unique bead
string that connects $u$ to $t$ in the recomputed block tree.

\smallskip

\noindent{\it Base case:} When $u=t$, output the $st$-path $\pi_s$.

\smallskip

\noindent{\it Recursive rule:} Let $\setofpaths(\pi_s, u,
\beadstring)$ denote the set of $st$-paths to be listed by the current
recursive call. Then, it is the union of the following two disjoint
sets, for an edge $e=(u,v)$ incident to~$u$:
\begin{itemize}
\item \emph{Left branching:} the $st$-paths in $\setofpaths(\pi_s
  \cdot e, v, \vbeadstring)$ that use $e$, where $\vbeadstring$ is the
  unique bead string connecting $v$ to $t$ in the block tree resulting
  from the deletion of vertex $u$ from $\beadstring$.
\item \emph{Right branching:} the $st$-paths in $\setofpaths(\pi_s, u,
  \beadstring')$ that do \emph{not} use~$e$, where $\beadstring'$ is
  the unique bead string connecting $u$ to $t$ in the block tree
  resulting from the deletion of edge $e$ from $\beadstring$.
\end{itemize}

\noindent
Hence, $\setofpaths_{s,t}(\sbeadstring)$ (and so
$\setofpaths_{s,t}(G)$) can be computed by invoking $\setofpaths(\{\},
s, \sbeadstring)$. The correctness and completeness of the above
approach is discussed in Section~\ref{sec:intro-cert}.

At this point, it should be clear why we introduce the notion of bead
strings in the binary partition. The existence of the partial path
$\pi_s$ and the bead string $\beadstring$ guarantees that there surely
exists at least one $st$-path. But there are two sides of the coin
when using $\beadstring$.

\begin{enumerate}
\item One advantage is that we can avoid useless recursive calls: If
  vertex $u$ has only one incident edge $e$, we just perform the left
  branching; otherwise, we can safely perform both the left and right
  branching since the \emph{first} bead in $\beadstring$ is always a
  biconnected component by definition (thus there exists both an
  $st$-path that traverses $e$ and one that does not).

\item \label{side_coin:2} The other side of the coin is that we have
  to maintain the bead string $\beadstring$ as $\vbeadstring$ in the
  left branching and as $\beadstring'$ in the right branching by
  Lemma~\ref{lemma:beadstring}. Note that these bead strings are
  surely non-empty since $\beadstring$ is non-empty by induction (we
  only perform either left or left/right branching when there are
  solutions by item~1).
\end{enumerate}

To efficiently address point~\ref{side_coin:2}, we need to introduce
the notion of certificate as described next.

\subsubsection{Introducing the certificate}
\label{sec:intro-cert}

Given the bead string $\beadstring$, we call the \emph{head} of
$\beadstring$, denoted by $\head$, the first biconnected component in
$\beadstring$, where $u \in \head$. Consider a DFS tree of
$\beadstring$ rooted at $u$ that changes along with $\beadstring$, and
classify the edges in $\beadstring$ as tree edges or back edges (there
are no cross edges since the graph is undirected).

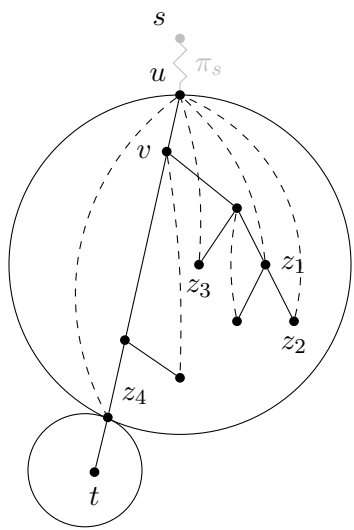
\begin{figure}[htbp]
	\centering
\begin{tikzpicture}
[nodeDecorate/.style={shape=circle,inner sep=1pt,draw,thick,fill=black},%
  lineDecorate/.style={-,dashed},%
  elipseDecorate/.style={color=gray!30},
  scale=0.25]
\draw (10,22) circle (9);
\draw (5,11.1) circle (3);

\node (s) at (10,34) [nodeDecorate,color=lightgray,label=above left:$s$] {};
\node (u) at (10,31) [nodeDecorate,label=above left:{ $u$}] {};

\node (tp) at (6.2,13.9) [nodeDecorate,label=above right:{ $z_4$}] {};
\node (t) at (5.5,11) [nodeDecorate,label=below:$t$] {};

\path {
	(s) edge[snake,-,color=lightgray] node {\quad\quad$\pi_s$} (u)
	(u) edge node {} (tp)
	(tp) edge node {} (t)
};

\node (a) at (9.3,28) [nodeDecorate,label=left:$v$] {};
\node (b) at (7.1,18) [nodeDecorate,] {};
\node (c) at (13,25) [nodeDecorate,] {};
\node (d) at (14.5,22) [nodeDecorate,label=right:$z_1$] {};
\node (e) at (11,22) [nodeDecorate,label=below:$z_3$] {};
\node (f) at (16,19) [nodeDecorate,label=below:$z_2$] {};
\node (g) at (13,19) [nodeDecorate,] {};
\node (h) at (10,16) [nodeDecorate,] {};

\path {
	(a) edge node {} (c)
	(c) edge node {} (d)
	(d) edge node {} (f)
	(c) edge node {} (e)
	(d) edge node {} (g)
	(b) edge node {} (h)
};

\path {
	(u) edge[dashed,bend left=-40] node {} (tp)
	(f) edge[dashed,bend left=-40] node {} (u)
	(g) edge[dashed,bend left=10] node {} (c)
	(e) edge[dashed,bend left=-10] node {} (u)
	(d) edge[dashed,bend left=-20] node {} (u)
	(a) edge[dashed,bend left=5] node {} (h)
};

\end{tikzpicture}
\caption{Example certificate of $B_{u,t}$ \label{fig:Certificate}}
\end{figure}

To maintain $\beadstring$ (and so $\head$) during the recursive calls,
we introduce a \emph{certificate} $C$ (see
Fig.~\ref{fig:Certificate}): It is a suitable data structure that uses
the above classification of the edges in $\beadstring$, and supports
the following operations, required by the binary partition scheme.
\begin{itemize}
\item $\chooseedge(C,u)$: returns an edge $e = (u,v)$ with $v \in
  \head$ such that $\pi_s \cdot (u,v) \cdot u \leadsto t$ is an
  $st$-path such that $u \leadsto t$ is inside $\beadstring$.  Note
  that $e$ always exists since $\head$ is biconnected.  Also, the
  chosen $v$ is the last one in DFS postorder among the neighbors of
  $u$: in this way, the (only) tree edge $e$ is returned when there
  are no back edges leaving from~$u$.  (As it will be clear in
  Sections~\ref{sec:recursion-amortization} and~\ref{sec:certificate},
  this order facilitates the analysis and the implementation of the
  certificate.)
\item $\oracleleft(C,e)$: for the given $e=(u,v)$, it obtains
  $\vbeadstring$ from $\beadstring$ as discussed in
  Section~\ref{sec:basic-scheme}. This implies updating also $\head$,
  $C$, and the block tree, since the recursion continues on~$v$. It
  returns bookkeeping information $I$ for what is updated, so that it
  is possible to revert to $\beadstring$, $\head$, $C$, and the block
  tree, to their status before this operation.
\item $\oracleright(C,e)$: for the given $e=(u,v)$, it obtains
  $\beadstring'$ from $\beadstring$ as discussed in
  Section~\ref{sec:basic-scheme}, which implies updating also $\head$,
  $C$, and the block tree. It returns bookkeeping information $I$ as
  in the case of $\oracleleft(C,e)$.
\item $\undooracle(C,I)$: reverts the bead string to $\beadstring$,
  the head $\head$, the certificate $C$, and the block tree, to their
  status before operation $I := \oracleleft(C,e)$ or $I :=
  \oracleright(C,e)$ was issued (in the same recursive call).
\end{itemize}

Note that a notion of certificate in listing problems has been
introduced in~\cite{Ferreira11}, but it cannot be directly applied to
our case due to the different nature of the problems and our use of
more complex structures such as biconnected components.

Using our certificate and its operations, we can now formalize the
binary partition and its recursive calls $\setofpaths(\pi_s, u,
\beadstring)$ described in Section~\ref{sec:basic-scheme} as
Algorithm~\ref{alg:liststpaths}, where $\beadstring$ is replaced by
its certificate $C$.

\begin{algorithm}[htbp] 
  \caption{$\liststpaths(\pi_s,\,u,\,C)$} \label{alg:liststpaths}
  
  \If{$u=t$}{
    $\routput(\pi_s)$ \label{code:base} \\
    $\return$ \label{code:returnbase}
  }
  $e = (u,v) := \chooseedge( C, u )$ \label{code:choose} \\
  \If{ $e \text{ is back edge}$ \label{code:if_back}}{
    $I := \oracleright(C,e)$  \label{code:right_update} \\
    $\liststpaths(\pi_s,\, u,\,C)$ \label{code:right_branch} \\
    $\undooracle(C, I)$ \label{code:right_undo}
  }
  $I := \oracleleft(C,e)$ \label{code:left_update} \\
  $\liststpaths( \pi_s \cdot (u,v),\, v,\, C)$ \label{code:left_branch} \\
  $\undooracle(C, I)$ \label{code:left_undo}
\end{algorithm}

The base case ($u=t$) corresponds to lines~1--4 of
Algorithm~\ref{alg:liststpaths}. During recursion, the left branching
corresponds to lines~5 and~11-13, while the right branching to
lines~5--10. Note that we perform only the left branching when there
is only one incident edge in $u$, which is a tree edge by definition
of $\chooseedge$. Also, lines~9 and~13 are needed to restore the
parameters to their values when returning from the recursive calls.

\begin{lemma} \label{lemma:correctness_algo_listpaths}
  Given a correct implementation of the certificate $C$ and its
  supported operations, Algorithm~\ref{alg:liststpaths} correctly
  lists all the $st$-paths in $\setofpaths_{s,t}(G)$.
\end{lemma}
\begin{proof}
  For a given vertex $u$ the function $\chooseedge(C, u)$ returns an
  edge $e$ incident to $u$. We maintain the invariant that $\pi_s$ is
  a path $s \leadsto u$, since at the point of the recursive call in
  line~\ref{code:left_branch}: (i) is connected as we append edge
  $(u,v)$ to $\pi_s$ and; (ii) it is simple as vertex $u$ is removed
  from the graph $G$ in the call to $\oracleleft(C,e)$ in
  line~\ref{code:left_update}. In the case of recursive call in
  line~\ref{code:right_branch} the invariant is trivially maintained
  as $\pi_s$ does not change.
  The algorithm only outputs $st$-paths since $\pi_s$ is
  a $s \leadsto u$ path and $u=t$ when the algorithm outputs, in
  line~\ref{code:base}. 

  The paths with prefix $\pi_s$ that do not use $e$ are listed by the
  recursive call in line~\ref{code:right_branch}. This is done by
  removing~$e$ from the graph (line~\ref{code:right_update}) and thus
  no path can include $e$. Paths that use $e$ are listed in
  line~\ref{code:left_branch} since in the recursive call $e$ is added
  to $\pi_s$. Given that the tree edge incident to $u$ is the last one
  to be returned by $\chooseedge(C,u)$, there is no path that does not
  use this edge, therefore it is not necessary to call
  line~\ref{code:right_branch} for this edge.
\end{proof}

A natural question is what is the time complexity: we must account for
the cost of maintaining~$C$ and for the cost of the recursive calls of
Algorithm~\ref{alg:liststpaths}. Since we cannot always maintain the
certificate in $O(1)$ time, the ideal situation for attaining an
optimal cost is taking $O(\mu)$ time if at least $\mu$ $st$-paths are
listed in the current call (and its nested calls). Unfortunately, we
cannot estimate~$\mu$ efficiently and cannot design
Algorithm~\ref{alg:liststpaths} so that it takes $O(\mu)$ adaptively.
We circumvent this by using a different cost scheme in
Section~\ref{sub:recursion-tree-cost} that is based on the recursion
tree induced by Algorithm~\ref{alg:liststpaths}.
Section~\ref{sec:certificate} is devoted to the efficient
implementation of the above certificate operations according to the
cost scheme that we discuss next.

\subsubsection{Recursion tree and cost amortization}
\label{sub:recursion-tree-cost}

We now show how to distribute the costs among the several recursive
calls of Algorithm~\ref{alg:liststpaths} so that optimality is
achieved. Consider a generic execution on the bead string
$\beadstring$. We trace this execution by using a binary recursion
tree $R$. The nodes of $R$ are labeled by the arguments of
Algorithm~\ref{alg:liststpaths}: specifically, we denote a node in $R$
by the triple $x = \langle \pi_s, u, C \rangle$ iff it represents the
call with arguments $\pi_s$, $u$, and~$C$.\footnote{For clarity, we
  use ``nodes'' when referring to $R$ and ``vertices'' when referring
  to $\beadstring$.}  The left branching is represented by the left
child, and the right branching (if any) by the right child of the
current node.

\begin{lemma}
\label{lem:properties_recursion}
The binary recursion tree $R$ for $\beadstring$ has the following
properties:
\begin{enumerate}
  \setlength{\itemsep}{0pt} 
  \item \label{item:R1} There is a one-to-one correspondence between
    the paths in $\setofpaths_{s,t}(\beadstring)$ and the leaves in
    the recursion tree rooted at node $\langle \pi_s, u, C \rangle$.
  \item \label{item:R2} Consider any leaf and its corresponding
    $st$-path $\pi$: there are $|\pi|$ left branches in the
    corresponding root-to-leaf trace.
  \item \label{item:R3} Consider the instruction $e:=\chooseedge(C,u)$
    in Algorithm~\ref{alg:liststpaths}: unary (i.e.\mbox{}
    single-child) nodes correspond to left branches ($e$ is a tree
    edge) while binary nodes correspond to left and right branches
    ($e$ is a back edge).
  \item \label{item:R4} The number of binary nodes is
    $|\setofpaths_{s,t}(\beadstring)| - 1$.
\end{enumerate}
\end{lemma}
\begin{proof}
We proceed in order as follows.
\begin{enumerate}
  \item We only output a solution in a leaf and we only do recursive
    calls that lead us to a solution. Moreover every node partitions
    the set of solutions in the ones that use an edge and the ones
    that do not use it. This guarantees that the leaves in the left
    subtree of the node corresponding to the recursive call and the
    leaves in the right subtree do not intersect. This implies that
    different leaves correspond to different paths from $s$ to $t$,
    and that for each path there is a corresponding leaf.
  \item Each left branch corresponds to the inclusion of an edge in
    the path $\pi$.
  \item Since we are in a biconnected component, there is always a
    left branch. There can be no unary node as a right branch: indeed
    for any edge of $\beadstring$ there exists always a path from $s$
    to $t$ passing through that edge. Since the tree edge is always
    the last one to be chosen, unary nodes cannot correspond to back
    edges and binary nodes are always back edges.
  \item It follows from point \ref{item:R1} and from the fact that the
    recursion tree is a binary tree. (In any binary tree, the number
    of binary nodes is equal to the number of leaves minus 1.)
\end{enumerate}
\end{proof}

We define a \emph{spine} of $R$ to be a subset of $R$'s nodes linked
as follows: the first node is a node $x$ that is either the left child
of its parent or the root of $R$, and the other nodes are those
reachable from $x$ by right branching in $R$. Let $x = \langle \pi_s,
u, C \rangle$ be the first node in a spine $S$. The nodes in $S$
correspond to the edges that are incident to vertex $u$ in
$\beadstring$: hence their number equals the degree $d(u)$ of $u$ in
$\beadstring$, and the deepest (last) node in $S$ is always a tree
edge in $\beadstring$ while the others are back edges.
Fig.~\ref{fig:spine} shows the spine corresponding to $B_{u,t}$ in
Fig.~\ref{fig:Certificate}. Summing up, $R$ can be seen as composed by
spines, unary nodes, and leaves where each spine has a unary node as
deepest node. This gives a global picture of $R$ that we now exploit
for the analysis.

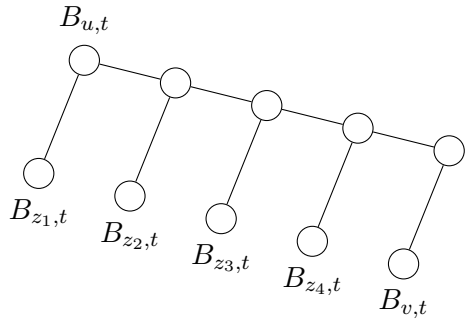
\begin{figure}[htbp]
	\centering
\begin{tikzpicture}
[nodeDecorate/.style={shape=circle,inner sep=4pt,draw,fill=white},%
  lineDecorate/.style={-,dashed},%
  elipseDecorate/.style={color=gray!30},
  scale=0.30]

  \node (a) at (0,0) [nodeDecorate,label=above:$B_{u,t}$] {};
  \node (b) at (4,-1) [nodeDecorate] {};
  \node (c) at (8,-2) [nodeDecorate] {};
  \node (d) at (12,-3) [nodeDecorate] {};
  \node (e) at (16,-4) [nodeDecorate] {};

  \node (a2) at (-2,-5) [nodeDecorate,label=below:$B_{z_1,t}$] {};
  \node (b2) at (2,-6) [nodeDecorate,label=below:$B_{z_2,t}$] {};
  \node (c2) at (6,-7) [nodeDecorate,label=below:$B_{z_3,t}$] {};
  \node (d2) at (10,-8) [nodeDecorate,label=below:$B_{z_4,t}$] {};
  \node (e2) at (14,-9) [nodeDecorate,label=below:$B_{v,t}$] {};

\path {
	(a) edge node {} (b)
	(b) edge node {} (c)
	(c) edge node {} (d)
	(d) edge node {} (e)
	(a) edge node {} (a2)
	(b) edge node {} (b2)
	(c) edge node {} (c2)
	(d) edge node {} (d2)
	(e) edge node {} (e2)
};
\end{tikzpicture}
\caption{Spine of the recursion tree \label{fig:spine}}
\end{figure}


We define the \emph{compact head}, denoted by \mbox{$\Chead = (V_X,
  E_X)$}, as the (multi)graph obtained by compacting the maximal
chains of degree-2 vertices, except $u$, $t$, and the vertices that
are the leaves of its DFS tree rooted at~$u$.

The rationale behind the above definition is that the costs defined in
terms of $\Chead$ amortize well, as the size of $\Chead$ and the
number of $st$-paths in the subtree of $R$ rooted at node $x = \langle
\pi_s, u, C \rangle$ are intimately related (see
Lemma~\ref{lemma:lower_bound_paths_beadstring} in
Section~\ref{sec:recursion-amortization}) while this is not
necessarily true for $\head$.

Recall that each leaf corresponds to a path $\pi$ and each spine
corresponds to a compact head $\Chead=(V_X,E_X)$. We now define the
following abstract cost for spines, unary nodes, and leaves of $R$,
for a sufficiently large constant $c_0 > 0$, that
Algorithm~\ref{alg:liststpaths} must fulfill:
\begin{equation}
  \label{eq:abstrac_cost}
  T(r) =
  \left\{\begin{array}{ll}
      c_0 & \mbox{if $r$ is unary}\\ 
      c_0 |\pi| & \mbox{if $r$ is a leaf}\\ 
      c_0 (|V_X|+|E_X|) \quad & \mbox{if $r$ is a spine}
\end{array}\right. 
\end{equation}

\begin{lemma} \label{lemma:total_cost_recursion_tree}
  The sum of the costs in the nodes of the recursion tree $\sum_{r \in
    R} T(r) = O(\sum_{\pi \in
    \setofpaths_{s,t}(\beadstring)}{|\pi|})$.
\end{lemma}

Section~\ref{sec:recursion-amortization} contains the proof of
Lemma~\ref{lemma:total_cost_recursion_tree} and related properties.
Setting $u:=s$, we obtain that the cost in
Lemma~\ref{lemma:total_cost_recursion_tree} is optimal, by
Lemma~\ref{lemma:beadstring}.

\begin{theorem}  \label{theorem:optimal_paths}
  Algorithm~\ref{alg:liststpaths} solves problem
  Problem~\ref{prob:liststpaths} in optimal $O(m + \sum_{\pi \in
    \setofpaths_{s,t}(G)}{|\pi|})$ time.
\end{theorem}

By Lemma~\ref{lemma:reduction}, we obtain an optimal result for
listing cycles.

\begin{theorem} \label{theorem:optimal_cycles}
  Problem~\ref{prob:listcycles} can be optimally solved in \mbox{$O(m
    + \sum_{c \in \setofcycles(G)}{|c|})$} time.
\end{theorem}

\subsection{Amortization strategy}
\label{sec:recursion-amortization}

We devote this section to prove
Lemma~\ref{lemma:total_cost_recursion_tree}. Let us split the sum in
Eq.~\eqref{eq:abstrac_cost} in three parts, and bound each part
individually, as
\begin{equation}
  \label{eq:sum_R}
  \sum_{r \in R} T(r) \leq \sum_{r:\,\mathrm{unary}} T(r) + \sum_{r:\,\mathrm{leaf}} T(r) + \sum_{r:\,\mathrm{spine}} T(r).
\end{equation}

We have that $\sum_{r:\,\mathrm{unary}} T(r) = O(\sum_{\pi \in
  \setofpaths_{s,t}(G)}{|\pi|})$, since there are
$|\setofpaths_{s,t}(G)|$ leaves, and the root-to-leaf trace leading to
the leaf for $\pi$ contains at most $|\pi|$ unary nodes by
Lemma~\ref{lem:properties_recursion}, where each unary node has cost
$O(1)$ by Eq.~\eqref{eq:abstrac_cost}.

Also, $\sum_{r:\,\mathrm{leaf}} T(r) = O(\sum_{\pi \in
  \setofpaths_{s,t}(G)}{|\pi|})$, since the leaf $r$ for $\pi$ has
cost $O(|\pi|)$ by Eq.~\eqref{eq:abstrac_cost}.

It remains to bound $\sum_{r\,\mathrm{spine}} T(r)$. By
Eq.~\eqref{eq:abstrac_cost}, we can rewrite this cost as
$\sum_{\Chead} c_0 (|V_X| + |E_X|)$, where the sum ranges over the
compacted heads $\Chead$ associated with the spines $r$. We use the
following lemma to provide a lower bound on the number of $st$-paths
descending from $r$.

\begin{lemma} \label{lemma:lower_bound_paths_beadstring}
  Given a spine $r$, and its bead string $\beadstring$ with head
  $\head$, there are at least $|E_X| - |V_X| + 1$ $st$-paths in $G$
  that have prefix $\pi_s = s \leadsto u$ and suffix $u \leadsto t$
  internal to $\beadstring$, where the compacted head is $\Chead =
  (V_X, E_X)$.
\end{lemma}
\begin{proof}
  $\Chead$ is biconnected. In any biconnected graph $B = (V_B, E_B)$
  there are at least $|E_B| - |V_B| + 1$ $xy$-paths for any $x,y \in
  V_B$. Find an ear decomposition (see Chapter~\ref{chap:back},
  Lemma~\ref{lem:back:ear}) of $B$ and consider the process of forming
  $B$ by adding ears one at the time, starting from a single cycle
  including $x$ and $y$. Initially $|V_B|=|E_B|$ and there are 2
  $xy$-paths. Each new ear forms a path connecting two vertices that
  are part of a $xy$-path, increasing the number of paths by at least
  1. If the ear has $k$ edges, its addition increases $V$ by $k-1$,
  $E$ by $k$, and the number of $xy$-paths by at least 1. The result
  follows by induction.
\end{proof}

The implication of Lemma~\ref{lemma:lower_bound_paths_beadstring} is
that there are at least $|E_X| - |V_X| + 1$ leaves descending from the
given spine $r$. Hence, we can charge to each of them a cost of
$\frac{c_0 (|V_X| + |E_X|)}{|E_X| - |V_X| +
  1}$. Lemma~\ref{lemma:density} allows us to prove that the latter
cost is $O(1)$ when $\head$ is different from a single edge or a
cycle. (If $\head$ is a single edge or a cycle, $\Chead$ is a single
or double edge, and the cost is trivially a constant.)
\begin{lemma} \label{lemma:density}
  For a compacted head $\Chead = (V_X, E_X)$, its density is
  $\frac{|E_X|}{|V_X|} \geq \frac{11}{10}$.
\end{lemma}
\begin{proof}
  Consider the following partition $V_X = \{r\} \cup V_2 \cup V_3$
  where: $r$ is the root; $V_2$ is the set of vertices with degree 2
  and; $V_3$, the vertices with degree $\geq 3$.  Since $H_X$ is
  compacted DFS tree of a biconnected graph, we have that $V_2$ is a
  \emph{subset} of the leaves and $V_3$ contains the set of internal
  vertices (except $r$). There are no vertices with degree 1 and $d(r)
  \geq 2$. Let $x = \sum_{v \in V_3} d(v)$ and $y = \sum_{v \in V_2}
  d(v)$.  We can write the density as a function of $x$ and $y$,
  namely,
  $$\frac{|E_X|}{|V_X|} = \frac{x + y + d(r)}{2 (|V_3| + |V_2| + 1)} $$
  
  Note that $|V_3| \le \frac{x}{3}$ as the vertices in $V_3$ have at
  least degree 3, $|V_2| = \frac{y}{2}$ as vertices in $V_2$ have
  degree exactly 2. Since $d(r) \ge 2$, we derive the following bound
  $$\frac{|E_X|}{|V_X|} \ge \frac{x + y + 2}{\frac{2}{3}x + y + 2} $$
  
  Consider any graph with $|V_X|>3$ and its DFS tree rooted at
  $r$. Note that: (i) there are no tree edges between any two leaves,
  (ii) every vertex in $V_2$ is a leaf and (iii) no leaf is a child of
  $r$.  Therefore, every tree edge incident in a vertex of $V_2$ is
  also incident in a vertex of $V_3$. Since exactly half the incident
  edges to $V_2$ are tree edges (the other half are back edges) we get
  that $y \le 2x$.
  
  With $|V_X| \ge 3$ there exists at least one internal vertex in
  the DFS tree and therefore $x \ge 3$.
  \begin{alignat*}{2}
	 \text{minimize }\quad   & \frac{x + y + 2}{\frac{2}{3}x + y + 2}\ \\
	 \text{subject to }\quad  & 0 \le y \le 2x, \\
	 		    & x \ge 3.
  \end{alignat*}

 Since for any $x$ the function is minimized by the maximum $y$ s.t.
 $y \le 2x$ and for any $y$ by the minimum $x$, we get
 $$\frac{|E_X|}{|V_X|} \ge \frac{9x + 6}{8x + 6} \ge
 \frac{11}{10}.$$ 
\end{proof}

Specifically, let $\alpha = \frac{11}{10}$ and write $\alpha = 1 +
2/\beta$ for a constant $\beta$: we have that $|E_X| + |V_X| = (|E_X|
- |V_X|) + 2 |V_X| \leq (|E_X| - |V_X|) + \beta (|E_X| - |V_X|) =
\frac{\alpha+1}{\alpha-1} (|E_X| - |V_X|)$. Thus, we can charge each
leaf with a cost of $\frac{c_0 (|V_X| + |E_X|)}{|E_X| - |V_X| + 1}
\leq c_0 \frac{\alpha+1}{\alpha-1} = O(1)$. This motivates the
definition of $\Chead$, since Lemma~\ref{lemma:density} does not
necessarily hold for the head $\head$ (due to the unary nodes in its
DFS tree).

One last step to bound $\sum_{\Chead} c_0 (|V_X| + |E_X|)$: as noted
before, a root-to-leaf trace for the string storing $\pi$ has $|\pi|$
left branches by Lemma~\ref{lem:properties_recursion}, and as many
spines, each spine charging $c_0 \frac{\alpha+1}{\alpha-1} = O(1)$ to
the leaf at hand. This means that each of the $|\setofpaths_{s,t}(G)|$
leaves is charged for a cost of $O(|\pi|)$, thus bounding the sum as
$\sum_{r\,\mathrm{spine}} T(r) = \sum_{\Chead} c_0 (|V_X| + |E_X|) =
O(\sum_{\pi \in \setofpaths_{s,t}(G)}{|\pi|})$. This completes the
proof of Lemma~\ref{lemma:total_cost_recursion_tree}. As a corollary,
we obtain the following result.

\begin{lemma} \label{lemma:amortized_cost_per_path}
  The recursion tree $R$ with cost as in Eq.~\eqref{eq:abstrac_cost}
  induces an $O(|\pi|)$ amortized cost for each $st$-path $\pi$.
\end{lemma}

\subsection{Certificate implementation and maintenance}
\label{sec:certificate}


The certificate $C$ associated with a node $\langle \pi_s, u, C
\rangle$ in the recursion tree is a compacted and augmented DFS tree
of bead string $\beadstring$, rooted at vertex~$u$. The DFS tree
changes over time along with $\beadstring$, and is maintained in such
a way that $t$ is in the leftmost path of the tree.
We compact the DFS tree by contracting the vertices that have
degree~2, except $u$, $t$, and the leaves (the latter surely have
incident back edges). Maintaining this compacted representation is not
a difficult data-structure problem. From now on we can assume
w.l.o.g.\mbox{} that $C$ is an augmented DFS tree rooted at $u$ where
internal nodes of the DFS tree have degree $\ge 3$, and each vertex
$v$ has associated the following information.

\begin{enumerate}
  \setlength{\itemsep}{0pt} 
  \item A doubly-linked list $lb(v)$ of back edges linking $v$
    to its descendants $w$ sorted by postorder DFS numbering.
  \item A doubly-linked list $ab(v)$ of back edges linking $v$
    to its ancestors $w$ sorted by preorder DFS numbering.
  \item \label{item:point3} An integer $\gamma(v)$, such that if $v$ is an
    ancestor of $w$ then $\gamma(v) < \gamma(w)$.
  \item \label{item:point4} The smallest $\gamma(w)$ over all
    $w$, such that $(h,w)$ is a back edge and $h$ is in the
    subtree of $v$, denoted by $\mathit{lowpoint}(v)$.
\end{enumerate}

Given three vertices $v,w,x \in C$ such that $v$ is the parent of $w$
and $x$ is not in the subtree\footnote{The second condition is always
  satisfied when $w$ is not in the leftmost path, since $t$ is not in
  the subtree of $w$.} of $w$, we can efficiently test if $v$ is an
articulation point, i.e.\mbox{} $\mathit{lowpoint}(w) \leq
\gamma(v)$. (Note that we adopt a variant of $\mathit{lowpoint}$ using
$\gamma(v)$ in place of the preorder numbering~\cite{Tarjan72}: it has
the same effect whereas using $\gamma(v)$ is preferable since it is
easier to dynamically maintain.)

\begin{lemma}
  \label{lem:certificate_scratch}
 The certificate associated with the root of the recursion can be
 computed in $O(m)$ time.
\end{lemma}

\begin{proof}
	In order to set $t$ to be in the leftmost path, we perform a
        DFS traversal of graph $G$ starting from $s$ and stop when we
        reach vertex $t$. We then compute the DFS tree, traversing the
        path $s \leadsto t$ first. When visiting vertex $v$, we set
        $\gamma(v)$ to depth of $v$ in the DFS. Before going up on the
        traversal, we compute the lowpoints using the lowpoints of the
        children. Let $z$ be the parent of $v$. If
        $\mathit{lowpoint}(v) \leq \gamma(z)$ and $v$ is not in the
        leftmost path in the DFS, we cut the subtree of $v$ as it does
        not belong to $B_{s,t}$.  When first exploring the
        neighborhood of $v$, if $w$ was already visited,
        i.e. $e=(u,w)$ is a back edge, and $w$ is a descendant of $v$;
        we add $e$ to $ab(w)$. This maintains the DFS preordering in
        the ancestor back edge list. Now, after the first scan of
        $N(v)$ is over and all the recursive calls returned (all the
        children were explored), we re-scan the neighborhood of
        $v$. If $e=(v,w)$ is a back edge and $w$ is an ancestor of
        $v$, we add $e$ to $lb(w)$. This maintains the DFS
        postorder in the descendant back edge list. This procedure
        takes at most two DFS traversals in $O(m)$ time.  This DFS
        tree can be compacted in the same time bound.  
\end{proof}

\begin{lemma}
	\label{lem:choose}
	Operation $\chooseedge(C,u)$ can be implemented in $O(1)$ time.
\end{lemma}
\begin{proof}
If the list $lb(v)$ is empty, return the tree edge $e=(u,v)$ linking $u$
to its only child $v$ (there are no other children).  Else,
return the last edge in $lb(v)$. 
\end{proof}

We analyze the cost of updating and restoring the
certificate $C$. We can reuse parts of~$C$,
namely, those corresponding to the vertices that are not in the
compacted head $H_X = (V_X,E_X)$ as defined in
Section~\ref{sub:recursion-tree-cost}.
%
%
%
We prove that, given a unary node $u$ and its tree edge $e=(u,v)$, the
subtree of $v$ in~$C$ can be easily made a certificate for the left
branch of the recursion.

\begin{lemma}
	\label{lem:unary_left}
	On a unary node, $\oracleleft(C,e)$ takes $O(1)$
	time.
\end{lemma}
\begin{proof}
	Take edge $e=(u,v)$. Remove edge $e$ and set $v$ as the root
	of the certificate. Since $e$ is the only edge incident in
	$v$, the subtree $v$ is still a DFS tree. Cut the list of children
	of $v$ keeping only the first child. (The other children are no
	longer in the bead string and become part of $I$.) There is
	no need to update $\gamma(v)$. 
\end{proof}

We now devote the rest of this section to show how to efficiently
maintain $C$ on a spine.  Consider removing a back edge $e$ from $u$:
the compacted head $H_X=(V_X,E_X)$ of the bead string can be divided
into smaller biconnected components.  Many of those can be excluded
from the certificate (i.e. they are no longer in the new bead string,
and so they are bookkept in $I$) and additionally we have to update
the lowpoints that change. We prove that this operation can be
performed in $O(|V_X|)$ total time on a spine of the recursion tree.

\begin{lemma}
  \label{lem:removebackedge}
  The total cost of all the operations $\oracleright(C,e)$ in a
  spine is $O(|V_X|)$ time.
\end{lemma}
\begin{proof}
  In the right branches along a spine, we remove all back edges in
  $lb(u)$. This is done by starting from the last edge in $lb(u)$,
  i.e. proceeding in reverse DFS postorder.
  For back edge $b_i = (z_i,u)$, we traverse the vertices in the path
  from $z_i$ towards the root $u$, as these are the only lowpoints
  that can change.
  While moving upwards on the tree, on each vertex $w$, we update
  $\mathit{lowpoint}(w)$. This is done by taking the endpoint $y$ of
  the first edge in $ab(w)$ (the back edge that goes the topmost in
  the tree) and choosing the minimum between $\gamma(y)$ and the
  lowpoint of each child\footnote{If $\mathit{lowpoint}(w)$ does not
    change we cannot pay to explore its children. For each vertex we
    dynamically maintain a list $l(w)$ of its children that have
    lowpoint equal to $\gamma(u)$. Then, we can test in constant time
    if $l(w) \neq \emptyset$ and $y$ is not the root $u$. If both
    conditions are true $\mathit{lowpoint}(w)$ changes, otherwise it
    remains equal to $\gamma(u)$ and we stop.} of $w$. We stop when
  the updated $\mathit{lowpoint}(w) = \gamma(u)$ since it implies that
  the lowpoint of the vertex can not be further reduced.  Note that we
  stop before $u$, except when removing the last back edge in $lb(u)$.
  
  To prune the branches of the DFS tree that are no longer in
  $B_{u,t}$, consider again each vertex $w$ in the path from $z_i$
  towards the root $u$ and its parent $y$. We check if the updated
  $\mathit{lowpoint}(w) \leq \gamma(y)$ and $w$ is not in the leftmost
  path of the DFS. If both conditions are satisfied, we have that $w
  \notin B_{u,t}$, and therefore we cut the subtree of $w$ and keep it
  in $I$ to restore later. We use the same halting criterion as in the
  previous paragraph.
  
  The cost of removing all back edges in the spine is $O(|V_X|)$:
  there are $O(|V_X|)$ tree edges and, in the paths from $z_i$ to $u$,
  we do not traverse the same tree edge twice since the process
  described stops at the first common ancestor of endpoints of back
  edges $b_i$. Additionally, we take $O(1)$ time to cut a subtree of
  an articulation point in the DFS tree.
\end{proof}

To compute $\oracleleft(C,e)$ in the binary nodes of a spine, we use
the fact that in every left branching from that spine, the graph is
the same (in a spine we only remove edges incident to $u$ and on a
left branch from the spine we remove the vertex $u$) and therefore its
block tree is also the same. However, the certificates on these nodes
are not the same, as they are rooted at different vertices. Using the
reverse DFS postorder of the edges, we are able to traverse each edge
in $H_X$ only a constant number of times in the spine.

\begin{lemma} \label{lem:promotebackedge}
  The total cost of all operations $\oracleleft(C,e)$ in a spine is
  amortized $O(|E_X|)$.
\end{lemma}
\begin{proof}
  Let $t'$ be the last vertex in the path $u \leadsto t$ s.t.  $t' \in
  V_X$. Since $t'$ is an articulation point, the subtree of the DFS
  tree rooted in $t'$ is maintained in the case of removal of vertex
  $u$. Therefore the only modifications of the DFS tree occur in the
  compacted head $H_X$ of $B_{u,t}$.
  Let us compute the certificate $C_i$: this is the certificate of the
  left branch of the $i$th node of the spine where we augment the path
  with the back edge $b_i = (z_i,u)$ of $lb(u)$ in the order defined
  by $\chooseedge(C,u)$.

  For the case of $C_1$, we remove $u$ and rebuild the certificate
  starting form $z_1$ (the last edge in $lb(u)$) using the algorithm
  from Lemma~\ref{lem:certificate_scratch} restricted to $H_X$ and
  using $t'$ as target and $\gamma(t')$ as a baseline to $\gamma$
  (instead of the depth). This takes $O(|E_X|)$ time.

  For the general case of $C_i$ with $i>1$ we also rebuild (part) of
  the certificate starting from $z_i$ using the procedure from
  Lemma~\ref{lem:certificate_scratch} but we use information gathered
  in $C_{i-1}$ to avoid exploring useless branches of the DFS
  tree. The key point is that, when we reach the first bead in common
  to both $B_{z_i,t}$ and $B_{z_{i-1},t}$, we only explore edges
  internal to this bead.  If an edge $e$ leaving the bead leads to
  $t$, we can reuse a subtree of $C_{i-1}$. If $e$ does not lead to
  $t$, then it has already been explored (and cut) in $C_{i-1}$ and
  there is no need to explore it again since it will be discarded.
  Given the order we take $b_i$, each bead is not added more than
  once, and the total cost over the spine is $O(|E_X|)$.

  Nevertheless, the internal edges $E_X'$ of the first bead in common
  between $B_{z_i,t}$ and $B_{z_{i-1},t}$ can be explored several
  times during this procedure.\footnote{Consider the case where $z_i,
    \ldots, z_j$ are all in the same bead after the removal of
    $u$. The bead strings are the same, but the roots $z_i, \ldots,
    z_j$ are different, so we have to compute the corresponding DFS of
    the first component $|j-i|$ times.}  We can charge the cost
  $O(|E'_X|)$ of exploring those edges to another node in the
  recursion tree, since this common bead is the head of at least one
  certificate in the recursion subtree of the left child of the $i$th
  node of the spine.  Specifically, we charge the first node in the
  \emph{leftmost} path of the $i$th node of the spine that has exactly
  the edges $E'_X$ as head of its bead string: (i) if $|E'_X| \le 1$
  it corresponds to a unary node or a leaf in the recursion tree and
  therefore we can charge it with $O(1)$ cost; (ii) otherwise it
  corresponds to a first node of a spine and therefore we can also
  charge it with $O(|E'_X|)$. We use this charging scheme when $i \neq
  1$ and the cost is always charged in the leftmost recursion path of
  $i$th node of the spine.  Consequently, we never charge a node in
  the recursion tree more than once.
\end{proof}

\begin{lemma} \label{lem:restore} 
  On each node of the recursion tree, $\undooracle(C,I)$ takes time
  proportional to the size of the modifications kept in $I$.
\end{lemma}
\begin{proof}
  We use standard data structures (i.e.~linked lists) for the
  representation of certificate $C$.  Persistent versions of these
  data structures exist that maintain a stack of modifications applied
  to them and that can restore its contents to their previous states.
  Given the modifications in $I$, these data structures take $O(|I|)$
  time to restore the previous version of $C$.

  Let us consider the case of performing $\oracleleft(C,e)$. We cut at
  most $O(|V_X|)$ edges from $C$. Note that, although we conceptually
  remove whole branches of the DFS tree, we only remove edges that
  attach those branches to the DFS tree. The other vertices and edges
  are left in the certificate but, as they no longer remain attached
  to $B_{u,t}$, they will never be reached or explored. In the case of
  $\oracleright(C,e)$, we have a similar situation, with at most
  $O(|E_X|$) edges being modified along the spine of the recursion
  tree.
\end{proof}

From Lemmas~\ref{lem:choose} and
\ref{lem:removebackedge}--\ref{lem:restore}, it follows that on a
spine of the recursion tree we have the costs: $\chooseedge(u)$ on
each node which is bounded by $O(|V_X|)$ time as there are at most
$|V_X|$ back edges in $u$; $\oracleright(C,e)$, $\undooracle(C,I)$
take $O(|V_X|)$ time; $\oracleleft(C,e)$ and $\undooracle(C,I)$ are
charged $O(|V_X|+|E_X|)$ time.  We thus have the following result,
completing the proof of Theorem~\ref{theorem:optimal_paths}.

\begin{lemma} 	\label{lem:algo_cost}
  Algorithm~$\ref{alg:liststpaths}$ can be implemented with a cost
  fulfilling Eq.~\eqref{eq:abstrac_cost}, thus it takes total
  \mbox{$O(m+\sum_{r \in R} T(r)) = O(m+\sum_{\pi \in
      \setofpaths_{s,t}(\beadstring)}{|\pi|})$} time.
\end{lemma}

\subsection{Extended analysis of operations}
\label{sec:extend-analys-oper}

In this section, we present all details and illustrate with figures
the operations $\oracleright(C,e)$ and $\oracleleft(C,e)$ that are
performed along a spine of the recursion tree. In order to better
detail the procedures in Lemma~\ref{lem:removebackedge} and
Lemma~\ref{lem:promotebackedge}, we divide them in smaller parts.  We
use bead string $B_{u,t}$ from Fig.~\ref{fig:Certificate} and the
respective spine from Fig.~\ref{fig:spine} as the base for the
examples. This spine contains four binary nodes corresponding to the
back edges in $lb(u)$ and an unary node corresponding to the tree edge
$(u,v)$. Note that edges are taken in order of the endpoints
$z_1,z_2,z_3,z_4,v$ as defined in operation $\chooseedge(C,u)$.

By Lemma~\ref{lemma:beadstring}, the impact of operations
$\oracleright(C,e)$ and $\oracleleft(C,e)$ in the certificate is
restricted to the biconnected component of $u$. Thus we mainly focus
on maintaining the compacted head $H_X = (V_X,E_X)$ of the bead
string~$B_{u,t}$.

\subsubsection{Operation $\oracleright(C,e)$} 

\begin{figure}[!ht]
\centering
\subfloat[Step 1]{
\begin{tikzpicture}
[nodeDecorate/.style={shape=circle,inner sep=1pt,draw,thick,fill=black},%
  lineDecorate/.style={-,dashed},%
  elipseDecorate/.style={color=gray!30},
  scale=0.20]

\draw (10,22) circle (9);
\draw (5,11.1) circle (3);

\node (s) at (10,34) [nodeDecorate,color=lightgray,label=above left:$s$] {};
\node (u) at (10,31) [nodeDecorate,label=above left:{ $u$}] {};

\node (tp) at (6.2,13.9) [nodeDecorate,label=right:{ $z_4$}] {};
\node (t) at (5.5,11) [nodeDecorate,label=left:$t$] {};

\path {
	(s) edge[snake,-,color=lightgray] node {\quad\quad$\pi_s$} (u)
	(u) edge node {} (tp)
	(tp) edge node {} (t)
};

\node (a) at (9.3,28) [nodeDecorate,label=left:$v$] {};
\node (b) at (7.1,18) [nodeDecorate,] {};
\node (c) at (13,25) [nodeDecorate,] {};
\node (d) at (14.5,22) [nodeDecorate,label=right:$z_1$] {};
\node (e) at (11,22) [nodeDecorate,label=below:$z_3$] {};
\node (f) at (16,19) [nodeDecorate,label=below:$z_2$] {};
\node (g) at (13,19) [nodeDecorate,] {};
\node (h) at (10,16) [nodeDecorate,] {};

\path {
	(a) edge node {} (c)
	(c) edge node {} (d)
	(d) edge node {} (f)
	(c) edge node {} (e)
	(d) edge node {} (g)
	(b) edge node {} (h)
};

\path {
	(u) edge[dashed,bend left=-40] node {} (tp)
	(f) edge[dashed,bend left=-40] node {} (u)
	(g) edge[dashed,bend left=10] node {} (c)
	(e) edge[dashed,bend left=-10] node {} (u)
	(a) edge[dashed,bend left=5] node {} (h)
};
\end{tikzpicture}
}
\subfloat[Step 2]{
\begin{tikzpicture}
[nodeDecorate/.style={shape=circle,inner sep=1pt,draw,thick,fill=black},%
  lineDecorate/.style={-,dashed},%
  elipseDecorate/.style={color=gray!30},
  scale=0.20]
\draw (10,22) circle (9);
\draw (5,11.1) circle (3);

\node (s) at (10,34) [nodeDecorate,color=lightgray,label=above left:$s$] {};
\node (u) at (10,31) [nodeDecorate,label=above left:{ $u$}] {};

\node (tp) at (6.2,13.9) [nodeDecorate,label=right:{ $z_4$}] {};
\node (t) at (5.5,11) [nodeDecorate,label=left:$t$] {};

\path {
	(s) edge[snake,-,color=lightgray] node {\quad\quad$\pi_s$} (u)
	(u) edge node {} (tp)
	(tp) edge node {} (t)
};

\node (a) at (9.3,28) [nodeDecorate,label=left:$v$] {};
\node (b) at (7.1,18) [nodeDecorate,] {};
\node (c) at (13,25) [nodeDecorate,] {};
\node (e) at (11,22) [nodeDecorate,label=below:$z_3$] {};
\node (h) at (10,16) [nodeDecorate,] {};

\path {
	(a) edge node {} (c)
	(c) edge node {} (e)
	(b) edge node {} (h)
};

\path {
	(u) edge[dashed,bend left=-40] node {} (tp)
	(e) edge[dashed,bend left=-10] node {} (u)
	(a) edge[dashed,bend left=5] node {} (h)
};

\end{tikzpicture}
}
\subfloat[Step 3]{
\begin{tikzpicture}
[nodeDecorate/.style={shape=circle,inner sep=1pt,draw,thick,fill=black},%
  lineDecorate/.style={-,dashed},%
  elipseDecorate/.style={color=gray!30},
  scale=0.20]
\draw (10,22) circle (9);
\draw (5,11.1) circle (3);

\node (s) at (10,34) [nodeDecorate,color=lightgray,label=above left:$s$] {};
\node (u) at (10,31) [nodeDecorate,label=above left:{ $u$}] {};

\node (tp) at (6.2,13.9) [nodeDecorate,label=right:{ $z_4$}] {};
\node (t) at (5.5,11) [nodeDecorate,label=left:$t$] {};

\path {
	(s) edge[snake,-,color=lightgray] node {\quad\quad$\pi_s$} (u)
	(u) edge node {} (tp)
	(tp) edge node {} (t)
};

\node (a) at (9.3,28) [nodeDecorate,label=left:$v$] {};
\node (b) at (7.1,18) [nodeDecorate,] {};
\node (h) at (10,16) [nodeDecorate,] {};

\path {
	(b) edge node {} (h)
};

\path {
	(u) edge[dashed,bend left=-40] node {} (tp)
	(a) edge[dashed,bend left=5] node {} (h)
};

\end{tikzpicture}
}
\subfloat[Step 4 (final)]{

\begin{tikzpicture}
[nodeDecorate/.style={shape=circle,inner sep=1pt,draw,thick,fill=black},%
  lineDecorate/.style={-,dashed},%
  elipseDecorate/.style={color=gray!30},
  scale=0.20]
\draw (10,22) circle (9) [color=gray!30];
\draw (5,11.1) circle (3);
\draw (9.6,29.5) circle (1.5);
\draw (8,23) circle (5);
\draw (6.7,15.9) circle (2.1);

\node (s) at (10,34) [nodeDecorate,color=lightgray,label=above left:$s$] {};
\node (u) at (10,31) [nodeDecorate,label=above left:{ $u$}] {};

\node (tp) at (6.2,13.9) [nodeDecorate,label=right:{ $z_4$}] {};
\node (t) at (5.5,11) [nodeDecorate,label=left:$t$] {};

\path {
	(s) edge[snake,-,color=lightgray] node {\quad\quad$\pi_s$} (u)
	(u) edge node {} (tp)
	(tp) edge node {} (t)
};

\node (a) at (9.3,28) [nodeDecorate,label=below left:$v$] {};
\node (b) at (7.1,18) [nodeDecorate,] {};
\node (h) at (10,19) [nodeDecorate,] {};

\path {
	(b) edge node {} (h)
};

\path {
	(a) edge[dashed,bend left=5] node {} (h)
};

\end{tikzpicture}
}

\caption{Example application of $\oracleright(C,e)$ on a spine of the recursion tree}
\label{fig:oracleleftexample}
\end{figure}
\begin{lemma}
  \emph{(Lemma~\ref{lem:removebackedge} restated)} In a spine of the
  recursion tree, operations $\oracleright(C,e)$ can be implemented in
  $O(|V_X|)$ total time.
\end{lemma}

In the right branches along a spine, we remove all back edges in
$lb(u)$. This is done by starting from the last edge in $lb(u)$,
i.e. proceeding in reverse DFS postorder. In the example from
Fig.~\ref{fig:Certificate}, we remove the back edges $(z_1,u) \ldots
(z_4,u)$. To update the certificate corresponding to $B_{u,t}$, we
have to (i) update the lowpoints in each vertex of $H_X$; (ii) prune
vertices that cease to be in $B_{u,t}$ after removing a back edge. For
a vertex $w$ in the tree, there is no need to update~$\gamma(w)$.

Consider the update of lowpoints in the DFS tree.  For a back edge
$b_i = (z_i,u)$, we traverse the vertices in the path from $z_i$
towards the root $u$. By definition of lowpoint, these are the only
lowpoints that can change.  Suppose that we remove back edge $(z_4,u)$
in the example from Fig.~\ref{fig:Certificate}, only the lowpoints of
the vertices in the path from $z_4$ towards the root $u$ change.
Furthermore, consider a vertex $w$ in the tree that is an ancestor of
at least two endpoints $z_i, z_j$ of back edges $b_i$, $b_j$. The
lowpoint of $w$ does not change when we remove $b_i$.  These
observations lead us to the following lemma.

\begin{lemma}  \label{lem:lowpoints}
  In a spine of the recursion tree, the update of lowpoints in the
  certificate by operation $\oracleright(C,e)$ can be done in
  $O(|V_X|)$ total time. 
\end{lemma}
\begin{proof}
  Take each back edge $b_i = (z_i,u)$ in the order defined by
  $\chooseedge(C,u)$. Remove $b_i$ from $lb(u)$ and $ab(z_i)$.
  Starting from $z_i$, consider each vertex $w$ in the path from $z_i$
  towards the root $u$.  On vertex $w$, we update
  $\mathit{lowpoint}(w)$ using the standard procedure: take the
  endpoint $y$ of the first edge in $ab(w)$ (the back edge that goes
  the nearest to the root of the tree) and choosing the minimum
  between $\gamma(y)$ and the lowpoint of each child of $w$.  When the
  updated $\mathit{lowpoint}(w) = \gamma(u)$, we stop examining the
  path from $z_i$ to $u$ since it implies that the lowpoint of the
  vertex can not be further reduced (i.e. $w$ is both an ancestor to
  both $z_i$ and $z_{i+1}$).

  If $\mathit{lowpoint}(w)$ does not change we cannot pay to explore
  its children. In order to get around this, for each vertex we
  dynamically maintain, throughout the spine, a list $l(w)$ of its
  children that have lowpoint equal to $\gamma(u)$. Then, we can test
  in constant time if $l(w) \neq \emptyset$ and $y$ (the endpoint of
  the first edge in $ab(w)$) is not the root $u$. If both conditions
  are satisfied $\mathit{lowpoint}(w)$ changes, otherwise it remains
  equal to $\gamma(u)$ and we stop. The total time to create the lists
  is $O(|V_X|)$ and the time to update is bounded by the number of
  tree edges traversed, shown to be $O(|V_X|)$ in the next paragraph.

  The cost of updating the lowpoints when removing all back edges
  $b_i$ is $O(|V_X|)$: there are $O(|V_X|)$ tree edges and we do not
  traverse the same tree edge twice since the process described stops
  at the first common ancestor of endpoints of back edges $b_i$ and
  $b_{i+1}$. By contradiction: if a tree edge $(x,y)$ would be
  traversed twice when removing back edges $b_i$ and $b_{i+1}$, it
  would imply that both $x$ and $y$ are ancestors of $z_i$ and
  $z_{i+1}$ (as edge $(x,y)$ is both in the path $z_i$ to $u$ and the
  path $z_{i+1}$ to $u$) but we stop at the first ancestor of $z_i$
  and $z_{i+1}$.
\end{proof}

Let us now consider the removal of vertices that are no longer in
$B_{u,t}$ as consequence of operation $\oracleright(C,e)$ in a spine
of the recursion tree. By removing a back edge $b_i = (z_i,u)$, it is
possible that a vertex $w$ previously in $H_X$ is no longer in the
bead string $B_{u,t}$ (e.g. $w$ is no longer biconnected to $u$ and
thus there is no simple path $u \leadsto w \leadsto t$).

\begin{lemma} \label{lem:cutbranches}
  In a spine of the recursion tree, the branches of the DFS that are
  no longer in $B_{u,t}$ due to operation $\oracleright(C,e)$ can be
  removed from the certificate in $O(|V_X|)$ total time.
\end{lemma}
\begin{proof}
  To prune the branches of the DFS tree that are no longer in $H_X$,
  consider again each vertex $w$ in the path from $z_i$ towards the
  root $u$ and the vertex $y$, parent of $w$. It is easy to check if
  $y$ is an articulation point by verifying if the updated
  $\mathit{lowpoint}(w) \leq \gamma(y)$ and there exists $x$ not in
  the subtree of $w$. If $w$ is not in the leftmost path, then $t$ is
  not in the subtree of $w$. If that is the case, we have that $w
  \notin B_{u,t}$, and therefore we cut the subtree of $w$ and
  bookkeep it in $I$ to restore later. Like in the update the
  lowpoints, we stop examining the path $z_i$ towards $u$ in a vertex
  $w$ when $\mathit{lowpoint}(w) = \gamma(u)$ (the lowpoints and
  biconnected components in the path from $w$ to $u$ do not change).
  When cutting the subtree of $w$, note that there are no back edges
  connecting it to $B_{u,t}$ ($w$ is an articulation point) and
  therefore there are no updates to the lists $lb$ and $ab$ of the
  vertices in $B_{u,t}$.  Like in the case of updating the lowpoints,
  we do not traverse the same tree edge twice (we use the same halting
  criterion).
\end{proof}

With Lemma~\ref{lem:lowpoints} and Lemma~\ref{lem:cutbranches} we
finalize the proof of Lemma~\ref{lem:removebackedge}.
Fig.~\ref{fig:oracleleftexample} shows the changes the bead string
$B_{u,t}$ from Fig.~\ref{fig:Certificate} goes through in the
corresponding spine of the recursion tree.

\subsubsection{Operation $\oracleleft(C,e)$} 
In the binary nodes of a spine, we use the fact that in every left
branching from that spine the graph is the same (in a spine we only
remove edges incident to $u$ and on a left branch from the spine we
remove the vertex $u$) and therefore its block tree is also the
same. In Fig.~\ref{fig:blocktree_without_u}, we show the resulting
block tree of the graph from Fig.~\ref{fig:Certificate} after having
removed vertex $u$. However, the certificates on these left branches
are not the same, as they are rooted at different vertices. In the
example we must compute the certificates $C_1 \ldots C_4$
corresponding to bead strings $B_{z_1,t} \ldots B_{z_4,t}$. We do not
account for the cost of the left branch on the last node of spine
(corresponding to $B_{v,t}$) as the node is unary and we have shown in
Lemma~\ref{lem:unary_left} how to maintain the certificate in $O(1)$
time.

By using the reverse DFS postorder of the back edges, we are able to
traverse each edge in $H_X$ only an amortized constant number of times
in the spine.
\begin{lemma}
  \emph{(Lemma~\ref{lem:promotebackedge} restated)} The calls to
  operation $\oracleleft(C,e)$ in a spine of the recursion tree can be
  charged with a time cost of $O(|E_X|)$ to that spine.
\end{lemma}

To achieve this time cost, for each back edge $b_i = (z_i,u)$, we
compute the certificate corresponding to $B_{z_i,t}$ based on the
certificate of $B_{z_{i-1},t}$. Consider the compacted head $H_X =
(V_X , E_X )$ of the bead string $B_{u,t}$. We use $O(|E_X|)$ time to
compute the first certificate $C_1$ corresponding to bead string
$B_{z_1,t}$. Fig.~\ref{fig:spine-left-update} shows bead string
$B_{z_1,t}$ from the example of Fig.~\ref{fig:Certificate}.

\begin{lemma}
  The certificate $C_1$, corresponding to bead string $B_{z_1,t}$, can
  be computed in $O(|E_X|)$ time.
\end{lemma}
\begin{proof}
  Let $t'$ be the last vertex in the path $u \leadsto t$ s.t.  $t' \in
  V_X$. Since $t'$ is an articulation point, the subtree of the DFS
  tree rooted in $t'$ is maintained in the case of removal vertex
  $u$. Therefore the only modifications of the DFS tree occur in head
  $H_X$ of $B_{u,t}$.

  To compute $C_1$, we remove $u$ and rebuild the certificate starting
  form $z_1$ using the algorithm from
  Lemma~\ref{lem:certificate_scratch} restricted to $H_X$ and using
  $t'$ as target and $\gamma(t')$ as a baseline to $\gamma$ (instead
  of the depth). In particular we do the following.  To set $t'$ to be
  in the leftmost path, we perform a DFS traversal of graph $H_X$
  starting from $z_1$ and stop when we reach vertex $t'$. Then compute
  the DFS tree, traversing the path $z_1 \leadsto t'$ first.
	
  {\it Update of $\gamma$.} For each tree edge $(v,w)$ in the $t'
  \leadsto z_1$ path, we set $\gamma(v)=\gamma(w)-1$, using
  $\gamma(t')$ as a baseline.  During the rest of the traversal, when
  visiting vertex $v$, let $w$ be the parent of $v$ in the DFS
  tree. We set $\gamma(v)=\gamma(w)+1$. This maintains the property
  that $\gamma(v)>\gamma(w)$ for any $w$ ancestor of $v$.

  {\it Lowpoints and pruning the tree.}  Bottom-up in the DFS-tree,
  compute the lowpoints using the lowpoints of the children.  For $z$
  the parent of $v$, if $\mathit{lowpoint}(v) \leq \gamma(z)$ and $v$
  is not in the leftmost path in the DFS, cut the subtree of $v$ as it
  does not belong to~$B_{z_1,t}$.

  {\it Computing $lb$ and $ab$.} In the traversal, when finding a back
  edge $e=(v,w)$, if $w$ is a descendant of $v$ we append $e$ to
  $ab(w)$. This maintains the DFS preorder in the ancestor back edge
  list. After the first scan of $N(v)$ is over and all the recursive
  calls returned, re-scan the neighborhood of $v$. If $e=(v,w)$ is a
  back edge and $w$ is an ancestor of $v$, we add $e$ to $lb(w)$. This
  maintains the DFS postorder in the descendant back edge list.  This
  procedure takes $O(|E_X|)$ time.
\end{proof}

\begin{figure}[t!]
\centering
\begin{tikzpicture}
[nodeDecorate/.style={shape=circle,inner sep=1pt,draw,thick,fill=black},%
  lineDecorate/.style={-,dashed},%
  elipseDecorate/.style={color=gray!30},
  scale=0.25]
\draw (5,11.1) circle (3);
\draw[rotate around={-10:(8,23)}] (8,23) ellipse (3 and 5);
\draw[rotate around={-23:(12.6,26.5)}] (12.6,26.5) ellipse (4 and 1);
\draw[rotate around={55:(15,23.5)}] (15,23.5) ellipse (2.1 and 0.5);
\draw (6.7,15.9) circle (2.1);
\draw (17.5,23.5) circle (2.0);
\draw (20.8,23) circle (1.5);

\node (s) at (10,34) [nodeDecorate,color=lightgray,label=above left:$s$] {};
\node (u) at (10,31) [nodeDecorate,color=lightgray,label=above left:{ $u$}] {};

\node (tp) at (6.2,13.9) [nodeDecorate,label=right:{ $z_4$}] {};
\node (t) at (5.5,11) [nodeDecorate,label=below:$t$] {};
\node (a) at (9.3,28) [nodeDecorate,label=below left:$v$] {};

\path {
	(s) edge[snake,-,color=lightgray] node {\quad\quad$\pi_s$} (u)
	(u) edge[color=lightgray] node {} (a)
	(a) edge node {} (tp)
	(tp) edge node {} (t)
};

\node (b) at (7.1,18) [nodeDecorate,] {};
\node (c) at (16,25) [nodeDecorate,] {};
\node (d) at (19.3,23.5) [nodeDecorate,label=above:$~z_1$] {};
\node (e) at (14,22) [nodeDecorate,label=below:$z_3$] {};
\node (f) at (22,22.2) [nodeDecorate,label=right:$z_2$] {};
\node (g) at (17,22) [nodeDecorate,] {};
\node (h) at (10,22) [nodeDecorate,] {};

\path {
	(a) edge node {} (c)
	(c) edge node {} (d)
	(d) edge node {} (f)
	(c) edge node {} (e)
	(d) edge node {} (g)
	(b) edge node {} (h)
};

\path {
	(g) edge node {} (c)
	(a) edge node {} (h)


	(u) edge[bend left=-50,color=lightgray!50] node {} (tp)
	(f) edge[bend left=-40,color=lightgray!50] node {} (u)
	(e) edge[bend left=-10,color=lightgray!50] node {} (u)
	(d) edge[bend left=-20,color=lightgray!50] node {} (u)
};

\end{tikzpicture}
\caption{Block tree after removing vertex $u$}
\label{fig:blocktree_without_u}
\end{figure}
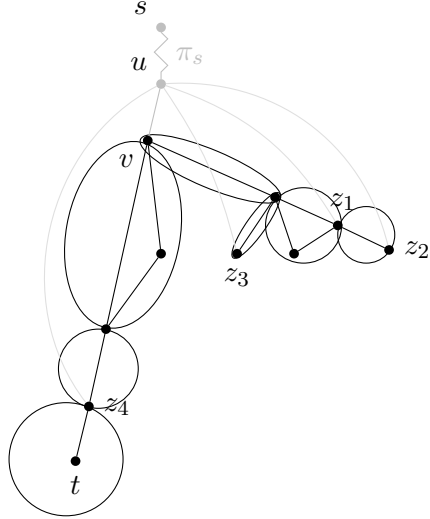

To compute each certificate $C_i$, corresponding to bead string
$B_{z_i,t}$, we are able to avoid visiting most of the edges that
belong $B_{z_{i-1},t}$. Since we take $z_i$ in reverse DFS postorder,
on the spine of the recursion we visit $O(|E_X|)$ edges plus a term
that can be amortized.

\begin{lemma} \label{lem:cost-spine-not-amortized}
  For each back edge $b_i = (z_i,u)$ with $i>1$, let ${E_X}_i'$ be the
  edges in the first bead in common between $B_{z_i,t}$ and
  $B_{z_{i-1},t}$. The total cost of computing all certificates
  $B_{z_i,t}$ in a spine of the recursion tree is: $O(|E_X| +
  \sum_{i>1}{|{E_X}_i'|})$.
\end{lemma}
\begin{proof}
  Let us compute the certificate $C_i$: the certificate of the left
  branch of the $i$th node of the spine where we augment the path with
  back edge $b_i = (z_i,u)$ of $lb(u)$.

  For the general case of $C_i$ with $i>1$ we also rebuild (part) of
  the certificate starting from $z_i$ using the procedure from
  Lemma~\ref{lem:certificate_scratch} but we use information gathered
  in $C_{i-1}$ to avoid exploring useless branches of the DFS
  tree. The key point is that, when we reach the first bead in common
  to both $B_{z_i,t}$ and $B_{z_{i-1},t}$, we only explore edges
  internal to this bead.  If an edge $e$ that leaves the bead leads to
  $t$, we can reuse a subtree of $C_{i-1}$. If $e$ does not lead to
  $t$, then it has already been explored (and cut) in $C_{i-1}$ and
  there is no need to explore it again since it is going to be
  discarded.

  In detail, we start computing a DFS from $z_i$ in $B_{u,t}$ until we
  reach a vertex $t' \in B_{z_{i-1},t}$. Note that the bead of $t'$
  has one entry point and one exit point in $C_{i-1}$. After reaching
  $t'$ we proceed with the traversal using only edges already in
  $C_{i-1}$. When arriving at a vertex $w$ that is not in the same
  bead of $t'$, we stop the traversal. If $w$ is in a bead towards
  $t$, we reuse the subtree of $w$ and use $\gamma(w)$ as a baseline
  of the numbering $\gamma$. Otherwise $w$ is in a bead towards
  $z_{i-1}$ and we cut this branch of the certificate. When all edges
  in the bead of $t'$ are traversed, we proceed with visit in the
  standard way.

  Given the order we take $b_i$, each bead is not added more than once
  to a certificate $C_i$, therefore the total cost over the spine is
  $O(|E_X|)$.  Nevertheless, the internal edges ${E_X}_i'$ of the
  first bead in common between $B_{z_i,t}$ and $B_{z_{i-1},t}$ are
  explored for each back edge $b_i$.
\end{proof}

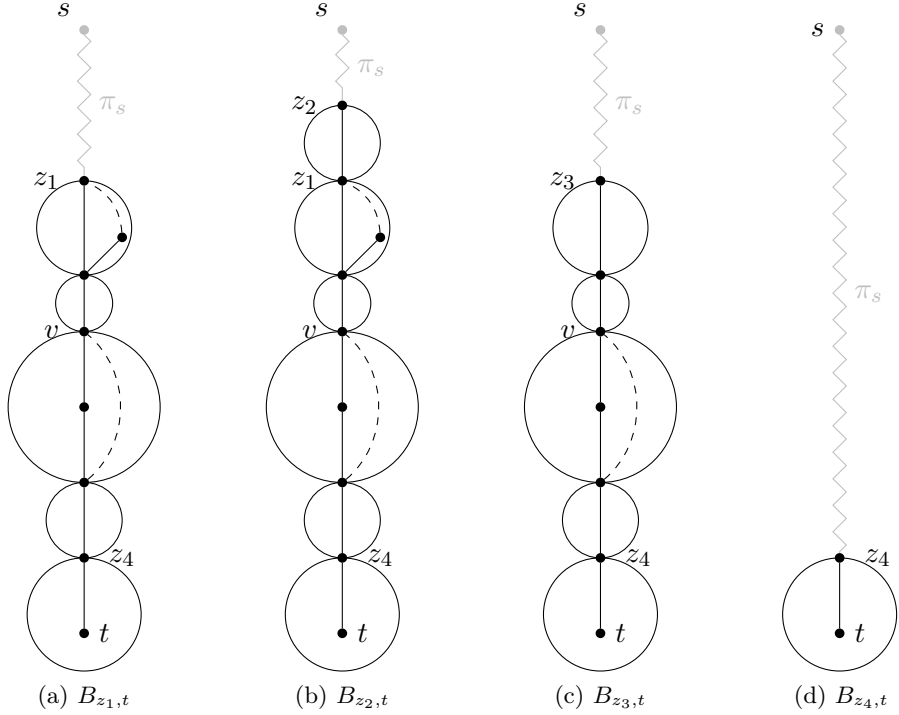
\begin{figure}[t]
\centering
\subfloat[$B_{z_1,t}$] {

\begin{tikzpicture}
[nodeDecorate/.style={shape=circle,inner sep=1pt,draw,thick,fill=black},%
  lineDecorate/.style={-,dashed},%
  elipseDecorate/.style={color=gray!30},
  scale=0.25]
\draw (10,23.5) circle (2.5);
\draw (10,19.5) circle (1.5);
\draw (10,14) circle (4);
\draw (10,8) circle (2);
\draw (10,3) circle (3);

\node (s) at (10,34) [nodeDecorate,color=lightgray,label=above left:$s$] {};

\node (z2) at (10,26) [nodeDecorate,label=left:$z_1~$] {};
\node (a) at (10,21) [nodeDecorate] {};
\node (b) at (12,23) [nodeDecorate] {};
\node (v) at (10,18) [nodeDecorate,label=left:$v~$] {};
\node (c) at (10,14) [nodeDecorate] {};
\node (d) at (10,10) [nodeDecorate] {};
\node (z4) at (10,6) [nodeDecorate,label=right:$~z_4$] {};
\node (t) at (10,2) [nodeDecorate,label=right:$t$] {};

\path {
	(s) edge[snake,-,color=lightgray] node {\quad\quad$\pi_s$} (z2)
	(z2) edge node {} (a)
	(a) edge node {} (b)
	(a) edge node {} (v)
	(v) edge node {} (c)
	(c) edge node {} (d)
	(d) edge node {} (z4)
	(z4) edge node {} (t)
};

\path {
	(d) edge[dashed,bend left=-50] node {} (v)
	(b) edge[dashed,bend left=-30] node {} (z2)
};
\end{tikzpicture}
}
\hspace{1cm}
\subfloat[$B_{z_2,t}$] {

\begin{tikzpicture}
[nodeDecorate/.style={shape=circle,inner sep=1pt,draw,thick,fill=black},%
  lineDecorate/.style={-,dashed},%
  elipseDecorate/.style={color=gray!30},
  scale=0.25]
\draw (10,28) circle (2);
\draw (10,23.5) circle (2.5);
\draw (10,19.5) circle (1.5);
\draw (10,14) circle (4);
\draw (10,8) circle (2);
\draw (10,3) circle (3);

\node (s) at (10,34) [nodeDecorate,color=lightgray,label=above left:$s$] {};

\node (z1) at (10,30) [nodeDecorate,label=left:$z_2~$] {};
\node (z2) at (10,26) [nodeDecorate,label=left:$z_1~$] {};
\node (a) at (10,21) [nodeDecorate] {};
\node (b) at (12,23) [nodeDecorate] {};
\node (v) at (10,18) [nodeDecorate,label=left:$v~$] {};
\node (c) at (10,14) [nodeDecorate] {};
\node (d) at (10,10) [nodeDecorate] {};
\node (z4) at (10,6) [nodeDecorate,label=right:$~z_4$] {};
\node (t) at (10,2) [nodeDecorate,label=right:$t$] {};

\path {
	(s) edge[snake,-,color=lightgray] node {\quad\quad$\pi_s$} (z1)
	(z1) edge node {} (z2)
	(z2) edge node {} (a)
	(a) edge node {} (b)
	(a) edge node {} (v)
	(v) edge node {} (c)
	(c) edge node {} (d)
	(d) edge node {} (z4)
	(z4) edge node {} (t)
};

\path {
	(d) edge[dashed,bend left=-50] node {} (v)
	(b) edge[dashed,bend left=-30] node {} (z2)
};
\end{tikzpicture}
}
\hspace{1cm}
\subfloat[$B_{z_3,t}$] {

\begin{tikzpicture}
[nodeDecorate/.style={shape=circle,inner sep=1pt,draw,thick,fill=black},%
  lineDecorate/.style={-,dashed},%
  elipseDecorate/.style={color=gray!30},
  scale=0.25]
\draw (10,23.5) circle (2.5);
\draw (10,19.5) circle (1.5);
\draw (10,14) circle (4);
\draw (10,8) circle (2);
\draw (10,3) circle (3);

\node (s) at (10,34) [nodeDecorate,color=lightgray,label=above left:$s$] {};

\node (z3) at (10,26) [nodeDecorate,label=left:$z_3~$] {};
\node (a) at (10,21) [nodeDecorate] {};
\node (v) at (10,18) [nodeDecorate,label=left:$v~$] {};
\node (c) at (10,14) [nodeDecorate] {};
\node (d) at (10,10) [nodeDecorate] {};
\node (z4) at (10,6) [nodeDecorate,label=right:$~z_4$] {};
\node (t) at (10,2) [nodeDecorate,label=right:$t$] {};

\path {
	(s) edge[snake,-,color=lightgray] node {\quad\quad$\pi_s$} (z3)
	(z3) edge node {} (a)
	(a) edge node {} (v)
	(v) edge node {} (c)
	(c) edge node {} (d)
	(d) edge node {} (z4)
	(z4) edge node {} (t)
};

\path {
	(d) edge[dashed,bend left=-50] node {} (v)
};
\end{tikzpicture}
}
\hspace{1cm}
\subfloat[$B_{z_4,t}$] {

\begin{tikzpicture}
[nodeDecorate/.style={shape=circle,inner sep=1pt,draw,thick,fill=black},%
  lineDecorate/.style={-,dashed},%
  elipseDecorate/.style={color=gray!30},
  scale=0.25]
\draw (10,3) circle (3);

\node (s) at (10,34) [nodeDecorate,color=lightgray,label=left:$s$] {};

\node (z4) at (10,6) [nodeDecorate,label=right:$~z_4$] {};
\node (t) at (10,2) [nodeDecorate,label=right:$t$] {};

\path {
	(s) edge[snake,-,color=lightgray] node {\quad\quad$\pi_s$} (z4)
	(z4) edge node {} (t)
};

\end{tikzpicture}
}
\caption{Certificates of the left branches of a spine}
\label{fig:spine-left-update}
\end{figure}

Although the edges in ${E_X}_i'$ are in a common bead between
$B_{z_i,t}$ and $B_{z_{i-1},t}$, these edges must be visited. The
entry point in the common bead can be different for $z_i$ and
$z_{i-1}$, the DFS tree of that bead can also be different. For an
example, consider the case where $z_i, \ldots, z_j$ are all in the
same bead after the removal of $u$. The bead strings $B_{z_i,t} \ldots
B_{z_j,t}$ are the same, but the roots $z_i, \ldots, z_j$ of the
certificate are different, so we have to compute the corresponding DFS
of the first bead $|j-i|$ times. Note that this is not the case for
the other beads in common: the entry point is always the same.

\begin{lemma} \label{lem:left-amortize}
  The cost $O(|E_X| + \sum_{i>1}{|{E_X}_i'}|)$ on a spine of the
  recursion tree can be amortized to $O(|E_X|)$.
\end{lemma}
\begin{proof}
  We can charge the cost $O(|{E_X}_i'|)$ of exploring the edges in the
  first bead in common between $B_{z_i,t}$ and $B_{z_{i-1},t}$ to
  another node in the recursion tree. Since this common bead is the
  head of at least one certificate in the recursion subtree of the
  left child of the $i$th node of the spine.  Specifically, we charge
  the first and only node in the \emph{leftmost} path of the $i$th
  child of the spine that has exactly the edges ${E_X}_i'$ as head of
  its bead string: (i) if $|{E_X}_i'| \le 1$ it corresponds to a unary
  node or a leaf in the recursion tree and therefore we can charge it
  with $O(1)$ cost; (ii) otherwise it corresponds to a first node of a
  spine and therefore we can also charge it with $O(|{E_X}_i'|)$. We
  use this charging scheme when $i \neq 1$ and the cost is always
  charged in the leftmost recursion path of $i$th node of the spine,
  consequently we never charge a node in the recursion tree more than
  once.
\end{proof}

Lemmas~\ref{lem:cost-spine-not-amortized} and~\ref{lem:left-amortize}
finalize the proof of Lemma~\ref{lem:promotebackedge}.
Fig.~\ref{fig:spine-left-update} shows the certificates of bead
strings $B_{z_i,t}$ on the left branches of the spine from
Figure~\ref{fig:spine}.


\bigskip
\bigskip


\section{Discussion and conclusions}
In the first part of this chapter, we showed that it is possible
(Algorithm~\ref{cyclealgo}) to list all bubbles with a given source in
a directed graph with linear delay, thus solving
Problem~\ref{prob:unweighted:stbubbles}. Moreover, it is possible
(Algorithm~\ref{alg:main}) to enumerate all bubbles, for all possible
sources, thus solving Problem~\ref{prob:unweighted:bubbles}, in
$O((m+n)(\eta+n))$ total time, where $\eta$ is the number of bubbles.

Unfortunately, this algorithm is not a good replacement for \ks's
listing algorithm (Section~\ref{sec:kissplice:algorithm}), since for
the task listing bubbles corresponding to AS events, in practice, the
latter performs better.  Recall that \ks searches for \emph{cycles}
satisfying conditions (i) to (iv) of
Section~\ref{sec:kissplice:algorithm}, the cycles satisfying condition
(i) correspond to the $(s,t)$-bubbles, the remaining conditions, (ii)
to (iv), are constraints for the length of the sequences corresponding
to each path. In \ks's listing algorithm, several prunings based on
these constraints are applied to avoid the enumeration of bubbles that
are guarantee not to satisfy the constraints. On the other hand,
Algorithm~\ref{alg:main} efficiently lists all bubbles directly,
i.e. cycles satisfying condition (i), but it is not evident how to
apply the same prunings for constraints (ii) to (iv). In the end, we
have to list all bubbles and, in a post-processing step, filter out
the ones not satisfying the constraints. Since in typical cases, the
number of bubbles satisfying the constraints is small compared to the
total\footnote{Bubbles not satisfying the constraints correspond to,
among others, de Bruijn graph artifacts, other genomic polymorphisms
(i.e. inversions), repeat related structures, and, more rarely,
multiple exclusive exons.}  number of bubbles, this approach is worse
than \ks's listing algorithm. In Chapter~\ref{chap:weighted}, we
present a \emph{practical} polynomial delay algorithm that directly
lists bubbles satisfying constraints (ii) and (iv).

Nonetheless, the problem of listing bubbles in a directed graph is
interesting from a theoretical point of view, since $(s,t)$-bubbles
are natural substructures in directed graphs\footnote{For instance,
$(s,t)$-bubbles are related to 2-vertex-connected directed graphs
(\cite{Bang-Jensen08}) where every pair of vertices are extremities of
at least one bubble. }. Moreover, Algorithm~\ref{cyclealgo} required a
non-trivial adaptation of Johnson's algorithm (\cite{Johnson75}) for
listing cycles in directed graphs, and is the first linear delay
algorithm to list all bubbles with a given source in a directed graph.

In the second part of this chapter, we showed that Johnson's
algorithm, the long-standing best known solution to list cycles, is
surprisingly not optimal for undirected graphs. We then presented an
$O(m + \sum_{c \in \setofcycles(G)}{|c|})$ algorithm to list cycles in
undirected graphs, where $\setofcycles(G)$ in the set of cycles and
$|c|$ the length of cycle $c$. Clearly, $\Omega(m)$ time is necessary
to read the graph and $\Omega(\sum_{c \in \setofcycles(G)}{|c|}))$
time to list the output. Thus, our algorithm is optimal. Actually, we
presented an optimal algorithm to list $st$-paths in undirected graphs
and used an optimality preserving reduction from cycle listing to
$st$-path listing.

This chapter raises some interesting questions, for instance, whether
it is possible to directly list the $(s,t)$-bubbles satisfying path
length constraints. In Chapter~\ref{chap:weighted}, we give an
affirmative answer to this question. Another natural question is
whether it is possible to apply techniques similar to the ones
presented in Section~\ref{sec:unweighted:cycle} to improve Johnson's
algorithm for directed graphs or Algorithm~\ref{cyclealgo}. An
important invariant maintained by our optimal $st$-path listing
algorithm is the following: in the beginning of every recursive call
every edge in the graph is contained in some $st$-path. Intuitively,
this means that at every step the graph is cleaned and only the
necessary edges are kept. However, in directed graphs is NP-hard to
decide if a given a arc belongs to a $st$-path or $(s,t)$-bubble
(\cite{Fortune80}). Thus, it is unlikely that the same kind of
cleaning can be done in directed graphs. This seems a hard barrier to
overcome. The last question is, provided we are only interested in
counting\footnote{Counting $st$-paths is \#P-hard, so an algorithm
polynomial in the size of the graph is very unlikely.}, whether it is
possible to improve our $st$-path listing algorithm to
$O(|\setofpaths_{st}(G)|)$. In other words, is it the possible improve
our algorithm to spend only a constant time per path if it is not
required to output each $st$-path. This may seem impossible, but there
are listing algorithms achieving this complexity, e.g. listing
spanning trees (\cite{Marino14}).

\chapter{Listing in weighted graphs}
\label{chap:weighted}
\minitoc
In this chapter, we present efficient algorithms to list paths and
bubbles, satisfying path length constraints in \emph{weighted}
directed graphs. The chapter is divided in two main parts.

The first part (Section~\ref{sec:weighted:bubble}) is strongly based
on our paper \cite{Sacomoto13}, and its goal is to present a
polynomial delay algorithm to list all bubbles in weighted directed
graphs, such that each path $p_1,p_2$ in the bubble has length bounded
by $\alpha_1,\alpha_2$ respectively. For a directed graph with $n$
vertices and $m$ arcs, the method we propose lists all bubbles with a
given source in $O(n(m + n \log))$ delay. Moreover, we experimentally
show that this algorithm is significantly faster than the listing
algorithm of \ks (version 1.6) to identify bubbles corresponding to
alternative splicing events.

The second part (Section~\ref{sec:weighted:path}) is strongly based on
our paper \cite{Grossi14} (in preparation), and its goal is to present
a general scheme to list bounded length $st$-paths in weighted
directed or undirected graphs using memory linear in the size of the
graph, independent of the number of paths output. For undirected
non-negatively weighted graphs, we also show an improved algorithm
that lists all $st$-paths with length bounded by $\alpha$ in $O( (m +
t(n,m)) \gamma)$ total time, where $\gamma$ is the number the
$st$-paths with length bounded by $\alpha$ and $t(m,n)$ is the time to
compute a shortest path tree. In particular, this is $O(m \gamma)$ for
unit weights and $O((m + n \log n) \gamma)$ for general non-negative
weights. Moreover, we show how to modify the general scheme to output
the paths in increasing order of their lengths.


\bigskip
\bigskip


\section{Listing bounded length bubbles in weighted directed graphs} \label{sec:weighted:bubble}

\subsection{Introduction}
In the previous chapter, we proposed a linear delay algorithm to list
all bubbles in a directed graph, in particular, applicable also to de
Bruijn graphs. Although interesting from a theoretical point of view,
the algorithm cannot replace the listing algorithm of \ks
(Section~\ref{sec:kissplice:algorithm}), since for the task of listing
bubbles corresponding to AS events, in practice, the latter performs
better. Indeed, this is due to the fact that the bubbles corresponding
to alternative splicing events (excluding mutually exclusive exons)
satisfy some path length constraints. We can use this information in
simple backtracking algorithm of \ks (version 1.6) to efficiently
prune the branches of the search tree; we cannot, however, give any
theoretical guarantees. In the worst case, the algorithm is still
exponential in the number of bubbles and the size of the graph. On the
other hand, it is not clear how to incorporate the same prunings in
the linear delay algorithm of the last chapter. As a result, the
algorithm lists a huge number of bubbles that have to be checked, in a
post-processing step, for the path length constraints. In this
chapter, we present a polynomial delay algorithm to directly list
bubbles satisfying the path constraints. Moreover, we experimentally
show that the algorithm is several orders of magnitude faster
than \ks's (version 1.6) listing algorithm.

As stated in Section~\ref{sec:unweighted:bubble}, the problem of
identifying bubbles with path length constraints was considered before
in the genome assembly (\cite{Soapdenovo,Idba,Velvet,Abyss}) and in
the variant finding (\cite{Cortex,Bubbleparse}) contexts. However, in
the first case the goal was not to list all bubbles. In general,
assemblers perform a greedy search for bubbles in order to
``linearize'' a de Bruijn graph. Moreover, the path length constraints
are symmetric, that is both paths should satisfy the same length
constraint. In the second case, the goal is really to list bubbles,
but in \cite{Cortex} the search is restricted to non-branching
bubbles, while in \cite{Bubbleparse} this constraint is relaxed to a
bounded (small) number branching internal vertices. Additionally, in
\cite{Bubbleparse} there is no strong theoretical guarantee for the
time complexity; the algorithm is basically an unconstrained DFS,
similar to the listing algorithm of \ks, where the search is truncated
at a given depth.

In this chapter, we introduce the first polynomial delay algorithm to
list all bubbles with length constraints in a weighted directed graph.
Its complexity for general non-negatively weighted graphs is
$O(n(m+n \log n))$ (Section~\ref{sec:alg_delay}) where $n$ is the
number of vertices in the graph, $m$ the number of arcs.  In the
particular case of de Bruijn graphs, the complexity is $O(n(m+n \log
\alpha))$ (Section~\ref{subsec:dijkstra}) where $\alpha$ is a constant
related to the length of the skipped part in an alternative splicing
event. In practice, an algorithmic solution in $O(nm\log n)$
(Section~\ref{subsec:comp_kissplice}) appears to work better on de
Bruijn graphs built from such data.  We implemented the latter, show
that it is more efficient than previous approaches and outline that it
allows to discover novel long alternative splicing events.

\subsection{De Bruijn graphs and bounded length bubbles} \label{sec:debruijn_and_as}
As was shown in Chapter~\ref{chap:kissplice}, polymorphisms
(i.e. variable parts) in a transcriptome (including alternative
splicing events) correspond to recognizable patterns in the DBG that
are precisely the $(s,t)$-bubbles
(Definition~\ref{def:unweighted:bubble}). Intuitively, the variable
parts correspond to alternative paths and the common parts correspond
to the beginning and end points of those paths. More formally, any
process generating patterns $awb$ and $aw'b$ in the sequences, with
$a,b,w,w' \in \Sigma^*$, $|a| \geq k, |b|\geq k$ and $w$ and $w'$ not
sharing any $k$-mer, creates a $(s,t)$-bubble in the DBG. In the
special case of AS events (excluding mutually exclusive exons), since
$w'$ is empty, one of the paths corresponds to the \emph{junction} of
$ab$, i.e. to $k$-mers that contain at least one letter of each
sequence. Thus the number of vertices of this path in the DBG is
predictable: it is at most\footnote{The size is \emph{exactly} $k-1$
if $w$ has no common prefix with $b$ and no common suffix with $a$.}
$k-1$. An example is given in Fig.~\ref{fig:weighted:ex_bubble}. In
practice (see Section~\ref{sec:kissplice:dbg_models}), an upper bound
$\alpha$ to the other path and a lower bound $\beta$ on both paths is
also imposed. In other words, an AS event corresponds to a
$(s,t)$-bubble with paths $p_1$ and $p_2$ such that $p_1$ has at most
$\alpha$ vertices, $p_2$ at most $k-1$ and both have at least $\beta$
vertices.

\smallskip
\begin{figure}[htb]
\centering
\resizebox{!}{2cm}{%
  \begin{tikzpicture}[->,>=stealth',shorten >=1pt,auto,node distance=1.5cm,
      thick,main node/.style={rectangle,draw,font=\sffamily\bfseries}]

    \node[main node] (1) {\color{red}ACT};
    \node[main node] (2) [right of=1] {{\color{red}CT}{\color{blue}G}};
    \node[main node] (3) [above right of=2, yshift=-0.2cm] {{\color{red}T}{\color{black}GG}};
    \node[main node] (4) [right of=3] {\color{black}GGA};
    \node[main node] (5) [right of=4] {{\color{black}GA}{\color{cyan}G}};
    \node[main node] (6) [right of=5] {{\color{black}A}{\color{cyan}GC}};
    \node[main node] (7) [below right of=6] {\color{cyan}GCG};
    \node[main node] (8) [right=2.4cm of 2, yshift=-0.80cm] {{\color{red}T}{\color{cyan}GC}};
    
    \path[every node/.style={font=\sffamily\small}]
      (1) edge node [left] {} (2)
      (2) edge node [left] {} (3)
          edge node [left] {} (8)
      (3) edge node [left] {} (4)
      (4) edge node [left] {} (5)
      (5) edge node [left] {} (6)
      (6) edge node [left] {} (7)
      (8) edge node [left] {} (7);
  \end{tikzpicture}
}
\caption{DBG with $k=3$ for the sequences: {\color{red}
    ACT}{\color{black}GGA}{\color{cyan}GCG} ($awb$) and {\color{red}
    ACT}{\color{cyan}GCG} ($ab$). The pattern in the sequence
  generates a $(s,t)$-bubble, from {\color{red}CT}{\color{blue}G} to
  {\color{cyan}GCG}. In this case, $b=$ {\color{cyan}GCG} and $w=$ GGA
  have their first letter {\color{blue}G} in common, so the path
  corresponding to the junction $ab$ has $k-1-1 = 1$ vertex.}
\label{fig:weighted:ex_bubble}
\end{figure}
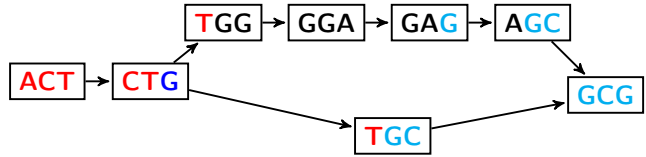

Given a directed graph $G$ with non-negative arc weights $w: E \mapsto
\mathbb{Q}_{\geq 0}$, we can extend
Definition~\ref{def:unweighted:bubble} to $G$ by considering
$(s,t)$-bubbles with length constraints in both paths.

\begin{definition}[$(s,t,\alpha_1,\alpha_2)$-bubble] 
  A \emph{$(s,t, \alpha_1,\alpha_2)$-bubble} in a weighted directed
  graph is a $(s,t)$-bubble with paths $p_1,p_2$ satisfying $w(p_1)
  \leq \alpha_1$ and $w(p_2) \leq \alpha_2$.
\end{definition}

As stated in Chapter~\ref{chap:kissplice}, when dealing with DBGs
built from RNA-seq data, in a lossless preprocessing step, all maximal
non-branching linear paths of the graph (i.e. paths containing only
vertices with in and out-degree 1) are compressed each into one single
vertex, whose label corresponds to the label of the path (i.e. it is
the concatenation of the labels of the vertices in the path without
the overlapping part(s)). The resulting graph is the compressed de
Bruijn graph (cDBG). In the cDBG, the vertices can have labels larger
than $k$, but an arc still indicates a suffix-prefix overlap of size
$k-1$. Finally, since the only property of a bubble corresponding to
an AS event is the constraint on the length of the path, we can
disregard the labels from the cDBG and only keep for each vertex its
label length\footnote{Resulting in a graph with weights in the
  vertices. Here, however, we consider the weights in the arcs. Since
  this is more standard and, in our case, both alternatives are
  equivalent, we can transform one into another by splitting vertices
  or arcs.}. In this way, searching for bubbles corresponding to AS
events in a cDBG can be seen as a particular case of looking for
$(s,t, \alpha_1, \alpha_2)$-bubbles satisfying the lower bound $\beta$
in a non-negative weighted directed graph.

Actually, it is not hard to see that the enumeration of $(s,t,
\alpha_1,\alpha_2)$-bubbles, for all $s$ and $t$, satisfying the lower
bound $\beta$ is NP-hard. Indeed, deciding the existence of at least
one $(s,t, \alpha_1,\alpha_2)$-bubble, for some $s$ and $t$, with the
lower bound $\beta$ in a weighted directed graph where all the weights
are 1 is NP-complete. It follows by a simple reduction from the
Hamiltonian path problem (\cite{Garey79}): given a directed graph $G =
(V,E)$ and two vertices $s$ and $t$, build the graph $G'$ by adding to
$G$ the vertices $s'$ and $t'$, the arcs $(s,s')$ and $(t,t')$, and a
new path from $s'$ to $t'$ with exactly $|V|$ nodes. There is a
$(x,y,|V|+2,|V|+2)$-bubble, for some $x$ and $y$, satisfying the lower
bound $\beta = |V| + 2$ in $G'$ if and only if there is a Hamiltonian
path from $s$ to $t$ in $G$.

From now on, we consider the more general problem of listing
$(s,t,\alpha_1,\alpha_2)$-bubbles (without the lower bound) for an
arbitrary non-negative weighted directed graph $G$ (not restricted to
a cDBG).

\begin{problem}[Listing bounded length bubbles] \label{prob:weighted:bubbles}
  Given a non-negatively weighted directed graph $G = (V,E)$, output
  all $(s,t,\alpha_1,\alpha_2)$-bubbles, for all pairs $s,t \in V$.
\end{problem}

In order to solve Problem~\ref{prob:unweighted:bubbles}, we consider
the problem of listing all bubbles with a given source
(Problem~\ref{prob:unweighted:stbubbles}). Indeed, by trying all
possible sources $s$ we can list all $(s,t)$-bubbles.

\begin{problem}[Listing $(s,*,\alpha_1,\alpha_2)$-bubbles] \label{prob:weighted:sbubbles}
  Given a non-negatively weighted directed graph $G = (V,E)$ and
  vertex $s$, output all $(s,t,\alpha_1,\alpha_2)$-bubbles, for all $t
  \in V$.
\end{problem}

The number of vertices and arcs of $G$ is denoted by $n$ and
$m$, respectively.

\subsection{An $O(n (m + n \log n))$ delay algorithm} \label{sec:alg_delay}
In this section, we present an $O(n (m + n \log n))$ delay algorithm
to enumerate, for a fixed source $s$, all
$(s,t,\alpha_1,\alpha_2)$-bubbles in a general directed graph $G$ with
non-negative weights.  The pseudocode is shown in
Algorithm~\ref{alg:weighted:listbubbles}. It is important to stress
that this pseudocode uses high-level primitives, e.g. the tests in
lines~\ref{alg2:initial_test}, \ref{alg2:include} and
\ref{alg2:exclude}. An efficient implementation for the test in
line~\ref{alg2:include}, along with its correctness and analysis, is
implicitly given in Lemma~\ref{lem:test2}. This is a central result in
this section. For its proof we need Lemma~\ref{lem:bubbles:dist}.

Algorithm~\ref{alg:weighted:listbubbles} uses a recursive strategy,
inspired by the binary partition method, that successively divides the
solution space at every call until the considered subspace is a
singleton.  In order to have a more symmetric structure for the
subproblems, we define the notion of a \emph{pair of compatible
  paths}, which is an object that generalizes the definition of a
$(s,t,\alpha_1,\alpha_2)$-bubble. Given two vertices $s_1,s_2 \in V$
and upper bounds $\alpha_1, \alpha_2 \in \mathbb{Q}_{\geq0}$, the
paths $p_1 = s_1 \leadsto t_1$ and $p_2 = s_2 \leadsto t_2$ are a
\emph{pair of compatible paths} for $s_1$ and $s_2$ if $t_1 = t_2$,
$w(p_1) \leq \alpha_1$, $w(p_2) \leq \alpha_2$ and the paths are
internally vertex-disjoint. Clearly, every
$(s,t,\alpha_1,\alpha_2)$-bubble is also a pair of compatible paths
for $s_1 = s_2 = s$ and some $t$.

Given a vertex $v$, the set of out-neighbors of $v$ is denoted by
$N^+(v)$. Let now $\mathcal{P}_{\alpha_1,\alpha_2}(s_1,s_2,G)$ be
the set of all pairs of compatible paths for $s_1$, $s_2$, $\alpha_1$
and $\alpha_2$ in $G$. We have\footnote{The same relation is true
  using $s_1$ instead of $s_2$.} that:
\begin{equation} \label{eq:bubbles:partition}
\mathcal{P}_{\alpha_1, \alpha_2}(s_1,s_2,G) = \mathcal{P}_{\alpha_1, \alpha_2}(s_1,s_2,G') 
                     \bigcup_{v \in N^+(s_2)} (s_2,v) \mathcal{P}_{\alpha_1, \alpha_2'} (s_1,v,G - s_2), 
\end{equation}
where $\alpha_2' = \alpha_2 - w(s_2,v)$ and $G' = G - \{(s_2,v) | v
\in N^+(s_2) \}$. In other words, the set of pairs of compatible
paths for $s_1$ and $s_2$ can be partitioned into:
$\mathcal{P}_{\alpha_1, \alpha_2'} (s_1,v,G - s_2)$, the sets of pairs
of paths containing the arc $(s_2,v)$, for each $v \in N^+(s_2)$;
and $\mathcal{P}_{\alpha_1, \alpha_2}(s_1,s_2,G')$, the set of pairs
of paths that do not contain any of them. Algorithm~\ref{alg:weighted:listbubbles}
implements this recursive partition strategy. The solutions are only
output in the leaves of the recursion tree (line~\ref{alg2:output}),
where the partition is always a singleton.  Moreover, in order to
guarantee that every leaf in the recursion tree outputs at least one
solution, we have to test if $\mathcal{P}_{\alpha_1, \alpha_2'}
(s_1,v,G - s_2)$ (and $\mathcal{P}_{\alpha_1, \alpha_2}(s_1,s_2,G')$)
is not empty before making the recursive call
(lines~\ref{alg2:include} and \ref{alg2:exclude}).

\begin{algorithm} 
\caption{$\listbubbles(s_1,\alpha_1,s_2,\alpha_2, B, G)$} \label{alg:weighted:listbubbles}
\If{$s_1 = s_2$}{ 
  \uIf{$B \neq \emptyset$}{
    output(B) \\ \label{alg2:output}
    \bf return 
  } 
  \ElseIf{there is no $(s, t, \alpha_1,\alpha_2)$-bubble, where $s = s_1 = s_2$}{ \label{alg2:initial_test} 
    \bf return
  } 
} 
choose $u \in \{s_1,s_2\}$, such that $N^+(u) \neq \emptyset$ \\
\For{$v \in N^+(u)$}{
  \If{there is a pair of compatible paths using $(u,v)$ in $G$}{ \label{alg2:include}
    \uIf{$u = s_1$}{ 
      $\listbubbles(v, \alpha_1 - w(s_1,v), s_2, \alpha_2, B \cup (s_1,v), G - s_1)$ \label{alg2:rec_include1}
    } 
    \Else{ 
      $\listbubbles(s_1, \alpha_1, v, \alpha_2 - w(s_2,v), B \cup (s_2,v), G - s_2)$ \label{alg2:rec_include2}
    } 
  } 
}
\If{there is a pair of compatible paths in $G - \{(u,v) | v \in N^+(u) \}$}{ \label{alg2:exclude}
  $\listbubbles(s_1, \alpha_1, s_2, \alpha_2, B, G - \{(u,v) | v \in N^+(u) \})$ \label{alg2:rec_exclude}
}
\end{algorithm}

The correctness of Algorithm~\ref{alg:weighted:listbubbles} follows
directly from the relation given in Eq.~\ref{eq:bubbles:partition} and
the correctness of the tests performed in lines \ref{alg2:include}
and \ref{alg2:exclude}. In the remaining of this section, we describe
a possible implementation for the tests, prove correctness and analyze
the time complexity. Finally, we prove that
Algorithm~\ref{alg:weighted:listbubbles} has an $O(n(m + n \log n))$
delay.

\begin{lemma} \label{lem:bubbles:dist}
  There exists a pair of compatible paths for $s_1 \neq s_2$ in $G$ if
  and only if there exists $t$ such that $d(s_1,t) \leq \alpha_1$ and
  $d(s_2, t) \leq \alpha_2$.
\end{lemma}
\begin{proof}
  Clearly this is a necessary condition. Let us prove that it is also
  sufficient. Consider the paths $p_1 = s_1 \leadsto t$ and $p_2 = s_2
  \leadsto t$, such that $w(p_1) \leq \alpha_1$ and $w(p_2) \leq
  \alpha_2$. Let $t'$ be the first vertex in common between $p_1$ and
  $p_2$. The sub-paths $p_1' = s_1 \leadsto t'$ and $p_2' = s_2
  \leadsto t'$ are internally vertex-disjoint, and since the weights
  are non-negative, they also satisfy $w(p_1') \leq w(p_1) \leq
  \alpha_1$ and $w(p_2') \leq w(p_2) \leq \alpha_2$.
\end{proof}

Using this lemma, we can test for the existence of a pair of
compatible paths for $s_1 \neq s_2$ in $O(m + n \log n)$ time. Indeed,
let $T_1$ be a shortest path tree of $G$ rooted in $s_1$ and truncated
at distance $\alpha_1$, the same for $T_2$, meaning that, for any
vertex $w$ in $T_1$ (resp. $T_2$), the tree path between $s_1$ and $w$
(resp. $s_2$ and $w$) is a shortest one.  It is not difficult to prove
that the intersection $T_1 \cap T_2$ is not empty if and only if there
is a pair of compatible paths for $s_1$ and $s_2$ in $G$. Moreover,
each shortest path tree can be computed in $O(m + n\log n)$ time,
using Dijkstra's algorithm (\cite{Cormen01}). Thus, in order to test
for the existence of a $(s, t, \alpha_1,\alpha_2)$-bubble for some $t$
in $G$, we can test, for each arc $(s,v)$ outgoing from $s$, the
existence of a pair of compatible paths for $s \neq v$ and $v$ in
$G$. Since $s$ has at most $n$ out-neighbors, we obtain
Lemma~\ref{lem:initial_test}.
  
\begin{lemma} \label{lem:initial_test}
  The test of line \ref{alg2:initial_test} can be performed in $O(n (m
  + n \log n))$.
\end{lemma}

The test of line~\ref{alg2:include} could be implemented using the
same idea. For each $v \in N^+(u)$, we test for the existence of
a pair of compatible paths for, say, $u = s_2$ (the same would apply
for $s_1$) and $v$ in $G - u$, that is $v$ is in the subgraph of $G$
obtained by eliminating from $G$ the vertex $u$ and all the arcs
incoming to or outgoing from $u$.  This would lead to a total cost of
$O(n(m+ n \log n))$ for all tests of line~\ref{alg2:include} in each
call. However, this is not enough to achieve an $O(n(m + n \log n))$
delay. In Lemma~\ref{lem:test2}, we present an improved strategy to
perform these tests in $O(m+ n \log n)$ total time.

\begin{lemma} \label{lem:test2}
  The test of line~\ref{alg2:include}, for all $v \in N^+(u)$,
  can be performed in $O(m + n \log n)$ total time.
\end{lemma}
\begin{proof}
  Let us assume that $u = s_2$, the case $u = s_1$ is
  symmetric. From Lemma~\ref{lem:bubbles:dist}, for each $v \in N^+(u)$,
  we have that deciding if there exists a pair of compatible paths for
  $s_1$ and $s_2$ in $G$ that uses $(u,v)$ is equivalent to deciding
  if there exists $t$ satisfying (i) $d(s_1,t) \leq \alpha_1$ and (ii)
  $d(v,t) \leq \alpha_2 - w(u,v)$ in $G - u$.

  First, we compute a shortest path tree rooted in $s_1$ for
  $G-u$. Let $V_{\alpha_1}$ be the set of vertices at a distance at
  most $\alpha_1$ from $s_1$. We build a graph $G'$ by adding a new
  vertex $r$ to $G-u$, and for each $y \in V_{\alpha_1}$, we add the
  arcs $(y,r)$ with weight $w(y,r) = 0$.  We claim that there exists
  $t$ in $G-u$ satisfying conditions (i) and (ii) if and only if
  $d(v,r) \leq \alpha_2 - w(u,v)$ in $G'$. Indeed, if $t$ satisfies
  (i) we have that the arc $(t,r)$ is in $G'$, so $d(t,r) = 0$. From
  the triangle inequality and (ii), $d(v,r) \leq d(v,t) + d(t,r) =
  d(v,t) \leq \alpha_2 - w(u,v)$. The other direction is trivial.

  Finally, we compute a shortest path tree $T_r$ rooted in $r$ for the
  reverse graph $G'^R$, obtained by reversing the direction of the
  arcs of $G'$. With $T_r$, we have the distance from any vertex to
  $r$ in $G'$, i.e. we can answer the query $d(v,r) \leq \alpha_2 -
  w(u,v)$ in constant time. Observe that the construction of $T_r$
  depends only on $G-u$, $s_1$ and $\alpha_1$, i.e. $T_r$ is the same
  for all out-neighbors $v \in N^+(u)$. Therefore, we can build
  $T_r$ only once in $O(m + n \log n)$ time, with two iterations of
  Dijkstra's algorithm, and use it to answer each test of
  line~\ref{alg2:include} in constant time.
\end{proof}

\begin{theorem}
  Algorithm~\ref{alg:weighted:listbubbles} has $O(n(m + n \log n))$ delay.
\end{theorem}
\begin{proof}
  The height of the recursion tree is bounded by $2n$ since at each
  call the size of the graph is reduced either by one vertex
  (lines~\ref{alg2:rec_include1} and \ref{alg2:rec_include2}) or all
  its out-neighborhood (line~\ref{alg2:rec_exclude}). After at most
  $2n$ recursive calls, the graph is empty. Since every leaf of the
  recursion tree outputs a solution and the distance between two
  leaves is bounded by $4n$, the delay is $O(n)$ multiplied by the
  cost per node (call) in the recursion tree. From
  Lemma~\ref{lem:bubbles:dist}, line~\ref{alg2:exclude} takes $O(m +
  n \log n)$ time, and from Lemma~\ref{lem:test2},
  line~\ref{alg2:include} takes $O(m + n \log n)$ total time. This
  leads to an $O(m + n \log n)$ time per call, excluding
  line~\ref{alg2:initial_test}. Lemma~\ref{lem:initial_test} states
  that the cost for the test in line~\ref{alg2:initial_test} is $O(n(m
  + n \log n))$, but this line is executed only once, at the root of
  the recursion tree. Therefore, the delay is $O(n (m + n \log n))$.
\end{proof}

\subsection{Implementation and experimental results}
We now discuss the details necessary for an efficient implementation
of Algorithm~\ref{alg:weighted:listbubbles} and the results on two
sets of experimental tests. For the first set, our goal is to compare
the running time of Dijkstra's algorithm (for typical cDBGs arising
from applications) using several priority queue implementations. With
the second set, our objective is to compare an implementation of
Algorithm~\ref{alg:weighted:listbubbles} to the \ks listing algorithm
given in Section~\ref{sec:kissplice:algorithm}.  For both cases, we
retrieved from the \emph{Short Read Archive} (accession code
ERX141791) 14M Illumina 79bp single-ended reads of a \emph{Drosophila
  melanogaster} RNA-seq experiment. We then built the de Bruijn graph
for this dataset with $k = 31$. In order to remove likely sequencing
errors, we discarded all $k$-mers that are present less than 3 times
in the dataset. The resulting graph contained 22M $k$-mers, which
after compressing all maximal linear paths, corresponded to 600k
vertices.

In order to perform a fair comparison with \ks, we pre-processed the
graph as described in Section~\ref{sec:kissplice:algorithm}. Namely,
we decomposed the underlying undirected graph into biconnected
components (BCCs) and compressed all non-branching bubbles with equal
path lengths. In the end, after discarding all BCCs with less than 4
vertices (as they cannot contain a bubble), we obtained 7113 BCCs, the
largest one containing 24977 vertices.  This pre-processing is
lossless, i.e. every bubble in the original graph is entirely
contained in exactly one BCC. In \ks, the enumeration is then done in
each BCC independently.

\subsubsection{Dijkstra's algorithm with different priority queues} \label{subsec:dijkstra}
Dijkstra's algorithm is an important subroutine of
Algorithm~\ref{alg:weighted:listbubbles} that may have a big influence
on its running time. Actually, the time complexity of
Algorithm~\ref{alg:weighted:listbubbles} can be written as $O(n
t(n,m))$, where $t(n,m)$ is the complexity of Dijkstra's
algorithm. There are several variants of this algorithm
(\cite{Cormen01}), with different complexities depending on the
priority queue used, including binary heaps ($O(m \log n)$) and
Fibonacci heaps ($O(m + n \log n)$). In the particular case where all
the weights are non-negative integers bounded by $C$, Dijkstra's
algorithm can be implemented using radix heaps ($O(m + n \log C)$)
(\cite{Tarjan90}). As stated in Section~\ref{sec:debruijn_and_as}, the
weights of the de Bruijn graphs considered here are integer, but not
necessarily bounded.  However, we can remove from the graph all arcs
with weights greater than $\alpha_1$ since these are not part of any
$(s,t,\alpha_1, \alpha_2)$-bubble. This results in a complexity of
$O(m + n \log \alpha_1)$ for Dijkstra's algorithm.

\begin{figure}[Htbp]
  \center
  \includegraphics[width=0.5\linewidth]{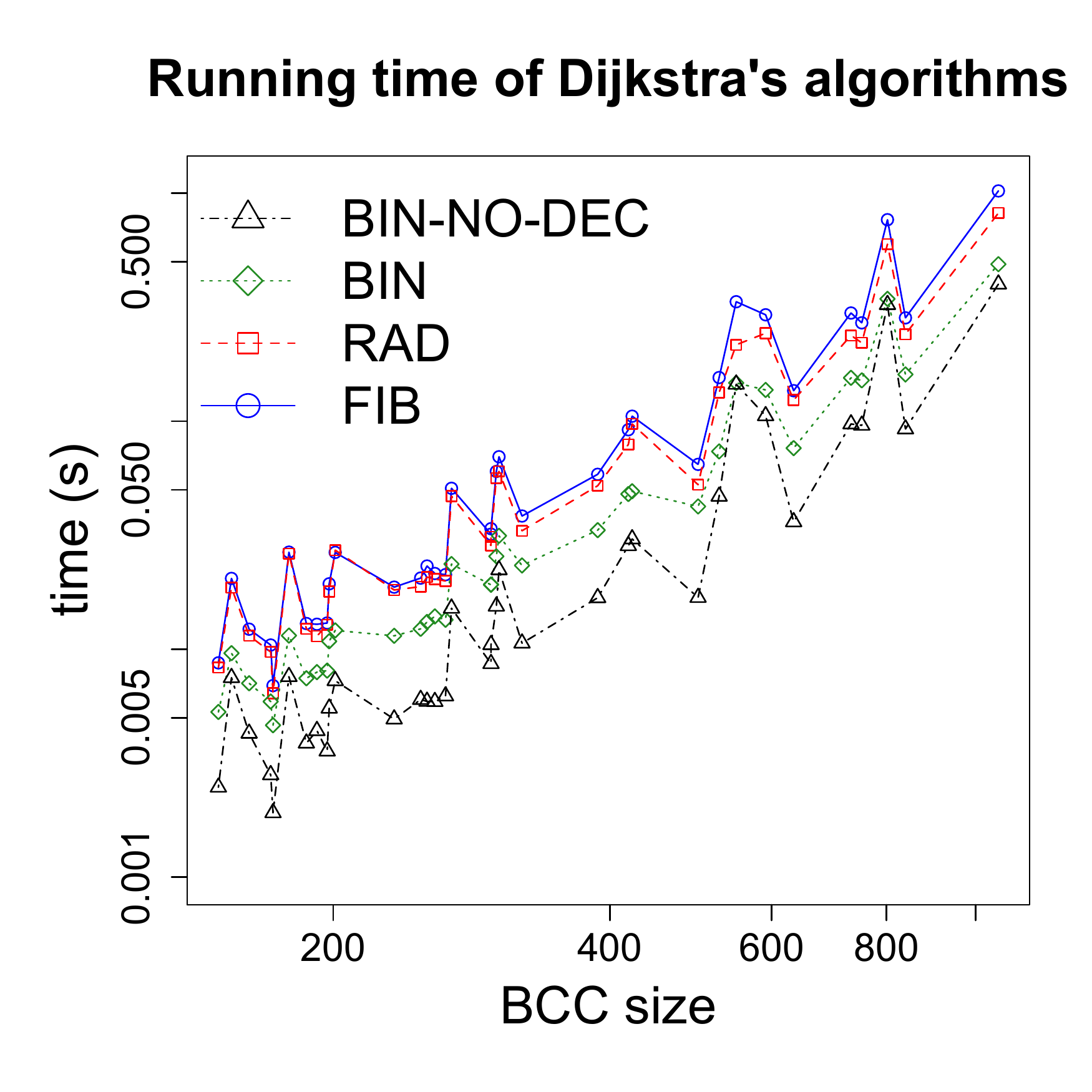}\label{fig:running_time_dijkstra}
  \caption{Running times for each version of Dijkstra's algorithm:
    using Fibonacci heaps (FIB), using radix heaps (RAD), using binary
    heaps (BIN) and using binary heaps without the decrease-key
    operation (BIN-NO-DEC).  The tests were done including all BCCs
    with more than 150 vertices. Both axes are in logarithmic
    scale.}  \label{fig:dijkstra}
\end{figure}

We implemented four versions of Lemma~\ref{lem:initial_test} (for
deciding whether there exists a $(s,t,\alpha_1, \alpha_2)$-bubble for
a given $s$) each using a different version of Dijkstra's algorithm:
with Fibonacci heaps (FIB), with radix heaps (RAD), with binary heaps
(BIN) and with binary heaps without decrease-key operation
(BIN-NO-DEC). The last version is Dijkstra's modified in order not to
use the decrease-key operation so that we can use a simpler binary
heap that does not support such operation (\cite{Chen07}).  We then
ran the four versions, using $\alpha_1 = 1000$ and $\alpha_2 = 2k - 2
= 60$, for each vertex in all the BCCs with more than 150
vertices. The results are shown\footnote{The results for the largest
  BCC were omitted from the plot to improve the visualization. It took
  942.15s for FIB and 419.84s for BIN-NO-DEC.}  in
Fig.~\ref{fig:dijkstra}. Contrary to the theoretical predictions, the
versions with the best complexities, FIB and RAD, have the worst
results on this type of instances. It is clear that the best version
is BIN-NO-DEC, which is at least 2.2 times and at most 4.3 times
faster than FIB.  One of the factors possibly contributing to a
better performance of BIN and BIN-NO-DEC is the fact that cDBGs, as
stated in Section~\ref{sec:debruijn_and_as}, have bounded degree and
are therefore sparse.

\subsubsection{Comparison with the \ks algorithm} \label{subsec:comp_kissplice}

In this section, we compare Algorithm~\ref{alg:weighted:listbubbles}
to the \ks enumeration algorithm given in
Section~\ref{sec:kissplice:algorithm}. To this purpose, we implemented
Algorithm~\ref{alg:weighted:listbubbles} using Dijkstra's algorithm
with binary heaps without the decrease-key operation for all shortest
paths computation. In this way, the delay of
Algorithm~\ref{alg:weighted:listbubbles} becomes $O(nm \log n)$, which
is worse than the one using Fibonacci or radix heaps, but is faster in
practice. The goal of the \ks enumeration is to find all the potential
alternative splicing events in a BCC, i.e. to find all
$(s,t,\alpha_1,\alpha_2)$-bubbles satisfying also the lower bound
constraint (Section~\ref{sec:debruijn_and_as}).  In order to compare
\ks (version 1.6) to Algorithm~\ref{alg:weighted:listbubbles}, we
(naively) modified the latter so that, whenever a
$(s,t,\alpha_1,\alpha_2)$-bubble is found, we check whether it also
satisfies the lower bound constraints and output it only if it does.

In \ks, the upper bound $\alpha_1$ is an open parameter, $\alpha_2 =
k-1$ and the lower bound is $k - 7$. Moreover, there are two stop
conditions: either when more than 10000
$(s,t,\alpha_1,\alpha_2)$-bubbles satisfying the lower bound
constraint have been enumerated or a 900s timeout has been reached.
We ran both \ks (version 1.6) and the modified
Algorithm~\ref{alg:weighted:listbubbles}, with the stop conditions,
for all 7113 BCCs, using $\alpha_2 = 60$, a lower bound of $54$ and
$\alpha_1 = 250,500,750$ and $1000$. The running times for all BCCs
with more than 150 vertices (there are 37) is shown\footnote{The BCCs
  where \emph{both} algorithms reach the timeout were omitted from the
  plots to improve the visualization. For $\alpha_1 = 250, 500, 750$
  and $1000$ there are 1, 2, 3 and 3 BCCs omitted, respectively.}  in
Fig.~\ref{fig:running_time}. For the BCCs smaller than 150 vertices,
both algorithms have comparable (very small) running times. For
instance, with $\alpha_1 = 250$, \ks runs in 17.44s for \emph{all}
7113 BCCs with less than 150 vertices, while
Algorithm~\ref{alg:weighted:listbubbles} runs in 15.26s.

The plots in Fig.~\ref{fig:running_time} show a trend of increasing
running times for larger BCCs, but the graphs are not very smooth,
i.e. there are some sudden decreases and increases in the running
times observed. This is in part due to the fact that the time complexity of 
Algorithm~\ref{alg:weighted:listbubbles} is output sensitive. The delay of the
algorithm is $O(nm \log n)$, but the total time complexity is
$O(|\mathcal{B}|nm \log n)$, where $|\mathcal{B}|$ is the number of
$(s,t,\alpha_1,\alpha_2)$-bubbles in the graph. The number of bubbles
in the graph depends on its internal structure. A large graph does not
necessarily have a large number of bubbles, while a small graph may have
an exponential number of bubbles. Therefore, the value of
$|\mathcal{B}|nm \log n$ can decrease by increasing the size of the
graph.

Concerning now the comparison between the algorithms, as we can see in
Fig.~\ref{fig:running_time}, Algorithm~\ref{alg:weighted:listbubbles}
is usually several times faster (keep in mind that the axes are in
logarithmic scale) than \ks, with larger differences when $\alpha_1$
increases (10 to 1000 times faster when $\alpha_1 = 1000$).  In some
instances however, \ks is faster than
Algorithm~\ref{alg:weighted:listbubbles}, but (with only one exception
for $\alpha_1 = 250$ and $\alpha_1 = 500$) they correspond either to
very small instances or to cases where only 10000 bubbles were
enumerated and the stop condition was met.  Finally, using
Algorithm~\ref{alg:weighted:listbubbles}, the computation finished
within 900s for all but 3 BCCs, whereas using \ks, 11 BCCs remained
unfinished after 900s.  The improvement in time therefore enables us
to have access to bubbles that could not be enumerated with the
previous approach.

\begin{figure}[Htbp]
  \center
  \subfloat[]{\includegraphics[width=0.5\linewidth]{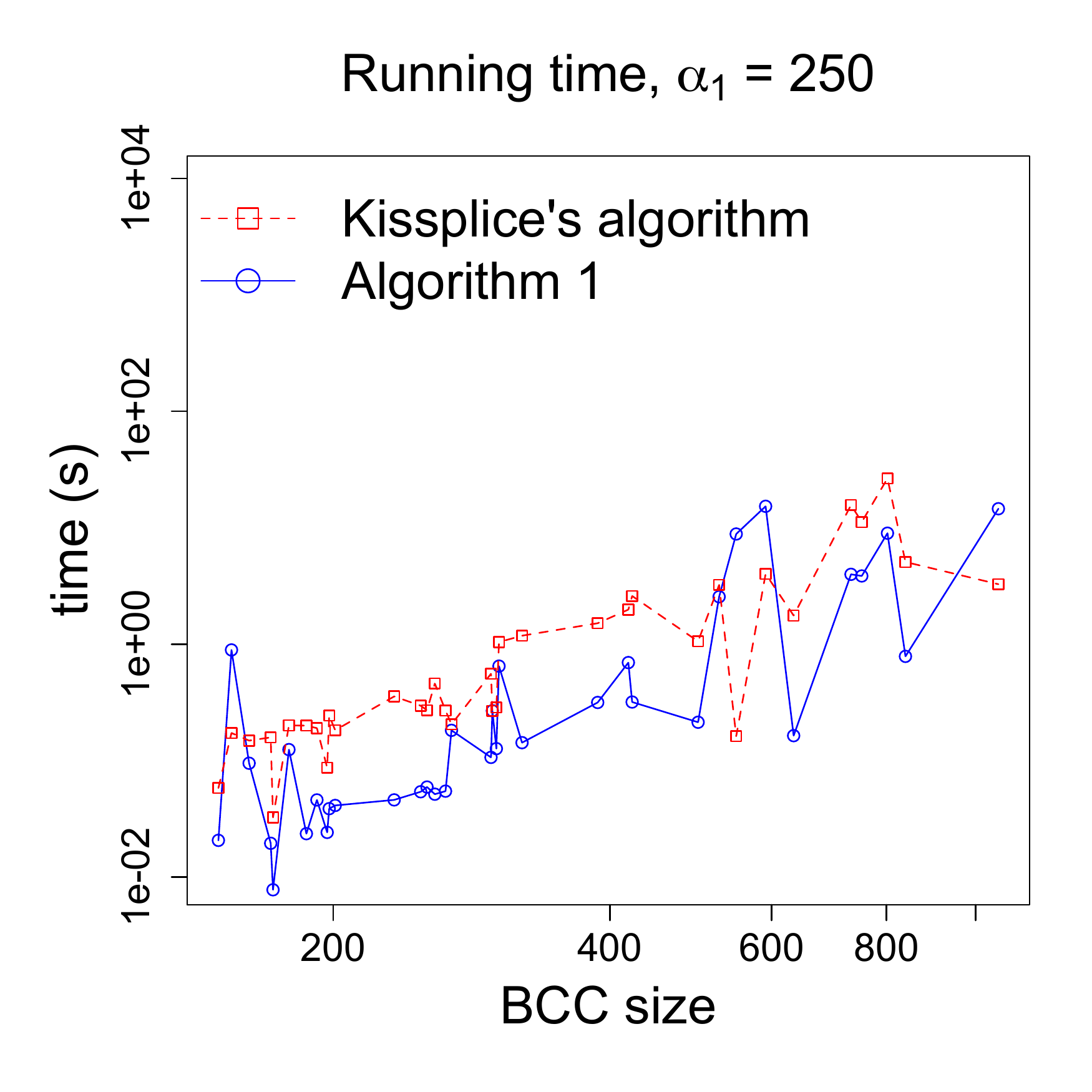}\label{fig:running_time_250}} 
  \subfloat[]{\includegraphics[width=0.5\linewidth]{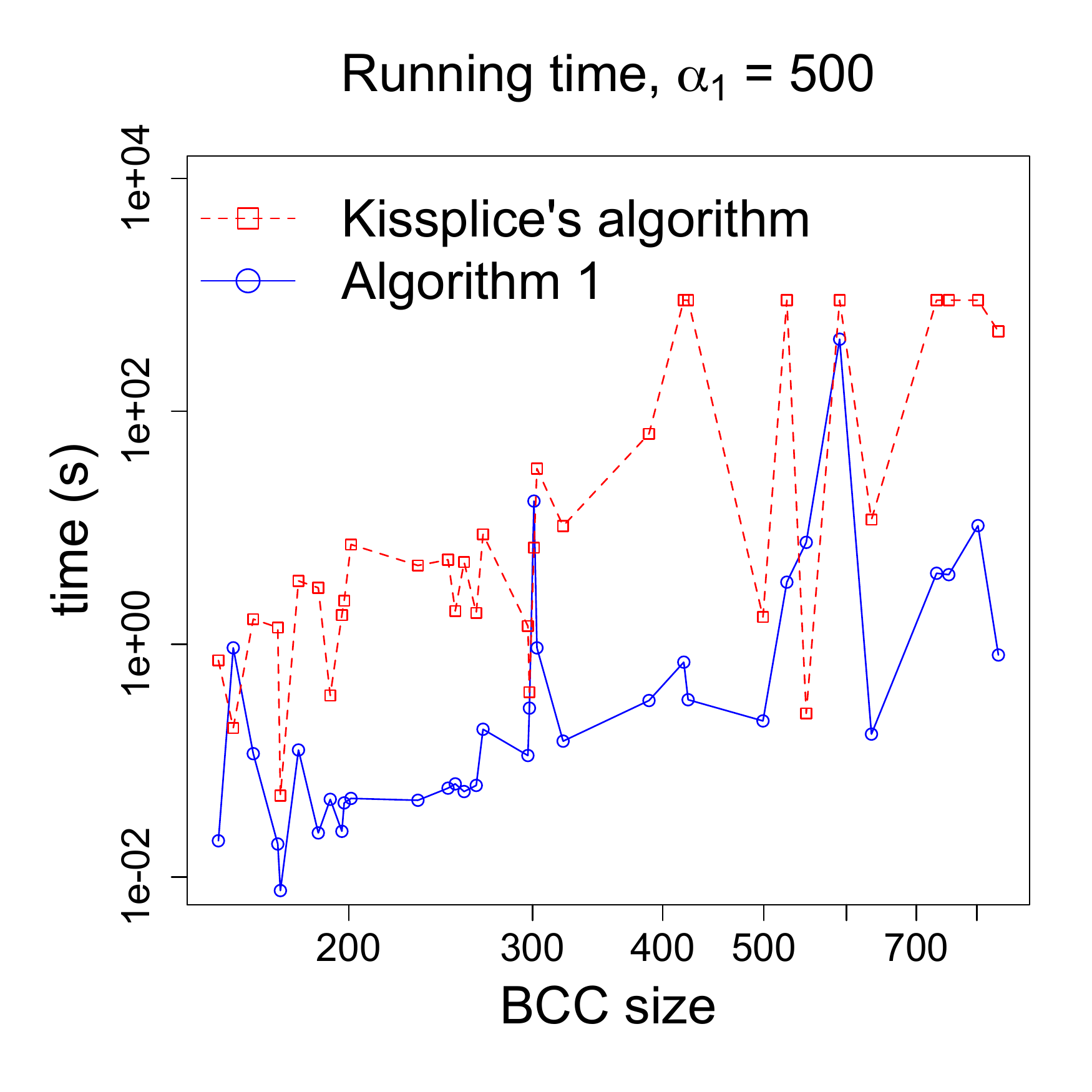}\label{fig:running_time_500}} \\
  \subfloat[]{\includegraphics[width=0.5\linewidth]{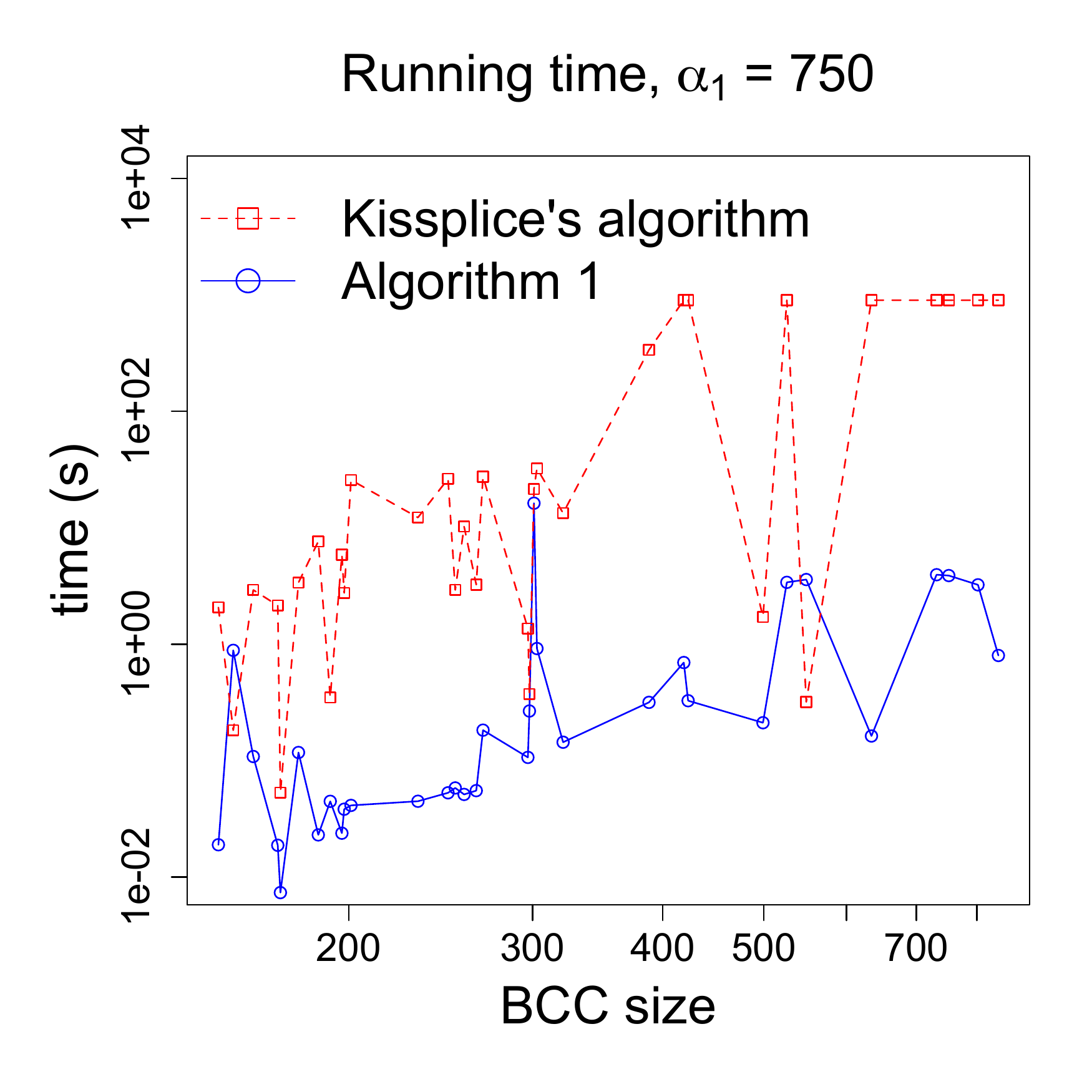} \label{fig:running_time_750}} 
  \subfloat[]{\includegraphics[width=0.5\linewidth]{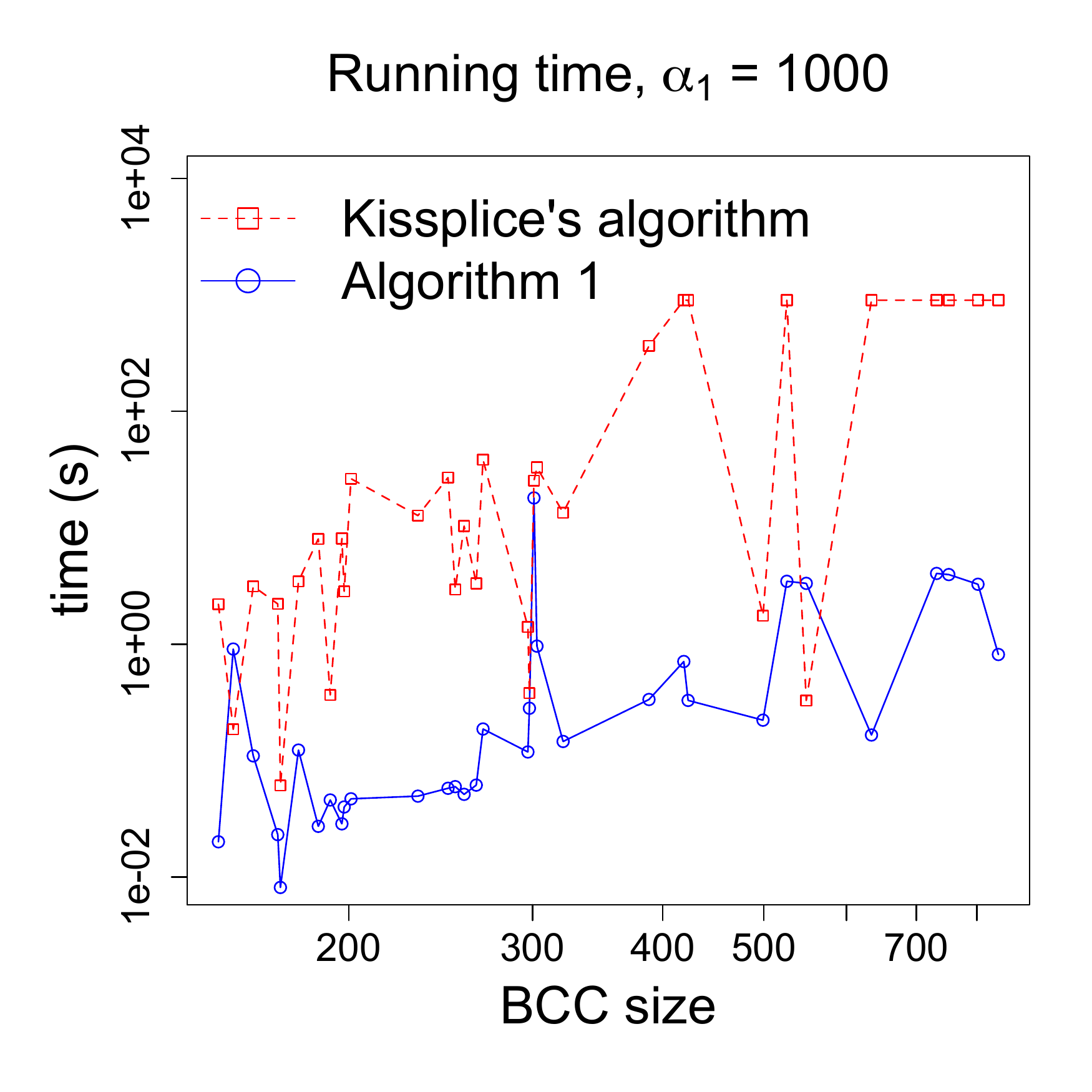}\label{fig:running_time_1000}}
  \caption{Running times of Algorithm~\ref{alg:weighted:listbubbles}
    and of the \ks bubble listing algorithm for all the BCCs with more
    than 150 vertices. Each graph (a), (b), (c) and (d) shows the
    running time of both algorithms for $\alpha_1 = 250, 500, 750$ and
    $1000$, respectively.} \label{fig:running_time}
\end{figure} 

\subsubsection{On the usefulness of larger values of $\alpha_1$}
In \ks (version 1.6), the value of $\alpha_1$ was experimentally set
to 1000 due to performance issues, as indeed the algorithm quickly
becomes impractical for larger values. On the other hand, the results
of Section~\ref{subsec:comp_kissplice} suggest that
Algorithm~\ref{alg:weighted:listbubbles}, that is faster than \ks, can
deal with larger values of $\alpha_1$.  From a biological point of
view, it is a priori possible to argue that $\alpha_1 = 1000$ is a
reasonable choice, because 87\% of annotated exons in Drosophila
indeed are shorter than 1000 bp (\cite{Refseq}). However, missing the
top 13\% may have a big impact on downstream analyses of AS, not to
mention the possibility that not yet annotated AS events could be
enriched in long skipped exons.  In this section, we outline that
larger values of $\alpha_1$ indeed produces more results that are
biologically relevant. For this, we exploit another RNA-seq dataset,
with deeper coverage.

\begin{figure}[Htbp]
  \center
  \includegraphics[width=\linewidth]{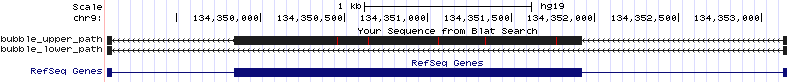}
  \caption{One of the bubbles with longest path larger than 1000 bp
    found by Algorithm~\ref{alg:weighted:listbubbles} with the corresponding
    sequences mapped to the reference genome and visualized using the
    UCSC Genome Browser. The first two lines correspond to the
    sequences of, respectively, the shortest (exon exclusion variant)
    and longest paths of the bubble mapped to the genome.  The blue
    lines are the UCSC human transcript annotations. }
  \label{fig:weighted:exon_skipping}
\end{figure}

To this purpose, we retrieved 32M RNA-seq reads from the human brain
and 39M from the human liver from the Short Read Archive (accession
number ERP000546). Next, we built the de Bruijn graph with $k=31$ for
both datasets, then merged and decomposed the DBG into 5692 BCCs
(containing more than 10 vertices). We ran
Algorithm~\ref{alg:weighted:listbubbles} for each BCC with $\alpha_1 =
5000$. It took 4min25s for Algorithm~\ref{alg:weighted:listbubbles} to
run on all BCCs, whereas \ks, even using $\alpha_1 = 1000$, took
31min45s, almost 8 times more. There were 59 BCCs containing at least
one bubble with the length of the longest path strictly larger than
1000bp potentially corresponding to alternative splicing events. In
Fig.~\ref{fig:weighted:exon_skipping}, we show one of those bubbles
mapped to the reference genome. It corresponds to an exon skipping in
the PRRC2B human gene, the skipped exon containing 2069 bp. While the
transcript containing the exon is annotated, the variant with the exon
skipped is not annotated.

Furthermore, we ran \tri on the same dataset and found that it was
unable to report this novel variant.  Our method therefore enables us
to find new AS events, reported by no other method. This is, of
course, just an indication of the usefulness of our approach when
compared to a full-transcriptome assembler.

\subsection{A natural generalization} \label{sec:gen}
\subsubsection{An intractable case: Paths with length constraints} \label{subsec:intractable}

For the sake of theoretical completeness, in this section, we extend
the definition of $(s,t,\alpha_1,\alpha_2)$-bubble to the case where
the length constraints concern $d$ vertex-disjoint paths, for an
arbitrary but fixed $d$. This situation also arises in real data, when
more than 2 variants share the same flanking splice sites (for
instance for single and double exon skipping), or when a SNP has 3
variants.

\begin{definition}[$(s,t, A)$-$d$-bubble] \label{def:kbubble}
  Let $d$ be a natural number and $A = \{\alpha_1, \ldots, \alpha_d\}
  \subset \mathbb{Q}_{\geq0}$.  Given a directed weighted graph $G$
  and two vertices $s$ and $t$, an \emph{$(s,t,A)$-$d$-bubble} is a
  set of $d$ pairwise internally vertex-disjoint paths $\{p_1, \ldots
  p_d\}$, satisfying $p_i = s \leadsto t$ and $w(p_i) \leq \alpha_i$,
  for all $i \in [1,d]$.
\end{definition} 

Analogously to $(s,t,\alpha_1,\alpha_2)$-bubbles, we can define two
variants of the enumeration problem: all bubbles with a given source
($s$ fixed) and all bubbles with a given source and target ($s$ and
$t$ fixed). In both cases, the first step is to decide the existence
of at least one $(s,t, A)$-$d$-bubble in the graph.

\begin{problem}[$(s,t,A)$-$d$-bubble decision problem] \label{prob:bounded_bubble_st}
  Given a non-negatively weighted directed graph $G$, two vertices
  $s,t$, a set $A = \{\alpha_1, \ldots, \alpha_d\} \subset
  \mathbb{Q}_{\geq0}$ and $d \in \mathbb{N}$, decide if there exists a
  $(s,t, A)$-$d$-bubble.
\end{problem} 

This problem is a generalization of the two-disjoint-paths problem
with a min-max objective function, which is NP-complete
(\cite{Chung90}). More formally, this problem can be stated as
follows: given a directed graph $G$ with non-negative weights, two
vertices $s,t \in V$, and a maximum length $M$, decide if there exists
a pair of vertex-disjoint paths such that the maximum of their lengths
is less than $M$. The $(s,t, A)$-$d$-bubble decision problem, with $A
= \{M,M\}$ and $d=2$, is precisely this problem.

\begin{problem}[$(s,*,A)$-$d$-bubble decision problem] \label{prob:bounded_bubble_s*}
  Given a non-negatively weighted directed graph $G$, a vertex $s$, a
  set $A = \{\alpha_1, \ldots, \alpha_d\} \subset \mathbb{Q}_{\geq0}$
  and $d \in \mathbb{N}$, decide if there exists a $(s,t,
  A)$-$d$-bubble, for some $t \in V$.
\end{problem} 

The two-disjoint-path problem with a min-max objective function is
NP-complete even for strictly positive weighted graphs. Let us reduce
Problem~\ref{prob:bounded_bubble_s*} to it.  Consider a graph $G$ with
strictly positive weights, two vertices $s,t \in V$, and a maximum
length $M$. Construct the graph $G'$ by adding an arc with weights $0$
from $s$ to $t$ and use this as input for the
$(s,*,\{M,M,0\})$-$3$-bubble decision problem. Since $G$ has strictly
positive weights, the only path with length $0$ from $s$ to $t$ in
$G'$ is the added arc. Thus, there is a $(s,*,\{M,M,0\})$-$3$-bubble
in $G'$ if and only if there are two vertex-disjoint paths in $G$ each
with a length $\leq M$.

Therefore, the decision problem for fixed $s$
(Problem~\ref{prob:bounded_bubble_st}) is NP-hard for $d \geq 2$, and
for fixed $s$ and $t$ (Problem~\ref{prob:bounded_bubble_s*}) is
NP-hard for $d \geq 3$.  In other words, the only tractable case is
the enumeration of $(s,t, A)$-$2$-bubbles with fixed $s$, the one
considered in Section~\ref{sec:alg_delay}.

\subsubsection{A tractable case: Paths without length constraints}

In the previous section, we showed that a natural generalization of
$(s,t,\alpha_1,\alpha_2)$-bubbles to contain more than two
vertex-disjoint paths satisfying length constraints leads to an
NP-hard enumeration problem.  Indeed, even deciding the existence of
at least one $(s,t, {\cal A})$-$d$-bubble is NP-hard. In this section,
we consider a similar generalization for $(s,t)$-bubbles instead of
$(s,t,\alpha_1,\alpha_2)$-bubbles, that is, we consider bubbles
containing more than two vertex-disjoint paths without any path length
constraints. The formal definition is given below.

\begin{definition}[$(s,t)$-$d$-bubble] \label{def:kbubble_unweighted}
  Let $d$ be a natural number. Given a directed graph $G$ and two
  vertices $s$ and $t$, a \emph{$(s,t)$-$d$-bubble} is a set of $d$
  pairwise internally vertex-disjoint paths $\{p_1, \ldots p_d\}$.
\end{definition} 

Clearly, this definition is a special case of
Definition~\ref{def:kbubble}: consider a weighted graph $G = (V,E)$
with unitary weights (i.e. an unweighted graph), the $(s,t, {\cal
  A})$-$d$-bubbles with $\alpha_i = |V|$ for $i \in [1,d]$ are
precisely the $(s,t)$-$d$-bubbles of $G$. As in
Section~\ref{subsec:intractable}, let us first consider the problem of
deciding whether a graph contains a $(s,t)$-$d$-bubble for fixed $s$
and $t$.

\begin{problem}[$(s,t)$-$d$-bubble decision problem] \label{prob:stbubble_un}
  Given a directed graph $G$ and two vertices $s,t$, decide whether
  there exists a $(s,t)$-$d$-bubble in $G$.
\end{problem} 

Contrary to Problem~\ref{prob:bounded_bubble_st}, this problem can be
decided in polynomial time. Indeed, given a directed graph $G = (V,
A)$ and two vertices $s$ and $t$, construct the graph $G' = (V', A')$
by splitting each vertex $v \in V$ in two vertices: an incoming part
$v_{in}$ with all the arcs entering $v$, and an outgoing part
$v_{out}$ with all the arcs leaving $v$; and add the arc
$(v_{in},v_{out})$. More formally, $G'$ is defined as $V' = \{
\{v_{in},v_{out}\} | v \in V \}$ and $A' = \{ (u_{out},v_{in}) | (u,v)
\in A \} \cup \{ (v_{in},v_{out}) | v \in V \}$. Now, it is not hard
to prove that every set of \emph{arc-disjoint} paths in $G'$
corresponds to a set of vertex-disjoint paths in $G$. Thus,
considering $G'$ a network with unitary arc capacities
(\cite{Cormen01}), we have that $G$ contains a $(s,t)$-$d$-bubble if
and only if $G'$ contains a $(s,t)$-flow $f$ such that $|f| \geq
d$. Therefore, using the augmenting path algorithm (\cite{Cormen01})
for the max-flow problem, we can decide if there exists a
$(s,t)$-$d$-bubble in $G$ in $O(md)$ time. Actually, using an
iterative decomposition of the $(s,t)$-flow $f$ into $(s,t)$-paths, we
can explicitly find a $(s,t)$-$d$-bubble in the time bound.

\begin{lemma}
  Given a directed graph $G = (V,A)$ and two vertices $s,t \in V$, a
  $(s,t)$-$d$-bubble in $G$ can be found in $O(md)$ time.
\end{lemma}

We now consider the problem of enumerating $(s,t)$-$d$-bubbles in $G$
for fixed $s$ and $t$. The reduction from $(s,t)$-$d$-bubbles to
$(s,t)$-flows used in the last paragraph may induce us to think that
we can enumerate $(s,t)$-$d$-bubbles in $G$ by enumerating
$(s,t)$-flows in $G'$, and since there is a polynomial delay algorithm
for the latter (\cite{Bussieck98}), we would be done. Unfortunately,
there is no one-to-one correspondence between $(s,t)$-flows in $G'$
and $(s,t)$-$d$-bubbles in $G$: we can always add a circulation $c$ to
a $(s,t)$-flow $f$ to obtain a new $(s,t)$-flow $f'$, but $f$ and $f'$
correspond to the same $(s,t)$-$d$-bubble. In fact, there can be
exponentially more $(s,t)$-flows in $G'$ than $(s,t)$-$d$-bubbles in
$G$. On the other hand, the strategy used in
Algorithm~\ref{alg:weighted:listbubbles} can be adapted to enumerate
$(s,t)$-$d$-bubbles.

Similarly to Section~\ref{sec:alg_delay}, in order to have a more
symmetric structure for the subproblems, we define the notion of a
\emph{set of compatible paths}, which is an object that generalizes
the definition of a $(s,t)$-$d$-bubble. Given a set of sources $S =
\{s_1, \ldots, s_d\}$ and a target $t$, a set of paths $P_t = \{p_1,
\ldots, p_d\}$ is compatible if $p_i = s_i \leadsto t$ and they are
internally vertex-disjoint. We then focus on the more general problem
of enumerating sets of compatible paths. Let $\mathcal{P}(S,t,G)$ be
the set of all compatible paths for $S$ and $t$ in $G$. The same
partition given in Eq.~\ref{eq:bubbles:partition} is also valid for
$\mathcal{P}(S,t,G)$. Namely, for any $s \in S$ such that $\delta^+(s)
\neq \emptyset$,
\begin{equation} \label{eq:bubbles:partition2}
\mathcal{P}(S,t,G) = \mathcal{P}(S,t,G') 
                     \bigcup_{v \in \delta^+(s)} (s,v) \mathcal{P}(S\setminus \{s\} \cup \{v\},t,G - s), 
\end{equation}
where $G' = G - \{(s,v) | v \in \delta^+(s)\}$.  Now, adding a new
source to $G$ with an arc to each vertex in $S$, we can use an
augmenting path algorithm to test whether $\mathcal{P}(S,t,G) \neq
\emptyset$ in $O(md)$ time. That way, an algorithm implementing the
partition scheme of Eq.~\ref{eq:bubbles:partition2} can enumerate
$(s,t)$-$d$-bubbles in $O(n^2md)$ delay, where the bound on the delay
holds since each node of the recursion tree costs $O(nmd)$ (at most
$n$ emptiness checks are performed) and the height of the tree is
bounded by $n$.

\begin{theorem}
  Given a directed graph $G$ and two vertices $s,t$, the
  $(s,t)$-$d$-bubbles in $G$ can be enumerated in $O(n^2md)$ delay.
\end{theorem}


\bigskip
\bigskip


\section{Listing bounded length paths} \label{sec:weighted:path}

\subsection{Introduction}
A natural generalization of the problem of listing $st$-paths in
undirected graphs (Section~\ref{sec:unweighted:cycle}) is obtained by
imposing a length constraint for the paths, that is, listing only the
$st$-paths such that the length is bounded by some constant. The
problem of listing $st$-paths in a \emph{weighted} directed graph with
lengths bounded by a constant is a further generalization of that
problem. In this section, we consider this more general problem along
with restrictions to undirected and unweighted graphs.

The shortest path problem is probably one of the most studied ones in
computer science with a huge number of applications; it would be
infeasible to list any reasonable subset of them here. A natural
generalization of it, falling into the enumeration context, is the
$K$-shortest paths problem, that consists in returning the first $K$
distinct shortest $st$-paths, where both the graph and the parameter
$K$ are part of the input. There are several applications for this
problem ranging from finding suboptimal solutions in sequence
alignment problems (\cite{Waterman83,Byers84}), to heuristics to solve
NP-hard multi-criteria path optimization problems
(\cite{Bahgat88,El-Amin93}). See \cite{Eppstein99} for further
references.

The $K$-shortest paths problem has been studied since the early 1960s
(see the references in \cite{Dreyfus69}). However, the first efficient
algorithm for this problem in directed graphs with non-negative
weights only appeared 10 years later, in the early 1970s, by
\cite{Yen71} and \cite{Lawler72}. With Dijkstra's algorithm
implemented with Fibonacci heaps (\cite{Cormen01}), their algorithm
runs in $O(K (mn + n^2 \log n))$ time, where $m,n$ are the number of
arcs and vertices, respectively. More recently, \cite{Eppstein99}
showed that if the paths can have cycles, i.e. they are walks, then
the problem can be solved in $O(K + m + n \log n)$ time. When the
input graph is undirected, the $K$-shortest \emph{simple} paths
problem is solvable in $O(K(m+n \log n))$ time (\cite{Katoh82}). For
directed unweighted graphs, the best known algorithm for the problem
is the $O(Km\sqrt{n})$ time\footnote{Polylog factors are omitted.}
randomized algorithm of \cite{Roditty05}.  In a different direction,
\cite{Roditty07} noticed that the $K$-shortest simple paths can be
efficiently approximated. Building upon his work, \cite{Bernstein10}
presented an $O(Km/\epsilon)$ time\footnote{Polylog factors are
  omitted.} algorithm for a $(1 + \epsilon)$-approximation. When the
paths are to be computed exactly, however, the best running time is
still the $O(K(mn + n^2 \log n))$ time of Yen and Lawler's algorithm
for directed graphs and the $O(K(m+n \log n))$ time of Katoh's
algorithm for undirected graphs. Both algorithms use $O(Kn + m)$
memory.

The problems of listing $st$-paths with a length bounded by $\alpha$
on one hand and of the $K$-shortest $st$-paths on the other are
closely related, even though the first problem cannot be solved in
polynomial time under the standard definition, i.e. there can be an
exponential number of bounded length $st$-paths. Intuitively, they are
both the same problem with different parameterizations; in the first
case the enumeration is constrained by the maximum length of the path
and in the second by the maximum number of paths. Although very
similar, the problem of listing bounded length $st$-paths has not, to
the best of our knowledge, been explicitly considered before, except
for \cite{Eppstein99} who mentions that his algorithm can be modified
to the bounded length case maintaining the same time and space
complexity. Actually, Yen and Lawler's algorithm can be modified to
solve the bounded length $st$-path problem, but in this case the
memory used by the algorithm is the same as in the original version,
i.e. proportional to the number of bounded length $st$-paths output,
which is potentially exponential in the size of the graph. We show
here that it is possible to list bounded length $st$-paths using space
that is only linear in the size of the graph.

In the remainder of the chapter, we consider the problem of listing
all $st$-paths with length bounded by $\alpha$ in a graph $G$ with $n$
vertices and $m$ edges/arcs. We give a general $O(nt(n,m))$ delay
algorithm, where $t(n,m)$ is the cost for a single source shortest
paths computation, to list them in weighted (including negative
values) directed graphs (Section~\ref{sec:simple_alg}) using $O(m+n)$
space. Next, we improve the total complexity of this algorithm to
$O((m + t(n,m)) \gamma)$, where $\gamma$ is the number of paths output,
for undirected graphs with non-negative weights
(Section~\ref{sec:weighted:improved}) while maintaining the same
memory complexity. Finally, we modify the general algorithm to output
the paths in increasing order of their lengths
(Section~\ref{sec:simple_alg}). This algorithm can be used to solve
the $K$-shortest paths problem.

\subsection{Preliminaries}
Given a weighted (directed or undirected) graph $G$ with weights $w :
E \mapsto \mathbb{Q}$, we say that a path $p$ is
\emph{$\alpha$-bounded} if the weight, or \emph{length}, of the path
satisfies $w(p) \leq \alpha$ and $\alpha \in \mathbb{Q}$, in the
particular case of unitary weights (i.e. unweighted graphs), we say
that $p$ is
\emph{$k$-bounded} if $w(p) \leq k$ and $k \in \mathbb{Z}_{\geq
  0}$. The general problem, formally defined below, with which we are
concerned in this section is listing $\alpha$-bounded $st$-paths in
$G$.

\begin{problem}[Listing $\alpha$-bounded $st$-paths] 
  Given a weighted directed graph $G = (V,E)$, two vertices $s,t \in
  V$, and an upper bound $\alpha \in Q$, output all $\alpha$-bounded
  $st$-paths.
\end{problem}

The general problem is stated in terms of \emph{directed} weighted
graph, because any solution for directed graphs also applies to the
undirected graphs, and in fact in Section~\ref{sec:simple_alg} we only
provide a solution to the directed case.  Moreover, whenever $G$
contains negative weight arcs, we assume that $G$ does not contain any
negative cycle, otherwise the shortest paths cannot be efficiently
computed (\cite{Cormen01}). Finally, we assume that all directed
graphs considered here are weakly connected and all undirected graphs
are connected, that way $m \geq n$, where $n$ is the number of
vertices and $m$ the number of arcs (edges).

\subsection{A simple polynomial delay algorithm} \label{sec:simple_alg}

In this section, we present a simple polynomial delay algorithm to
list all $st$-paths with length bounded by $\alpha$ in a weighted
directed graph $G$. This is the most general version of the problem.
Consequently, the algorithm works for any version of the problem,
weighted (including negative weights) or unweighted, directed or
undirected.  However, the complexity is different for each version of
the problem.  The algorithm, inspired by the binary partition method,
recursively partitions the solution space at every call until the
considered subspace is a singleton (contains only one solution) and in
that case outputs the corresponding solution. It is important to
stress that the order in which the solutions are output is fixed, but
arbitrary.  The pseudocode is given in Algorithm~\ref{alg:simple}.

Let us describe the partition scheme. Let
$\mathcal{P}_{\alpha}(s,t,G)$ be the set of all paths from $s$ to $t$
in $G$. Assuming $s \neq t$, we have that
\begin{equation} \label{eq:path:partition}
\mathcal{P}_{\alpha}(s,t,G) = \bigcup_{v \in N^+(s)} (s,v) \mathcal{P}_{\alpha'} (v,t,G-s), 
\end{equation}
where $\alpha' = \alpha - w(s,v)$. In other words, the set of paths
from $s$ to $t$ can be partitioned into the union of
$(s,v)\mathcal{P}_{\alpha} (v, t,G - s)$, the sets of paths containing
the edge $(s,v)$, for each $v \in N^+(s)$. Indeed, since $s \neq t$,
every path in $\mathcal{P}_{\alpha} (s, t, G)$ necessarily contains an
edge $(s,v)$, where $v \in N^+(s)$.

Algorithm~\ref{alg:simple} implements this recursive partition
strategy. The solutions are only output in the leaves of the recursion
tree (line~\ref{alg:simple:output}), where the partition is always a
singleton. Moreover, in order to guarantee that every leaf in the
recursion tree outputs one solution, we have to test if
$\mathcal{P}_{\alpha'} (v, t, G - u)$, where $\alpha' = \alpha -
w(u,v)$, is not empty before the recursive call
(line~\ref{alg:simple:test}). This set is not empty if and only if the
weight of the shortest path from $v$ to $t$ in $G-u$ is at most
$\alpha'$, i.e. $d_{G-u}(v,t) \leq \alpha' = \alpha - w(u,v)$. Hence,
to perform this test it is enough to compute all the distances from
$t$ in the graph $G^R - u$, where $G^R$ is the graph $G$ with all arcs
reversed.

\begin{algorithm} 
\caption{$\listpaths(u,t, \alpha, \pi_{su}, G)$} \label{alg:simple}
\If{$u = t$}{ 
  output($\pi_{su}$) \\ \label{alg:simple:output}
  \bf return 
}
compute the distances from $t$ in 
$G^R - u$ \\
\For{$v \in N^+(u)$}{
  \If{$d(v,t) \leq \alpha - w(u,v)$}{ \label{alg:simple:test} 
    $\listpaths(v,t, \alpha - w(u,v), \pi_{su} (u,v), G - u)$ 
  } 
}
\end{algorithm}

The correctness of Algorithm~\ref{alg:simple} follows directly from
the relation given in Eq.~\ref{eq:path:partition} and the correctness
of the tests of line \ref{alg:simple:test}. We can perform those tests
in $O(1)$ by pre-computing the distances from $t$ to all vertices
(single source shortest paths) in the reverse graph $G^R - u$, which
can be computed in $O(t(n,m))$. The height of the recursion tree is
bounded by $n$, since at every level of the recursion tree a new
vertex is added to the current solution and any solution has at most
$n$ vertices.  In that way, the path between any two leaves in the
recursion tree has at most $2n$ nodes. Thus, the time elapsed between
two solutions being output is $O(nt(n,m))$. Moreover, the space
complexity of the algorithm is $O(m)$, since for each recursive call,
we can store the difference with the previous graph.

\begin{theorem} \label{thm:simple_complexity}
  Algorithm~\ref{alg:simple} has delay $O(n t(n,m))$, where $t(n,m)$
  is the cost to compute a shortest path tree, and uses $O(m)$ space.
\end{theorem}

For unweighted (directed and undirected) graphs, the single source
shortest paths can be computed using breadth-first search (BFS)
running in $O(m)$, so Theorem~\ref{thm:simple_complexity} guarantees
an $O(km)$ delay to list all $k$-paths, since the height of the
recursion tree is bounded by $k$ instead of $n$. In the case of
non-negative weights the single source shortest paths can be computed
using Dijkstra's algorithm in $O(m + n \log n)$, resulting in an $O(nm
+ n^2 \log n)$ delay. Finally, for general weights, the single source
shortest paths can be computed using the Bellman-Ford algorithm in
$O(mn)$ time, resulting in an $O(mn^2)$ delay.

\subsection{An improved algorithm for undirected graphs} \label{sec:weighted:improved}

In this section, we improve the total time complexity of
Algorithm~\ref{alg:simple} from $O(nt(n,m) \gamma)$ to $O((m+t(n,m))
\gamma)$ in the case of \emph{non-negatively} weighted undirected
graphs, where $\gamma = |\mathcal{P}_{\alpha}(s,t,G)|$ is the number
of $\alpha$-bounded $st$-paths .  In other words, for undirected
graphs we can list all $\alpha$-bounded $st$-paths in $O((m+n \log n)
\gamma)$ and all $k$-bounded $st$-paths in $O(m \gamma)$. However, the
delay of the algorithm is still $O(nt(n,m))$ in the worst case,
although the (worst case) average delay is $O(m + t(n,m))$. From now
on, all the graphs considered are undirected unless otherwise stated.

The basis to improve the complexity of Algorithm~\ref{alg:simple} is
to explore the structure of $\mathcal{P}_{\alpha}(s,t,G)$ to reduce
the number of nodes in the recursion tree. More precisely, at every
call, we identify the longest common prefix of
$\mathcal{P}_{\alpha}(s,t,G)$, i.e. the longest (considering the
number of edges) path $\pi_{ss'}$ such that
$\mathcal{P}_{\alpha}(s,t,G) = \pi_{ss'}\mathcal{P}_{\alpha}(s',t,G)$,
and append it to the current path prefix being considered in the
recursive call. The pseudocode for this algorithm is very similar to
Algorithm~\ref{alg:simple} and, for the sake of completeness, is given
in Algorithm~\ref{alg:improved}. We postpone the description of the
$\lcp(u, t, \alpha, G)$ function to the next section, along with a
discussion about the difficulties to extend it to directed graphs or
general weights graphs.

\begin{algorithm} 
\caption{$\listpaths(u,t, \alpha, \pi_{su}, G)$} \label{alg:improved}
$\pi_{uu'}$ = $\lcp(u, t, \alpha, G)$ \\
\uIf{$u' = t$}{ 
  output($\pi_{su}\pi_{uu'}$) \\ 
  \bf return 
}
\Else{
  compute a shortest path tree $T'_t$ from $t$ in $G^R - \pi_{uu'}$ \\
  \For{$v \in N(u')$}{
    \If{$d(v,t) + w(u,v) \leq \alpha$}{ 
      $\listpaths(v,t, \alpha - w(\pi_{uu'}) - w(u',v), \pi_{su} \pi_{uu'} (u',v), G - \pi_{uu'})$ 
    } 
  }
}
\end{algorithm}

The correctness of Algorithm~\ref{alg:improved} follows directly from
the correctness of Algorithm~\ref{alg:simple}. The space used is the
same of Algorithm~\ref{alg:simple}, provided that $\lcp(u, t, \alpha,
G)$ uses linear space, which, as we show in the next section, is
indeed the case (Theorem~\ref{teo:lcp}).

Let us now analyze the total complexity of
Algorithm~\ref{alg:improved} as a function of the input size and of
$\gamma$, the number of $\alpha$-bounded ($k$-bounded) $st$-paths.  Let
$R$ be the recursion tree of Algorithm~\ref{alg:improved} and $T(r)$
the cost of a given node $r \in R$. The total cost of the algorithm
can be split in two parts, which we later bound individually, in the
following way:

\begin{equation}
   \sum_{r \in R} T(r) = \sum_{r: internal} T(r) + \sum_{r: leaf} T(r). \label{eq:total_cost}
\end{equation}

We have that $\sum_{r: leaf} T(r) = O((m + t(m,n))\gamma)$, since
leaves and solutions are in one-to-one correspondence and the cost for
each leaf is dominated by the cost of $\lcp(u, t, \alpha, G)$, that is
$O(m + t(m,n))$ (Theorem~\ref{teo:lcp}). Now, we have that every
internal node of the recursion has at least two children, otherwise
$\pi_{uu'}$ would not be the longest common prefix of
$\mathcal{P}_{\alpha}(u,t,G)$. Thus, $\sum_{r: internal} T(r) = O((m +
t(m,n))\gamma)$ since each internal node costs $O(m + t(m,n))$, the
cost is also dominated by the cost of the longest prefix computation,
and in any tree the number of branching nodes is at most the number of
leaves. Therefore, the total complexity of
Algorithm~\ref{alg:improved} is $O((m + t(n,m))\gamma)$. This
completes the proof of Theorem~\ref{thm:cost_improved}.

\begin{theorem} \label{thm:cost_improved}
  Algorithm~\ref{alg:improved} outputs all $\alpha$-bounded (or
  $k$-bounded) $st$-paths in $O((m+t(n,m))\gamma)$ using $O(m)$ space.
\end{theorem}

This means that for unweighted graphs it is possible to list all
$k$-bounded $st$-paths in $O(m)$ per path.  Moreover, for
non-negatively weighted graphs, it is possible to list all
$\alpha$-bounded $st$-paths in $O(m + n \log n)$ per path.

\subsubsection{Computing the longest common prefix of $\mathcal{P}_{\alpha}(s,t,G)$}

In this section, we present an efficient algorithm to compute the
longest common prefix of the set of $\alpha$-paths from $s$ to $t$,
completing the description of Algorithm~\ref{alg:improved}. The naive
algorithm for this problem runs in $O(nt(n,m))$, so that using it in
Algorithm~\ref{alg:improved} would not improve the total complexity
compared to Algorithm~\ref{alg:simple}. Basically, the naive algorithm
computes a shortest path $\pi_{st}$ and then for each prefix in
increasing order of length tests if there are at least two distinct
extensions each with total weight less than $\alpha$.  In order to
test the extensions, for each prefix $\pi_{su}$, we recompute the
distances from $t$ in the graph $G - \pi_{su}$, thus performing $n$
shortest path tree computations ($k$ computations in the unweighted
case) in the worst case.

Algorithm~\ref{alg:lcp} improves the naive algorithm by avoiding those
recomputations.  However, before entering the description of
Algorithm~\ref{alg:lcp}, we need a better characterization of the
structure of the longest common prefix of
$\mathcal{P}_{\alpha}(s,t,G)$. Lemma~\ref{lem:lcp} gives this. It does
so by considering a shortest path tree rooted at $s$, denoted by
$T_s$.  Recall that $T_s$ is a subgraph of $G$ and induces a partition
of the edges of $G$ into tree edges and non-tree edges. In this tree,
the longest common prefix of $\mathcal{P}_{\alpha}(s,t,G)$ is a prefix
of the tree path from the root $s$ to $t$. Additionally, any $st$-path
in $G$, excluding the tree path, necessarily passes through at least
one non-tree edge.  The lemma characterizes the longest common prefix
in terms of the non-tree edges from the subtrees rooted at siblings of
the vertices in the tree path from $s$ to $t$.

\begin{figure}[htbp]
  \centering \def\svgwidth{0.8\linewidth}
  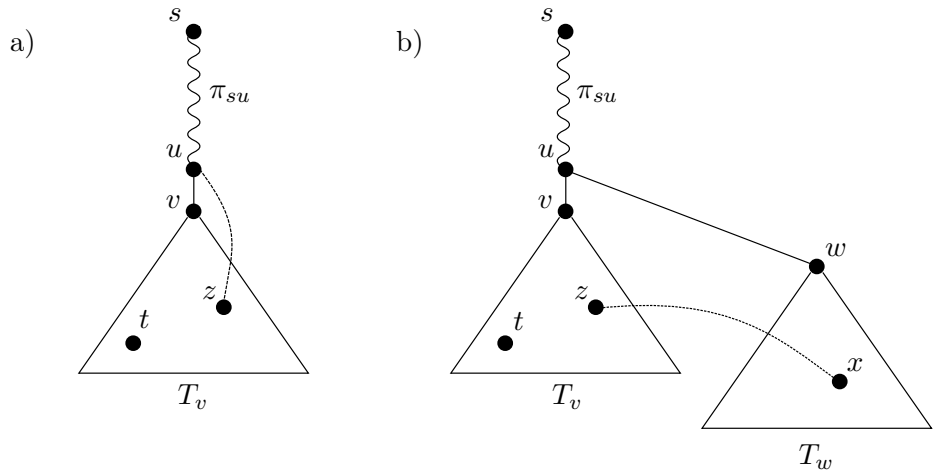 
  \caption{The common prefix $\pi_{su}$ of
    $\mathcal{P}_{\alpha}(s,t,G)$ can always be extended into an
    $st$-path using the tree path of $T_s$ from $u$ to $t$. The path
    $\pi_{su}$ is the longest common prefix if and only if it can also
    be extended with a path containing a non-tree edge $(x,z)$ such
    that $z \in T_v$ and (a) $x = u$ or (b) $x \in T_w$ and $w$ is
    sibling of $v$; and $d_{G'}(s,x) + w(x,z) + d_{G'}(z,t) \leq
    \alpha$, where $G' = G - (u,v)$.}  \label{fig:weighted:lemma}
\end{figure}

\begin{lemma} \label{lem:lcp}
 Let $\pi_{su} = (s = v_0, v_1), \ldots, (v_{l-1}, v_l = u)$ be a
 common prefix of all paths in $\mathcal{P}_{\alpha}(s,t,G) \neq
 \emptyset$ and $T_s$ a shortest path tree rooted at $s$. Then,
 \begin{enumerate}
 \item the path $\pi_{su}(u,v)$ is a common prefix of
   $\mathcal{P}_{\alpha}(s,t,G)$, if there is no edge $(x,z)$, with $z
   \in T_v$ and $x = u$ or $x \in T_w$ where $w$ is a sibling of $v$
   in the tree $T_s$, such that $d_{G'}(s,x) + w(x,z) + d_{G'}(z,t)
   \leq \alpha$, where $G' = G - (u,v)$; (see
   Fig.~\ref{fig:weighted:lemma})
 \item $\pi_{su}$ is the longest common prefix of
   $\mathcal{P}_{\alpha}(s,t,G)$, otherwise.
 \end{enumerate}
\end{lemma}
\begin{proof}
  Let us prove that if there exists a path $\pi_{ut}$, not containing
  the tree edge $(u,v)$, extending $\pi_{su}$ such that
  $w(\pi_{su}\pi_{ut}) \leq \alpha$, then there is a non-tree edge
  $(x,z)$ such that $d_{G'}(s,x) + w(x,z) + d_{G'}(z,t) \leq \alpha$
  and $z \in T_v$. For the moment, we do not impose that $x \in T_w$
  or $x = u$, we deal with this condition later. The paths of
  $\mathcal{P}_{\alpha}(s,t,G)$ that do not pass through $(u,v)$
  necessarily use some non-tree edge $(x,z)$, where $z \in T_v$, since
  $t$ belongs to $T_v$. Now, consider the path $\pi_{su}\pi_{ut}$ and
  let $(x,z)$ be the last non-tree edge of this path that enters
  $T_v$. This path can be rewritten as $\pi_{sx}(x,z)\pi_{zt}$.  We
  have that $w(\pi_{sx}(x,z)\pi_{zt}) = w(\pi_{sx}) + w(x,z) +
  w(\pi_{zt})$, where the path $\pi_{zt}$ is entirely contained in the
  induced subgraph of the vertices of $T_v$, because of our choice of
  $(x,z)$. Thus, $w(\pi_{zt}) \geq d_{G'}(z,t)$. Moreover, the path
  $\pi_{sx}$ does not contain $(u,v)$, since $\pi_{su}\pi_{ut}$ is a
  simple path and $(u,v)$ is not the first edge of $\pi_{ut}$. Thus,
  $w(\pi_{sx}) \geq d_{G'}(s,x)$. Therefore, combining the two
  inequalities, we have that $d_{G'}(s,x) + w(x,z) + d_{G'}(z,t) \leq
  w(\pi_{sx}) + w(x,z) + w(\pi_{zt}) \leq \alpha$ and $(x,z)$ is a
  non-tree edge with $z \in T_v$.
  
  It remains to prove that it is sufficient to consider only the
  non-tree edges $(x,z)$ entering $T_v$, such that $x = u$ or $x \in
  T_w$ where $w$ is sibling of $v$ in the tree. Let $\pi_{st}$ be a
  path in $\mathcal{P}_{\alpha}(s,t,G)$ including $(x,z)$ that does
  not have $\pi_{su}(u,v)$ as a prefix. This path can be rewritten as
  $\pi_{st} = \pi_{sx}(x,z)\pi_{zt}$. Since $\pi_{su}$ is a common
  prefix of $\mathcal{P}_{\alpha}(s,t,G)$, we have that $\pi_{su}$ is
  a prefix of $\pi_{sx}(x,z)\pi_{zt}$. Thus, $\pi_{sx}
  = \pi_{su}\pi_{ux}$, and either $\pi_{ux}$ is empty, so $(x,z)$ is a
  non-tree edge from $u$, or $\pi_{ux}$ enters a sibling subtree of
  $v$. This completes the proof of the first part of the lemma.

  Let us prove the second part of the lemma. There is at least one
  non-tree edge $(x,z)$ entering $T_v$ from $u$ or $T_w$, a sibling of
  $v$, such that $d_{G'}(s,x) + w(x,z) + d_{G'}(z,t) \leq
  \alpha$. Thus, the concatenation $\pi_{sx}(x,z)\pi_{zt}$ of the
  shortest paths contains $\pi_{su}(u,y)$ as prefix, where $y$ is a
  neighbor of $u$. Moreover, there is a subpath $\pi^*_{st}$ of
  $\pi_{sx}(x,z)\pi_{zt}$ that is simple and $w(\pi^*_{st}) \leq
  \alpha$, which also has $\pi_{su}(u,y)$ as prefix. Therefore,
  $\pi_{su}$ has two possible extensions, using $(u,y)$ or the tree
  edge $(u,v)$.
\end{proof}

In order to use the characterization of Lemma~\ref{lem:lcp} for the
longest prefix of $\mathcal{P}_{\alpha}(s,t,G)$, we need to be able to
efficiently test for the weight condition given in item~1, namely
$d_{G'}(s,x) + w(x,z) + d_{G'}(z,t) \leq \alpha$, where $G' = G -
(u,v)$ and $(u,v)$ belongs to the tree path from $s$ to $t$. We have
that $d_{G'}(s,x) = d_{G}(s,x)$, since $x$ does not belong to the
subtree of $v$ in the shortest path tree $T_s$. Indeed, only the
distances of vertices in the subtree $T_v$ can possibly change after
the removal of the tree edge $(u,v)$. On the other hand, in principle
we have no guarantee that $d_{G'}(z,t)$ also remains unchanged: recall
that to maintain the distances from $t$ we need a tree rooted at $t$
not at $s$. Clearly, we cannot compute the shortest path tree from $t$
for each $G'$, in the worst case, this would imply the computation of
$n$ shortest path trees. For this reason, we need
Lemma~\ref{lem:path:dist}.  It states that, in the specific case of
the vertices $z$ we need to compute the distance to $t$ in $G'$, we
have that $d_{G'}(z,t) = d_G(z,t)$. A similar result was proved in
\cite{Hershberger01}.

\begin{lemma} \label{lem:path:dist}
  Let $T_s$ be a shortest path tree rooted at $s$ and $t$ a vertex of
  $G$. Then, for any edge $(u,v)$, with $v$ closer to $t$, in the
  shortest path $\pi_{st}$ in the tree $T_s$, we have that $d_G(z,t) =
  d_{G'}(z,t)$, where $z \in T_v$ and $G' = G - (u,v)$.
\end{lemma}
\begin{proof}
  Suppose that $(u,v)$ belongs to the shortest path $\pi'_{zt}$ in
  $G$. This path can be written as the concatenation
  $\pi'_{zv}\pi'_{vt}$ (assuming wlog $v$ closer than $u$ to $t$),
  where $\pi'_{zv}$ and $\pi'_{vt}$ are both simple paths. We also
  have that $t$ and $z$ belong to the subtree $T_v$ which does not
  include the edge $(u,v)$, so the paths $\pi_{vt}$ and $\pi_{vz}$ in
  the tree $T_v$ are shortest paths that do not include $(u,v)$. The
  concatenation of $\pi_{vt}$ and $\pi_{vz}$ contains a subpath
  $\pi_{zt}$ from $z$ to $t$ such that $w(\pi_{zt}) \leq w(\pi_{vt}) +
  w(\pi_{vz})$. On the other hand, $w(\pi_{vt}) + w(\pi_{vz}) \leq
  w(\pi'_{zv}) + w(\pi'_{vt})$, since $\pi_{vt}$ and $\pi_{vz}$ are
  both shortest paths. Thus, $w(\pi_{zt}) \leq w(\pi'_{zv}) +
  w(\pi'_{vt})$.  Therefore, the concatenation of $\pi_{vt}$ and
  $\pi_{vz}$ contains as a subpath a shortest path from $z$ to $t$
  that does not include $(u,v)$.
\end{proof}
 
It is not hard to verify that Lemma~\ref{lem:lcp} is also valid for
directed graphs. Indeed, in the proof above, the fact that $G$ is
undirected is not used. On the other hand, the non-negative hypothesis
for the weights is necessary; more specifically, we need the
monotonicity property for path weights which states that for any path
the weight of any sub-path is not greater than the weight of the full
path. Now, in Lemma~\ref{lem:path:dist} both the path monotonicity
property and the fact that the graph is undirected are
necessary. Since these two lemmas are the base for the efficiency of
Algorithm~\ref{alg:lcp}, it seems difficult to extend it to general
weights and/or directed graphs.

Algorithm~\ref{alg:lcp} implements the strategy suggested by
Lemma~\ref{lem:lcp}. Given a shortest path tree $T_s$ of $G$ rooted at
$s$, the algorithm traverses each vertex $v_i$ in the tree path $s =
v_0, \ldots, v_n = t$ from the root $s$ to $t$, and at every step
finds all non-tree edges $(x,z)$ entering the subtree rooted at
$v_{i+1}$ from a sibling subtree, i.e. a subtree rooted at $w \in
N^+(v_i) \setminus \{v_{i+1}\}$. For each non-tree $(x,z)$ linking the
sibling subtrees found, it checks if it satisfies the weight condition
$d_{G'}(s,x) + w(x,z) + d_{G'}(z,t) \leq \alpha$, where $G' = G
\setminus (v_i,v_{i+1})$, given in item~1 of
Lemma~\ref{lem:lcp}. Item~2 of the same lemma implies that the first
time an edge $(x,z)$ satisfies the weight condition, the tree path
traversed so far is the longest common prefix of
$\mathcal{P}_{\alpha}(s,t,G)$. In order to test the weight conditions,
as stated previously, we have that $d_{G'}(s,x) = d_G(s,x)$, since $x$
does not belong to the subtree of $v$ in $T_s$. In addition,
Lemma~\ref{lem:path:dist} guarantees that $d_{G'}(z,t) =
d_G(z,t)$. Thus, it is sufficient for the algorithm to compute only
the shortest path trees from $t$ and from $s$ in $G$.

\begin{algorithm} 
\caption{$\lcp(s, t, \alpha, G)$} \label{alg:lcp}
compute $T_s$, a shortest path tree from $s$ in $G$ \\
compute $T_t$, a shortest path tree from $t$ in $G$ \\
let $\pi_{st} = (s = v_0, v_1) \ldots (v_{n-1}, v_n = t)$ be the shortest path in $T_s$ \\
\For{$v_i \in \{v_1, \ldots, v_n\}$} { \label{alg:lcp:for}
  \For{$w \in N^+(v_i) \setminus \{v_{i+1}\}$}{
    let $T_w$ be the subtree of $T_s$ rooted at $w$ \\
    \For{$(x,z) \in G$, s.t. $x \in T_w$ {\bf or} $x = v_i$}{
      \If{$z \in T_{v_{i+1}}$ {\bf and}  $d_G(s,x) + w(x,z) + d_G(z,t) \leq \alpha$}{ \label{alg:lcp:edges}
        {\bf break} \\
      }
    }
  }
}
{\bf return} $\pi_{sv_{i-1}}$
\end{algorithm}

\begin{theorem} \label{teo:lcp}
  Algorithm~\ref{alg:lcp} finds the longest common prefix of
  $\mathcal{P}_{\alpha}(s,t,G)$ in $O(m+t(n,m))$ using $O(m)$ space.
\end{theorem}
\begin{proof}
  The cost of the algorithm can be divided in two parts: the cost to
  compute the shortest path trees $T_s$ and $T_t$, and the cost of the
  loop in line~\ref{alg:lcp:for}. The first part is bounded by
  $O(t(n,m))$. Let us now prove that the second part is bounded by
  $O(m+n)$. The cost of each execution of line~\ref{alg:lcp:edges} is
  $O(1)$, since we only need distances from $s$ and $t$ and the
  shortest path trees from $s$ and $t$ are already computed, and we
  pre-process the tree to decide in $O(1)$ if a vertex belongs to a
  subtree.  In that way, the cost of the loop is bounded by the number
  of times line~\ref{alg:lcp:edges} is
  executed. Line~\ref{alg:lcp:edges} is executed at most $m$ times,
  the neighborhood of each vertex is visited at most once, since the
  subtrees $T_w$ are disjoint, they are rooted at vertices adjacent to
  some vertex in the tree path $\pi_{su}$ but not included in it.
\end{proof}

\subsection{Listing paths in increasing order of their lengths} \label{sec:weighted:ordered}
In this section, we modify Algorithm~\ref{alg:simple} to output the
$\alpha$-bounded $st$-paths in increasing order of their length, while
maintaining (almost) the same time complexity but increasing the
memory usage. As for Algorithm~\ref{alg:simple}, this algorithm works
for any version of the problem, directed or undirected graphs with
general weights, and the complexity depends on the cost to compute
a shortest path tree. The pseudocode is shown in
Algorithm~\ref{alg:iterative}. This is a generic description of the
algorithm, the container $Q$ is not specified in the pseudocode, the
only requirement is the support for two operations: \emph{push}, to
insert a new element in $Q$; and \emph{pop}, to remove and return an
element of $Q$.

\begin{algorithm}[htbp]
\caption{$\listpathsiter(u,t, \alpha, \pi_{su}, G)$}  \label{alg:iterative}   
push $\langle s, t, \emptyset, G\rangle$ in $Q$ \\
\While{$Q$ is not empty} {
  $\langle u, t, \pi_{su}, G\rangle = Q.pop()$ \\
  \uIf{$u = t$}{ 
    output($\pi_{su}$) \\ 
  }
  \Else {
    compute a shortest path tree $T_t$ from $t$ in $G^R - u$ \\
    \For{$v \in N^+(u)$}{
      \If{$d(v,t) \leq \alpha - w(u,v)$}{ \label{alg:iterative:test}
        push $\langle v,t, \alpha - w(u,v), \pi_{su} (u,v), G - u \rangle$ in $Q$ 
      } 
    }
  }
}
\end{algorithm}

Algorithm~\ref{alg:iterative} is a non-recursive version of
Algorithm~\ref{alg:simple}, and uses the same strategy to partition
the solution space (Eq.~\ref{eq:path:partition}). However, the order
in which the partitions are explored is not necessarily the same,
depending on the type of container used for $Q$. We show that if $Q$
is a stack then the solutions are output in the reverse order of
Algorithm~\ref{alg:simple}, and the maximum size of the stack is
linear in the size of the input.  If on the other hand, $Q$ is a heap,
using a suitable key, the solutions are output in increasing order of
their lengths, but in this case the maximum size of the heap is linear
in the number of solutions, which is not polynomial in the size of the
input.


The recursive partition of $\mathcal{P}_{\alpha}(s,t,G)$, i.e. the set
of $\alpha$-bounded $st$-paths in $G = (V,E)$, according to
Eq.~\ref{eq:path:partition} has a rooted tree structure. Indeed, the
nodes are the sets $\mathcal{P}_{\alpha'}(v,t,G')$, where $G' =
(V',E')$ is a subgraph of $G$, $\alpha' \in \mathbb{Q}$, and $v \in
V'$; for a given node the children are the sets in the partition of
Eq.~\ref{eq:path:partition} satisfying the condition of
line~\ref{alg:iterative:test}, i.e. the non-empty sets; the root is
$\mathcal{P}_{\alpha}(s,t,G)$; and the leaves are the singletons
$\mathcal{P}_{\alpha'}(t,t,G')$, which are in a one-to-one
correspondence with the $\alpha$-bounded $st$-paths. We denote this
rooted tree by $\mathcal{T}$.

For any container $Q$ supporting push and pop operations,
Algorithm~\ref{alg:iterative} visits each node of $\mathcal{T}$
exactly once, since at every iteration a node from $Q$ is deleted and
its children are inserted in $Q$, and $\mathcal{T}$ is a tree.  In
particular, this guarantees that every leaf of $\mathcal{T}$ is
visited exactly once, thus proving the following lemma.

\begin{lemma}
  Algorithm~\ref{alg:iterative} outputs all $\alpha$-bounded
  $st$-paths.
\end{lemma}

Let us consider the case where $Q$ is a stack. It is not hard to prove
that Algorithm~\ref{alg:simple} is a DFS traversal of $\mathcal{T}$
starting from the root, while Algorithm~\ref{alg:iterative} is an
\emph{iterative} DFS (\cite{Sedgewick01}) traversal of $\mathcal{T}$
also starting from the root. Basically, an iterative DFS keeps the
vertices of the fringe of the non-visited subgraph in a stack, at each
iteration the next vertex to be explored is popped from the stack, and
recursive calls are replaced by pushing vertices in the stack. Now,
for a fixed permutation of the children of each node in $\mathcal{T}$,
the nodes visited in an iterative DFS traversal are in the reverse
order of the nodes visited in a recursive DFS traversal
(\cite{Sedgewick01}), thus proving Lemma~\ref{lem:weighted:dfs_order}.

\begin{lemma} \label{lem:weighted:dfs_order}
  If $Q$ is a stack, then Algorithm~\ref{alg:iterative} outputs the
  $\alpha$-bounded $st$-path in the reverse order of
  Algorithm~\ref{alg:simple}.
\end{lemma}

For any rooted tree, at any moment during an iterative DFS traversal,
the number of nodes in the stack is bounded by the sum of the degrees
of the root-to-leaf path currently being explored. Recall that every
leaf in $\mathcal{T}$ corresponds to a path in
$\mathcal{P}_{\alpha}(s,t,G)$. Actually, there is a one-to-one
correspondence between the nodes of a root-to-leaf path $P$ in
$\mathcal{T}$ and the vertices of the $\alpha$-bounded $st$-path $\pi$
associated to that leaf. Hence, the sum of the degrees of the nodes of
$P$ in $\mathcal{T}$ is equal to the sum of the degrees of the
vertices $\pi$ in $G$, which is bounded by $m$, thus proving
Lemma~\ref{lem:weighted:dfs_size}. 

\begin{lemma} \label{lem:weighted:dfs_size}
  The maximum number of elements in the stack of
  Algorithm~\ref{alg:iterative} over all iterations is bounded by $m$.
\end{lemma}

Let us consider now the case where $Q$ is a heap. There is a
one-to-many correspondence between arcs in $G$ and arcs in
$\mathcal{T}$, i.e. if $\mathcal{P}_{\alpha''}(v,t,G'')$ is a child of
$\mathcal{P}_{\alpha'}(u,t,G')$ in $\mathcal{T}$ then $(u,v)$ is an
arc of $G$. For every arc of $\mathcal{T}$ let us associate the weight
of the corresponding arc in $G$. Intuitively,
Algorithm~\ref{alg:iterative} using a priority queue with $w(\pi_{su})
+ d_G(u,t)$ as keys performs a Dijkstra-like traversal in a weighted
version of $\mathcal{T}$ starting from the root, where for a node
$\langle u, t, \pi_{su}, G\rangle$ the distance from the root is
$w(\pi_{su})$ and $d_G(u,t)$ is a (precise) estimation of the distance
from $\langle u, t, \pi_{su}, G\rangle$ to the closest leaf of
$\mathcal{T}$. In other words, it is an $A^*$-like traversal
(\cite{Dechter85}) in the weighted rooted tree $\mathcal{T}$, using
the (optimal) heuristic $d_G(u,t)$. As such,
Algorithm~\ref{alg:iterative} explores first the nodes of
$\mathcal{T}$ leading to the cheapest non-visited leaf. This is
formally stated in Lemma~\ref{lem:heap_order}.

\begin{lemma} \label{lem:heap_order}
  If $Q$ is a priority queue with $w(\pi_{su}) + d_{G'}(u,t)$ as the
  priority key of $\langle u, t, \pi_{su}, G'\rangle$, then
  Algorithm~\ref{alg:iterative} outputs the $\alpha$-bounded
  $st$-paths in increasing order of their lengths.
\end{lemma}
\begin{proof}
  The priority of a node $N_u = \langle u, t, \pi_{su}, G'\rangle$
  (i.e. $\mathcal{P}_{\alpha - w(\pi_{su})} (u,t,G')$) is the weight
  of the path $\pi_{su}$ plus the weight of a shortest path
  $\pi^*_{ut}$ from $u$ to $t$ in $G'$. Let $v$ be an out-neighbor of
  $u$ in $G'$, then the node $N_v = \langle v, t, \pi_{su}(u,v), G' -
  u\rangle$ is a child of $N_u$, and the priority of $N_v$ is greater
  or equal to the priority of $N_u$. Indeed, suppose it is strictly
  smaller, then the path $\pi_{su}(u,v)$ concatenated with the
  shortest path from $v$ to $t$ in $G'-u$ is shorter than
  $\pi_{su}\pi^*_{ut}$, contradicting the fact that $\pi^*_{ut}$ is a
  shortest path of $G'$. Hence, the priorities of the nodes removed
  from $Q$ are not decreasing, since for every node removed only nodes
  with greater or equal priorities are inserted. Moreover, the
  priority of a leaf $\langle t, t, \pi_{st}, G'\rangle$ is precisely
  $w(\pi_{st})$, the weight of a path in
  $\mathcal{P}_{\alpha}(s,t,G)$. Therefore, the leaves are visited in
  increasing order of the length of their corresponding $st$-path.
\end{proof}

For any choice of the container $Q$, every node of $\mathcal{T}$ is
visited exactly once, that is, each node of $\mathcal{T}$ is pushed at
most once in $Q$. This proves Lemma~\ref{lem:heap_size}.

\begin{lemma} \label{lem:heap_size}
  The maximum number of elements in a priority queue of
  Algorithm~\ref{alg:iterative} over all iterations is bounded by $\gamma$.
\end{lemma}

Algorithm~\ref{alg:iterative} uses $O(m\gamma)$ space, since for every
node inserted in the heap, we also have to store the corresponding
graph. Moreover, using a binary heap (\cite{Cormen01}) as a priority
queue, the push and pop operations can be performed in $O(\log
\gamma)$ each, where by Lemma~\ref{lem:heap_size} $\gamma$ is the
maximum size of the heap. Therefore, combining it with
Lemma~\ref{lem:heap_order} we obtain the following theorem.

\begin{theorem}
  Algorithm~\ref{alg:iterative} using a heap outputs all
  $\alpha$-bounded $st$-paths in increasing order of their lengths in
  $O((nt(n,m) + \log \gamma)\gamma)$ total time, using $O(m\gamma)$
  space.
\end{theorem}


\bigskip
\bigskip


\section{Discussion and conclusions}
In the first part of this chapter, we introduced a polynomial delay
algorithm to list all bubbles with path length constraints in weighted
directed graphs. This is a theoretically sound approach that in
practice is considerably faster than the bubble listing algorithm of
\ks (Section~\ref{sec:kissplice:algorithm}), and as a result enables
us to enumerate more bubbles. Additionally, we gave an indication that
these additional bubbles correspond to longer AS events, overseen
previously but biologically very relevant.  Moreover, as shown
in~\cite{Tarjan90}, by combining radix and Fibonacci heaps in
Dijkstra, we can achieve a $O(n(m + n \sqrt{ \log \alpha_1)})$ delay
for Algorithm~\ref{alg:weighted:listbubbles} in cDGBs. The current
implementation of \ks (version 2.0) uses
Algorithm~\ref{alg:weighted:listbubbles} to list bubbles.

In the second part of this chapter, we introduced a general framework
to list bounded length $st$-paths in weighted directed graphs. In the
particular case of undirected graphs, we showed an improved algorithm
to list bounded length $st$-paths in $O((m + n \log n) \gamma)$ time
for non-negative weights and $O(m \gamma)$ time for unit weights,
where $\gamma$ is the number of bounded length paths. Moreover, we
showed how to modify the general algorithm to output the paths in
increasing order of their length, thus providing an alternative
solution to the classical $K$-shortest paths problem, which does not
improve the complexity but is simpler than previous approaches.

Actually, the general framework of Section~\ref{sec:weighted:path} can
be seen as a ``simplification'' of the bubble listing algorithm of
Section~\ref{sec:weighted:bubble} (extended to general weights). More
precisely, listing bounded length $st$-paths can be reduced to listing
bounded length bubbles with a given source $s$. Indeed, consider an
instance of the first problem, a graph $G$ and two vertices $s,t$, and
build the graph $G'$ by adding an arc $(s,t)$ with weight $\alpha'$,
strictly smaller than the sum of all negative weight arcs (if any) of
$G$; listing $st$-paths with a length bounded by $\alpha$ in $G$ is
equivalent to listing $(s,t,\alpha',\alpha)$-bubbles in $G'$.

\chapter{Memory efficient de Bruijn graph representation}
\label{chap:dbg}
\minitoc
\renewcommand{\arraystretch}{1.5}
\setlength{\tabcolsep}{8pt}

This chapter is strongly based on our paper \cite{Salikhov13}. As
shown in Chapter~\ref{chap:kissplice}, the de Bruijn graph
construction and representation are the memory bottleneck of \ks. In
this chapter, we consider the problem of compactly representing a de
Bruijn graph. We show how to reduce the memory required by the
algorithm of \cite{Chikhi12}, that represents de Brujin graphs using
Bloom filters. Our method requires 30\% to 40\% less memory with
respect to their method, with insignificant impact to construction
time. At the same time, our experiments showed a better query time
compared to their method. This is, to our knowledge, the best
\emph{practical} representation for de Bruijn graphs. The current
implementation of \ks (version 2.0) uses the de Bruijn graph
representation and construction presented in this chapter.


\bigskip
\bigskip


\section{Introduction}
As shown in Chapter~\ref{chap:back}, \ks is not the only NGS data
analysis method using de Bruijn graphs. In fact, the majority of the
more recent genome and transcriptome assemblers and some metagenome
assemblers (\cite{Meta-idba,Namiki11}) use de Bruijn graphs.  Due to
the very large size of NGS datasets, it is essential to represent de
Bruijn graphs as compactly as possible.  This has been a very active
line of research. Recently, several papers have been published that
propose different approaches to compressing de Bruijn graphs
\cite{Conway11,Ye12,Chikhi12,Bowe12,Pell09}.

\cite{Conway11} proposed a method based on classical succinct data
structures, i.e. bitmaps with efficient rank/select operations.  On
the same direction, \cite{Bowe12} proposed a very interesting succinct
representation that, assuming only one string (read) is present, uses
only $4m$ bits, where $m$ is the number of arcs in the graph. The more
realistic case, where there are $M$ reads, can be easily reduced to
the one string case by concatenating all $M$ reads using a special
separator character. However, in this case the size of the structure
is $4m + O(M \log m)$ bits (\cite{Bowe12}, Theorem 1). Since the
multiplicative constant of the second term is hidden by the asymptotic
notation, it is hard to know precisely what would be the size of this
structure in practice.

\cite{Ye12} proposed a different method based on a sparse
representation of de Bruijn graphs, where only a subset of $k$-mers
present in the dataset are stored. \cite{Pell09} proposed a method to
represent it approximately, the so called \emph{probabilistic de
  Bruijn graph}. In their representation a vertex have a small
probability to be a false positive, i.e. the $k$-mer is not present in
the dataset. Finally, \cite{Chikhi12} improved Pell's scheme in order
to obtain an exact representation of the de Bruijn graph. This was, to
our knowledge, the best \emph{practical} representation of an exact de
Bruijn graph.

In this chapter, we focus on the method proposed in \cite{Chikhi12}
which is based on Bloom filters.  They were first used in
\cite{Pell09} to provide a very space-efficient representation of a
subset of a given set (in our case, a subset of $k$-mers), at the
price of allowing {\em one-sided errors}, namely {\em false
  positives}. The method of \cite{Chikhi12} is based on the following
idea: if all queried vertices ($k$-mers) are only those which are
reachable from some vertex known to belong to the graph, then only a
fraction of all false positives can actually occur. Storing these
false positives explicitly leads to an exact (false positive free) and
space-efficient representation of the de Bruijn graph.

Our contribution is an improvement of this scheme by changing the
representation of the set of false positives. We achieve this by
iteratively applying a Bloom filter to represent the set of false
positives, then the set of ``false false positives'' etc. We show
analytically that this cascade of Bloom filters allows for a
considerable further economy of memory, improving the method of
\cite{Chikhi12}. Depending on the value of $k$, our method requires
30\% to 40\% less memory with respect to the method of
\cite{Chikhi12}. Moreover, with our method, the memory grows very
little as $k$ grows. Finally, we implemented our method and tested it
against \cite{Chikhi12} on real datasets. The tests confirm the
theoretical predictions for the size of structure and show a 20\% to
30\% \emph{improvement} in query times.

\section{Preliminaries} \label{sec:prelim}
A \emph{Bloom filter} is a space-efficient data structure for
representing a given subset of elements $T \subseteq U$, with support
for efficient membership queries with one-sided error. That is, if a
query for an element $x \in U$ returns \emph{no} then $x \notin T$,
but if it returns \emph{yes} then $x$ may or not belong to $T$,
i.e. with small probability $x \notin T$ (false positive). It consists
of a bitmap (array of bits) $B$ with size $m$ and a set of $p$
distinct hash functions $\{h_1,\ldots, h_p\}$, where $h_i: U \mapsto
\{0,\ldots, m-1\}$. Initially, all bits of $B$ are set to $0$. An
insertion of an element $x \in T$ is done by setting the elements of
$B$ with indices $h_1(x), \ldots, h_p(x)$ to $1$, i.e. $B[h_i(x)] = 1$
for all $i \in [1,p]$. The membership queries are done symmetrically,
returning \emph{yes} if all $B[h_i(x)]$ are equal $1$ and \emph{no}
otherwise. As shown in \cite{Kirsch08}, when considering hash
functions that yield equally likely positions in the bit array, and
for large enough array size $m$ and number of inserted elements $n$,
the false positive rate $\mathcal{F}$ is
\begin{equation}
  \mathcal{F} \approx (1 - e^{-pn/m})^p = (1 - e^{-p/r})^p,
\end{equation}
where $r=m/n$ is the number of bits (of the bitmap $B$) per element
(of $T$ represented). It is not hard to see that this expression is
minimized when $p = r \ln 2$, giving a false positive rate of
\begin{equation}
  \mathcal{F} \approx (1 - e^{-p/r})^p = (1/2)^p \approx 0.6185^r.
\end{equation}

A de Bruijn graph, as defined in Chapter~\ref{chap:back}
(Definition~\ref{def:dbg}), is entirely determined by the set of
$k$-mers (vertices) and $(k+1)$-mers (arcs) of the read set
$\mathcal{R} \subseteq \Sigma^* = \{A,C,T,G\}^*$. For reasons that
will be clear soon, we relax this definition, dropping the bijection
between that $(k+1)$-mers and arcs but keeping the $k-1$ suffix-prefix
overlap requirement. That way, a de Bruijn graph, for a given
parameter $k$, of a set of reads $\mathcal{R}$ is entirely defined by
the set $T \subseteq U = \Sigma^k$ of $k$-mers present in
$\mathcal{R}$. Indeed, the vertices of the graph are precisely the
$k$-mers of $T$ and for any two vertices $u,v \in T$, there is an arc
from $u$ to $v$ if and only if the suffix of $u$ of size $k-1$ is
equal to the prefix of $v$ of the same size. Therefore, given a set $T
\subseteq U$ of $k$-mers we can represent its de Bruijn graph using a
Bloom filter $B$. This representation has the disadvantage of having
false positive vertices, as direct consequence of the false positive
queries in the Bloom filter, which can create false connections in the
graph (see \cite{Pell09} for the influence of false positive vertices
on the topology of the graph). The naive way to remove those false
positives vertices, by explicitly storing (e.g. using a hash table)
the set of all false positives of $B$, is clearly inefficient, as the
expected number of elements to be explicitly stored is $|U|\mathcal{F}
= 4^k \mathcal{F}$.

The key idea of \cite{Chikhi12} is to explicitly store only a subset
of all false positives of $B$, the so-called {\em critical false
  positives}. This is possible because in order to perform an exact
(without false positive vertices) graph traversal, only potential
neighbors of vertices in $T$ are queried. In other words, the set of
critical false positives consists of the potential neighbors of $T$
that are false positives of $B$, i.e. the $k$-mers from $U$ that
overlap the $k$-mers from $T$ by $k-1$ letters and are false positives
of $B$. Thus, the size of the set of critical false positives is
bounded by $8|T|$, since each vertex of $T$ has at most $2|\Sigma| =
8$ neighbors (for each vertex, there are $|\Sigma|$ $k$-mers
overlapping the $k-1$ suffix and $|\Sigma|$ overlapping the $k-1$
prefix). Therefore, the expected number of critical false positives is
bounded above by $8 |T| \mathcal{F}$.

\section{Cascading Bloom filter}

Let $\mathcal{R}$ be a set of reads and $T_0$ be the set of occurring
$k$-mers (vertices of the de Brujin graph) that we want to store. As
stated in Section~\ref{sec:prelim}, the method of \cite{Chikhi12}
stores $T_0$ via a bitmap $B_1$ using a Bloom filter, together with
the set $T_1$ of critical false positives. $T_1$ consists of those
$k$-mers which have a $k-1$ overlap with $k$-mers from $T_0$ but which
are stored in $B_1$ ``by mistake'', i.e. belong\footnote{By a slight
  abuse of language, we say that ``an element belongs to $B_j$'' if it
  is accepted by the corresponding Bloom filter.} to $B_1$ but not to
$T_0$. $B_1$ and $T_1$ are sufficient to represent the graph provided
that the only queried $k$-mers are those which are potential neighbors
of $k$-mers of $T_0$.

The idea we introduce here is to use this structure recursively and
represent the set $T_1$ by a new bitmap $B_2$ and a new set $T_2$,
then represent $T_2$ by $B_3$ and $T_3$, and so on.  More formally,
starting from $B_1$ and $T_1$ defined as above, we define a series of
bitmaps $B_1, B_2, \ldots$ and a series of sets $T_1,T_2,\ldots$ as
follows.  $B_2$ stores the set of false positives $T_1$ using another
Bloom filter, and the set $T_2$ contains the critical false positives
of $B_2$, i.e. ``true vertices'' from $T_0$ that are stored in $B_2$ ``by
mistake'' (we call them {\bf false$^{2}$} positives). $B_3$ and $T_3$,
and, generally, $B_i$ and $T_i$ are defined similarly: $B_i$ stores
$k$-mers of $T_{i-1}$ using a Bloom filter, and $T_i$ contains
$k$-mers stored in $B_i$ ``by mistake'', i.e. those $k$-mers that do
not belong to $T_{i-1}$ but belong to $T_{i-2}$ (we call them {\bf
  false}$^{i}$ positives). Observe that $T_0\cap T_1=\emptyset$, $T_0
\supseteq T_2 \supseteq T_4 \ldots$ and $T_1 \supseteq T_3 \supseteq
T_5 \ldots$.

The following lemma shows that the construction is correct, that is it
allows one to verify whether or not a given $k$-mer belongs to the set
$T_0$.

\begin{lemma}
\label{mainlemma}
Given a $k$-mer (vertex) $K$, consider the smallest $i$ such that $K
\not\in B_{i+1}$ (if $K \not\in B_1$, we define $i=0$).  Then, if $i$
is odd, then $K\in T_0$, and if $i$ is even (including $0$), then
$K\not\in T_0$.
\end{lemma}
\begin{proof}
Observe that $K\not\in B_{i+1}$ implies $K\not\in T_i$ by the basic
property of Bloom filters that membership queries have one-sided
error, i.e. there are no false negatives. We first check the Lemma for
$i=0,1$.

For $i=0$, we have $K \not\in B_1$, and then $K \not\in T_0$.

For $i=1$, we have $K\in B_1$ but $K\not\in B_2$. The latter implies
that $K\not\in T_1$, and then $K$ must be a false$^2$ positive, that
is $K\in T_0$.  Note that here we use the fact that the only queried
$k$-mers $K$ are either vertices of $T_0$ or their neighbors in the graph
(see \cite{Chikhi12}), and therefore if $K\in B_1$ and $K\not\in T_0$
then $K\in T_1$.

For the general case $i \geq 2$, we show by induction that $K\in
T_{i-1}$. Indeed, $K\in B_1\cap \ldots\cap B_i$ implies $K\in
T_{i-1}\cup T_i$ (which, again, is easily seen by induction), and
$K\not\in B_{i+1}$ implies $K\not\in T_i$.

Since $T_{i-1}\subseteq T_0$ for odd $i$, and $T_{i-1}\subseteq T_1$
for even $i$ (for $T_0\cap T_1=\emptyset$), the lemma follows.
\end{proof}

Naturally, the lemma provides an algorithm to check if a given $k$-mer
$K$ belongs to the graph: it suffices to check successively if it
belongs to $B_1,B_2,\ldots$ until we encounter the first $B_{i+1}$
which does not contain $K$. Then, the answer will simply depend on
whether $i$ is even or odd: $K$ belongs to the graph if and only if
$i$ is odd.

In our reasoning so far, we assumed an infinite number of bitmaps
$B_i$.  Of course, in practice we cannot store infinitely many (and
even simply many) bitmaps. Therefore, we ``truncate'' the construction
at some step $t$ and store a finite set of bitmaps $B_1, B_2,
\ldots,B_t$ together with an explicit representation of $T_t$. The
procedure of Lemma~\ref{mainlemma} is extended in the obvious way: if
for all $1\leq i \leq t$, $K\in B_i$, then the answer is determined by
directly checking $K\in T_t$.

\section{Memory and time usage} \label{sec:memory_time}

First, we estimate the memory needed by our data structure, under the
assumption of an infinite number of bitmaps. Let $N$ be the number of
``true positives'', i.e. vertices of $T_0$. As stated in
Section~\ref{sec:prelim}, if $T_0$ has to be stored via a bitmap $B_1$
of size $rN$, the false positive rate can be estimated as $c^r$, where
$c=0.6185$. And, the expected number of critical false positive
vertices (set $T_1$) has been estimated in \cite{Chikhi12} to be
$8Nc^r$, as every vertex has eight extensions, i.e. potential
neighbors in the graph. We slightly refine this estimation to $6Nc^r$
by noticing that for most of the graph vertices, two out of these
eight extensions belong to $T_0$ (are real vertices) and thus only six
are potential false positives.  Furthermore, to store these $6Nc^r$
critical false positive vertices, we use a bitmap $B_2$ of size
$6rNc^r$. Bitmap $B_3$ is used for storing vertices of $T_0$ which are
stored in $B_2$ ``by mistake'' (set $T_2$).  We estimate the number of
these vertices as the fraction $c^r$ (false positive rate of filter
$B_2$) of $N$ (size of $T_0$), that is $Nc^r$.  Similarly, the number
of vertices we need to put to $B_4$ is $6Nc^r$ multiplied by $c^r$,
i.e. $6Nc^{2r}$. Continuing in this way, the memory needed for the
whole structure is $rN + 6rNc^r + rNc^r+ 6rNc^{2r} + rNc^{2r} + ...$
bits.  The number of bits per $k$-mer is then

\begin{equation}
  r + 6rc^r + rc^r + 6rc^{2r} + ... = (r+6rc^r )(1 + c^r + c^{2r}+...)= (1+6c^r) \frac{r}{1-c^r}. \label{formula}
\end{equation}
A simple calculation shows that the minimum of this expression is
achieved when $r = 5.464$, and then the minimum memory used per
$k$-mer is $8.45$ bits.

As mentioned earlier, in practice we store only a finite number of
bitmaps $B_1,\ldots,B_t$ together with an explicit representation
(such as array or hash table) of $T_t$.
In this case, the memory taken by the bitmaps is a truncated sum $rN +
6rNc^r + rNc^r+ ..$, and a data structure storing $T_t$ takes either
$2k \cdot Nc^{\lceil\frac{t}{2}\rceil r}$ or $2k \cdot
6Nc^{\lceil\frac{t}{2}\rceil r}$ bits, depending on whether $t$ is
even or odd. The latter follows from the observations that we need to
store $Nc^{\lceil\frac{t}{2}\rceil r}$ (or
$6rNc^{\lceil\frac{t}{2}\rceil r}$) $k$-mers, each taking $2k$ bits of
memory. Consequently, we have to adjust the optimal value of $r$
minimizing the total space, and re-estimate the resulting space spent
on one $k$-mer.

Table~\ref{table1} shows estimations for optimal values of $r$ and the
corresponding space per $k$-mer for $t=4$ and $t=6$, and several
values of $k$. The data demonstrates that even such small values of
$t$ lead to considerable memory savings. It appears that the space per
$k$-mer is very close to the ``optimal'' space ($8.45$ bits) obtained
for the infinite number of filters. Table~\ref{table1} reveals another
advantage of our improvement: the number of bits per stored $k$-mer
remains almost constant for different values of $k$.

\begin{table}
\begin{center}
\begin{tabular}{|c|c|c|c|c|c|}
\hline
$k$ &optimal $r$& bits per $k$-mer & optimal $r$& bits per $k$-mer & bits per $k$-mer\\
 & for $t=4$ & for $t=4$ & for $t=6$ & for $t=6$ & for $t=1$ \\\hline\hline
16 & 5.777 & 8.556& 5.506 & 8.459 & 12.078 \\
\hline
32 & 6.049 & 8.664& 5.556 & 8.47 & 13.518 \\
\hline
64 & 6.399 & 8.824 & 5.641 & 8.49 & 14.958 \\
\hline
128 & 6.819 & 9.045 & 5.772 & 8.524 & 16.398\\
\hline
\end{tabular}
\end{center}
\caption{1st column: $k$-mer size; 2nd and 4th columns: optimal value
  of $r$ for Bloom filters (bitmap size per number of stored elements)
  for $t=4$ and $t=6$ respectively; 3rd and 5th columns: the resulting
  space per $k$-mer (for $t=4$ and $t=6$); 6th column: space per
  $k$-mer for the method of \cite{Chikhi12} ($t=1$) }\label{table1}
\end{table}

The last column of Table~\ref{table1} shows the memory usage of the
original method of \cite{Chikhi12}, obtained using the estimation
$(1.44 \log_2 (\frac{16k}{2.08}) + 2.08)$ the authors provided.  Note
that according to that estimation, doubling the value of $k$ results
in a memory increment by $1.44$ bits, whereas in our method the
increment is of $0.11$ to $0.22$ bits.

Let us now estimate preprocessing and query times for our scheme. If
the value of $t$ is small (such as $t=4$, as in Table~\ref{table1}),
the preprocessing time grows insignificantly in comparison to the
original method of \cite{Chikhi12}. To construct each $B_i$, we need
to store $T_{i-2}$ (possibly on disk, if we want to save on the
internal memory used by the algorithm) in order to compute those
$k$-mers which are stored in $B_{i-1}$ ``by mistake''.  The
preprocessing time increases little in comparison to the original
method of \cite{Chikhi12}, as the size of $B_i$ decreases
exponentially and then the time spent to construct the whole structure
is linear on the size of $T_0$.

The query time can be split in two parts: the time spent on querying
$t$ Bloom filters and the time spent on querying $T_t$.  Clearly,
using $t$ Bloom filters instead of a single one introduces a
multiplicative factor of $t$ to the first part of the query time. On
the other hand, the set $T_t$ is generally much smaller than $T_1$,
due to the above-mentioned exponential decrease. Depending on the data
structure for storing $T_t$, the time saving in querying $T_t$
vs. $T_1$ may even dominate the time loss in querying multiple Bloom
filters. Our experimental results (Section~\ref{sec:imple} below)
confirm that this situation does indeed occur in practice. Note that
even in the case when querying $T_t$ weakly depends on its size
(e.g. when $T_t$ is implemented by a hash table), the query time will
not increase much, due to our choice of a small value for $t$, as
discussed earlier.

\subsection{Using different values of $r$ for different filters} \label{subsec:difr}
In the previous section, we assumed that each of our Bloom filters
uses the same value of $r$, the ratio of bitmap size to the number of
stored $k$-mers. However, formula (\ref{formula}) for the number of
bits per $k$-mer shows a difference for odd and even filter
indices. This suggests that using different parameters $r$ for
different filters, rather than the same for all filters, may reduce
the space even further.  If $r_i$ denotes the corresponding ratio for
filter $B_i$, then (\ref{formula}) should be rewritten to
\begin{equation}
r_1 + 6r_2c^{r_1} + r_3c^{r_2} + 6r_4c^{r_1+r_3} + ...,
\end{equation}
and the minimum value of this expression becomes $7.93$ (this value is
achieved with $r_1 = 4.41; r_i = 1.44, i > 1$).

In the same way, we can use different values of $r_i$ in the truncated
case. This leads to a small $2\%$ to $4\%$ improvement in comparison
with case of unique value of $r$.  Table~\ref{table3} shows results
for the case $t=4$ for different values of $k$.

\begin{table}[h]
\begin{center}
\begin{tabular}{|c|c|c|c|}
\hline
$k$ & $r_1, r_2, r_3, r_4$ &  bits per $k$-mer &  bits per $k$-mer \\
       & &  different values of $r$ & single value of $r$ \\\hline\hline
16 & 5.254, 3.541, 4.981, 8.653 & 8.336 &  8.556 \\
\hline
32 & 5.383, 3.899, 5.318, 9.108 & 8.404 & 8.664\\
\hline
64 & 5.572, 4.452, 5.681, 9.108 & 8.512 & 8.824\\
\hline
128 & 5.786, 5.108, 6.109, 9.109 & 8.669 & 9.045\\
\hline
\end{tabular}
\end{center}
\caption{Estimated memory occupation for the case of different values
  of $r$ vs. single value of $r$, for 4 Bloom filters ($t=4$). Numbers
  in the second column represent values of $r_i$ on which the minimum
  is achieved. For the case of single $r$, its value is shown in
  Table~\ref{table1}. }
\label{table3}
\end{table}

\subsection{Query distribution among filters} \label{sec:estim_distr}
The query algorithm of Lemma~\ref{mainlemma} simply queries Bloom
filters $B_1,\ldots,B_t$ successively as long as the returned answer
is positive. The query time then directly depends on the number of
filters applied before getting a negative answer. Therefore, it is
instructive to analyze how the query frequencies to different filters
are distributed when performing a graph traversal.  We provide such an
analysis in this section.

We analyze query frequencies during an exhaustive traversal of the de
Bruijn graph, when each true node is visited exactly once. We assume
that each time a true node is visited, all its eight potential
neighbors are queried, as there is no other way to tell which of those
neighbors are real. Note however that this assumption does not take
into account structural properties of the de Bruin graph, nor any
additional statistical properties of the genome (such as genomic word
frequencies).

For a filter $B_i$, we want to estimate the number of queried $k$-mers
resolved by $B_i$ during the traversal, that is queries on which $B_i$
returns
\textit{no}. This number is the difference of the number of queries
submitted to $B_i$ and the number of queries for which $B_i$ returns
\textit{yes}. Note that the queries submitted to $B_i$ are precisely
those on which the previous filter $B_{i-1}$ returns \textit{yes}.

If the input set $T_0$ contains $N$ $k$-mers, then the number of
queries in a graph traversal is $8N$, since for each true node each of
its $8$ potential neighbors are queried. Moreover, about $2N$ queries
correspond to true $k$-mers, as we assume that most of the graph nodes
have two true neighbors. Filter $B_1$ will return \textit{yes} on $2N
+ 6c^rN$ queries, corresponding to the number of true and false
positives respectively. For an arbitrary $i$, filter $B_i$
returns \textit{yes} precisely on the $k$-mers inserted to $B_i$
(i.e. $k$-mers $B_i$ is built on), and the $k$-mers which are inserted
to $B_{i+1}$ (which are the critical false positives for $B_i$). The
counts then easily follow from the analysis of
Section~\ref{sec:memory_time}.

\begin{table}[htbp]
\centering
\begin{tabular}{|c|c|c|c|c|}
\hline
 & $B_1$ &  $B_2$ &  $B_3$& $B_4$ \\
\hline
nb of queries  & $8N$ & $(2 + 6c^r)N$ & $(6c^r + 2c^r)N$ & $(2c^r + 6c^{2r})N$ \\
\hline
queries returning \textit{yes} & $(2 + 6c^r)N$ & $(6c^r + 2c^r)N$ & $(2c^r + 6c^{2r})N$ & $(6c^{2r} + 2c^{2r})N$ \\
\hline
queries returning \textit{no} & $(6 - 6c^r)N$ & $(2 - 2c^r)N$ & $(6c^r - 6c^{2r})N$ & $(2c^r - 2c^{2r})N$\\
\hline
resolved queries & $69.57\%$ & $23.19\%$ & $5.04\%$ & $1.68\%$\\
\hline
\end{tabular}
\caption{Estimations of the number of queries made to filters
  $B_1, B_2, B_3$, $B_4$ in the case of infinite number of
  filters. Last row: fraction of queries resolved by each filter,
  estimated for the optimal value $r=5.464$.}
\label{table:theordistrib}
\end{table}

Table~\ref{table:theordistrib} provides counts for the first four
filters, together with the estimated fraction of $k$-mers resolved by
each filter (last row), for the case of infinite number of
filters. The data shows that $99.48\%$ of all $k$-mers are resolved by
four filters. This suggests that a very small number of filters should
be sufficient to cover a vast majority of $k$-mers. Furthermore,
Table~\ref{table:theordistribcount} shows data for $1$-, $2$- and
$4$-filter setups, this time with the optimal value of $r$ for each
case. Even two filters are already sufficient to reduce the accesses
to $T_2$ to $2.08\%$. In case of four filters, $99.7\%$ of $k$-mers
are resolved before accessing $T_4$.

\begin{table}[htbp]
\centering
\begin{tabular}{|c|c|c|c|c|c|c|}
\hline
value of $t$ & $r$ & $B_1$ & $B_2$ & $B_3$ & $B_4$ & $T_t$ \\
\hline
$1$ & $11.44$ & $74.70\%$ & $0$ & $0$ & $0$ & $25.3\%$\\
\hline
$2$ & $8.060$ & $73.44\%$ & $24.48\%$ & $0$ & $0$ & $2.08\%$ \\
\hline
$4$ & $6.049$ & $70.90\%$ & $23.63\%$ & $3.88\%$ & $1.29\%$ & $0.3\%$ \\
\hline
\end{tabular}
\caption{Fractions of queries resolved by each filter for, $1$, $2$
  and $4$ filters. Estimations have been computed for $k=32$ and
  optimal values of $r$ shown in the second column.  Last column shows
  the fraction of queries resolved at the last step, by testing
  against the explicitly stored set
  $T_t$.}
\label{table:theordistribcount}
\end{table}

\section{Experimental results}
\subsection{Construction algorithm} \label{sec:imple}
In practice, constructing a cascading Bloom filter for a real-life
read set is a computationally intensive step. To perform it on a
commonly-used computer, the implementation makes an essential use of
external memory. Here we give a short description of the construction
algorithm for up to four Bloom filters. Extension for larger number of
filters is straightforward.

We start from the input set $T_0$ of $k$-mers written on disk.  We
build the Bloom filter $B_1$ of appropriate size by inserting elements
of $T_0$ successively.  Next, all possible extensions of each $k$-mer
in $T_0$ are queried against $B_1$, and those which return true are
written to the disk.  Then, in this set only the $k$-mers absent from
$T_0$ are kept, i.e. we perform a set difference from $T_0$. We cannot
afford to load $T_0$ entirely in memory, so we partition $T_0$ and
perform the set difference in several iterations, loading only one
partition of $T_0$ each time.  This results in the set $T_1$ of
critical false positives, which is also kept on disk. Up to this
point, the procedure is identical to that of
\cite{Chikhi12}.

Next, we insert all $k$-mers from $T_1$ into $B_2$ and to obtain
$T_2$, we check for each $k$-mer in $T_0$ if a query to $B_2$ returns
true. This results in the set $T_2$, which is directly stored on
disk. Thus, at this point we have $B_1$, $B_2$ and, by loading $T_2$
from the disk, a complete representation for $t = 2$. In order to
build the data structure for $t = 4$, we continue this process, by
inserting $T_2$ in $B_3$ and retrieving (and writing directly on disk)
$T_3$ from $T_1$ (stored on disk). It should be noted that to obtain
$T_i$ we need $T_{i-2}$, and by always directly storing it on disk we
guarantee not to use more memory than the size of the final structure.
The set $T_t$ (that is, $T_1$, $T_2$ or $T_4$ in our experiments) is
represented as a sorted array and is searched by a binary search. We
found this implementation more efficient than a hash table.

\subsection{Implementation and experimental setup}\label{exp-setup}
We implemented our method using {\sc Minia} software (\cite{Chikhi12})
and ran comparative tests for $2$ and $4$ Bloom filters ($t = 2,4$).
Note that since the only modified part of {\sc Minia} was the
construction step and the $k$-mer membership queries, this allows us
to precisely evaluate our method against the one of
\cite{Chikhi12}.

The first step of the implementation is to retrieve the list of
$k$-mers that appear more than $d$ times using DSK (\cite{Dsk}) -- a
constant memory streaming algorithm to count $k$-mers.  Note, as a
side remark, that performing counting allows us to perform off-line
deletions of $k$-mers. That is, if at some point of the scan of the
input set of $k$-mers (or reads) some of them should be deleted, it is
done by a simple decrement of the counter.

Assessing the query time is done through the procedure of graph
traversal, as it is implemented in
\cite{Chikhi12}. Since the procedure is identical and
independent on the data structure, the time spent on graph traversal
is a faithful estimator of the query time.

We compare three versions: $t=1$ (i.e. the version of
\cite{Chikhi12}), $t=2$ and $t=4$. For convenience, we
define $1$ Bloom, $2$ Bloom and $4$ Bloom as the versions with $t = 1,
2$ and $4$, respectively.

\subsection{\emph{E. coli} dataset, varying $k$} \label{subsec:varying_k}
In this set of tests, our main goal was to evaluate the influence of
the $k$-mer size on principal parameters: size of the whole data
structure, size of the set $T_t$, graph traversal time, and time of
construction of the data structure.  We retrieved 10M \emph{E. coli}
reads of 100bp from the \emph{Short Read Archive} (ERX008638) without
read pairing information and extracted all $k$-mers occurring at least
two times. The total number of $k$-mers considered varied, depending
on the value of $k$, from 6,967,781 ($k=15$) to 5,923,501 ($k = 63$).
We ran each version, 1 Bloom (\cite{Chikhi12}), 2 Bloom and 4 Bloom,
for values of $k$ ranging from $16$ to $64$. The results are shown in
Fig.~\ref{fig:ecoli}.

The total size of the structures in bits per stored $k$-mer, i.e. the
size of $B_1$ and $T_1$ (respectively, $B_1, B_2$,$T_2$ or $B_1, B_2,
B_3, B_4$,$T_4$) is shown in Fig.~\ref{fig:bits_per_kmer}. As
expected, the space for 4 Bloom filters is the smallest for all values
of $k$ considered, showing a considerable improvement, ranging from
32\% to 39\%, over the version of \cite{Chikhi12}. Even the version
with just 2 Bloom filters shows an improvement of at least 20\% over
\cite{Chikhi12}, for all values of $k$. Regarding the influence of the
$k$-mer size on the structure size, we observe that for 4 Bloom
filters the structure size is almost constant, the minimum value is
8.60 and the largest is 8.89, an increase of only 3\%. For 1 and 2
Bloom the same pattern is seen: a plateau from $k = 16$ to $32$, a
jump for $k=33$ and another plateau from $k = 33$ to $64$. The jump at
$k=32$ is due to switching from $64$-bit to $128$-bit representation
of $k$-mers in the table $T_t$.

The traversal times for each version is shown in
Fig.~\ref{fig:traversal}.  The fastest version is 4 Bloom, showing an
improvement over \cite{Chikhi12} of 18\% to 30\%, followed by 2 Bloom.
This result is surprising and may seem counter-intuitive, as we have
four filters to apply to the queried $k$-mer rather than a single
filter as in \cite{Chikhi12}. However, the size of $T_4$ (or even
$T_2$) is much smaller than $T_1$, as the size of $T_i$'s decreases
exponentially. As $T_t$ is stored in an array, the time economy in
searching $T_4$ (or $T_2$) compared to $T_1$ dominates the time lost
on querying additional Bloom filters, which explains the overall gain
in query time.

As far as the construction time is concerned
(Fig.~\ref{fig:construction}), our versions yielded also a faster
construction, with the 4 Bloom version being 5\% to 22\% faster than
that of \cite{Chikhi12}. The gain is explained by the time required
for sorting the array storing $T_t$, which is much higher for $T_0$
than for $T_2$ or $T_4$. However, the gain is less significant here,
and, on the other hand, was not observed for bigger datasets (see
Section~\ref{human}).

\begin{figure}[Htbp]
  \center
  \subfloat[]{\includegraphics[width=0.5\linewidth]{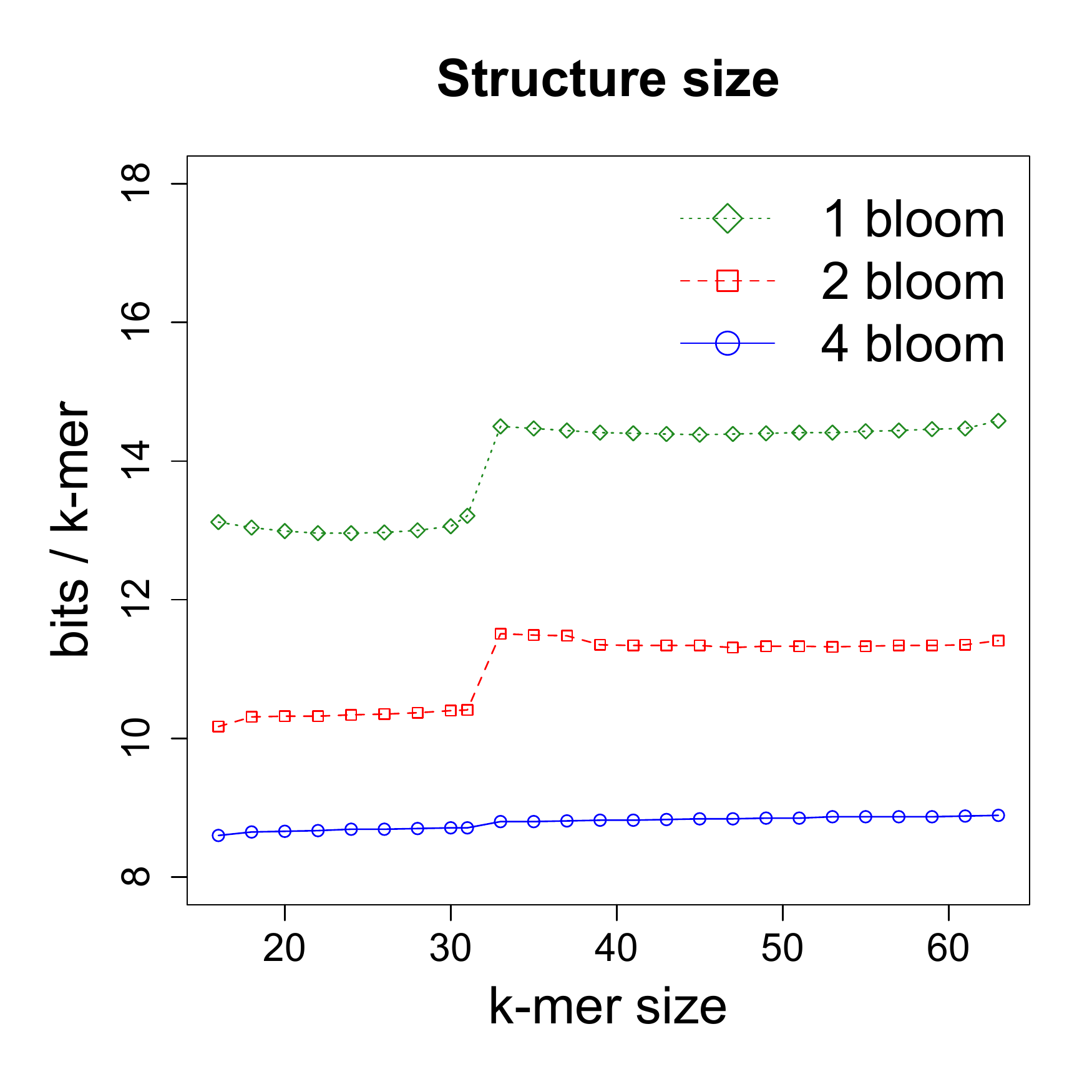}\label{fig:bits_per_kmer}}
  \subfloat[]{\includegraphics[width=0.5\linewidth]{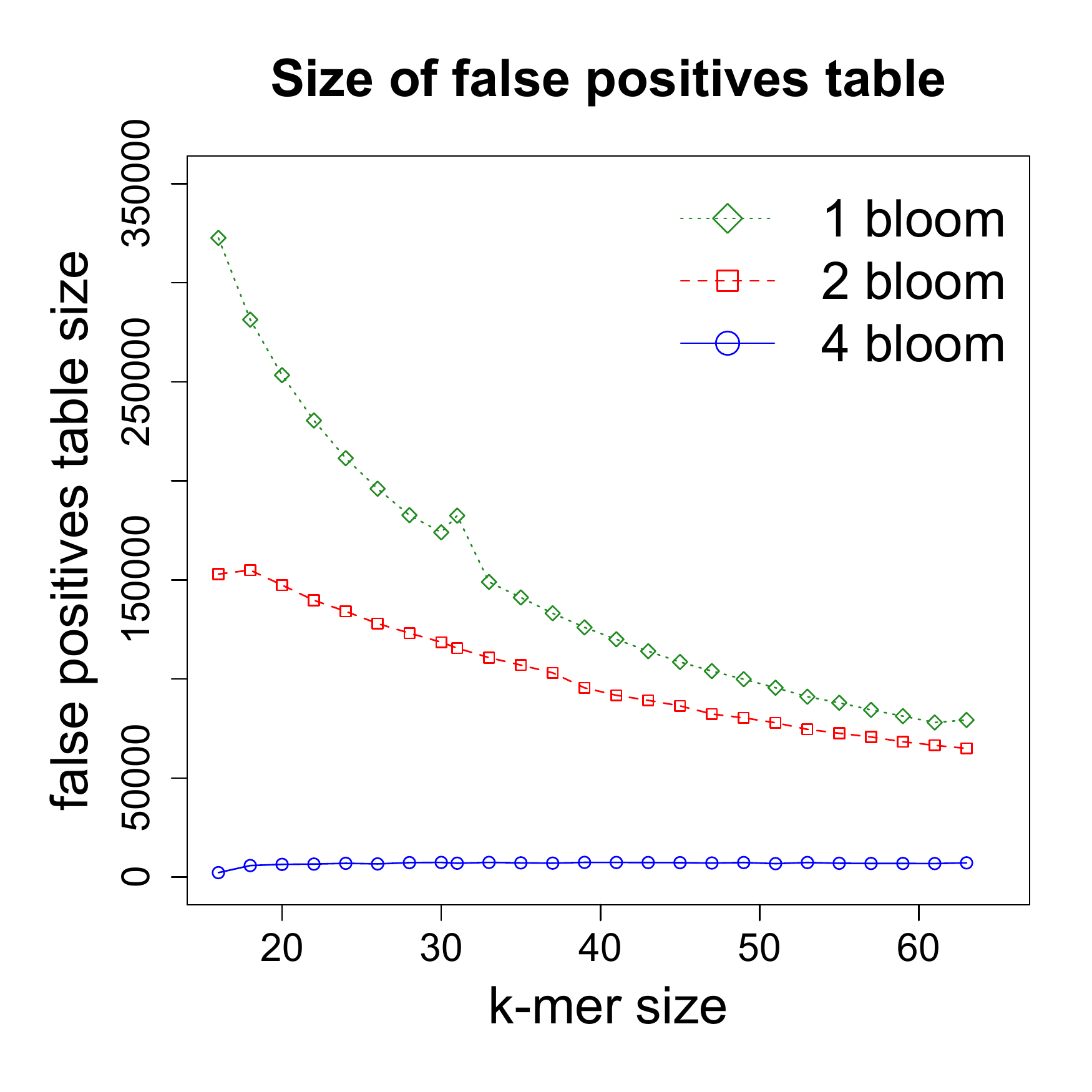}\label{fig:false_positive}} \\
  \subfloat[]{\includegraphics[width=0.5\linewidth]{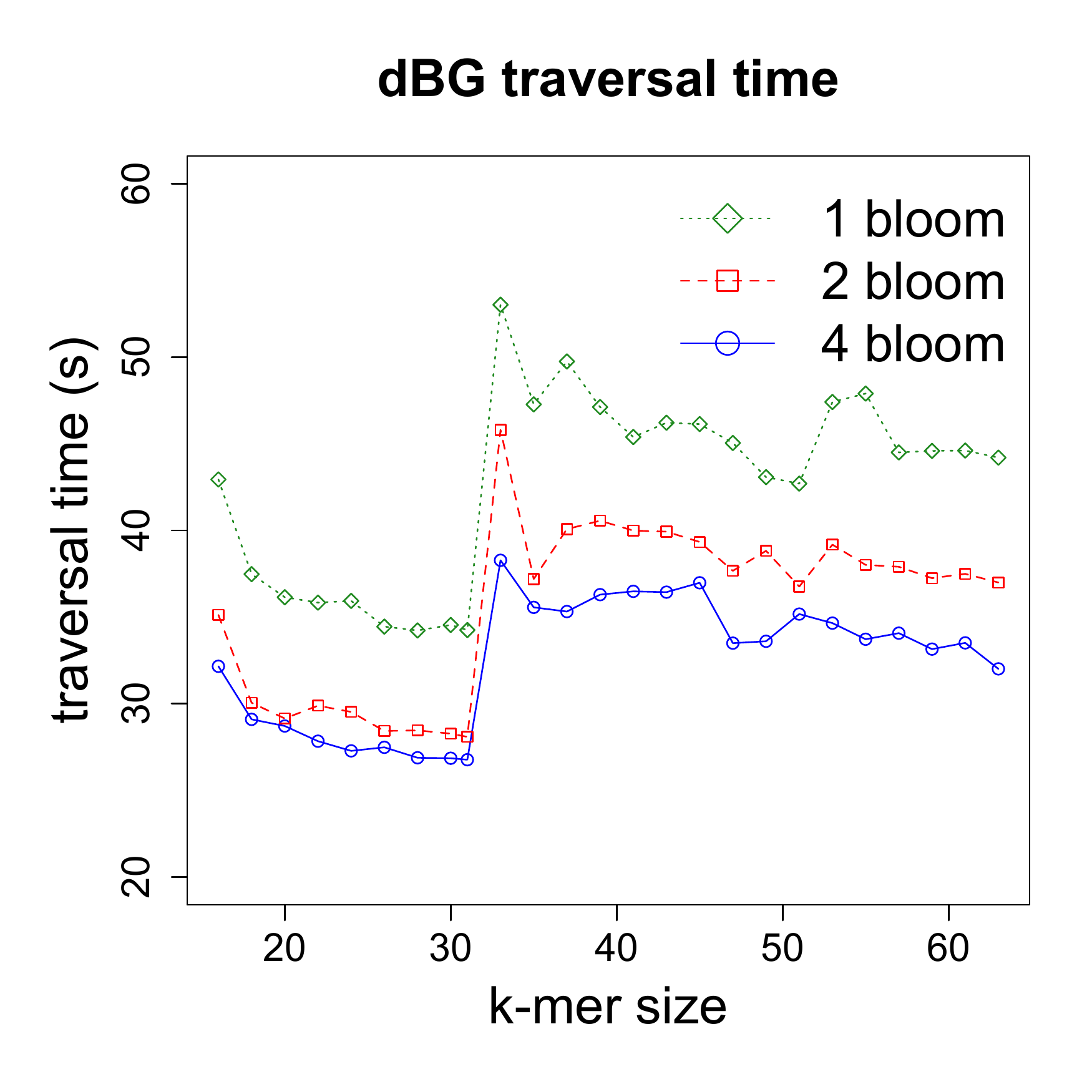} \label{fig:traversal}}
  \subfloat[]{\includegraphics[width=0.5\linewidth]{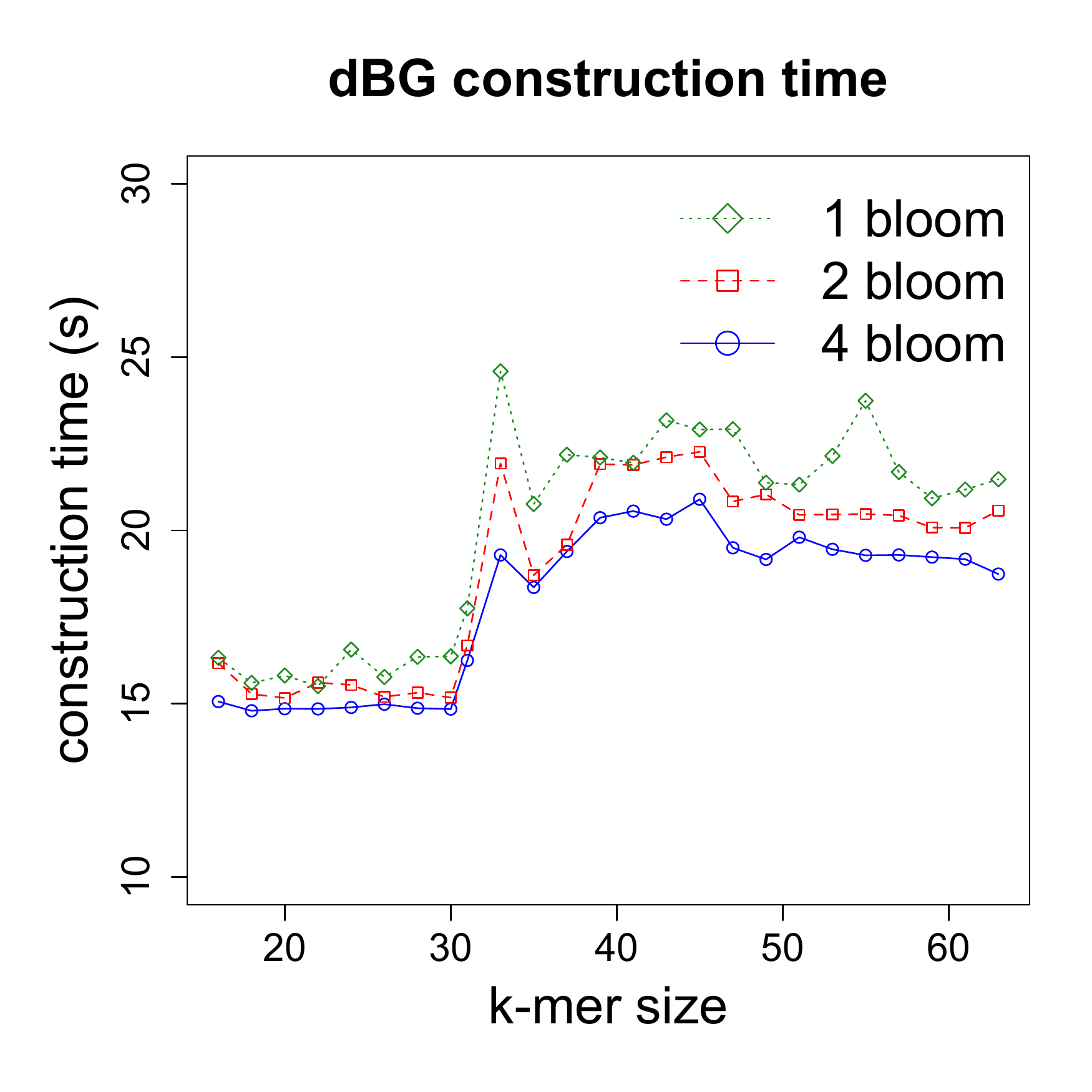}\label{fig:construction}} 
  \caption{Results for 10M E.coli reads of 100bp using several values
    of $k$. The \emph{1 Bloom} version corresponds to the one
    presented in \cite{Chikhi12}. (a) Size of the
    structure in bits used per $k$-mer stored. (b)
    Number of false positives stored in $T_1$, $T_2$ or $T_4$ for 1, 2
    or 4 Bloom filters, respectively. (c) De Bruijn graph
    construction time, excluding $k$-mer counting step. (d) De Bruijn
    graph traversal time, including branching $k$-mer indexing. }  \label{fig:ecoli}
\end{figure}

\subsection{\emph{E. coli} dataset, varying coverage} \label{subsec:cov}
From the complete \emph{E. coli} dataset ($\approx$44M reads) from the
previous section, we selected several samples ranging from 5M to 40M
reads in order to assess the impact of the coverage on the size of the
data structures. This strain \emph{E. coli} (K-12 MG1655) is estimated
to have a genome of 4.6M bp~\cite{ecoli}, implying that a sample of 5M
reads (of 100bp) corresponds to $\approx$100X coverage. We set $d = 3$
and $k = 27$. The results are shown in
Fig.~\ref{fig:ecoli_coverage}. As expected, the memory consumption per
$k$-mer remains almost constant for increasing coverage, with a slight
decrease for 2 and 4 Bloom. The best results are obtained with the 4
Bloom version, an improvement of 33\% over the 1 Bloom version of
\cite{Chikhi12}. On the other hand, the number of distinct $k$-mers
increases markedly (around 10\% for each 5M reads) with increasing
coverage, see Fig.~\ref{fig:ecoli_cov:kmers}.  This is due to
sequencing errors: an increase in coverage implies more errors with
higher coverage, which are not removed by our cutoff $d = 3$. This
suggests that the value of $d$ should be chosen according to the
coverage of the sample. Moreover, in the case where read qualities are
available, a quality control pre-processing step may help to reduce
the number of sequencing errors.

\begin{figure}[Htbp]
  \center
  \subfloat[]{\includegraphics[width=0.5\linewidth]{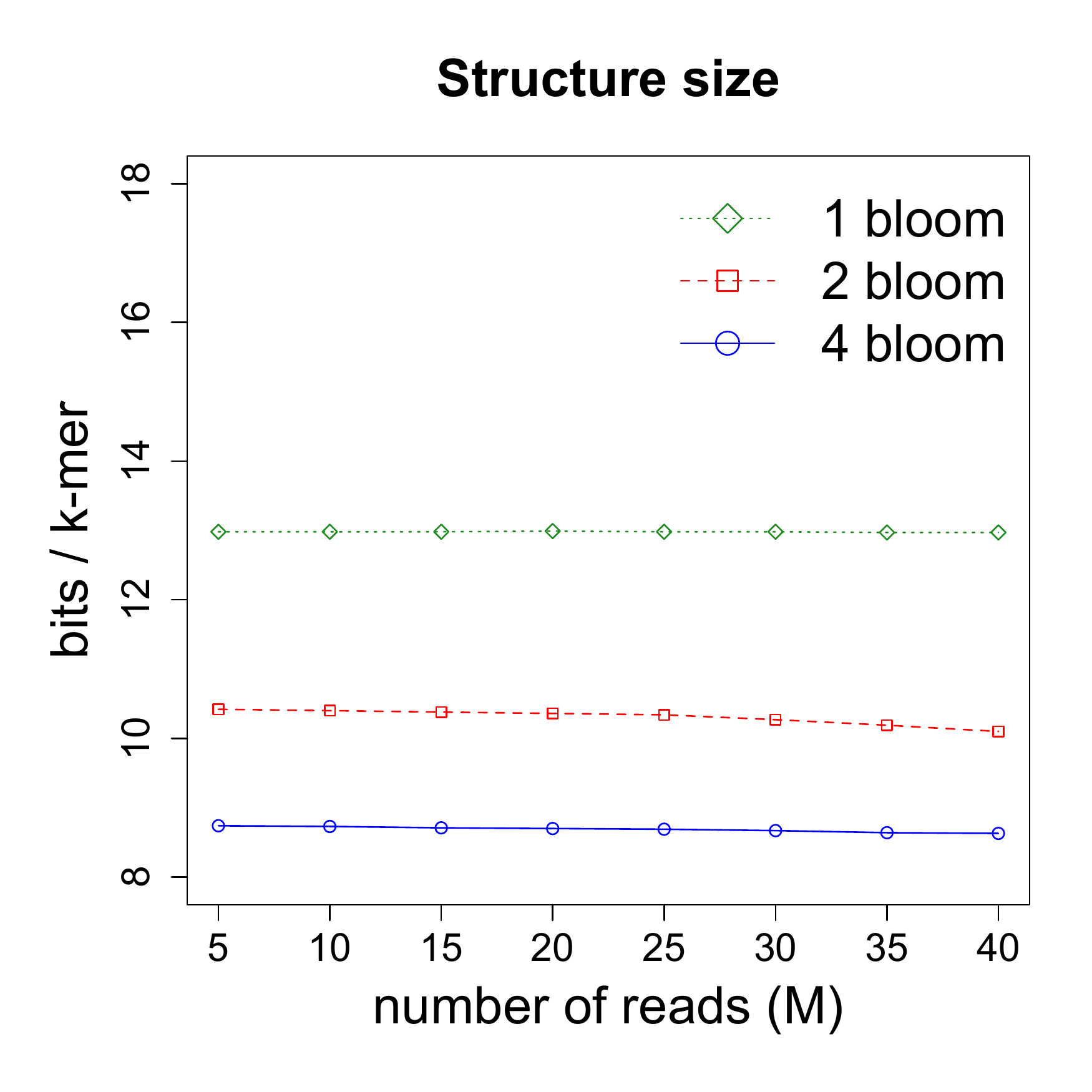}}
  \subfloat[]{\includegraphics[width=0.5\linewidth]{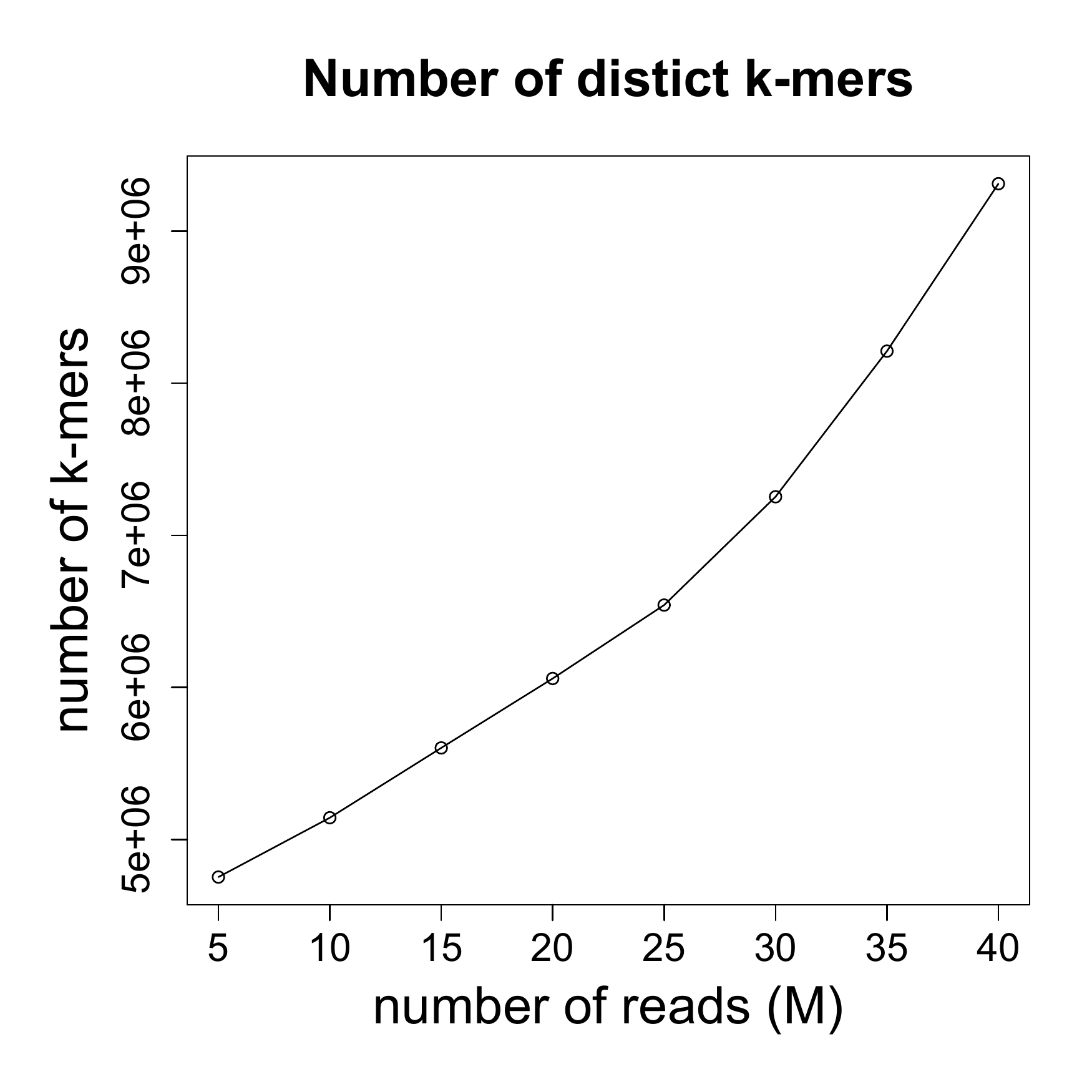} \label{fig:ecoli_cov:kmers}}
  \caption{Results for \emph{E.coli} reads of 100bp using $k =
    27$. The \emph{1 Bloom} version corresponds to the one presented
    in \cite{Chikhi12}. (a) Size of the structure in bits used per
    $k$-mer stored. (b) Number of distinct
    $k$-mers.}\label{fig:ecoli_coverage}
\end{figure}

\subsection{\emph{E. coli} dataset, query statistics}
In this set of tests we used the dataset of
Section~\ref{subsec:varying_k} to experimentally evaluate how the
queries are distributed among the Bloom filters. We ran the graph
traversal algorithm for each version, 1 Bloom (\cite{Chikhi12}), 2
Bloom and 4 Bloom, using values of $k$ ranging from $16$ to $64$ and
retrieved the number of queries resolved in each Bloom filter and the
table $T_t$. The results are shown in Fig.~\ref{fig:ecoli_query}. The
plots indicate that, for each version, the query distribution among
the Bloom filters is approximately invariant to the value of
$k$. Indeed, on average 74\%, 73\% and 70\% of the queries are
resolved in $B_1$ for the 1, 2 and 4 Bloom version, respectively, and
the variance is smaller than 0.01\% in each case. For the 4 Bloom
version, 70\%, 24\%, 4\%, 1\% and 0.2\% of the queries are resolved in
$B_1$, $B_2$, $B_3$, $B_4$ and $T_4$, respectively, showing that the
values estimated theoretically in Section~\ref{sec:estim_distr} (the
last row of Table~\ref{table:theordistribcount}) are very
precise. Furthermore, as a query to a Bloom filter is faster than to
$T_1$ and the majority of the queries to 4 and 2 Bloom versions, 94\%
and 95\% respectively, are resolved in the first two filters, it is
natural that on average queries to 1 Bloom version are slower than to
2 and 4 Bloom versions, corroborating the results of
Section~\ref{subsec:varying_k}.

\begin{figure}[htbp]
  \center
  \subfloat[]{\includegraphics[width=0.5\linewidth]{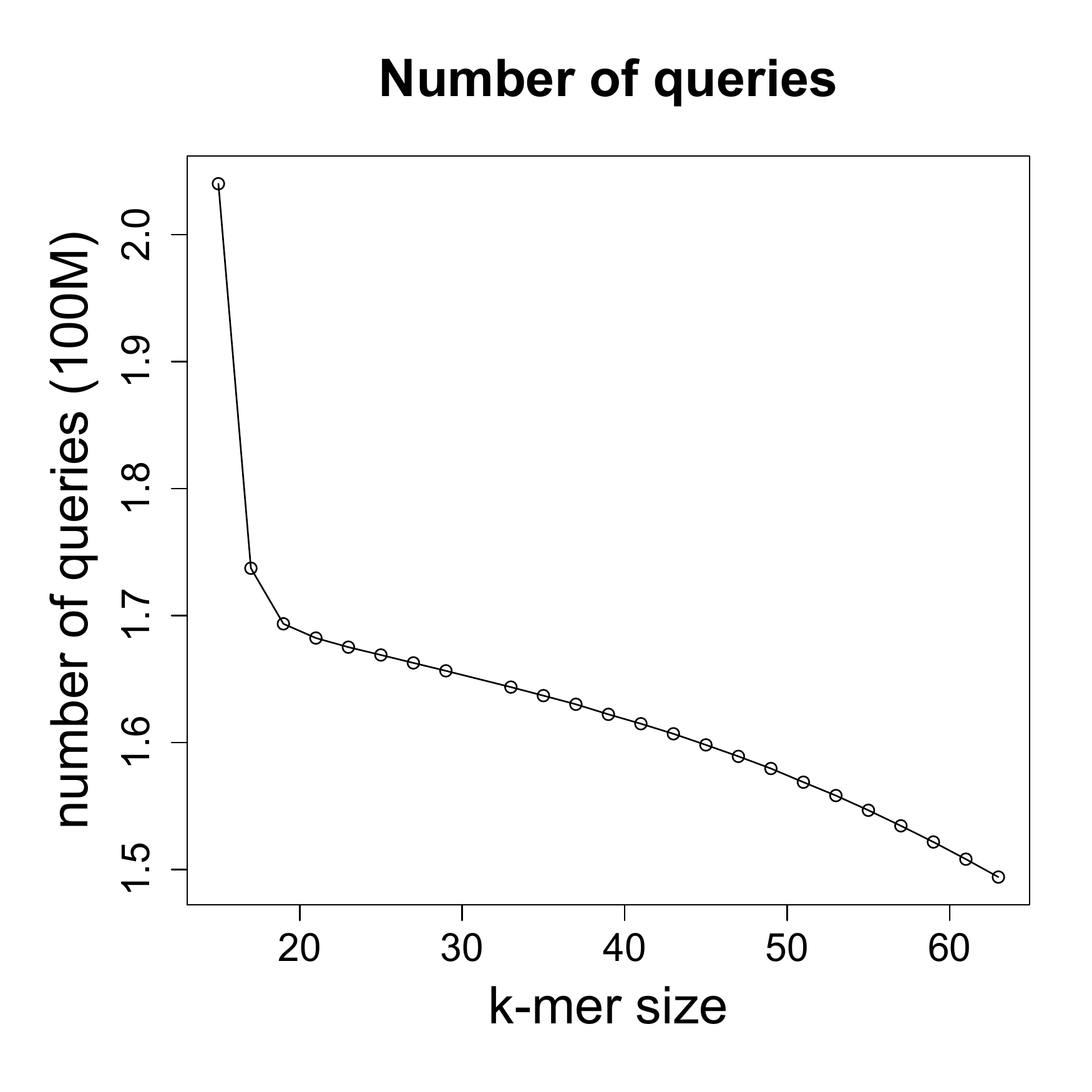}\label{fig:number_query}}
  \subfloat[]{\includegraphics[width=0.5\linewidth]{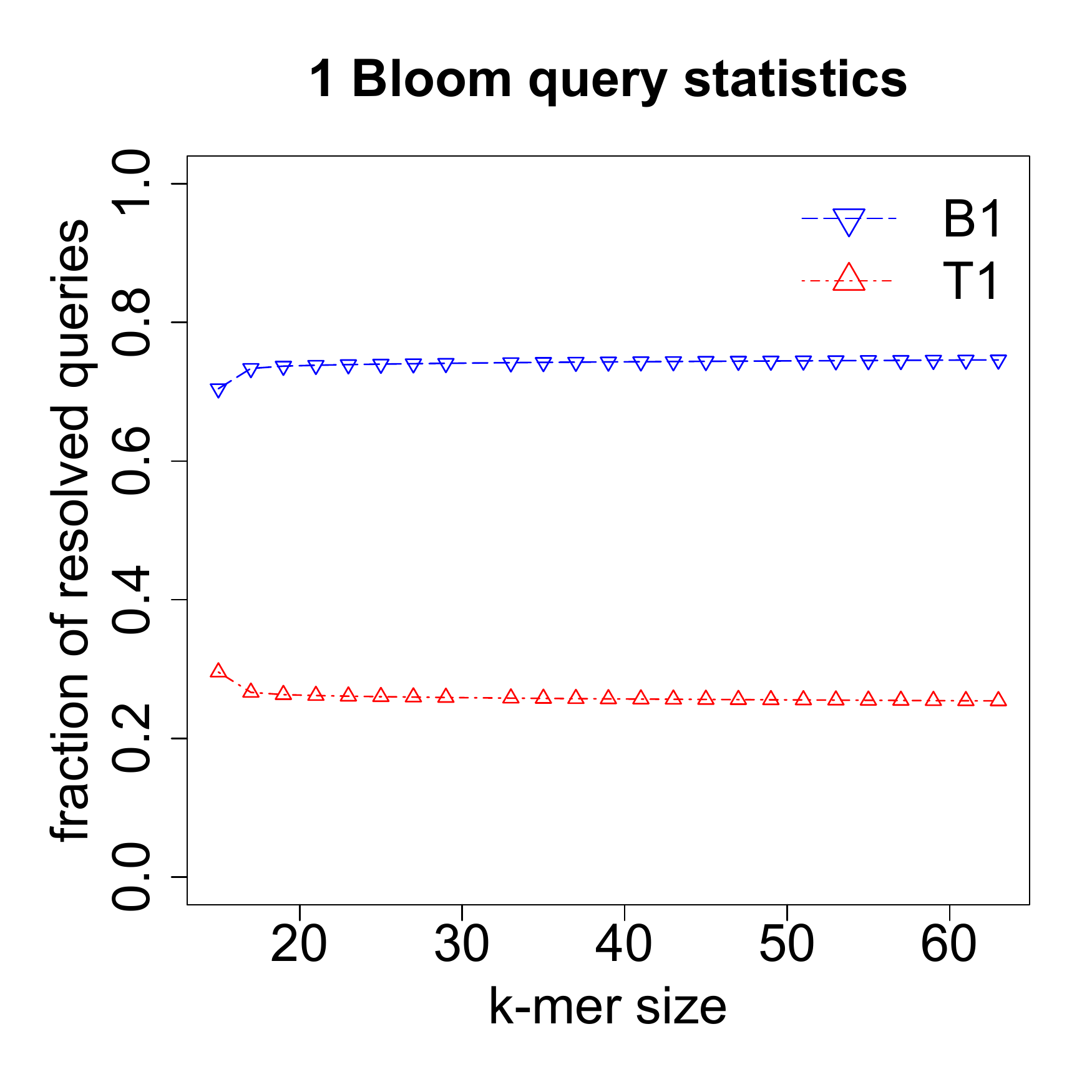}\label{fig:1bloom_query}} \\
  \subfloat[]{\includegraphics[width=0.5\linewidth]{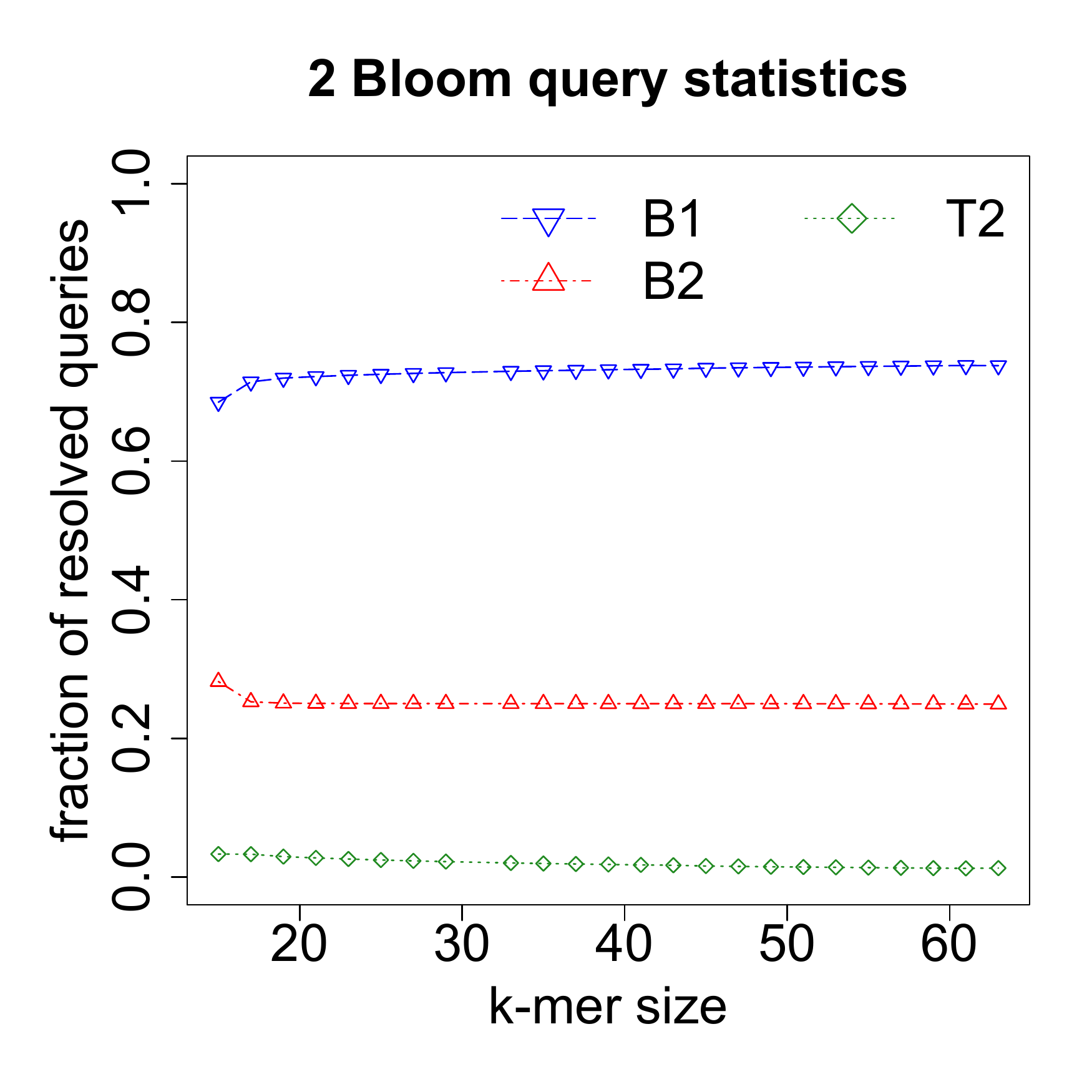}\label{fig:2bloom_query}}
  \subfloat[]{\includegraphics[width=0.5\linewidth]{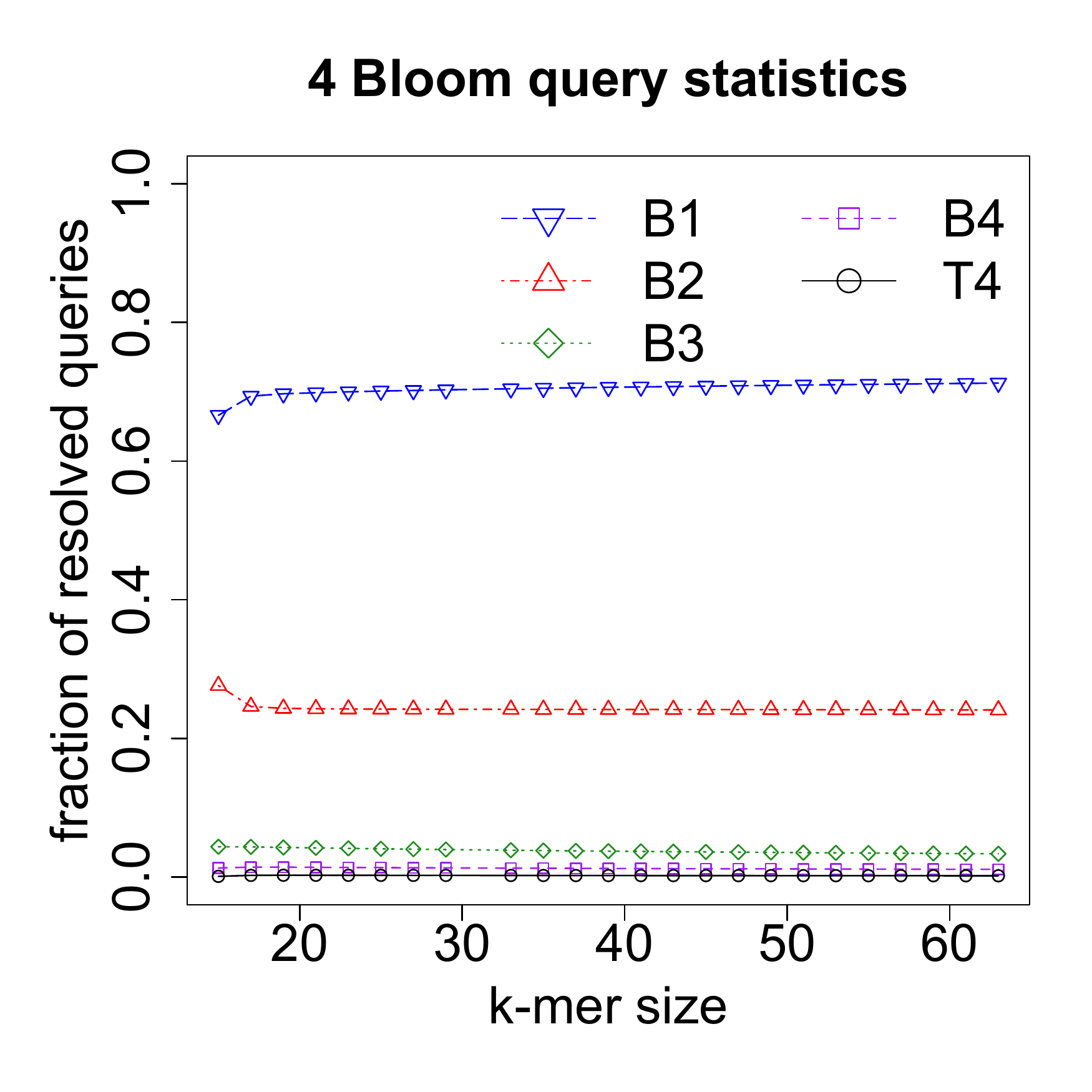}\label{fig:4bloom_qeury}} 
  \caption{Query statistics results for 10M E.coli reads of 100bp
    using several values of $k$. The \emph{1 Bloom} version
    corresponds to the one presented in
    \cite{Chikhi12}. (a) Total number of queries
    performed, for each value of $k$, during a graph traversal. (b)
    Fraction of resolved queries in $B_1$ and $T_1$ (1 Bloom version)
    for each value of $k$. (c) Fraction of resolved queries in
    $B_1$,$B_2$ and $T_2$ (2 Bloom version) for each value of $k$. (d)
    Fraction of resolved queries in $B_1$,$B_2$,$B_3$,$B_4$ and $T_4$
    for each value of $k$.}  \label{fig:ecoli_query}
  \end{figure}

\subsection{Human dataset}
\label{human}
We also compared 2 and 4 Bloom versions with the 1 Bloom version of
\cite{Chikhi12} on a large dataset. For that, we retrieved 564M Human
reads of 100bp (SRA: SRX016231) without pairing information and
discarded the reads occurring less than 3 times.  The dataset
corresponds to $\approx$17X coverage. A total of 2,455,753,508
$k$-mers were indexed. We ran each version, 1 Bloom (\cite{Chikhi12}),
2 Bloom and 4 Bloom with $k = 23$. The results are shown in
Table~\ref{tab:human}.

The results are in general consistent with the previous tests on
\emph{E.coli} datasets. There is an improvement of 34\% (21\%) for the
4 Bloom (2 Bloom) in the size of the structure. The graph traversal is
also 26\% faster in the 4 Bloom version. However, in contrast to the
previous results, the graph construction time {increased} by 10\% and
7\% for 4 and 2 Bloom versions respectively, when compared to the 1
Bloom version. This is due to the fact that disk writing/reading
operations now dominate the time for the graph construction, and 2 and
4 Bloom versions generate more disk accesses than 1 Bloom.  As stated
in Section~\ref{sec:imple}, when constructing the 1 Bloom structure,
the only part written on the disk is $T_1$ and it is read only once to
fill an array in memory. For 4 Bloom, $T_1$ and $T_2$ are written to
the disk, and $T_0$ and $T_1$ are read at least one time each to build
$B_2$ and $B_3$. Moreover, since the size coefficient of $B_1$
reduces, from $r = 11.10$ in 1 Bloom to $r = 5.97$ in 4 Bloom, the
number of false positives in $T_1$ increases.

\begin{table}[htbp]
\begin{center}
\begin{tabular}{|c|c|c|c|}
\hline
Method                          & 1 Bloom & 2 Bloom & 4 Bloom \\
\hline \hline
Construction time (s)           &  40160.7 & 43362.8 & 44300.7 \\
\hline
Traversal time (s)              &  46596.5 & 35909.3 & 34177.2 \\
\hline
$r$ coefficient                      &  11.10 &  7.80    & 5.97 \\
\hline
\multirow{4}{*}{Bloom filters size (MB)}  &  $B_1 = 3250.95$ & $B_1 = 2283.64$ & $B_1 = 1749.04$  \\
                                          &                  & $B_2 = 323.08$  & $B_2 = 591.57$ \\
                                          &                  &                 & $B_3 = 100.56$  \\
                                          &                  &                 & $B_4 = 34.01$ \\
\hline
False positive table size (MB)  &  $T_1 = 545.94$ & $T_2 = 425.74$ & $T_4 = 36.62$ \\
\hline
Total size (MB)                 &  3796.89 & 3032.46 & 2511.8 \\
\hline
\bf Size (bits/$k$-mer)             & \bf 12.96 & \bf 10.35 & {\bf 8.58} \\
\hline
\end{tabular}
\end{center}
\caption{Results of 1, 2 and 4 Bloom filters version for 564M Human
  reads of 100bp using $k = 23$. The \emph{1 Bloom} version
  corresponds to the one presented in
  \cite{Chikhi12}. }\label{tab:human}
\end{table}

\section{Discussion and conclusions}
Using cascading Bloom filters for storing de Bruijn graphs has clear
advantage over the single-filter method of \cite{Chikhi12}. In terms
of memory consumption, which is the main parameter here, we obtained
an improvement of around 30\%-40\% in all our experiments. Our data
structure takes 8.5 to 9 bits per stored $k$-mer, compared to 13 to 15
bits by the method of \cite{Chikhi12}.  This confirms our analytical
estimations. The above results were obtained using only four filters
and are very close to the estimated optimum (around 8.4 bits/$k$-mer)
produced by the infinite number of filters. An interesting
characteristic of our method is that the memory grows insignificantly
with the growth of $k$, even slower than with the method of
\cite{Chikhi12}.  Somewhat surprisingly, we also obtained a
significant decrease, of order 20\%-30\%, of query time. The
construction time of the data structure varied from being 10\% slower
(for the human dataset) to 22\% faster (for the bacterial dataset).

As stated previously, another compact encoding of de Bruijn graphs has
been proposed in \cite{Bowe12}, however no implementation of the
method was made available. For this reason, we could not
experimentally compare our method with the one of \cite{Bowe12}. We
remark, however, that the space bound of \cite{Bowe12} heavily
depends on the number of reads (i.e. coverage), while in our case, the
data structure size is almost invariant with respect to the coverage
(Section~\ref{subsec:cov}).

An interesting open question is whether the Bloom filter construction
can be made online, so that new $k$-mers (reads) can be inserted
without reconstructing the whole data structure from scratch.  Note
that the presented construction (Section~\ref{sec:imple}) is
inherently off-line, as all $k$-mers should be known before the data
structure is built.

Another interesting prospect for further possible improvements of our
method is offered by \cite{Porat09}, where an efficient replacement to
Bloom filter was introduced. The results of \cite{Porat09} suggest
that we could hope to reduce the memory to about $5$ bits per
$k$-mer. However, there exist obstacles on this way: an implementation
of such a structure would probably result in a significant
construction and query time increase.

\chapter*{Conclusion and perspectives}
\markboth{\bf Conclusion and perspectives}{\bf Conclusion and
  perspectives} \addstarredchapter{Conclusion and perspectives}
In this thesis, we presented \ks, a time and memory efficient method
to identify variations (alternative splicing and genomic
polymorphisms) by locally assembling RNA-seq data without using a
reference genome. The local nature of the \ks strategy allows to avoid
some of the difficulties faced by standard full-length transcriptome
assemblers, namely solving an ill-posed problem often formulated as a
NP-hard optimization problem. As a result, we can avoid an extensive
use of heuristics and, thus, obtain an overall more sensitive method
with stronger theoretical guarantees. A lot of effort was put in order
to make our method as scalable as possible in order to deal with
ever-increasing volumes of NGS data. We improved the state-of-the-art
de Bruijn graph construction and representation in an effort to reduce
the memory footprint of \ks. We also developed a new time-efficient
approach to list bubbles in de Bruijn graphs in order to reduce the
running time of our method.

The techniques we developed while studying the bubble listing problem
turned out to be useful in other enumeration contexts, namely cycle
listing and the $K$-shortest paths problem. The classical problem of
listing cycles in a graph has been studied since the early 70s,
however, as shown here, the best algorithm for undirected graph is not
optimal. In this thesis, we gave the first optimal algorithm to list
cycles in undirected graphs, along with the first optimal algorithm to
list $st$-paths. The classical $K$-shortest paths problem has been
studied since the early 60s, and the best algorithm for solving it
uses memory proportional to the number of path output, i.e. $K$. In
this thesis, we gave an alternative parameterization of the problem.
For this alternative version, we gave an algorithm that uses memory
linear in the size of the graph, independent of the number of paths
output.

In the past 3 years, \ks has evolved into not only a time and memory
efficient method, but also into an user-friendly \emph{software} for
the bioinformaticians and biologists. We are now a 5-persons team
actively developing \ks, which includes continuously: improving its
robustness, correcting bugs, improving the usability and the
documentation. It is important to highlight that, besides the
algorithmic improvements already mentioned, a lot of effort was put on
improving the \emph{implementation} of \ks, including: parallelization
of certain steps of the pipeline, careful implementation of the data
structures, systematic removal of memory leaks, among others. This
team effort produced a stable and user-friendly software.

\ks has been used in several projects, as evidenced by an average of
250 unique visitors per month to our website. A summary of some of the
projects in which we are directly involved is shown in the figure
below. We are aware that in order to further convince the biologists
that the new alternative splicing events found by \ks are \emph{real}
events, it is desirable to experimentally validate some of them, and
for that we are working in collaboration with D. Auboeuf's group to
validate AS events found by \ks in K562 cell lines. Since the human
genome is known and well annotated, our choice of human cell lines for
these experiments may seem a bit odd. Actually, our goal is to show
that \ks is useful even when a good reference genome is available. For
that, we need to validate AS events found by \ks that are not
annotated and not found by mapping approaches (\cite{Cufflinks})
either, even when there is a good (annotated) reference genome. The
preliminary results are promising: from the randomly selected events
only found by \ks, around 40\% of them were validated, while 30\%,
although not validated, are part of complex events (more than two
isoforms) and the isoform amplified in the experiment matches the AS
event found by \ks but not selected for validation.  The remaining
30\% corresponded to cases where the minor isoform had a relative
abundance of less than 15\%. Although for now we did not manage to
validate these cases, it does not yet mean that they are not
real. Indeed, since the experimental validation is based on RT-PCR, it
may be that the early rounds of the PCR favor the major isoform,
which causes the complete loss of the minor isoform in the final
rounds.  Finally, after clarifying how many of the novel events found
by \ks are real, there still remains the central question whether
these new isoforms are functional or just noise of the splicing
machinery. Our point of view is that an exhaustive description of all
isoforms present in the cell is a good prerequisite to help address
this central question.

\begin{figure}[htbp]
  \includegraphics[width=\linewidth]{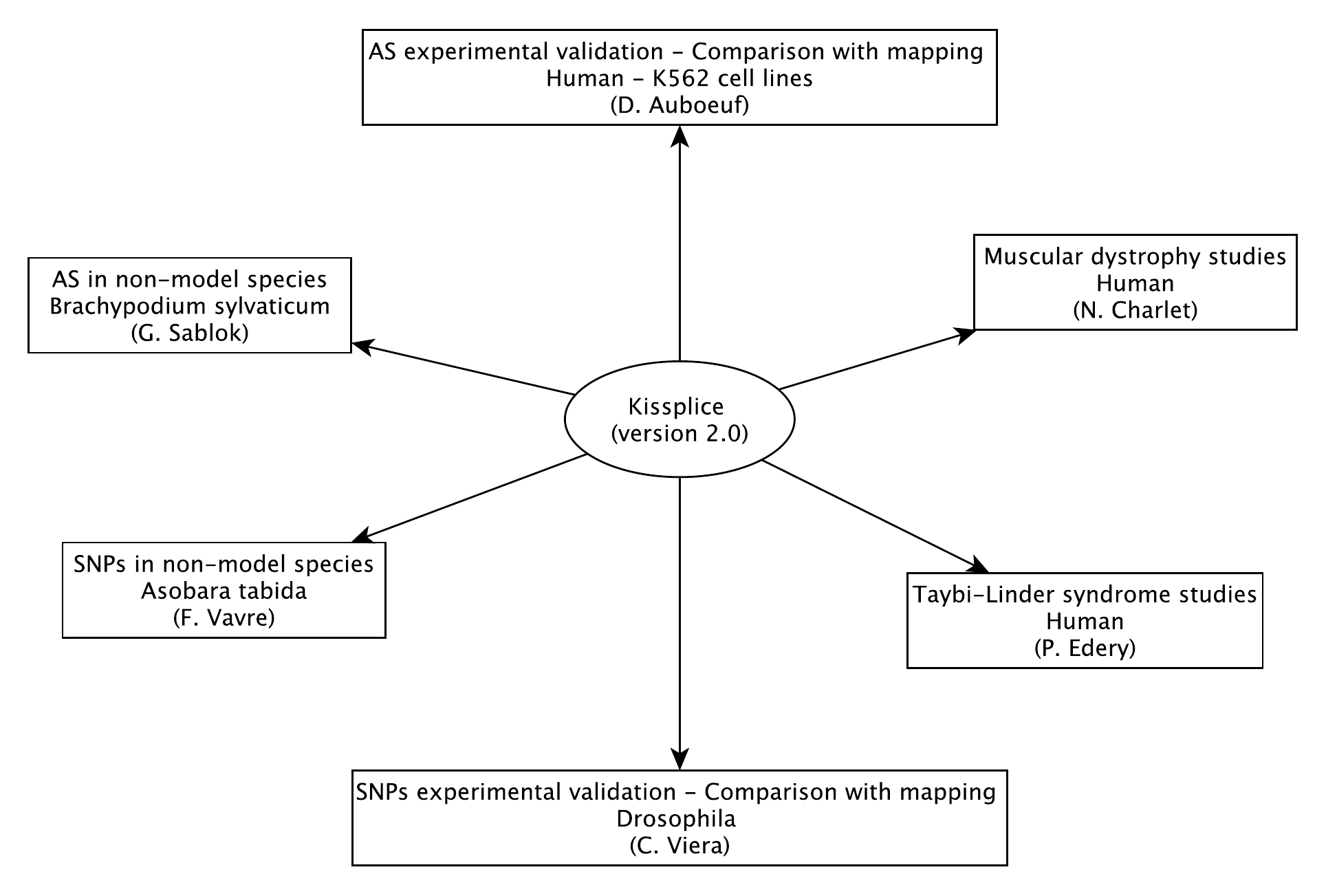} 
  \caption{Projects using \ks.}
\end{figure}

At the end of each chapter, open problems and perspectives were
discussed, providing the reader with an idea of possible extensions of
the methods and techniques presented in this work. For that reason, in
this chapter we focus on the main open problem concerning \ks: the
complex biconnected components (BCCs).

The second step of the \ks algorithm is the biconnected
decomposition. Since each bubble is entirely contained in one BCC,
after the BCC decomposition, the bubbles in each BCC are enumerated
independently. In the ideal case -- a repeat-free genome -- each
connected component in the DBG corresponds to a single gene, as well
as each BCC. In practice, the BCC decomposition works well: the vast
majority of the BCCs are relatively small and contain the sequence of
a single gene (or a family of paralogous genes). However, there is a
small number of large BCCs (often one or two) that contain the
sequences of several unrelated genes, and it is infeasible to
enumerate all the bubbles inside these complex BCCs. The problem is
not the efficiency of the algorithm, but the number of bubbles
satisfying our constraints. There are a huge number of bubbles in the
complex BCCs (usually more than in all the other BCCs together), and
most of them are repeat-associated bubbles.  A manual exploration of a
fraction of the bubbles contained in these complex BCCs led us to
think that they are generated by transposable elements (and to a much
lesser extent to other types of repeats). A transposable element (TE)
is a DNA sequence that can change its position within the genome
through a copy-and-paste or cut-and-paste process (\cite{Wicker07}).
Transposable elements are spread throughout the genome (including in
many transcribed regions).  We believe that old copies of TEs that
invaded the genome a long time ago are responsible for the complex
BCCs.  Since they are old, these copies diverged, however, they still
contain enough sequence similarity to merge several unrelated genes
inside the same BCC. More importantly, they are in transcribed
regions, mostly UTRs.  In human, the transposable elements in the Alu
family alone generate a BCC with millions of bubbles satisfying our
constraints, so it is infeasible to enumerate all of them. The problem
of simply ignoring the complex BCCs (that is what we have been doing
so far), is that they potentially contain true events ``trapped''
inside.

\begin{figure}[htbp]
  \includegraphics[width=\linewidth]{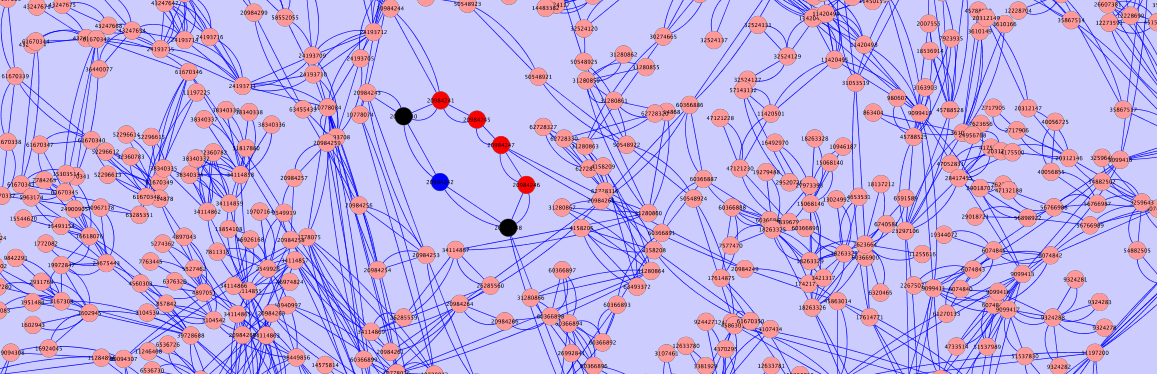}
  \caption{An alternative splicing event (intron retention) in the
    SCN5A gene (human) trapped inside a complex BCC. The switching
    vertices are shown in black.}
\end{figure}

Repeated elements are not a problem restricted to \ks or local
assembly strategies; global transcriptome assemblers are possibly even
more affected. Recall that the second step of the standard
transcriptome heuristic (see Section~\ref{sec:back:modeling}) is to
partition the graph into subgraphs corresponding to genes; in the
presence of repeated elements this is a much harder task. Several
genes are likely to be wrongly assigned to the same subgraph, and as a
result the heuristic is going to produce chimeric transcripts. On the
other hand, if to account for the presence of repeated elements, the
heuristic adopts a more stringent graph partition strategy, more genes
are likely to be split into several subgraphs, resulting in
transcripts only partially assembled.

A first step towards a solution to this issue is to solve the
following problem: given a DBG built from RNA-seq reads, identify the
subgraph corresponding to the transposable elements (they are not the
only repeated elements, but we believe they are the main source of
problems). Observe that, unlike genomic NGS data where it is possible
to use the coverage of a vertex as a proxy for uniqueness of that
sequence in the genome, in RNA-seq data it is not obvious how to
determine the uniqueness of the sequence corresponding to a vertex. A
solution to this problem would be useful in \ks as well as in
full-length transcriptome assemblers.

 \cleardoublepage
\phantomsection 
 \addstarredchapter{Bibliography}
\renewcommand{\leftmark}{\bf Bibliography}
\renewcommand{\rightmark}{\bf Bibliography}
\bibliographystyle{apalike}
\bibliography{thesis}

\newpage
\pagestyle{empty}
\strut
\newpage

\ifodd\thepage 
   \strut
   \newpage
\fi


\end{document}